\documentclass[12pt,a4paper,twoside]{article}

\usepackage{graphicx,hyperref,bm,url,microtype,color}
\usepackage{amsfonts,amssymb,amsthm,amsmath,MnSymbol}
\usepackage{pgfplots}
\usepackage[english]{babel}
\usepackage{natbib}
\usepackage{comment} 

\usepackage{float} 

\usepackage[a4paper,text={16cm,23.5cm},centering]{geometry}
\usepackage[compact,small]{titlesec}
\usepackage[utf8]{inputenc}
\usepackage{enumerate} 

\def \d{\text{\rm d}}

\def \E{\text{\rm E}}
\def \P{\text{\rm P}}

\def \st{=_{\rm st}}
\def \Lo{\le_{\rm L}}
\def \B{\mathbb{B}}
\newcommand{\sgn}{\mbox{\rm sgn}}
\newcommand{\R}{\ensuremath{\mathbb{R}}}

\newcommand{\I}{\ensuremath{\mathcal{I}}}
\newcommand{\N}{\mathbb{N}}
\newcommand{\p}{\partial}

\newcommand{\ext}{\ensuremath{\text{\rm Ext}}}

\DeclareMathOperator*{\esssup}{ess\,sup}

\DeclareMathOperator*{\diam}{diam}

\definecolor{violet}{rgb}{0.7,0,0.6}
\theoremstyle{definition}

\newtheorem{definition}{Definition}
\newtheorem{assumption}{Assumption}

\theoremstyle{theorem}

\newtheorem{theorem}{Theorem}
\newtheorem{proposition}{Proposition}
\newtheorem{corollary}{Corollary}
\newtheorem{lemma}{Lemma}

\title{\textbf{\Large Extremal points of Lorenz curves and \\ applications to inequality analysis }}
\author{Amparo Ba\'illo, Javier C\'{a}rcamo\footnote{Corresponding author}\, and Carlos Mora-Corral \\
{\normalsize Departamento de Matem\'{a}ticas, Universidad Aut\'{o}noma de Madrid, 28049 Madrid (SPAIN)}}
\date{\today}

\begin{document}

\maketitle

\begin{abstract}
We find the set of extremal points of Lorenz curves with fixed Gini index and compute the maximal $L^1$-distance between Lorenz curves with given values of their Gini coefficients. As an application we introduce a bidimensional index that simultaneously measures relative inequality and dissimilarity between two populations. This proposal employs the Gini indices of the variables and an $L^1$-distance between their Lorenz curves. The index takes values in a right-angled triangle, two of whose sides characterize perfect relative inequality---expressed by the Lorenz ordering between the underlying distributions. Further, the hypotenuse represents maximal distance between the two distributions. As a consequence, we construct a chart to, graphically, either see the evolution of (relative) inequality and distance between two income distributions over time or to compare the distribution of income of a specific population between a fixed time point and a range of years. We prove the mathematical results behind the above claims and provide a full description of the asymptotic properties of the plug-in estimator of this index. Finally, we apply the proposed bidimensional index to several real EU-SILC income datasets to illustrate its performance in practice.
\end{abstract}



\section{Introduction and motivation} \label{Section.Introduction}

Inequality is one of the main global issues in nowadays world; see \cite{Atkinson-2015}. It is commonly accepted that social inequality has increased across the globe over the last few decades; see for example \cite{Greselin-2014} for an analysis of income inequality in the United States of America and \cite{Bosmans-2014} who analyze inequality in many countries. The media often claim that the social gap has widened considerably; day by day the richest are getting richer and the poor poorer. This kind of assertions are commonly based on striking facts such as ``the world's richest 1\%, those with more than \$1 million, own 45\% of the world's wealth\footnote{See \url{https://inequality.org/facts/global-inequality/}}''. Obviously, social inequality has wide-ranging adverse impacts on both the society and economy; see \cite{Stiglitz-2012}, \cite{Bourguignon}, \cite{Jarman-2016}, and the references therein. Therefore, in the econometric literature there is a major interest to develop and study indices quantifying inequality with accuracy, as well as summaries of the dispersion/heterogeneity of income or wealth distribution in a given population. The analysis of suitable empirical counterparts to make statistical inferences related to inequality has also played a central role within this field of research.

The motivation of this work is to compare and quantify the inequality between two populations, bringing to light the intrinsic differences between the underlying distributions. The usual econometric tools for quantifying inequality are based on the comparison of one-dimensional indices that might coincide for very different distributions; see \cite{Yitzhaki-Schechtman-2012} and \citet[Appendix B]{Fontanari-Cirillo-Oosterlee-2018} for concrete examples regarding the Gini index. As a matter of fact, the first goal of this work is to specify how different two distributions can be (with respect to a suitably chosen distance) if we know their Gini indices. To this end, we have explicitly computed the set of extreme points of Lorenz curves with a fixed value of their Gini coefficient as well as the maximum $L^1$-distance between the Lorenz curves of the distributions when their Gini indices are given. We also want to identify \textit{extremal distributions}, that is, those pairs of distributions for which this maximal distance is attained. We believe that these issues are relevant by themselves and they turn out to be interesting and deep mathematical problems.

A second, though primary, goal of this work is to propose a new inequality index that amends (to the extent possible) the deficiencies of the current econometric approaches to compare inequality between two populations. Specifically, we aim at introducing an index (defined for pairs of distributions) that combines the main elements in social welfare evaluations: the Lorenz curve, the Gini index, and the Lorenz ordering. A detailed inspection of the minimum desired requirements such index should satisfy reveals that a statistic with appropriate properties cannot be one-dimensional, as are the most frequently used inequality measures; see Section \ref{Subsection.Ideal-properties}. For this reason we propose a two-dimensional index that measures at the same time the difference between the Gini indices and an $L^1$-distance between the Lorenz curves of the two populations. We can further normalize this index so that it takes values in a triangle in $\R^2$ with vertexes $(0,0)$, $(1,1)$ and $(-1,1)$. We prove that the legs of this triangle characterize perfect inequality---expressed as Lorenz ordering between the distributions---while the hypothenuse (upper face), maximum dissimilarity. This new index allows us to evaluate all at once the difference between the income distributions of the two populations and their relative inequality. We can use this index to visualize the evolution (over time) of relative inequality and distance.

This paper is structured as follows: Section \ref{Section.Inequality-comparisons} reviews the usual econometrical approaches to quantify and compare inequality. We also list the ideal properties a relative inequality index should satisfy. In Section \ref{Section.Lorenz-Gini} we revise the basic concepts related to inequality used throughout the paper: the Lorenz curve, the Gini index and the Lorenz ordering. In Section \ref{Section.Extremal} we consider the class of Lorenz curves with a given value of their Gini coefficients and show that this set is compact (and convex) in the space $L^1$. We further compute the extreme points of this set. Section \ref{Section.Maximum} is devoted to compute maximum distances between two Lorenz curves and to find extremal pairs of distributions with fixed values of their Gini indices. In Section \ref{Section.Index} we introduce the aforementioned inequality index and enumerate its main properties. Two simple normalizations of the index are also proposed to improve data visualization. In Section \ref{Section-Asymptotics} we follow a ``plug-in approach" to estimate the proposed indices. 
We show the strong consistency of the estimators and compute the limit distributions of their normalized versions.
Necessary and sufficient conditions for some of the statistics to be asymptotically normal are provided.
We also include the conclusions of a simulation study to evaluate the behaviour of the asymptotic results in finite samples. In Section \ref{Section.RealData} we compute the bidimensional index for various income datasets from EU-SILC (European Union Statistics on Income and Living Conditions).
Additional examples and material regarding the analysis of real data sets and the simulation study are included in the Supplementary Material file. Finally, the proofs of the main results are collected in Section \ref{Section.Appendix}, a technical appendix.

\section{Inequality: comparisons between distributions}\label{Section.Inequality-comparisons}

In this section, first we give a general review of the current techniques to quantify inequality as well as to compare income distribution across different populations. The most frequent approach is to summarize the distribution into a one-dimensional quantity. However, we conclude the section pointing out that there is no unidimensional index satisfying simultaneously all the reasonable properties a good relative inequality measure between two distributions should fulfill.

\subsection{Tools to compare inequality between distributions}

Comparing inequality in two populations is far from being new. The usual econometric tools to carry out comparisons among income distributions within countries or to analyze the evolution of inequality in different moments of time can be essentially divided into two groups.

\textsl{1. Inequality measures.} In the literature there is a great amount of statistics to assess economic inequality. We can mention the well-known \textit{Theil}, \textit{Hoover}, \textit{Amato}, \textit{Atkinson} and \textit{generalized entropy} indices. These are only a few examples among many others, and even new measures are introduced from time to time; see \cite{Prendergast-Staudte-2018} for a recent proposal. The interested reader might consult the book by \cite{Cowell} or \cite{Eliazar-Sokolov-2012} for a panoramic overview on equality indices. These measurements summarize and quantify---usually in a (normalized) single real number---the statistical dispersion and heterogeneity of income distribution in a population. There is no doubt that the most commonly used measure in this context is the \textit{Gini index} (see Section \ref{Section.Gini.index}), which is at the heart of social welfare evaluations. In practice, it is quite frequent to analyze the situation of two (or more) countries in terms of evenness by comparing their respective Gini indices. Most rankings where countries are ordered by income equality and poverty mappings (i.e., maps of income disparity) are usually obtained in this way. In this first group, the comparison of income distributions relies on the corresponding analysis of suitable income inequality metrics.

\vspace{1 mm}

\textsl{2. Stochastic comparisons.} An essentially different way to compare (income) distributions is to establish a stochastic ordering between them; see \cite{Sriboonchita-2009}. Stochastic orders, also known as \textit{stochastic dominance rules} in the economic literature, are nothing but partial order relations in the set of probability measures; see \cite{Shaked-Shanthikumar-2006}. Therefore, they allow comparing and ranking distributions according to some specific criterion. If such a criterion is evenness, the \textit{Lorenz order} is the  most commonly accepted rule, primarily in economic sciences; see \cite{Arnold-Sarabia-2018}. Two distributions are ordered with respect to this relation if one of their Lorenz curves (see the precise definition in Section \ref{Section-Lorenz.curve}) is completely above the other one. In terms of inequality, this roughly speaking means that the wealth is distributed in a fairer way in one of the two populations.
Hence, this second group of techniques consists in performing global comparisons of the distributions by means of stochastic dominance rules that take into account evenness.

\vspace{1 mm}

\textsl{Pros and cons.} Each of the previous two approaches has advantages and limitations. Inequality metrics provide useful summaries of income distributions that are simple and easily interpretable. These measurements can be used to make comparisons (related to inequality) among distributions by simply arranging the selected index. However, it is clear that a single real number cannot represent faithfully the distribution of income of a population. The same inequality index might correspond to many very different distributions. On the other hand, if two distributions are stochastically ordered---with respect to a certain relation that takes into account inequality---, we can derive many important consequences regarding the underlying distributions. Typically, stochastic dominance implies an ordering among many inequality measures simultaneously. For instance, if the Lorenz ordering holds, an inequality between expectations of convex functions of the involved variables is satisfied; see \cite{Arnold-Sarabia-2018}. Additionally, from the empirical point of view, there are various hypothesis tests to check whether it is reasonable to assume that two variables are ordered; see \cite{Anderson-1996}, \cite{Barrett-Donald-2003}, \cite{Zheng-2002}, \cite{Berrendero-Carcamo-2011}, \cite{Barrett-Donald-Bhattacharya-2014}, \cite{Sun-Beare-2021}, among others. Nevertheless, dominance rules are partial orders and, hence, not every pair of distributions can be arranged; see \cite{Davies-Hoy-1995}. Further, in general, the statement that two distributions are ordered cannot be proved statistically, as one would desire. This happens because a test with null hypothesis ``two variables are  \textit{not} ordered" and alternative ``the variables are ordered" is usually ill-posed: given a pair of ordered distributions, we can normally find pairs of non-ordered distributions arbitrarily close to the initial ones and hence the null and alternative hypotheses are indistinguishable; see \cite{Ermakov-2017}. 

\subsection{Ideal properties of a univariate relative inequality index}\label{Subsection.Ideal-properties}

Let us assume that we want to construct a unidimensional index, say $I$, to measure relative inequality between two populations $X_1$ and $X_2$. We might want this index to combine the most commonly employed tools in social welfare evaluations: the Gini index and the Lorenz ordering. One can easily describe the most desirable properties a reasonable index should fulfill.

\begin{enumerate}[(P$_1$)]
  \item \textsl{Normalization:} The index has to take values on a reference interval. As we intend to measure relative inequality (of one variable with respect to another one), this interval has to be symmetric. Let us assume that the normalization is such that $I(X_1,X_2)\in [-1,1]$. By convention, positive and negative values of $I(X_1,X_2)$ can indicate that $X_1$ is fairer (in some precise sense) than $X_2$, or the other way around. For instance, we might ask that $I(X_1,X_2)> 0$ if and only if the Gini index of $X_1$ (or any other reference inequality index) is less than the corresponding index of $X_2$.

  \item \textsl{Symmetry:} The index has to verify that $I(X_1,X_2)=-I(X_2,X_1)$.

  \item \textsl{Extreme values:} The endpoints of the reference interval should reflect that one distribution is uniformly more equitable than the other one. Ideally, this fact can be translated into the Lorenz ordering between the variables:
  \begin{enumerate}[(i)]
    \item $I(X_1,X_2)=+ 1$ if and only if $X_1$ is smaller than $X_2$ in the Lorenz order.
    \item $I(X_1,X_2)=-1$ if and only if $X_2$ is smaller than $X_1$ in the Lorenz order.
  \end{enumerate}
  \item \textsl{Value at the origin:} If $I(X_1,X_2)=0$, then $X_1$ and $X_2$ have to satisfy some kind of equality. A weak version of this property could be that $I(X_1,X_2)=0$ if and only if their Gini indices (or some other related statistic) coincide. However, preferably the equality should be in distribution. For instance, we could ask that `$I(X_1,X_2)=0$' be equivalent to `$X_1\st c X_2$', where $c>0$ is a constant and `$\st$' stands for stochastic equality. In other words, `$I(X_1,X_2)=0$' would mean that $X_1$ and $X_2$ distribute the wealth exactly in the same manner up to a size factor.
  \item \textsl{Continuity:} Let $\{X_{1, n_1}\}$ and $\{X_{2, n_2}\}$ be two sequences of random variables such that $X_{1,n_1}\rightsquigarrow X_1$ and $X_{2,n_2}\rightsquigarrow X_2$ (as $n_1,n_2\to\infty$), where `$\rightsquigarrow$' stands for some suitable mode of convergence of random variables. In this setting, continuity of the index amounts to $I(X_{1, n_1},X_{2, n_2})\to I(X_1,X_2)$ (as $n_1,n_2\to\infty$). This property is essential to estimate well the index in practice with random samples of the populations.
\end{enumerate}

Unfortunately, some of these properties are usually incompatible for a one-dimensional index. For instance, generally (P$_3$) and (P$_5$) cannot hold at the same time. The reason relies on the fact that there are pairs of ordered distributions---according to the Lorenz ordering---arbitrarily close. We might have that $X_{1,n_1} \rightsquigarrow X_1$ and $X_1$ is smaller than $X_{1,n_1}$ in the Lorenz order if $n_1$ even and the other way around if $n_1$ odd, so that under (P$_3$), $I(X_1,X_{1,n_1})=(-1)^{n_1}$. Therefore, if we want to construct an index satisfying similar properties to those enumerated before, we need to quantify the difference between two income distributions with more than one number.

\section{Lorenz curve, Gini index and Lorenz ordering}\label{Section.Lorenz-Gini}

Since their introduction by \cite{Lorenz-1905} and \cite{Gini-1914}, the Lorenz curve and the Gini index have been key tools in the analysis of economic inequality; see \cite{Kleiber-Kotz-2003}. An irrefutable proof of their historical transcendence is their continued use for more than a century. In this section we recall the precise definitions of these crucial concepts and set the notation used throughout the rest of the paper.

\subsection{The Lorenz curve}\label{Section-Lorenz.curve}

The Lorenz curve provides a graphical representation of the distribution of income or wealth in a population. The following conditions on the involved random variables will be assumed in the sequel: let $X$ be a positive random variable with finite mean $\mu>0$ and cumulative distribution function $F(x)=\P(X\le x)$, for $x\ge 0$.
Formally, the \textit{Lorenz curve} of the variable $X$ (or of the distribution $F$) is
\begin{equation}
\label{Lorenz-curve}
\ell(t)=\frac{1}{\mu}\int_0^t F^{-1}(x)\, \d x,\quad 0\le t \le 1,
\end{equation}
where
\begin{equation}\label{eq:F-1}
F^{-1}(x)=\inf\{ y \ge 0 : F(y)\ge x \}
\end{equation}
($0<x<1$) is the \textit{quantile function} of $X$, that is, the generalized inverse of $F$.
Hence, if $X$ measures income in a population, for each value $t\in [0,1]$, the function in \eqref{Lorenz-curve} gives us the (normalized) total income accumulated by the proportion $t$ of the poorest in that population.
Note that $F^{-1}$ is non-decreasing, $\mu=\int_0^1 F^{-1}(x)\, \d x$ and $\ell^\prime(t)=F^{-1}(t)/\mu$ a.e.\ $t\in(0,1)$. Therefore, $\ell$ is a convex and non-decreasing function such that $\ell(0)=0$ and $\ell(1)=1$. In particular, $\ell$ is continuous except perhaps at the point $1$ and has positive second derivative $\ell^{\prime\prime}$ a.e.
Moreover, as the quantile function characterizes the probability distribution, $\ell$ determines the distribution of the underlying variable up to a (positive) scale transformation. Explicit analytic expressions for the Lorenz curves of the usual parametric distributions can be found in \citet[Section 2.1.2]{Kleiber-Kotz-2003}.

By convexity, for every the Lorenz curve $\ell$ it holds that
\begin{equation}\label{inequalities-Lorenz}
\ell_{\rm pi}\le \ell \le \ell_{\rm pe},
\end{equation}
where
\begin{equation}\label{pe-pi-curve}
\ell_{\rm pe}(t)=t \quad (0\le t\le 1) \quad \text{and}\quad
\ell_{\rm pi}(t)=
\begin{cases}
0, & \text{if }0\le t < 1,\\
1, & \text{if }t=1.
\end{cases}
\end{equation}
Figure \ref{fi:Lorenz} shows a graphical representation of the inequalities in \eqref{inequalities-Lorenz}. The function $\ell_{\rm pe}$ is called the \textit{perfect equality curve} as it corresponds to the Lorenz curve of a Dirac delta measure, i.e., the probability measure corresponding to a population in which all individuals have equal (and positive) incomes. Additionally, $\ell_{\rm pi}$ is the \textit{perfect inequality curve} because it can be viewed as the limit (when the total number of individuals tends to infinity) of Lorenz curves in finite populations where only one person accumulates all the wealth.

\begin{figure}[h]
\centering{\includegraphics[width=9cm]{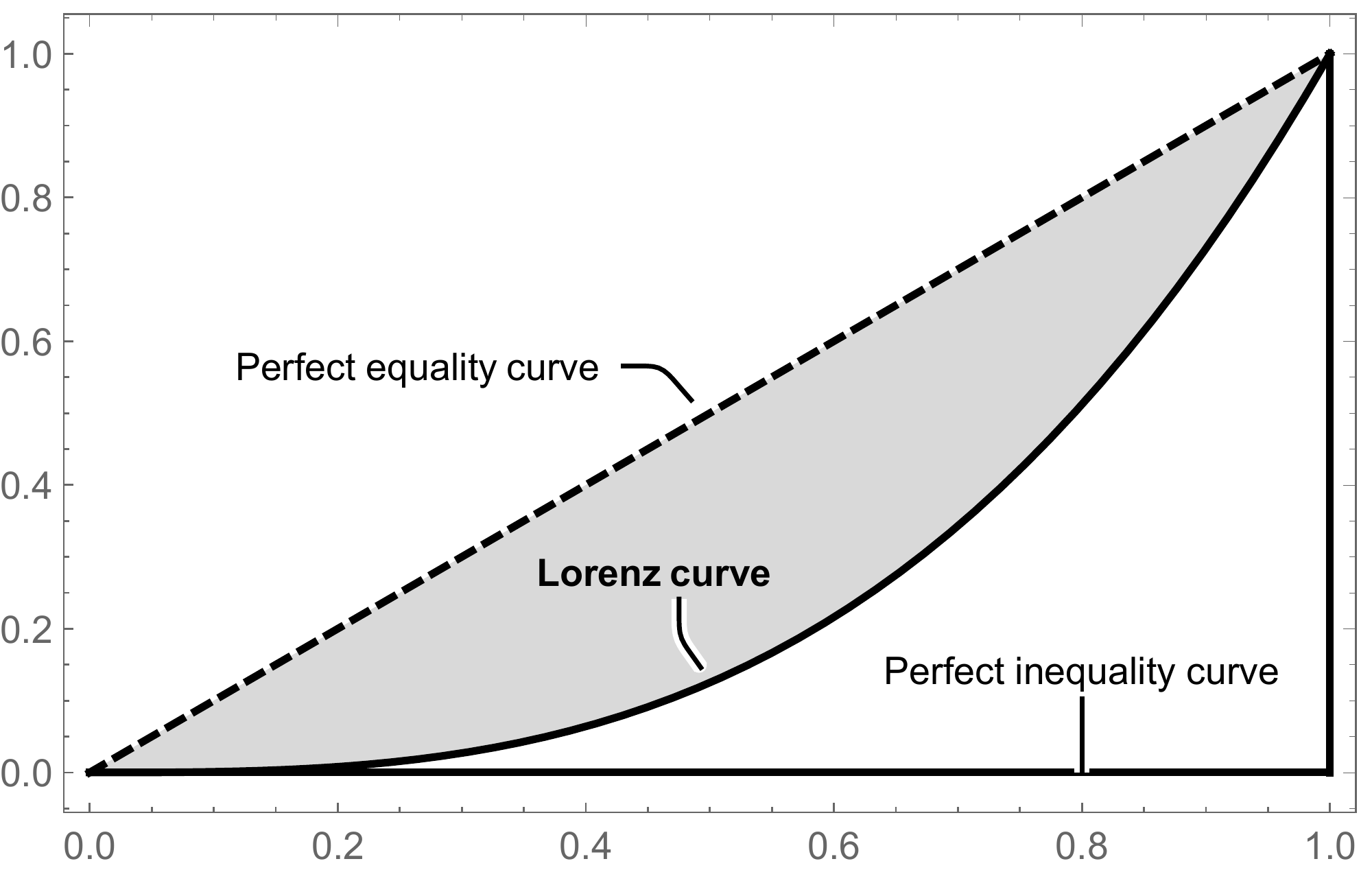}}
\caption{A Lorenz curve together with the perfect equality and inequality curves.}
\label{fi:Lorenz}
\end{figure}

\subsection{The Gini index}\label{Section.Gini.index}

Different characteristics, functionals and values of the Lorenz curve are employed to construct inequality indices; see \cite{Arnold-Sarabia-2018}. The Gini, Pietra, Amato, {20:20 ratio} and {Palma ratio} indices are some examples of inequality measures derived from the Lorenz curve. As we have mentioned before, a simple and effective comparison frequently used by the media can be made by analyzing the evolution of the proportion of income accumulated by the top (or bottom) 1\% of the population, which is nothing but the analysis of one single value of the Lorenz curve. 

The most popular inequality measure derived from the Lorenz curve is the \textit{Gini index}. This index has almost an uncountable number of interesting interpretations and representations; see \citet[Chapter 2]{Yitzhaki-Schechtman-2012}. 
One possible way to define it is the following:
\begin{equation*}
G(X)= 2 \int_0^1 (t-\ell (t))\, \d t.
\end{equation*}

In the sequel we denote by $L^1=L^1([0,1])\equiv$ the Banach space of equivalence
classes of measurable functions $f:[0,1]\to \R$ endowed with the usual $L^1$-norm,
\begin{equation*}
\Vert f\Vert =\int_0^1 |f(t)|\, \d t,\quad \text{for }f\in L^1.
\end{equation*}
From \eqref{inequalities-Lorenz}--\eqref{pe-pi-curve}, we have that $G(X)=2\| \ell_{\rm pe}-\ell  \| =  1-2 \|  \ell  \|$. Geometrically, the Gini index corresponds to twice the shaded area in Figure \ref{fi:Lorenz}. In particular,
\begin{equation}\label{Gini-index-2}
G(X)=\frac{\Vert \ell-\ell_{\rm pe} \Vert}{\Vert \ell_{\rm pe} -\ell_{\rm pi}  \Vert}.
\end{equation}
The denominator in \eqref{Gini-index-2} equals $1/2$ (the maximum $L^1$-distance between Lorenz curves) and acts as a normalizing constant so that $0\le G(X)\le 1$.

The Gini index has many desirable properties: it is scale-free (because the Lorenz curve is itself invariant under positive scaling); it can be computed whenever the considered random variable is integrable (finite second moment is not necessary); it is normalized so that it takes values between 0 (perfect equality) and 1 (perfect inequality); it has a simple and effective interpretation (small values of this index amount to fair income distributions, whereas high values indicate unequal distributions); it is a quasi-convex measure (see \cite{Blackorby-Donaldson-1980}), i.e., for all variables $X_1$, $X_2$ and $\lambda\in[0,1]$, $G(\lambda X_1+ (1-\lambda) X_2)\le \max\{ G(X_1),G(X_2)\}$.

\subsection{The Lorenz ordering}

Another important instrument to compare distributions according to inequality is the so-called Lorenz ordering. Let $X_1$ and $X_2$ be two variables with Lorenz curves $\ell_1$ and $\ell_2$, respectively. It is said that $X_1$ is less than or equal to $X_2$ in the \textit{Lorenz order}, written $X_1\Lo X_2$, if $\ell_1(t)\ge \ell_2(t)$, for all $t\in [0,1]$. In this case, we have that $\ell_{\rm pe}\ge \ell_1 \ge \ell_2$, where $\ell_{\rm pe}$ is the perfect equality curve defined in \eqref{pe-pi-curve}. In other words, income is distributed in a more equitable manner in $X_1$ than in $X_2$.

For many families of parametric distributions usually considered in applications, two members of the family differing in the dispersion parameter are usually ordered in accordance with this relation. For instance, Pareto, normal, lognormal, Gamma, Weibull distributions (among others) satisfy this property; see \cite{Kleiber-Kotz-2003}. 


\section{Extreme points of Lorenz curves with given Gini index}\label{Section.Extremal}

Let us consider
\begin{equation}\label{eq:L}
\mathcal{L}=\{  \ell:[0,1]\to[0,1] : \ell \text{ convex and } \ell(0)=0,\,  \ell(1)=1  \},
\end{equation}
the closure (with respect to the pointwise convergence) of the set of Lorenz curves of positive and integrable random variables with strictly positive expectation. For example, the function $\ell_{\rm pi}$ defined in (\ref{pe-pi-curve}) (see also Figure~\ref{fi:Lorenz}), which is \textit{not} a proper Lorenz curve, belongs to $\mathcal{L}$. For simplicity, we will refer to $\mathcal{L}$ as the class of Lorenz curves.

The Gini index of $\ell\in \mathcal{L}$ will be also denoted by $G(\ell)$. In other words,
\begin{equation}\label{Gini-Lorenz}
G(\ell)  =   1-2\Vert \ell\Vert.
\end{equation}
For $a\in[0,1],$ we define
\begin{equation}\label{eq:La}
\mathcal{L}_a=\{  \ell\in\mathcal{L} : G(\ell)= a  \},
\end{equation}
the collection of Lorenz curves with Gini index $a$. Note that $\mathcal{L}_a=\{  \ell\in\mathcal{L} : \| \ell \|= (1-a)/2  \}$.

The following proposition shows the compactness of $\mathcal{L}_a$ in the space $L^1$.
\begin{proposition}\label{pr:LaCompact}
For each $a \in [0,1]$, $\mathcal{L}_a$ is a compact and convex set in $L^1$.
\end{proposition}

To achieve a deeper understanding of the class $\mathcal{L}_a$ in \eqref{eq:La} we need some basic concepts about convex sets. The notion of {extreme point} plays a prominent role in convex analysis; see, for example, \citet[Section 8]{Simon-2011}. Roughly speaking, an extreme of a convex set is a point that cannot be expressed as a proper convex combination of other points within the set. Formally, given a convex set $C$, $x \in C$ is an \textit{extreme point} of $C$ if $x = t x_1 + (1-t) x_2$, for some $t \in (0,1)$ and $x_1, x_2 \in C$, implies that $x_1 = x_2$. In the following we denote by $\ext(C)$ the set of extreme points of $C$. The relevance of extreme points is comprehended through the Krein--Milman theorem (see, e.g., \citet[Theorem 8.14]{Simon-2011}), which is a central result in convex analysis. This theorem affirms that a convex and compact set in a locally convex space is the closed convex hull of its extreme points. Therefore, we can retrieve the entire convex set by knowing only the (usually much smaller) set of extreme points. Further, Bauer's maximum principle (see, e.g., \citet[Proposition 16.6]{Phelps-2001} or \citet[7.69]{Aliprantis-2006}) states that a convex, upper-semicontinuous functional on a non-empty, compact and convex set of a locally convex space attains its maximum at an extreme point. In consequence, the knowledge of extreme points is fundamental in mathematical optimization. The practical application of these powerful results goes through the explicit computation of the extreme points of the convex set under study, which is usually a difficult task in infinite-dimensional spaces.


The next theorem determines the set of extreme points of $\mathcal{L}_a$ in \eqref{eq:La}.
We believe that this result might be of independent interest and it is indeed necessary for further developments of this work.

\begin{theorem}\label{th:extLa}
For $a \in [0, 1]$, we have that
\[
 \ext(\mathcal{L}_a) = \left\{ \ell^a_{x_1} : x_1 \in [0, a] \right\} \cup \left\{ m^a_{x_2} : x_2 \in (a, 1) \right\} \cup \left\{ n^a_{x_1, x_2} : \, x_1 \in (0,a) , \, x_2 \in (a,1) \right\},
\]
where $\ell^a_{x_1}$, $m^a_{x_2}$ and $n^a_{x_1, x_2}$ are the piecewise affine functions of $\mathcal{L}_a$ such that
\begin{equation}\label{eq:extremePW}
 \ell^a_{x_1} : \begin{cases}
  0 \mapsto 0 \\
 x_1 \mapsto 0 \\
 1^- \mapsto \frac{1-a}{1-x_1} ,
 \end{cases} \qquad
 m^a_{x_2} : \begin{cases}
 0 \mapsto 0 \\
 x_2 \mapsto x_2 -a \\
 1 \mapsto 1,
 \end{cases} \qquad
 n^a_{x_1,x_2} : \begin{cases}
 0 \mapsto 0 \\
 x_1 \mapsto 0 \\
 x_2 \mapsto \frac{x_2 -a}{1-x_1} \\
 1 \mapsto 1
 \end{cases}
\end{equation}
(with the notation $\ell^a_{x_1}(1^-)=\lim_{t\uparrow 1}\ell^a_{x_1}(t)$ and the convention $\ell^1_{1}=\ell_{\rm pi}$ in \eqref{pe-pi-curve}). 
\end{theorem}

\begin{figure}[h]
\centering{\includegraphics[width=16cm]{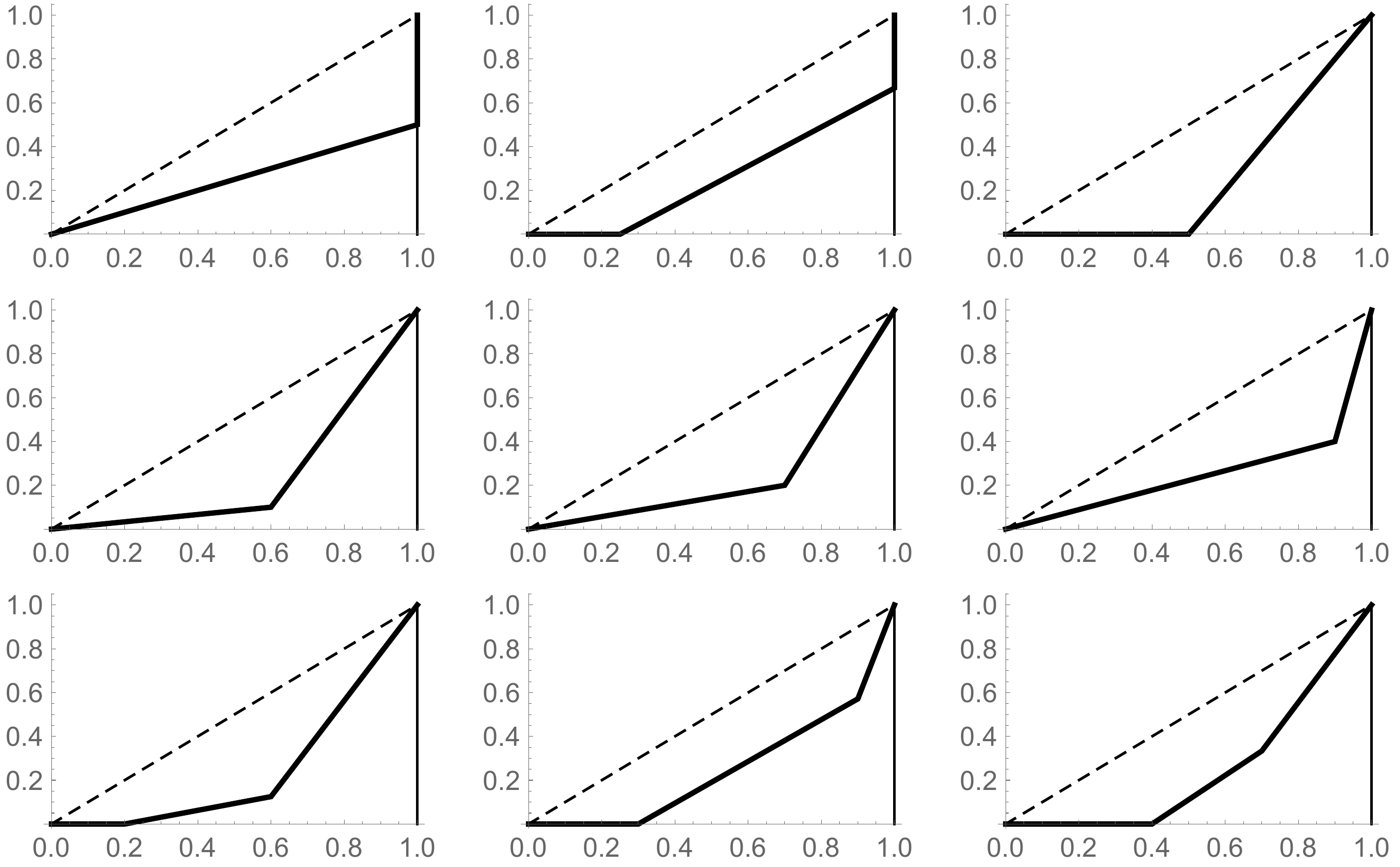}}
\caption{The functions $\ell^{a}_{x_1}$, with $a=0.5$ and $x_1=0$ (upper left panel), $x_1=0.25$ (upper middle panel) and $x_1=0.5$ (upper right panel). The functions $m^{a}_{x_2}$, with $a=0.5$ and $x_2=0.6$ (central left panel), $x_2=0.7$ (central middle panel) and $x_1=0.9$ (central right panel). The functions $n^{a}_{x_1,x_2}$, with $a=0.5$, $x_1=0.2$ and $x_2=0.6$ (lower left panel), $x_1=0.3$ and $x_2=0.9$ (lower middle panel) and $x_1=0.4$ and $x_2=0.7$ (lower right panel).}
\label{fi:extremal}
\end{figure}

Theorem \ref{th:extLa} summarizes the information of $\mathcal{L}_a$ (an infinite-dimensional collection) in the set of its extreme points, which has only dimension 2. Furthermore, as $\mathcal{L}_a$ is compact in $L^1$, it identifies all possible maximizers of convex and continuous functionals. 
To prove Theorem~\ref{th:extLa} we first show that twice differentiation determines an affine isomorphism between $\mathcal{L}_a$ and the set of non-negative measures on $(0,1)$ with some restrictions. Afterwards, we identify those combinations of delta measures that are extreme points. The last step of the proof of Theorem \ref{th:extLa} is related to the results of \cite{Winkler-1988} and \cite{Pinelis-2016}, where they analyze the set of extreme points of a subset of measures defined through some inequalities.

In Figure \ref{fi:extremal} we have depicted various extreme points of $\mathcal{L}_a$, with $a=0.5$. The probabilistic and economic meaning of some of these Lorenz curves is described in the next section.

\section{Maximum distance between Lorenz curves with fixed Gini}\label{Section.Maximum}

As stated in the introduction, for two distributions with fixed Gini indices, one aim of this work is to quantify how ``far'' they can be from one another. Specifically, we are interested in computing the value $d(\mathcal{L}_a,\mathcal{L}_b)$ ($a,b\in[0,1]$), for a suitable metric $d$ on $\mathcal{L}\times \mathcal{L}$, where $\mathcal{L}$ is defined in \eqref{eq:L}, and $\mathcal{L}_a$ and $\mathcal{L}_b$ as in \eqref{eq:La}. Theorem \ref{th:extLa} is extremely useful for this purpose. If $d$ is defined through a norm, $d$ is a convex and continuous functional on (the convex set) $\mathcal{L}_a \times \mathcal{L}_b$. Therefore, as long as $\mathcal{L}_a$ and $\mathcal{L}_b $ are compact, by Bauer's maximum principle, the supremum of $d$ on $ \mathcal{L}_a \times \mathcal{L}_b$ is attained in $\ext (\mathcal{L}_a \times \mathcal{L}_b)   =  \ext(\mathcal{L}_a) \times \ext(\mathcal{L}_b) $. Thus, thanks to Theorem \ref{th:extLa}, we reduce the calculation of $d(\mathcal{L}_a,\mathcal{L}_b)$ to a finite-dimensional problem.

The exact computation of $d(\mathcal{L}_a,\mathcal{L}_b)$ will eventually depend on the particular choice of the metric $d$. In Section \ref{Subsection-Metric}, we introduce a distance between Lorenz curves which is natural in this context. 
The computation of this maximal distance, carried out in Section \ref{Subsection.Maximal.distance}, as well as the characterization of the distributions where the maximum is attained, is crucial to define the bidimensional inequality index proposed in Section \ref{Section.Index}.

\subsection{The choice of the metric: the Lorenz distance}\label{Subsection-Metric}

Depending on the interests of the researcher and the problem at hand, there are many probability metrics that can be used to quantify the distance between two random variables; see the compilation volume on probability distances and their applications by \cite{Rachev-2013}. However, we note that the Gini coefficient itself is defined in terms of a (normalized) $L^1$-distance between Lorenz curves; see formula \eqref{Gini-index-2}. Therefore, a sensible and convenient choice
to measure dissimilarities between distributions is also a normalized $L^1$-norm of the difference between the corresponding Lorenz curves. The $L^1$ distance between Lorenz curves has also been used in \cite{Zheng2018} related to \textit{almost stochastic dominance} of \cite{Leshno-Levy-2002}.


Explicitly, given $X_1$ and $X_2$ two random variables with Lorenz curves $\ell_1$ and $\ell_2$, respectively, we define the \textit{Lorenz distance} between the variables as
\begin{equation*}
d_{\rm L}(X_1,X_2)=\frac{\Vert \ell_1-\ell_2\Vert}{\Vert \ell_{\rm pe} -\ell_{\rm pi}  \Vert}= 2  \Vert \ell_1-\ell_2\Vert.
\end{equation*}
We observe that $0\le d_{\rm L}(X_1,X_2)\le 1$ and $d_{\rm L}$ is actually a pseudo-metric because $d(X_1,X_2)=0$ holds if and only if $X_1\st c X_2$, where $c>0$ is a constant. Further, $d_{\rm L}$ can only achieve the value $1$ when the variables have the perfect equality and inequality Lorenz curves in \eqref{pe-pi-curve}. Observe that with this definition the Gini index of a variable is nothing but the Lorenz distance between the variable and a positive constant.


\subsection{Maximal distance and extremal distributions with given Gini indices}\label{Subsection.Maximal.distance}

We endow the set $\mathcal{L}$ in \eqref{eq:L} with the Lorenz distance
\begin{equation}\label{Distance}
d_{\rm L}(\ell_1,\ell_2)= 2  \Vert \ell_1-\ell_2 \Vert = 2\int_0^1 |\ell_1-\ell_2|,\quad \ell_1,\ell_2\in \mathcal{L}.
\end{equation}
By (\ref{inequalities-Lorenz}), the diameter of $\mathcal{L}$ with respect to the metric $d_{\rm L}$ is
\begin{equation*}
\diam(\mathcal{L})=\sup \{d_{\rm L} (\ell_1, \ell_2) :  \ell_1,\, \ell_2 \in \mathcal{L}   \}=d_{\rm L} (\ell_{\rm pe}, \ell_{\rm pi})  =1.
\end{equation*}
We further observe that $\mathcal{L}_a$ in \eqref{eq:La} is the set of $\ell\in\mathcal{L}$ such that $d_{\rm L} ( \ell, \ell_{\rm pe} ) = a$. For any fixed $a,b\in[0,1]$, $\mathcal{L}_a$ and $\mathcal{L}_b$ are compact sets in $L^1$ (see Proposition \ref{pr:LaCompact}). Therefore, from Theorem \ref{th:extLa}, the maximum
\begin{equation}\label{Mab}
M(a,b)=\max \{ d_{\rm L} (\ell_1, \ell_2)  : \ell_1 \in\mathcal{L}_a \text { and }\ell_2 \in\mathcal{L}_b \}
\end{equation}
is attained at $\ext(\mathcal{L}_a) \times \ext(\mathcal{L}_b)$.



\begin{definition}\label{Definition.extremal}
We say that the pair $(\ell_1,\ell_2)\in \mathcal{L}\times \mathcal{L}$ is \textit{extremal} if
\begin{equation*}
  d_{\rm L}(\ell_1,\ell_2)=M(G(\ell_1),G(\ell_2)).
\end{equation*}
The pair of probability distributions associated to an extremal pair of Lorenz curves will also be called \textit{extremal distributions}.
\end{definition}

For notational convenience, we rename the functions $\ell_a^a$ and $\ell^a_0$ in \eqref{eq:extremePW} as $\ell_a^-$ and $\ell_a^+$, respectively. In other words, for $0\le a \le 1$, $\ell_{a}^-, \ell_{a}^+ \in \mathcal{L}_a$ are defined as
\begin{equation}\label{la+-la-}
 \ell_{a}^- (t) = \max \left\{ 0,  \frac{t-a}{1-a}  \right\}
\quad\text{and}\quad
\ell_{a}^+ (t) = \begin{cases}
(1-a) t, & \text{if } 0\le t < 1 , \\
1, & \text{if } t=1\\
\end{cases}
\end{equation}
(with the agreement that $\ell_{1}^-\equiv \ell_{\rm pi}$ defined in (\ref{pe-pi-curve})). These two functions will play an essential role in the rest of the section. In Figure \ref{fi:l+-} we display two of these functions.

\begin{figure}[h]
\centering{\includegraphics[width=15.5cm]{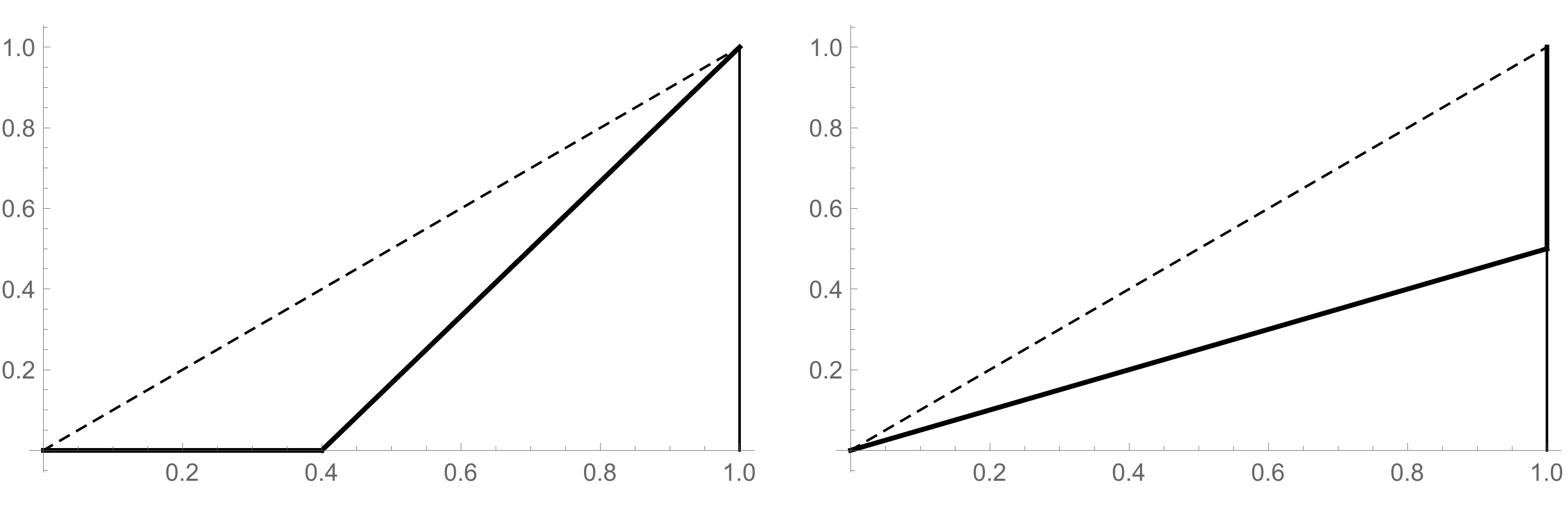}}
\caption{The functions $\ell_{0.4}^-$ (left panel) and $\ell_{0.5}^+$ (right panel).}
\label{fi:l+-}
\end{figure}

The following theorem, which is the main theoretical result of this section, provides an explicit expression for $M(a,b)$ and shows that this maximum distance is precisely attained at functions of the form \eqref{la+-la-}. It should be mentioned that the computation of $M(a,b)$ is a mathematical problem whose statement is very simple and seems to be deceptively easy. However, the proof of this result, which begins at Theorem \ref{th:extLa} and is collected in the technical appendix, reveals that this issue is indeed more delicate and complex than expected.

\begin{theorem}\label{Theorem.Mab}
For $0\le a,b\le 1$, let $M(a,b)$ be as in \eqref{Mab}. We have that
\begin{equation}\label{Mab.expression}
M(a,b)=\frac{(1-a)b^2 +(1-b)a^2}{a+b-ab}
\end{equation}
(the value $M(0,0)=0$ is taken by continuity). Moreover, $(\ell_{a}^-, \ell_{b}^+)$ and $(\ell_{a}^+,\ell_{b}^-)$ are pairs of extremal Lorenz curves within the set $\mathcal{L}_a\times \mathcal{L}_b$.
\end{theorem}

In Figure \ref{fi:Mab} we have plotted the function $M(a,b)$. From \eqref{Mab.expression}, it is easy to check that $(2-\sqrt{2},2-\sqrt{2})$ is a saddle point on the graph of the function $z=M(a,b)$.

\begin{figure}[h]
\centering{\includegraphics[width=16cm]{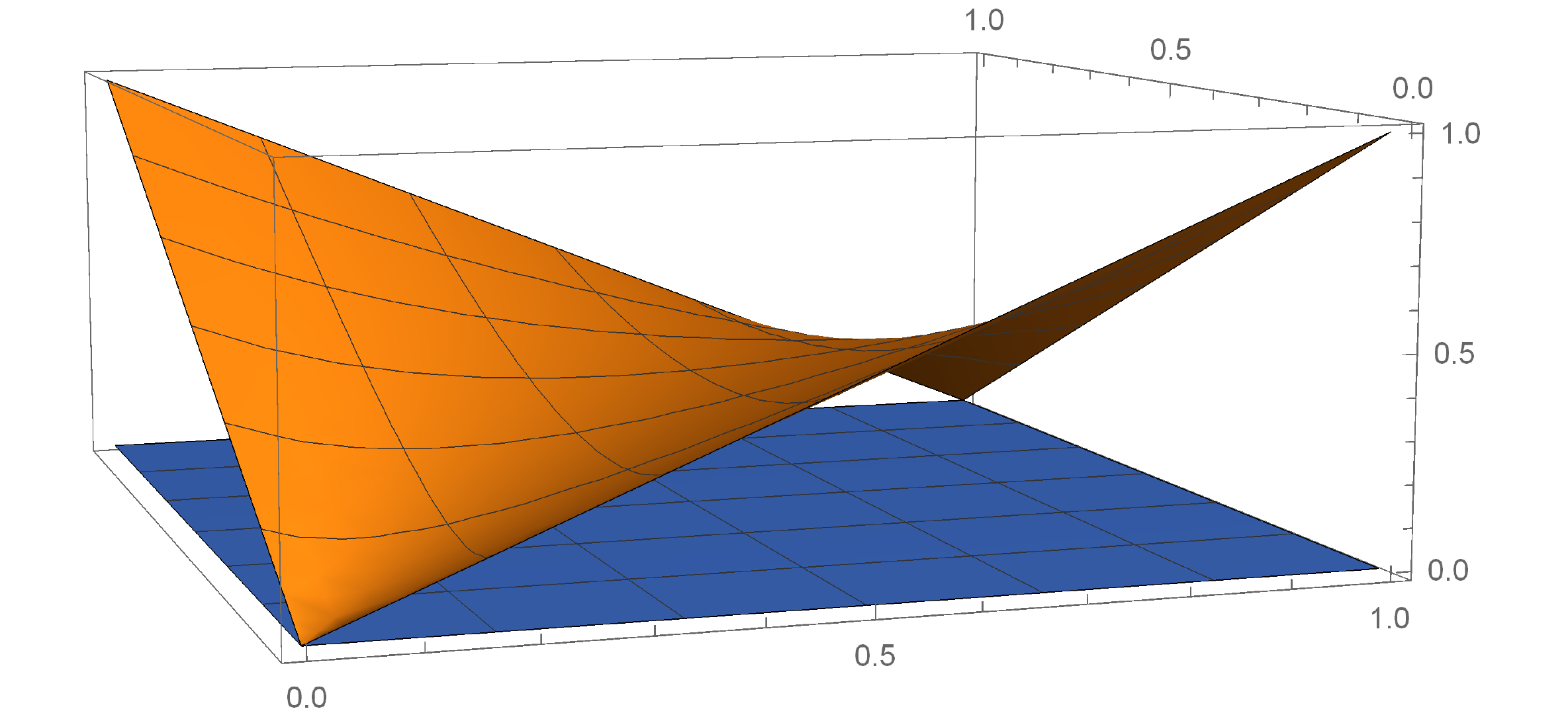}}
\caption{The maximal function $M(a,b)$, for $0\le a,b \le 1$.}
\label{fi:Mab}
\end{figure}

Theorem \ref{Theorem.Mab} asserts that the maximum distance in \eqref{Mab} is attained at the pairs $(\ell_{a}^-,\, \ell_{b}^+)$ and $(\ell_{a}^+,\, \ell_{b}^-)$. Hence, the associated probability distributions (unique up to positive scale transformations) are extremal. The function $\ell_{a}^-$ is the Lorenz curve of a population in which a proportion $a$ of the people have $0$ income and the rest, a proportion $1-a$, have equal and positive income. Also, $\ell_{a}^-$ is the Lorenz curve of a variable $X_a$ with Bernoulli distribution with parameter $1-a$, that is, $\P(X_a=0)=a$ and $\P(X_a=1)=1-a$. On the other hand, $\ell_{b}^+$ is not a proper Lorenz curve, but it can be expressed as the limit (as $n$ goes to infinity) of Lorenz curves of populations with $n$ individuals where $n-1$ of them fairly share a proportion $(1-b)$ of the wealth and there is only one ``lucky person" who accumulates the rest of the total wealth (the proportion $b$). From a probabilistic perspective, we have that $\ell_{b}^+=\lim_{p\to 0}\ell_{X(b,p)}$, where $\ell_{X(b,p)}$ is the Lorenz curve of $X(b, p)$, a random variable with distribution $\P(X(b, p)=1-b)=1-p$ and $\P(X(b, p)=1-b+b/p)=p$.

From Theorem \ref{Theorem.Mab} we can easily see the range of values of the distance $d_{\rm L}(\ell_1,\ell_2)$, when $(\ell_1,\ell_2)$ varies in $\mathcal{L}_a\times \mathcal{L}_b$.
\begin{corollary}\label{Coro.range}
For $0\le a,b\le 1$, we have that
\[
 \{ d_{\rm L} (\ell_1, \ell_2) : \, \ell_1 \in \mathcal{L}_a \text{ and } \ell_2  \in \mathcal{L}_b \} = \left[ \left| b-a \right| , M (a,b) \right],
\]
where $ M (a,b)$ is given in \eqref{Mab.expression}.
\end{corollary}

Theorem \ref{Theorem.Mab} also allows us to explicitly compute the maximum distance between Lorenz curves with a given difference of their Gini indices.

\begin{corollary}\label{Corollary.maximal}
For $-1\le c\le 1$, let us consider
\begin{equation*}
M^*(c)=\max\{M(a,b): a,b \in[0,1]\text{ and } b-a=c  \}.
\end{equation*}
We have that
\begin{equation}\label{ac}
M^*(c)=M(a_c,a_c+c)\quad \text{with}\quad a_c=(  4-c-\sqrt{8+c^2} )/2
\end{equation}
and
\begin{equation}\label{M*}
 M^*(c)=8-\frac{8+(c^2+8)^{3/2}}{c^2+4}.
\end{equation}
\end{corollary}

\begin{definition}\label{Definition.super-extremal}
We say that the pair $(\ell_1,\ell_2)\in \mathcal{L}^2$ is \textit{super-extremal} if
\begin{equation*}
  d_{\rm L}(\ell_1,\ell_2)=M^*(G(\ell_1)-G(\ell_2)).
\end{equation*}
The associated pairs of probability distributions will be also called \textit{super-extremal distributions}.
\end{definition}

Obviously, each super-extremal pair is extremal because it always holds that
\begin{equation}\label{M-and-M*}
   M(a,b)\le M^*(a-b), \quad \text{for }0\le a,b\le 1.
\end{equation}
However, from Theorem \ref{Theorem.Mab} and for any $0 \le c\le 1$, among all the pairs $(\ell_{a}^-,\ell_{a+c}^+)$ and $(\ell_{a+c}^-, \ell_{a}^+)$ (with $a\in[0,1-c]$) of extreme Lorenz curves with a value $c$ for the difference of their Gini indices there are only two super-extremal curves. Namely, the pairs corresponding to $a=a_c$ in \eqref{ac}.

Observe that $M^*(0)$ is the maximum possible distance between Lorenz curves with equal Gini indices. By Theorem \ref{Theorem.Mab} and Corollary \ref{Corollary.maximal}, we have that the maximum distance between two income distributions both with Gini indices equal to $a$ is
\begin{equation}\label{Maa}
M(a,a)=\frac{2 a(1-a)}{ 2-a},\quad \text{for }a\in[0,1],
\end{equation}
which attains its maximum at the point $a_0=2-\sqrt{2}\approx 0.59$. Therefore, the maximum (Lorenz) distance between distributions with the same Gini index is
\begin{equation*}
  M^*(0) =\max_{0\le a\le 1} M(a,a) = M(a_0,a_0)= 6 - 4 \sqrt{2}\approx 0.34.
\end{equation*}
Additionally, $M(a,a)$ is the $d_{\rm L}$-diameter of $\mathcal{L}_a$. A graphical representation of $M(a,a)$ and $M^*(0)$ is presented in Figure~\ref{fi:Maa}. The Lorenz curves $\ell^-_{a_0}$ and $\ell^+_{a_0}$ are hence super-extremal Lorenz curves with equal Gini indices; see  Figure~\ref{fi:maximales}.

\begin{figure}[h]
\centering{\includegraphics[width=10 cm]{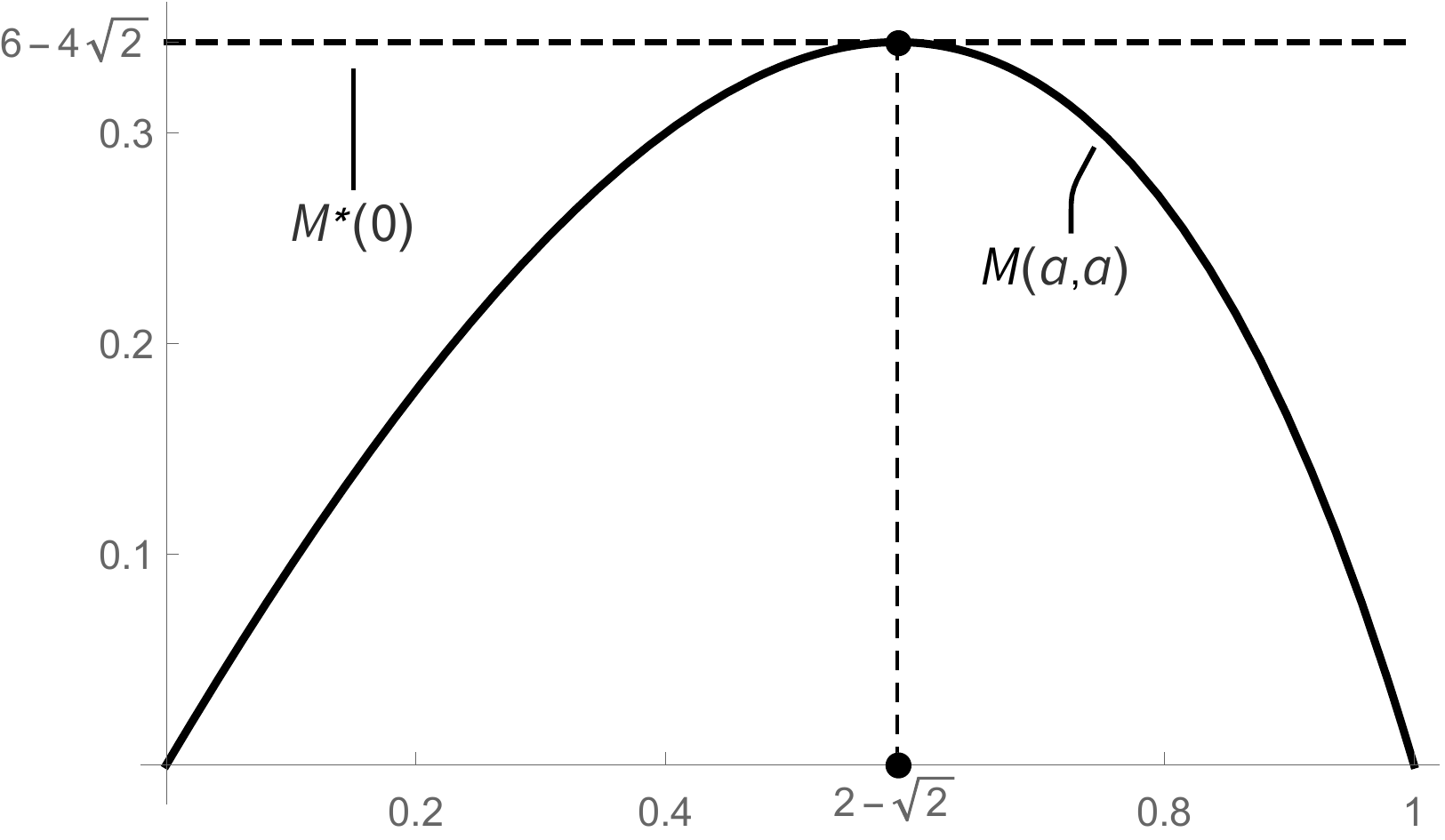}}
\caption{The funtion $M(a,a)$ $(0\le a\le 1)$ in \eqref{Maa} and its maximum value $M^*(0)=6-4\sqrt{2}$, attained at the point $2-\sqrt{2}$.}
\label{fi:Maa}
\end{figure}

\begin{figure}[h]
\centering{\includegraphics[width=10cm]{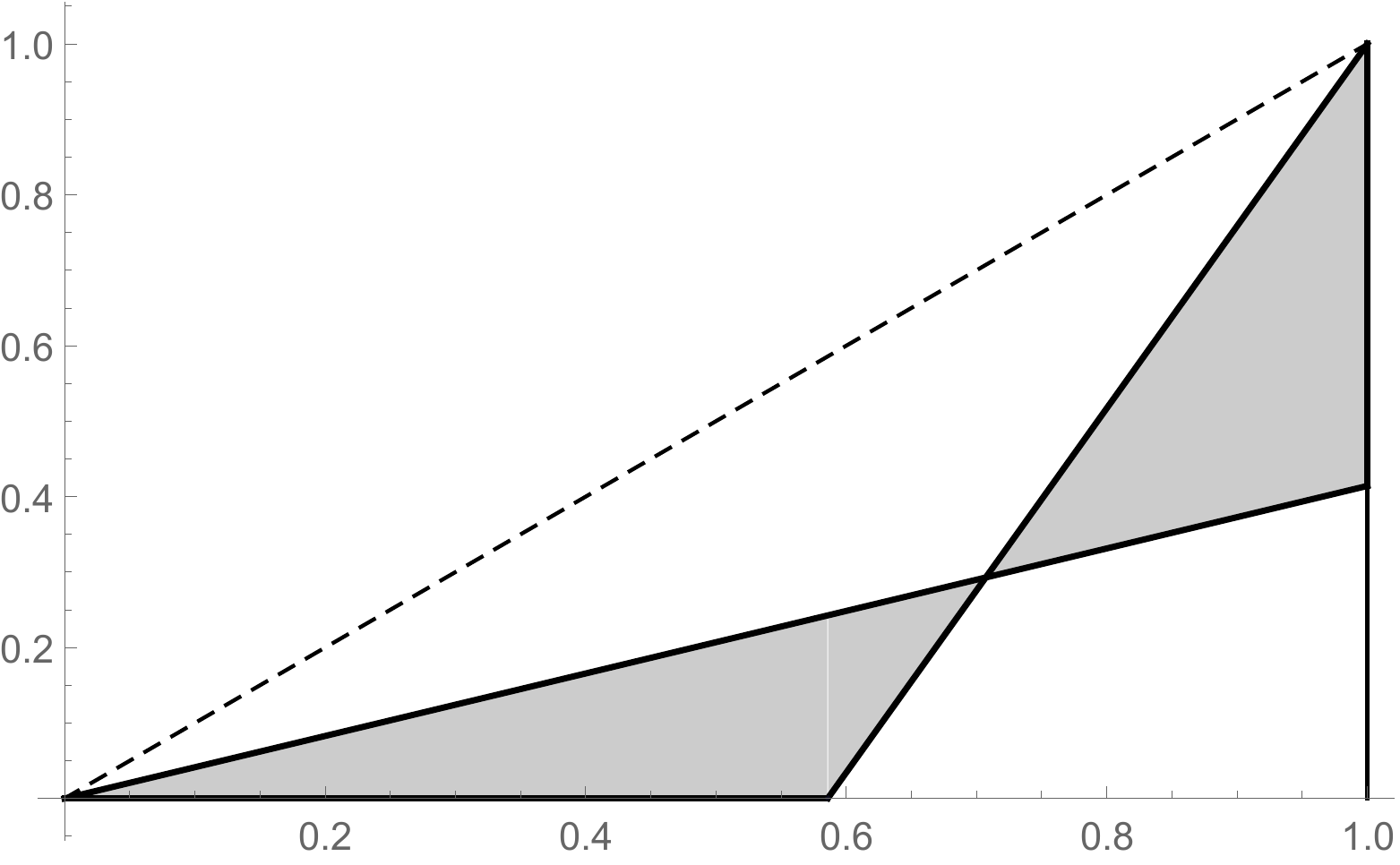}}
\caption{Lorenz curves with equal Gini indices and maximal distance.}
\label{fi:maximales}
\end{figure}

The curves $\ell_a^-$ and $\ell_a^+$ satisfy another extremal property related to inequality within the class $\mathcal{L}_a$.

\begin{proposition}\label{Proposition.Extremal.Third}
  Let $a\in [0,1]$ be fixed. For all $\ell\in \mathcal{L}_a$ and all $t\in[0,1]$, it holds that
  \begin{equation}\label{Extremality-la}
    \int_{0}^{t} \ell_a^+(x)\,\d x \ge \int_{0}^{t} \ell(x)\,\d x \ge \int_{0}^{t} \ell_a^-(x)\,\d x.
  \end{equation}

\end{proposition}

Observe that for $t\in[0,1]$, by Fubini's theorem, we have that
\begin{equation*}
  \int_{0}^{t} \ell(x)\,\d x = \frac{1}{\mu}\int_{0}^{t} (t-x)F^{-1}(x)\,\d x.
\end{equation*}
If we measure income, this quantity is a weighted average of the income accumulated by the proportion $t$ of the poorest in that population. The weight function, $w(x)=t-x$, for $0\le x\le t$, places more weight on the poorest. Therefore, inequalities in \eqref{Extremality-la} show that the Lorenz curve $\ell_a^+$ (respectively, $\ell_a^-$) is the most equitable (respectively, least equitable) within the class $\mathcal{L}_a$ in this precise sense. In other words, the distributions given by $\ell_a^-$ and $\ell_a^+$ are extremes for the stochastic relation given in \eqref{Extremality-la}. This relationship is closely related to a stochastic ordering called \textit{third order inverse stochastic dominance}; see \cite{Cal-Carcamo-2010}.

\section{A bidimensional inequality index}\label{Section.Index}

In this section we introduce a two-dimensional inequality index defined for pairs of distributions that combines the Gini coefficients of two variables with the Lorenz distance defined in \eqref{Distance}. Hence, the proposed index simultaneously measures relative inequality and dissimilarity between two populations. We will show that this bidimensional index satisfies many desirable properties. As the definitions in this section only involve Lorenz curves, we refer to pairs $(\ell_1,\ell_2)$, with $\ell_1, \ell_2 \in\mathcal{L}$, instead of considering random variables.





\subsection{Definition of the bidimensional index $\mathcal{I}$}


Let $\ell_1$ and $\ell_2$ be two Lorenz curves in $\mathcal{L}$. As a measure of relative inequality we simply consider the difference of the Gini indices, that is, $G(\ell_2)-G(\ell_1)$. To quantify dissimilarity we employ the distance $d_{\rm L}(\ell_1,\ell_2)$ in \eqref{Distance}. Therefore, a natural proposal for a new two-dimensional index is the following:
\begin{equation}\label{I0}
  \mathcal{I} (\ell_1,\ell_2) =  (G(\ell_2)-G(\ell_1),d_{\rm L}(\ell_1,\ell_2)).
\end{equation}

The next result provides the region of $\R^2$ where $\mathcal{I}$ takes values.
\begin{proposition}\label{Proposition:rangeDelta}
  Let  $\I$ be defined in \eqref{I0} and let us consider the region of $\R^2$ defined by
  \begin{equation*}
    \Delta= \mathcal{I}(\mathcal{L}\times \mathcal{L})\equiv  \{ \mathcal{I}(\ell_1,\ell_2)  : \ell_1,\ell_2\in\mathcal{L}   \}. 
  \end{equation*}
We have that
\begin{equation}\label{Delta-estrellita}
   \Delta= \left\{ (x,y)\in [-1,1]\times [0,1]  : |x|\le y \le M^*(x)         \right\},
\end{equation}
where the function $M^*$ is defined in \eqref{M*}.
\end{proposition}

\begin{figure}[h]
\centering{\includegraphics[width=14 cm]{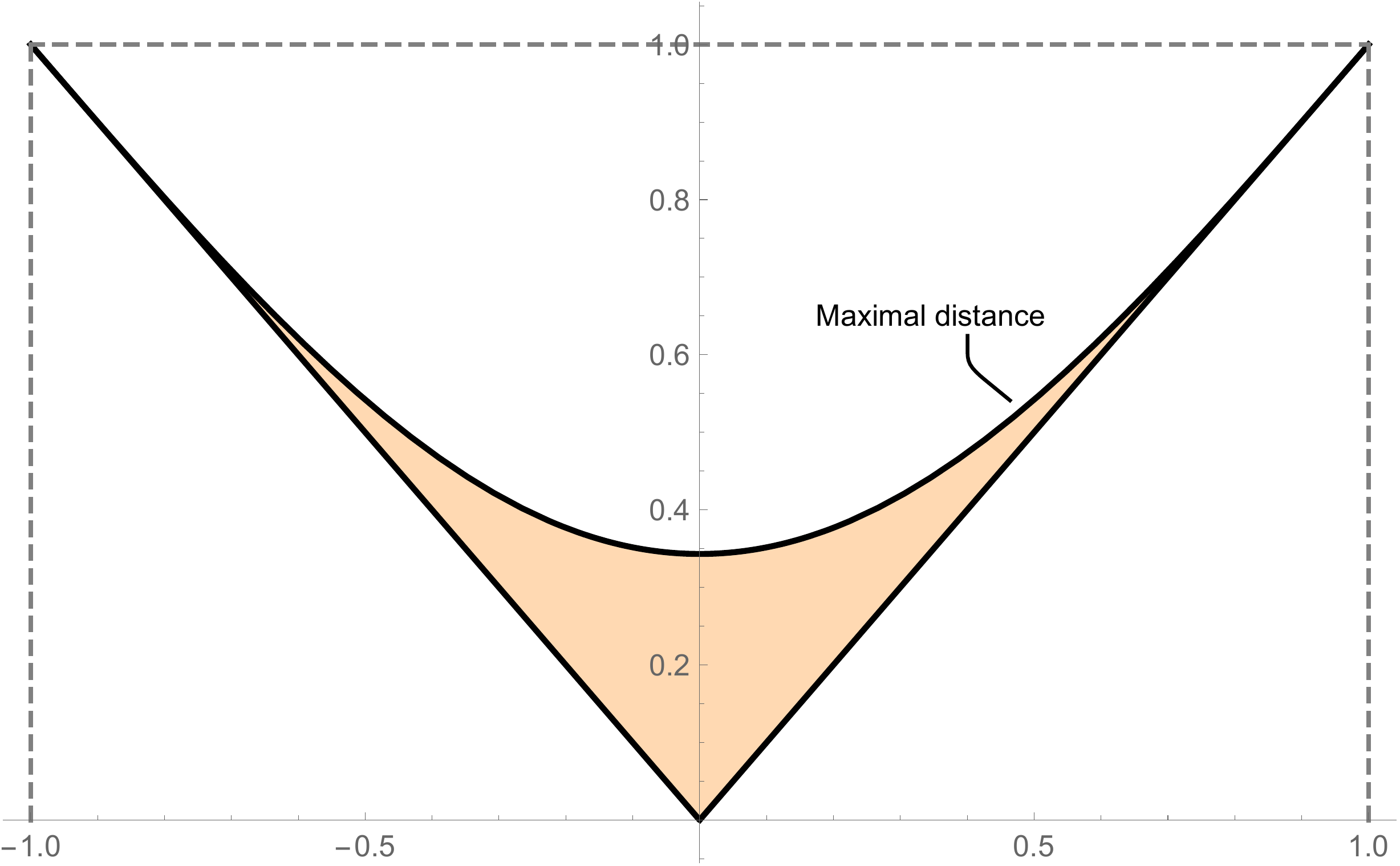}}
\caption{The region $\Delta$ in \eqref{Delta-estrellita}.}
\label{fi:delta-estrellita}
\end{figure}

In Figure~\ref{fi:delta-estrellita} we have plotted the region $\Delta$ specified in \eqref{Delta-estrellita}. This graphical representation is very informative because we see how different the Gini indices of two variables can be in accordance with the distance between their Lorenz curves. We first notice that the range of variation of the index $\I$ is limited as $\Delta$ has a very small area:
$$\text{Area}(\Delta)= 2(6-\pi-\log(16))\approx 0.17.$$

Figure~\ref{fi:delta-estrellita} also reveals that if two distributions are very different, that is, $d_{\rm L}(\ell_1,\ell_2)$ is large, then the difference of their Gini indices cannot vary too much. This is specially noticeable when $d_{\rm L}(\ell_1,\ell_2) > M^*(0)\approx 0.34$. If we look first at the $x$-axis in Figure~\ref{fi:delta-estrellita}, we see that the distance between variables with very different Gini indices has a very small range of variation. For example, if $|G(\ell_1)-G(\ell_2)|=0.35$, the quantity $d_{\rm L}(\ell_1,\ell_2)$ can only vary $0.09$. This small range of variation includes extremely different situations: from the case in which the variables are stochastically ordered in the Lorenz sense---which amounts to minimum distance---to the situation in which the two distributions are as different as possible (the extreme pairs of distributions obtained in Theorem \ref{Theorem.Mab}).


It is easy to check that the index $\I$ in \eqref{I0} satisfies all the ideal properties enumerated in Section \ref{Subsection.Ideal-properties}. Nevertheless, it has the disadvantage that its values are difficult to interpret because they are located in a narrow region. Therefore, in the rest of this section, we suggest two possible transformations of the index taking values in a more convenient region.

\subsection{Transforming the index to improve its visualization}\label{Subsetion.I*}

To facilitate the understanding of the graphical representation of the index, in this section we propose two normalizations of $\mathcal{I}$ in \eqref{I0} taking values on a simpler region of the plane, instead of lying on the set $\Delta$ displayed in Figure \ref{fi:delta-estrellita}. For example, we can transform---through a suitable homeomorphism---the set $\Delta$ in \eqref{Delta-estrellita} into the triangle
\begin{equation}\label{T}
  T=\{(x,y)\in\R^2 : |x|\le y \le 1\}.
\end{equation}

There are several alternatives to carry out this normalization. The simplest way to transform the set $\Delta$ into $T$ is by linearly stretching the segment
$[(x,|x|),(x,M^*(x))]$ into $[(x,|x|), (x,1)]$, for each $-1\le x\le 1$,
where $M^*$ is in \eqref{M*}. We thus consider the map $t_*:\Delta\to T$ given by
\begin{equation}\label{t*}
  t_*(x,y)=\left(x, |x| + \frac{(1-|x|)(y-|x|)}{M^*(x)-|x|}\right).
\end{equation}

We introduce the two-dimensional index defined by
\begin{equation}\label{I*}
  \mathcal{I}_*(\ell_1,\ell_2)=t_*(\mathcal{I}_0 (\ell_1,\ell_2))=t_*(G(\ell_1)-G(\ell_2),d_{\rm L } (\ell_1,\ell_2)),
\end{equation}
where $t_*$ is the homeomorphism defined in \eqref{t*}. By construction, we have that $\mathcal{I}_*$ takes values in the triangle \eqref{T}.

The mapping $t_*$ in \eqref{t*} is perhaps the most natural homeomorphism to transform $\Delta$ in \eqref{Delta-estrellita} into $T$ in \eqref{T}. Nevertheless, only super-extremal pairs of distributions (see Definition \ref{Definition.super-extremal}) lay on the uppermost side of the triangle $T$, that is, the segment $[-1,1]\times\{1\}$. For instance, among all pairs of extremal distributions with equal Gini, $\{(\ell^-_a,\ell^+_{a}) : a\in[0,1]$\}, only the pair with $a=2-\sqrt{2}$ achieves a value of $\mathcal{I}_*$ equal to $(0,1)$. This happens because $\mathcal{I}_*$ only takes into account \textit{the difference} between the Gini indices of the involved variables.

The second proposal is to incorporate to $\mathcal{I}$ in \eqref{I0} the value of the Gini indices of each variable separately. In this way, we can send all extremal pairs of distributions to the uppermost side of $T$. We start with the following proposition.

\begin{proposition}\label{Proposition-range-Delta*}
  Let us consider the region of $\R^3$ defined by
  \begin{equation*}
  \Delta^*  =\{ (G(\ell_1),G(\ell_2),d_{\rm L}(\ell_1,\ell_2))  : \ell_1,\ell_2\in\mathcal{L}   \}.
  \end{equation*}
We have that
\begin{equation}\label{Delta-set}
   \Delta^*= \left\{ (x,y,z)\in [0,1]^3  : |x-y|\le z \le M(x,y)         \right\},
\end{equation}
where $M$ is given in \eqref{Mab}.
\end{proposition}

\begin{figure}[h]
\centering{\includegraphics[width=14cm]{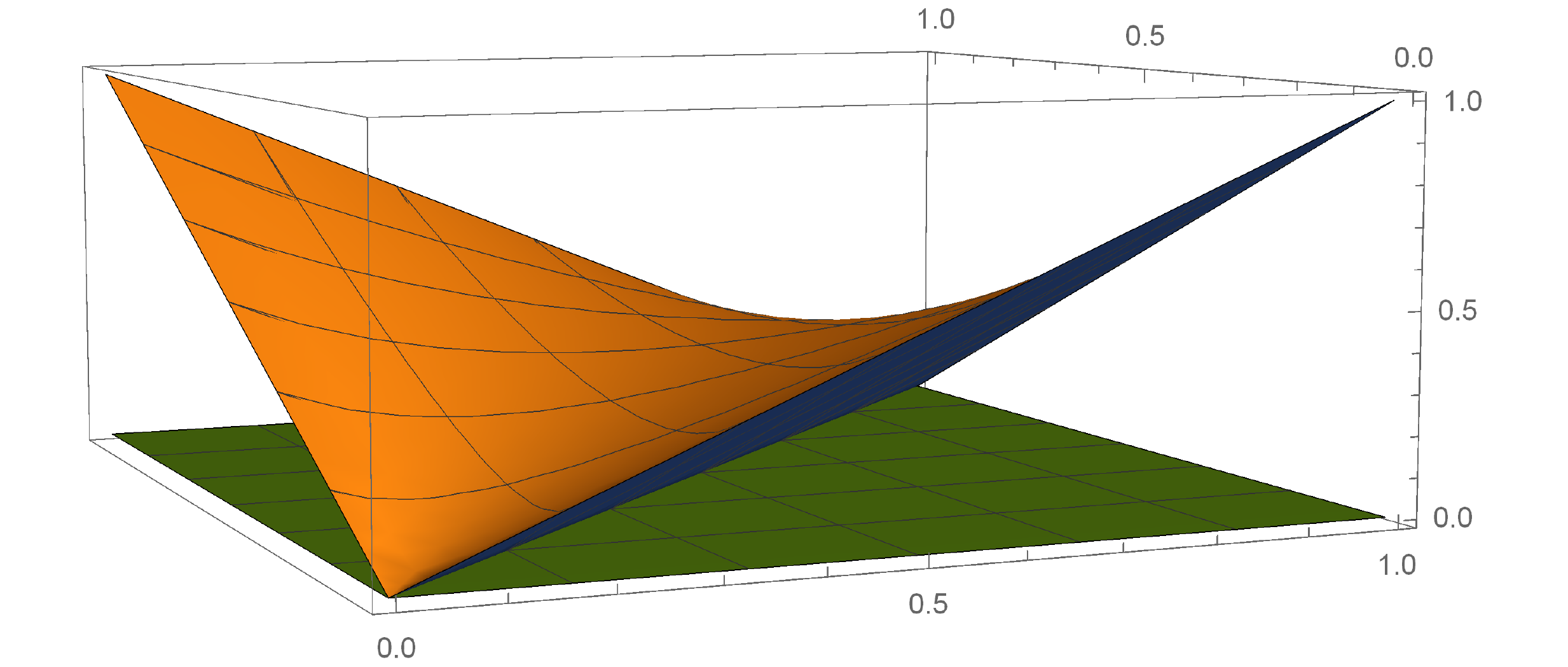}}
\caption{The region $\Delta^*$ in \eqref{Delta-set}.}
\label{fi:delta}
\end{figure}

Figure \ref{fi:delta} represents the set $\Delta^*$ in \eqref{Delta-set}. To understand the actual size of $\Delta^*$, we point out that its volume  is
\begin{equation*}
  \text{Vol} (\Delta^*)=\frac{2}{3} (10-\pi ^2)\approx 0.087.
\end{equation*}
However, the normalization used to construct $\I_*$ in Section \ref{Subsetion.I*} (through the function $M^*$) generates the set
\[ \left\{ (x,y,z)\in [0,1]^3  : |x-y|\le z \le M^*(x-y)         \right\}  \]
that contains $\Delta^*$ (by \eqref{M-and-M*}) and whose volume is (approximately) $0.14$ ($1.58$ times larger that the one of $\Delta^*$).

Next, we consider the map $t^*:\Delta^*\to T$ defined by
\begin{equation}\label{t}
  t^*(x,y,z)= \left(y-x, |y-x| + \frac{(1-|y-x|)(z-|y-x|)}{M(x,y)-|y-x|}\right).
\end{equation}
Observe that $t^*(\Delta)=T$, but $t^*$ is \textit{not} injective. This is not a problem as we want to send all extremal distributions with a given difference of their Gini indices to the same point on the frontier of $T$.

Finally, we define the bidimensional index $\mathcal{I}^*$ as
\begin{equation}\label{I}
  \I^*(\ell_1,\ell_2)= t^*(G(\ell_1),G(\ell_2),d_{\rm L}(\ell_1,\ell_2)).
\end{equation}

By construction, $\I^*$ takes values in $T$. Moreover, $\I^*$ sends all pairs of extremal distributions (see Definition \ref{Definition.extremal}) to the upper side of $T$. Observe that, from \eqref{M-and-M*}, the second component of $\I^*$ is always larger that the corresponding one of $\I_*$ in \eqref{I*}. In this regard, we highlight that $M(a,b)$ could be very different from $M^*(b-a)$. This is specially noticeable when the Gini indices of both variables are simultaneously small or large, as can be seen in Figure~\ref{fi:Maa}. Hence, this second proposal could be significatively different than the previous one in this situation.

Using the expression for $M$ in \eqref{Mab.expression}, we can rewrite the index $\I^*(\ell_1,\ell_2)$ in a slightly different way. For simplicity, let us set
\begin{equation*}
G(\ell_1)=a,\quad  G(\ell_2)=b\quad \text{and}\quad d_{\rm L}(\ell_1,\ell_2)=d.
\end{equation*}
We have that
\begin{equation*}
  \I^*(\ell_1,\ell_2)=
  \begin{cases}
   \displaystyle \left( b-a, a-b  + \frac{(1-a+b)(a+b-a b)}{2(1-a)b^2}\cdot (d-a+b)   \right), & \mbox{if } a\ge b, \\[5mm]
    \displaystyle \left( b-a,   b-a +\frac{(1+a-b)(a+b-ab)}{2 a^2 (1-b)}\cdot (d+a-b)   \right), & \mbox{if } b < a.
  \end{cases}
\end{equation*}

\subsection{Main properties of the index}

The following proposition enumerates the main properties of the indices $\I_*$ and $\I^*$ defined in \eqref{I*} and \eqref{I}, respectively.

\begin{proposition}\label{Proposition.Properties}
Let  $(X_1,X_2)$ be a pair of random variables with Lorenz curves $\ell_1$ and $\ell_2$. We consider the index $\I_*=(\I_{*1},\I_{*2})$ defined in \eqref{I}. The following properties hold:

  \begin{enumerate}[(i)]
    \item \textsl{Normalization:} We have that $\I_*(\mathcal{L}\times \mathcal{L})=T$ in \eqref{T}, i.e., $\I_*$ takes values in the triangle $T$.
    Moreover, $\I_{*1}(\ell_1,\ell_2)\ge 0$ if and only if $G(\ell_1)\le G(\ell_2)$, whereas $\I_{*1}(\ell_1,\ell_2)\le 0$ if and only if $G(\ell_2)\le G(\ell_1)$. In other words, positive (respectively, negatives) values of the first component of $\I_*$ indicates that $\ell_1$ (respectively, $\ell_2$) is fairer than $\ell_2$ (respectively, $\ell_1$) according to the Gini index.
    \item \textsl{Symmetry:} We have that $\I_*(\ell_2,\ell_1)=(-\I_{*1}(\ell_1,\ell_2),\I_{*2}(\ell_1,\ell_2))$.
    \item \textsl{Extreme values (frontier of $T$):}
    \begin{enumerate}[(1)]
    \item $\I_*(\ell_1,\ell_2)\in L_1=\{ (x,x) : x\in[0,1] \}$ if and only if $\ell_1\ge \ell_2$, i.e., $X_1\Lo X_2$.
    \item $\I_*(\ell_1,\ell_2)\in L_2=\{ (-x,x) : x\in[0,1] \}$ if and only if $\ell_1 \le \ell_2$, i.e., $X_2\Lo X_1$.
    \item $\I_*(\ell_1,\ell_2)\in L_3=[-1,1]\times\{1\}$ if and only if $(\ell_1,\ell_2)$ is extremal, i.e., $d_{\rm L}(\ell_1,\ell_2)=M(G(\ell_1),G(\ell_2))$, where $M$ is given in \eqref{Mab} and \eqref{Mab.expression}.
  \end{enumerate}
    \item \textsl{Value at extreme points of $T$:}
    \begin{enumerate}[(1)]
    \item $\I_*(\ell_1,\ell_2)=(0,0)$ if and only if $\ell_1=\ell_2$. Therefore, the value at the origin means that the associated distributions satisfy that $X_1\st c X_2$, where $c>0$ is a constant. 
    \item $\I_*(\ell_1,\ell_2)=(1,1)$ if and only if $\ell_1=\ell_{\rm pe}$ and $\ell_2=\ell_{\rm pi}$.
    \item $\I_*(\ell_1,\ell_2)=(-1,1)$ if and only if $\ell_1=\ell_{\rm pi}$ and $\ell_2=\ell_{\rm pe}$.
  \end{enumerate}
    \item \textsl{Continuity:} If $\{\ell_{1,n_1}\}_{n_1\ge 1}\subset \mathcal{L}$ and $\{\ell_{2,n_2}\}_{n_2\ge 1}\subset \mathcal{L}$ are sequences such that $\ell_{1,n_1}\to \ell_1$ and $\ell_{2,n_2}\to \ell_2$  pointwise as $n_1,n_2\to\infty$, then
         $\I_*(\ell_{1,n_1},\ell_{2,n_2})\to \I_*(\ell_1,\ell_2)$.
           \end{enumerate}
The same properties hold for the index $\I^*$ in \eqref{I*} by changing in (iii) (3) `extremal' for `super-extremal' and `$M(G(\ell_1),G(\ell_2))$' for `$M^*(G(\ell_1)-G(\ell_2))$'.
\end{proposition}

Figure \ref{fi:super} summarizes graphically the properties in Proposition \ref{Proposition.Properties}.

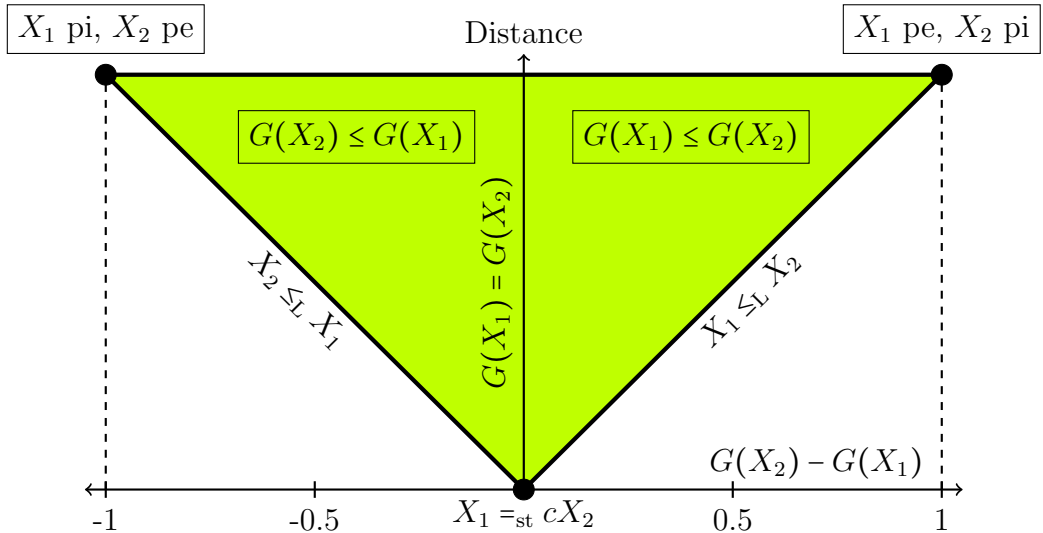
\begin{figure}[h]\centering{
\begin{tikzpicture}[scale=5.5]
\draw [black, fill=lime] (0,0) -- (1,1) -- (-1,1) -- (0,0);
\draw[fill] (1,1) circle [radius=0.025];
\draw[fill] (0,0) circle [radius=0.025];
\draw[fill] (-1,1) circle [radius=0.025];
\draw [ultra thick] (0,0) -- (1,1);
\draw [ultra thick] (0,0) -- (-1,1);
\draw [ultra thick] (1,1) -- (-1,1);
\draw (0,0) to node [sloped,below] {$X_1\Lo X_2$} (1,1);
\draw (0,0) to node [sloped,below] {$X_2\Lo X_1$} (-1,1);
\draw [thick] [<->] (0,0) -- (-1.05,0);
\draw [thick] [<->] (0,0) -- (1.05,0);
\draw [thick] [<->] (0,0) -- (0,1.05);
\node [above] at (0,1.05) {Distance};
\node [above] at (0.7,0) {$G(X_2)-G(X_1)$};
\node [below] at (0,0) {$X_1=_{\rm st} c X_2$};
\node[draw] [above] at (1,1.05) {$X_1$ pe, $X_2$ pi};
\node[draw] [above] at (-1,1.05) {$X_1$ pi, $X_2$ pe};
\draw [dashed, thick] (1,0) -- (1,1);
\draw [dashed, thick] (-1,0) -- (-1,1);
\draw [thick] (1,-0.02) node[below]{1} -- (1,0.02);
\draw [thick] (.5,-0.02) node[below]{0.5} -- (.5,0.02);
\node[draw]  at (.4,.85) {$G(X_1)\le G(X_2)$};
\node[draw]  at (-.4,.85) {$G(X_2)\le G(X_1)$};
\draw [thick] (-.5,-0.02) node[below]{-0.5} -- (-.5,0.02);
\draw [thick] (-1,-0.02) node[below]{-1} -- (-1,0.02);
\draw (0,0) to node [sloped,above] {$G(X_1)=G(X_2)$} (0,1);
\end{tikzpicture}}
\caption{The triangle $T$ (in green) in which the indices $\I_*$ and $\I^*$ take values. The $x$-axis represent the difference between the Gini indices of two variables and the $y$-axis the normalization of the Lorenz distance through $t_*$ or $t^*$ in \eqref{t*} and \eqref{t}, respectively.}
\label{fi:super}
\end{figure}


\section{Estimation and asymptotic properties}\label{Section-Asymptotics}

In this section we prove that the plug-in estimators of the indices defined in the previous section are strongly consistent. Moreover, we determine their asymptotic distributions and obtain necessary and sufficient conditions so that the estimator of $\I$ in \eqref{I0} is asymptotically normal. To finish this section, we have included the conclusions of a small simulation study with generalized beta-type distributions of the second kind to evaluate the behaviour of the asymptotic results in finite samples.


\subsection{Estimators of the indices and strong consistency}


Let $X_1$, $X_2$ be two random variables with distribution functions $F_1$ and $F_2$ and Lorenz curves $\ell_1$ and $\ell_2$, respectively. For $j=1,2$,  we consider random samples from $X_j$, $\{X_{j,i}\}_{i=1}^{n_j}$, $n_j \in\mathbb{N}$. For simplicity, we will assume that both samples are mutually independent. However, similar convergence results can be obtained when we observe ``matched pairs", $\{(X_{1,i},X_{2,i})\}_{i=1}^n$, drawn from a bivariate distribution $(X_1,X_2)$ with copula $C$ satisfying that its maximal correlation is strictly less than one; see \cite{Beare-2010}. As pointed out in \cite{Barrett-Donald-Bhattacharya-2014} and \cite{Sun-Beare-2021}, this second setting is more reasonable when we have one sample of individuals and two measures of welfare.

To simplify the notation, in the sequel all estimated quantities are denoted with a ``hat", and it will be implicitly understood the dependence on the corresponding sample sizes.
To estimate the inequality indices introduced in Section \ref{Section.Index}, the starting point is the natural estimator of the distribution function of the sample. Namely, for $j=1,2$, we denote by $\hat{F}_j$ the \emph{empirical distribution functions} of the samples, i.e.,
$$\hat{F}_j(x)=\frac{1}{n_j}\sum_{i=1}^{n_j} 1_{\{X_{j,i}\le x\}}, \quad x\in [0,\infty),$$
where $1_A$ stands for the indicator function of the set $A$.
The corresponding \emph{empirical quantile functions} are $\hat{F}^{-1}_j(x)=\inf\{ y \ge 0 : \hat F_j (y)\ge x \}$ ($0<x<1$) and the \emph{empirical Lorenz curves} are
\begin{equation*}
\hat\ell_j(t)=\frac{1}{\hat \mu_j }\int_{0}^{t} \hat{F}_j^{-1} (x)\, \d x, \quad t\in[0,1],
\end{equation*}
where $\hat \mu_j=\frac{1}{n_j}\sum_{i=1}^{n_j} X_{j,i}$ are the sample means. Therefore, the plug-in estimator of the indices  $\I$, $\I_*$ and $\I^*$ defined in equations \eqref{I0}, \eqref{I*} and \eqref{I} are respectively given by
\begin{equation}\label{estimators-indices}
\hat \I (\ell_1,\ell_2)= \I (\hat\ell_1 ,\hat \ell_2  ),\quad \hat \I_* (\ell_1,\ell_2)= \I_* (\hat\ell_1 ,\hat \ell_2  ),\quad \text{and} \quad \hat \I^* (\ell_1 , \ell_2)= \I^* (\hat\ell_1,\hat\ell_2).
\end{equation}

The next proposition shows the strong consistency of these estimators.

\begin{proposition}\label{Proposition.Strong.Consistency}
As $n_1,n_2 \to \infty$, we have that
\begin{equation*}
  \I (\hat \ell_1,\hat \ell_2) \to \I (\ell_1,\ell_2),\quad   \I_* ( \hat \ell_1,\hat \ell_2) \to \I_* (\ell_1,\ell_2)\quad \text{and}\quad  \I^* (\hat \ell_1,\hat \ell_2) \to \I^* (\ell_1,\ell_2)\quad  \text{a.s.}
\end{equation*}
\end{proposition}

The proof of Proposition \ref{Proposition.Strong.Consistency} (see the Appendix) shows that strong consistency of the estimators of the indices follows from the (almost surely) uniform convergence of the empirical Lorenz curves to its theoretical counterparts. We observe that this convergence can be derived under weaker assumptions regarding the samples of $X_1$ and $X_2$. For instance, in \citet[Theorem 2.1]{Csorgo-Yu-1999} strong uniform consistency of $\hat \ell_j$ to $\ell_j$ is obtained under very general conditions.


\subsection{Convergence of empirical Lorenz processes}


The computation of the asymptotic distribution of the indices relies on the convergence of the \emph{empirical Lorenz processes} (associated with $X_j$  with Lorenz curves $\ell_j$, respectively, for $j=1,2$) given by
\begin{equation}\label{Lorenz-processes}
\sqrt{n_j} \left(\hat \ell_j(t)-\ell_j(t)\right),\quad t\in[0,1],\quad j=1,2.
\end{equation}

The analysis of the convergence of Lorenz processes can be traced back to \cite{Goldie-1977}. However, we will use a recent result by \citet{Sun-Beare-2021} in which the weak joint convergence of the processes in \eqref{Lorenz-processes} is obtained by using a new result regarding the convergence of the quantile process in $L^1$ (see \cite{Kaji-2018} and \cite{Kaji-2019}) together with the functional delta method; see \citet[Section 3.9]{van der Vaart-Wellner}. Finally, we show the (directional) Hadamard differentiability of the map \eqref{I0}, which essentially follows from \citet[Lemma 4]{Carcamo}, and apply the (extended) functional delta method (see \citet[Theorem 2.1]{Shapiro-1990}) to derive the asymptotic distributions. Therefore, we need to impose various conditions on the variables so that the associated Lorenz processes converge in $L^1$.


\begin{assumption}[\textbf{Integrability condition}] \label{Assumption-integrability}
For $j=1,2$, it holds that $\Lambda_{2,1}(X_j)<\infty$, where
  \begin{equation}\label{Delta-21}
  \Lambda_{2,1}(X_j)=\int_{0}^{\infty} \sqrt{1-F_j(x)}\,\d x.
\end{equation}
\end{assumption}
\begin{assumption}[\textbf{Regularity condition}]\label{Assumption-regularity} For $j=1,2$, we have that $F_j(0)=0$ and $F_j$ has at most finitely many jumps and is continuously differentiable elsewhere with
strictly positive density.
\end{assumption}

Assumption \ref{Assumption-integrability} amounts to saying that the variables $X_j$ belong to the Lorentz space $\mathcal{L}^{2,1}$; see \citet[Section 1.4]{Grafakos}. This condition is equivalent to the convergence of the classical empirical process (associated with $F_j$) in the space $L^1$; see \citet[Theorem 2.1]{del Barrio}.
Condition $\Lambda_{2,1}(X_j)<\infty$ is slightly stronger than $\E X_j^2<\infty$: it holds for example when $\E X_j^{2+\epsilon}<\infty$, for some $\epsilon>0$.
The smoothness condition in Assumption \ref{Assumption-regularity} is necessary to conclude the convergence of the quantile process in $L^1$
through the differentiability of the inverse map plus the convergence of the empirical process in $L^1$; see \citet{Kaji-2019}.

Under these assumptions, and with the functional delta method, \citet[Lemma 2.1]{Sun-Beare-2021} obtained the asymptotic behaviour of the empirical Lorenz process in $C([0,1])\equiv$ the space of continuous real-valued functions on $[0,1]$. This result is collected in the following lemma where we use the arrow `$\rightsquigarrow$' to denote the weak convergence of probability measures in the sense of Hoffmann-J{\o}rgensen; see \cite{van der Vaart-Wellner}. Further, for $j=1,2$, $\B_j$  will denote two independent standard \textit{Brownian bridges} on $[0,1]$.

\begin{lemma}\label{Lemma-Lorenz-convergence}
For $j=1,2$, let us assume that $X_j$ satisfy Assumptions \ref{Assumption-integrability} and \ref{Assumption-regularity}. As $n_j \to\infty$, we have that $\sqrt{n_j} \left(\hat \ell_j-\ell_j\right) \rightsquigarrow     \mathbb{L}_j $ in $C([0,1])$, where $\mathbb{L}_j$  are (independent) centered Gaussian processes with continuous trajectories a.s. that can be expressed as
\begin{equation}\label{Lorenz-convergence}
   \mathbb{L}_j(t)  =  \ell_j(t)\int_{0}^{1}  \ell^{\prime\prime}_j(x)\mathbb{B}_j(x)\, \d x  - \int_{0}^{t}  \ell^{\prime\prime}_j(x)\mathbb{B}_j(x)\, \d x,\quad t\in[0,1],  \\
\end{equation}
where $\mathbb{B}_j$ are independent standard Brownian bridges.
\end{lemma}

Some comments should be made regarding the previous key lemma. First, we have opted for the less restrictive assumptions given in \citet[Proposition 4.2]{Kaji-2019} instead of those considered in \citet[Assumption 2.1]{Sun-Beare-2021} or the more demanding in \citet[Assumption 1]{Barrett-Donald-Bhattacharya-2014}. However, for simplicity, we assume that $X_1$ and $X_2$ are independent. If this is not the case, a similar result can be stated (see \citet[Lemma 2.1]{Sun-Beare-2021}): the joint limit distribution of $(\sqrt{n_1} \left(\hat \ell_1-\ell_1\right),\sqrt{n_2} \left(\hat \ell_2-\ell_2\right) )$ is again $(\mathbb{L}_1,\mathbb{L}_2)$, but in this case the Brownian bridges $\mathbb{B}_1$ and $\mathbb{B}_2$ in \eqref{Lorenz-convergence} are correlated.

\subsection{Asymptotic distribution of the estimator of $\I$}

The computation of the asymptotic distribution of the estimator of $\I$ follows from Lemma \ref{Lemma-Lorenz-convergence} together with the functional delta method. Traditionally, to apply this latter tool it is usually assumed that the considered maps are Hadamard differentiable. However, as showed by \cite{Shapiro-1991} (see also \cite{Dumbgen}) it is enough to have Hadamard \emph{directional} differentiability. We recall this concept in the following definition.

\begin{definition}\label{Definition.Hadamard}
Let $\mathcal{D}$ and $\mathcal{E}$ be real Banach spaces with norms $\Vert \cdot\Vert_\mathcal{D}$ and $\Vert \cdot\Vert_\mathcal{E}$, respectively. A map $\phi:\mathcal{D}\longrightarrow \mathcal{E}$ is said to be \textit{Hadamard directionally differentiable} at $\theta\in \mathcal{D}$ tangentially to a set $\mathcal{D}_0\subset \mathcal{D}$ if there exists a map $\phi^\prime_\theta: \mathcal{D}_0 \longrightarrow \mathcal{E}$ such that
\begin{equation}\label{Hadamard}
\left\Vert  \frac{\phi(\theta+t_n h_n)-\phi(\theta)}{t_n}-\phi^\prime_\theta(h)\right\Vert_\mathcal{E}\to 0,
\end{equation}
for all $h\in \mathcal{D}_0$ and all sequences $\{h_n\}\subset \mathcal{D}$, $\{t_n\}\subset \R$ such that $t_n\downarrow 0$ and $\Vert h_n-h\Vert_\mathcal{D}\to 0$.
\end{definition}

The main difference between full and directional Hadamard differentiability is that the derivative $\phi^\prime_\theta$ is not necessarily linear in Definition \ref{Definition.Hadamard}. However, if equation (\ref{Hadamard}) is satisfied, then $\phi^\prime_\theta$ is continuous and positive homogeneous of degree 1; see \citet[Proposition 3.1]{Shapiro-1990}.

The proof of following lemma follows from \citet[Lemma 4]{Carcamo}.

\begin{lemma}\label{Lemma-Hadamard}
The map $\delta:L^1 \to \R$ defined by $\delta(f)=\Vert f\Vert$ is Hadamard directionally differentiable at every $f\in L^1$. For $g\in L^1$, its derivative is given by
\begin{equation}\label{delta-derivative}
  \delta_f'(g)=\int_{\{ f = 0 \}} |g(x)|\, \d x + \int_{\{ f \ne 0 \}} g(x)\cdot \sgn(f(x))\, \d x,
\end{equation}
where $\sgn(\cdot)$ is the sign function.

In particular, if the Lebesgue measure of the set $\{ f = 0 \}=\{x\in[0,1] : f(x)=0\}$ is zero, we have that $\delta$ is (fully) Hadamard differentiable with (linear) derivative
\begin{equation*}
  \delta_f'(g)= \int_{0}^1 g(x)\cdot \sgn(f(x))\, \d x.
\end{equation*}
\end{lemma}

In the following proposition we establish the asymptotic behaviour of the normalized estimator of the index $\I$ in \eqref{I0}. We impose the following condition on the sample sizes.

\begin{assumption}[\textbf{Sampling condition}]\label{Assumption-sizes}
  The sample sizes $n_1$ and $n_2$ are (weakly) balanced, that is, as $n_1,n_2\to \infty$,  $n_1/(n_1+n_2)\to \lambda$, where $\lambda\in[0,1]$.
\end{assumption}

\begin{proposition}\label{Proposition-Asymptotics}
Let Assumptions \ref{Assumption-integrability}, \ref{Assumption-regularity} and \ref{Assumption-sizes} be fulfilled.
Then, as $n_1,n_2\to\infty$,
\begin{equation}\label{I0-normalized}
   \sqrt{\frac{n_1 n_2}{n_1+n_2}} \left( \I (\hat \ell_1,\hat \ell_2) -   \I ( \ell_1, \ell_2)        \right) \rightsquigarrow I = 2\,\left( \int_{0}^{1}  \mathbb{L}(t)\,\d t, \,      \delta^\prime_{\ell_1-\ell_2} (   \mathbb{L}  )   \right),
\end{equation}
where
\begin{equation}\label{L}
  \mathbb{L}= \sqrt{1-\lambda}\, \mathbb{L}_1-  \sqrt{\lambda}\,   \mathbb{L}_2,
\end{equation}
with $\mathbb{L}_j$ the independent centered Gaussian processes in \eqref{Lorenz-convergence} ($j=1,2$) and the derivative $\delta^\prime$ in \eqref{delta-derivative}.
\end{proposition}

The proof of Proposition \ref{Proposition-Asymptotics} (see the Appendix) relies on the joint convergence of the underlying Lorenz processes. Therefore, any sampling scheme ensuring this joint convergence is enough to derive the asymptotic distribution of the normalized estimator.

The following corollary provides necessary and sufficient conditions for the limit distribution in \eqref{I0-normalized} to be bivariate normal.

\begin{corollary}\label{Corollary-normal}
Under the conditions of Proposition \ref{Proposition-Asymptotics}, the following three assertions are equivalent:
\begin{enumerate}[(a)]
  \item The set $\{\ell_1=\ell_2\} = \{ x\in[0,1] : \ell_1(x)=\ell_2(x) \}$ has zero Lebesgue measure.

  \item The limit distribution in \eqref{I0-normalized} can be expressed as
  \begin{equation}\label{limit-distribution-2}
    I= 2\,\left( \int_{0}^{1}  \mathbb{L}(t)\,\d t, \,      \int_{0}^{1}  \mathbb{L}(t) \cdot \sgn(\ell_1(t)-\ell_2(t))  \,\d t   \right),
  \end{equation}
  where $\mathbb{L}$ is defined in \eqref{L}.
  \item The limit in \eqref{I0-normalized} has a centered bivariate normal distribution.
\end{enumerate}
\end{corollary}

We observe that $\{ \ell_1= \ell_2 \}$ is the set of \emph{crossing points} of the two Lorenz curves. The case when this set has zero Lebesgue measure (Corollary \ref{Corollary-normal} (a)) is actually a reasonable assumption when we consider two different populations in practice; for instance, when comparing the anual household income of two different countries. In this scenario, the asymptotic distribution of the index $\I$ is normal, which simplifies implementing the usual inferential procedures (confidence intervals, hypothesis testing). Otherwise, if the $\{ \ell_1= \ell_2 \}$  does have positive Lebesgue measure, the limit distribution in \eqref{I0-normalized} could be complicated to handle. Further, as the derivative appearing in the limit is not linear, the corresponding map $\delta$ is not fully Hadamard differentiable and, consequently, the standard bootstrap scheme fails. \cite{Fang-Santos} propose several methodologies to correct the bootstrap scheme in this situation.


\subsection{Asymptotic distributions of the estimators of $\I_*$ and $\I^*$}

Once the asymptotic distribution of the estimator of $\I$ has been established in the previous section, the corresponding distributions for the indices $\I_*$ and $\I^*$ can be derived thanks to the (traditional) delta method; see for instance \citet[Theorem 3.1]{van der Vaart}.

In the case of the index $\I_*$ in \eqref{I*}, we have that $\I_*(\ell_1,\ell_2) = t_* (\I_0  (\ell_1,\ell_2) ) $, where $t_*$ in \eqref{t*} is (by construction) a smooth map. Hence, we can state the following result.

\begin{proposition}\label{Proposition-I*}
Let $\I_*$ be the index defined in \eqref{I*}. In the conditions of Proposition \ref{Proposition-Asymptotics} we have that
\begin{equation}\label{asymptotic-I*}
 \sqrt{\frac{n_1 n_2}{n_1+n_2}} \left( \I_* (\hat \ell_1,\hat \ell_2) -   \I_* ( \ell_1, \ell_2)        \right)  \rightsquigarrow (t_*)^\prime_{\I ( \ell_1, \ell_2) } (I),
\end{equation}
where $I$ is the limit distribution in \eqref{I0-normalized}, $t_*=(t_{*1},t_{*2})$ is in \eqref{t*} and
\begin{equation*}
(t_*)^\prime_{( x, y ) } =
\begin{pmatrix}
\frac{\partial t_{*1}}{\partial x } ( x , y ) & \frac{\partial t_{*1}}{\partial y } ( x, y ) \\
\frac{\partial t_{*2}}{\partial x } (x, y  ) & \frac{\partial t_{*2}}{\partial y } (x, y )
\end{pmatrix}
=\begin{pmatrix}
1 &   0 \\
\frac{\partial t_{*2}}{\partial x }( x, y)  & \frac{1-|x|}{M^*(x)-|x|}
\end{pmatrix}.
\end{equation*}
\end{proposition}
The expression of $\frac{\partial t_{*2}}{\partial x }( x, y)$ can be easily computed, but it is too long to be included here.

We observe that the evaluation of the derivative of $t_*$ in \eqref{asymptotic-I*} is understood as a product of matrices. At least theoretically, Proposition \ref{Proposition-I*} provides the asymptotic distribution of $\hat\I_*$. Nevertheless, we point out that even when the distribution of $I$ is bivariate normal (see Corollary \ref{Corollary-normal}), the second component of the asymptotic distribution in \eqref{asymptotic-I*} could be complicated as it is expressed as a non-linear transformation of $I$.

The asymptotic distribution of the estimator of $\I^*$ in \eqref{I} can be computed by following the same steps as in the proof of Proposition \ref{Proposition-Asymptotics}. First, we consider the map $\psi : C([0,1])\times C([0,1]) \to\R^3$ given by
\begin{equation*}
  \psi(f,g)=(   1-2\Vert f\Vert, 1-2\Vert g\Vert, 2\Vert f-g \Vert ),\quad  f,g\in C([0,1]).
\end{equation*}
Let $\ell_1,\ell_2\in\mathcal{L}$. We observe that
\begin{equation}\label{I-psi}
  \psi(\ell_1,\ell_2)=(G(\ell_1),G(\ell_2),d_{\rm L}(\ell_1,\ell_2))\quad \text{and}\quad \I^*(\ell_1,\ell_2)=t^*(\psi(\ell_1,\ell_2)),
\end{equation}
where $t^*$ is in \eqref{t}. Again, it can be checked that $\psi$ is Hadamard directionally differentiable at $(\ell_1,\ell_2)$ with derivative given by
\begin{equation}\label{psi-derivative}
\psi^\prime_{(\ell_1,\ell_2)}(h_1,h_2)=2\left(    -\delta^\prime_{\ell_1}(h_1), -  \delta^\prime_{\ell_2}(h_2),\, \delta^\prime_{\ell_1-\ell_2}(h_1-h_2)            \right),\quad \text{for } h_1,h_2\in C([0,1]),
\end{equation}
where $\delta^\prime$ is defined in \eqref{delta-derivative}. Therefore, from \eqref{I-psi} and by the chain rule, we obtain the following result.

\begin{proposition}
Let $\I^*$ be the index defined in \eqref{I}. In the conditions of Proposition \ref{Proposition-Asymptotics} we have that
\begin{equation*}
 \sqrt{\frac{n_1 n_2}{n_1+n_2}} \left( \I^* (\hat \ell_1,\hat \ell_2) -   \I^* ( \ell_1, \ell_2)        \right)  \rightsquigarrow (t^*)^\prime_{\psi(\ell_1,\ell_2) } \left( \psi^\prime_{(\ell_1,\ell_2)}  \left(     \sqrt{1-\lambda}\,\mathbb{L}_1,\sqrt{\lambda}\,  \mathbb{L}_2     \right)     \right),
\end{equation*}
where $ \psi^\prime$ is in \eqref{psi-derivative}, $t^*=(t_1^*,t_2^*)$ is in \eqref{t} and
\begin{equation*}
t^\prime_{( x, y, z ) } =
\begin{pmatrix}
\frac{\partial t_1}{\partial x } ( x , y ,z) & \frac{\partial t_1}{\partial y } ( x, y ,z) & \frac{\partial t_1}{\partial z } ( x, y , z) \\
\frac{\partial t_2}{\partial x } (x, y ,z ) & \frac{\partial t_2}{\partial y } (x, y , z) & \frac{\partial t_2}{\partial z } (x, y , z)
\end{pmatrix}
=\begin{pmatrix}
-1 & 1  &  0 \\
\frac{\partial t_2}{\partial x }( x, y ,z )&  \frac{\partial t_2}{\partial y }( x, y ,z ) & \frac{1-|y-x|}{M(x,y)-|y-x|}
\end{pmatrix}.
\end{equation*}
\end{proposition}
The expressions of $\frac{\partial t_2}{\partial x }( x, y ,z )$ and $\frac{\partial t_2}{\partial y }( x, y ,z )$ can be easily computed, but they are too long to be included here.

\subsection{Simulation study}

We illustrate here some of the previous asymptotic results through a small simulation study. We focus on the index $\I$ in \eqref{I0} since the other proposals, $\I_*$ in \eqref{I*} and $\I^*$ in \eqref{I}, are smooth transformations of $\I$.

To carry out the simulations there are two options to generate the data: to consider parametric families of Lorenz curves (see for instance \cite{Sarabia-2008}) or to use parametric probability density functions. We have chosen to simulate the data from probability densities that are used in practice to model income distributions. Many families of probability distributions have two parameters that represent changes in location and scale; Weibull or Lognormal distributions are examples of such type of families. When this happens, in the Lorenz curve \eqref{Lorenz-curve} the location parameter disappears (because of the normalization) and only the scale parameter remains. Typically, by moving this dispersion parameter, distributions of the same family are ordered in the Lorenz sense. Therefore, the value of the index between two variables of the same (two-parameter) family is usually located on one of the 2 diagonals of the region $\Delta$ in \eqref{Delta-estrellita}. Consequently, to look for examples whose index lies within the triangle we resort to distributions with more than 2 parameters.

Here we propose to use generalized beta distributions of the second kind (GB2). This family has been previously considered as a model for the distribution of income; see \cite{Chotikapanich-et-al-2018} and \cite{McDonald-Ransom-2008}. The probability density of the GB2 distribution depends on 4 positive parameters: $ a, b, p, q $ and is given by
\begin{equation}\label{eq.densidad.GB2}
f(x|a,b,p,q) = \frac{a x^{ap-1}}{b^{ap} \text{B}(p,q) \left(  1+\left( {x}/{b}\right)^{a}    \right)^{p+q}},\quad x>0,
\end{equation}
where $ \text{B}(\cdot,\cdot)$ is the Euler beta function. We will denote $X\sim\text{GB2}(a,b,p,q)$ a random variable with this density.

From the expression of the density of $X\sim \text{GB2}(a,b,p,q)$ in \eqref{eq.densidad.GB2}, we see that
$ f(x | a,b,p,q)$ behaves as ${1}/{x^{aq+1}}$, as $x\to\infty$. In particular, if $\alpha>0$, we have that $\E X^\alpha<\infty$ if and only if $aq>\alpha$. Consequently, GB2 variables are integrable whenever $aq>1$, and, in such a case, we can compute their Lorenz curves. Further, to apply Proposition \ref{Proposition-Asymptotics} or Corollary \ref{Corollary-normal} it is sufficient that $aq>2$.


It can be checked that the Lorenz curve of $X\sim \text{GB2}(a,b,p,q)$ is given by
\begin{equation}\label{eq.GB2.Lorenz}
\ell(t) = \beta(  \beta^{-1} (t| p,q ) | p+1/a,q-1/a    ),\quad 0<t<1,
\end{equation}
where
$$\beta(x|p,q)=\frac{1}{\text{B}(p,q)}\int_0^x t^{p-1}(1-x)^{q-1}\, \d t,\quad 0<x<1,$$
is the incomplete beta function.

The Lorenz curve of $X\sim \text{GB2}(a,b,p,q)$ in \eqref{eq.GB2.Lorenz} does \textit{not} depend on $b$ as it is essentially a scale parameter. However, it is not convenient to select $b=1$ for the considered examples, as the mean of $X$ is
$$\E X = b  \frac{\text{B}(p+1/a,q-1/a)}{\text{B}(p,q)},$$
which also depends on the rest of the parameters. Therefore, to compare better the densities of the considered models in all the GB2 distributions we fix the value
$$ b = \frac{\text{B}(p,q)}{\text{B}(p+1/a,q-1/a)}.$$
In this way, the variables always have expectation 1.

We have considered 5 models of pairs of GB2 variables to reflect a wide range of possible situations regarding the value of the index $\I$ in \eqref{I0}.
\medskip

\textbf{Model 1:} For $i=1,2$, we consider $X_i\sim\text{GB2}(a_i,b_i,p_i,q_i)$ with
\begin{align*}
  a_1&=1.45, & b_1&=1.01,  & p_1& = 0.75, & q_1 & = 1.45, \\
  a_2&=10, &  b_2&=4.86, & p_2& = 0.035, & q_2 & = 7.
\end{align*}
In this example, we have that
$$\quad G(X_1)\approx 0.5886,\quad G(X_2)\approx 0.5923,\quad d_{\rm L}(X_1,X_2) \approx 0.0858.$$
Therefore, $\I(X_1,X_2)\approx (-0.0037 ,  0.0858 )$. Hence, the value of the index lies almost on the vertical line of equality of Gini indexes and the distance between the distribution is intermediate. 
We also note that $a_iq_i>2$ and this means that the integrability condition given in \eqref{Delta-21} is satisfied. In particular, we can apply Corollary \ref{Corollary-normal} to obtain that the (normalized) plug-in estimator of the index is asymptotically normal.
\medskip

\textbf{Model 2:} For $i=1,2$, we consider $X_i\sim\text{GB2}(a_i,b_i,p_i,q_i)$ with
\begin{align*}
  a_1&=1, & b_1& =1.375&  p_1& = 0.8, & q_1 & = 2.1, \\
  a_2&=2.6, &  b_2&= 0.469 & p_2& = 3, & q_2 & = 1.
\end{align*}
We can compute
$$ G(X_1)\approx 0.6828  ,\quad G(X_2)\approx 0.3318,\quad d_{\rm L}(X_1,X_2) \approx 0.3510.$$
Therefore, $\I(X_1,X_2)\approx ( -0.3510, 0.3510  )$. That is, the variables are ordered with respect to the Lorenz dominance. Further, again $a_iq_i>2$ and the integrability condition in \eqref{Delta-21} holds. By Corollary \ref{Corollary-normal} we see that the (normalized) plug-in estimator of the index is a asymptotically normal and the limit distribution is concentrated on the diagonal $L_2$.

\medskip

\textbf{Model 3:} For $i=1,2$, we consider $X_i\sim\text{GB2}(a_i,b_i,p_i,q_i)$ with
\begin{align*}
  a_1&=4, & b_1& =0.7&  p_1& = 0.8, & q_1 & = 0.6, \\
  a_2&=3, &  b_2&= 1.38 & p_2& = 0.5, & q_2 & = 1.3.
\end{align*}
We obtain that
$$G(X_1)\approx 0.3547,\quad G(X_2)\approx 0.3723,\quad d_{\rm L}(X_1,X_2) \approx 0.041.$$
Then, $\I(X_1,X_2)\approx ( -0.0176 , 0.0414)$. The two variables have a similar Gini index and small distance between them. The parameters satisfy $a_iq_i>2$ and we can use  Corollary \ref{Corollary-normal} to conclude that the (normalized) plug-in estimator of the index is asymptotically normal.

\medskip

\textbf{Model 4:} For $i=1,2$, we consider $X_i\sim\text{GB2}(a_i,b_i,p_i,q_i)$ with
\begin{align*}
  a_1&=2, & b_1& =0.072&  p_1& = 5, & q_1 & = 0.6, \\
  a_2&=2, &  b_2&= 13.18 & p_2& = 0.1, & q_2 & = 5.
\end{align*}
We have that
$$G(X_1)\approx 0.7546,\quad G(X_2)\approx 0.7553,\quad d_{\rm L}(X_1,X_2) \approx 0.1369.$$
We hence obtain that, $\I(X_1,X_2)\approx ( -0.0007, 0.1369  )$. We see that the variables have similar Gini index and a high distance between them. In this case, $a_1q_1=1.2$ and we cannot apply Proposition \ref{Proposition-Asymptotics} or Corollary \ref{Corollary-normal} because $X_1$ does not satisfy the integrability condition \eqref{Delta-21}. Only the consistency of the estimator is guaranteed by Proposition \ref{Proposition.Strong.Consistency}.

\medskip

\textbf{Model 5:} In this example $X_1=X_2$ (in distribution). We consider $X_i\sim\text{GB2}(a_i,b_i,p_i,q_i)$ with
\begin{equation*}
  a_1=a_2 = 2,\quad b_1=b_2 =\frac{4}{\pi},\quad  p_1=p_2 = 1,\quad  q_1  = q_2=2.
\end{equation*}
The index takes the value $(0,0)$ and as $a_iq_i=4$, we can apply Proposition \ref{Proposition.Strong.Consistency} to conclude that the (normalized) estimator of the index converges in distribution to a non-Gaussian random vector.

\medskip

Additional details of these simulations can be found in the Supplementary Material.
The main conclusions are the following: In Models 1--3, the asymptotic distribution of the estimator of the index is normally distributed. However, if the corresponding Lorenz curves are close to each other, then larger sample sizes are needed to observe a Gaussian distribution. This is reasonable because when the variables coincide (Model 5), the limit distribution is not normal (it has a second positive component). In the case of Model 4 (there is no convergence), we observe that we can estimate the index reasonably well as the estimator is consistent  (see Proposition \ref{Proposition.Strong.Consistency}).


\section{Application to EU-SILC income data} \label{Section.RealData}

EU-SILC (European Union Statistics on Income and Living Conditions) is the reference source for comparable longitudinal and cross-sectional microdata on income, living conditions, poverty and social inclusion in Europe.
As part of its objective of monitoring poverty and social inclusion in the EU, the EU-SILC project releases  statistics and reports on income and living conditions, for instance, indicators on the distribution of income.
The microdata are separately provided to EU-SILC by each country participating in the project, as collected by the administrative organism in charge of compiling the official statistics of that state.
In this section we compute the plug-in estimator (\ref{estimators-indices}) of the bidimensional inequality index introduced in (\ref{I0}) for cross-sectional income microdata (at the household level) obtained from EU-SILC collection.

The random variable \(X\) under consideration is the {\em equivalised disposable income}, the total disposable income of a private household divided by the equivalised household size. The total disposable income represents the total income of a household which is available for saving or spending.
The equivalised household size is the number of household members converted into ``equivalised'' adults by the modified OECD (Organisation for Economic Co-operation and Development) equivalence scale: the first household member aged 14 years or more counts as 1 person, each other household member aged 14 years or more counts as 0.5 person, each household member aged 13 years or less counts as 0.3 person.
The equivalised disposable income is one of the variables describing income at the household level which is used by Eurostat (the statistical office of the EU) to compute Gini coefficients and other inequality indicators.

We use the bidimensional inequality index $\mathcal{I}$ in (\ref{I0}) to compare the empirical Lorenz curves derived from the equivalised disposable income of two populations, \(X_1\) and \(X_2\). These are generated in two different ways:
\begin{enumerate}
\item[(i)] {\bf Cross-temporal comparisons within countries:}
Fix a \textit{reference year} and country. Analyze the evolution of the index in that country (relative to the initial year) over a period of time.
This gives insight into the evolution of income inequality in a specific country along the period under study.
\item[(ii)] {\bf Longitudinal comparisons between two countries:} Two different countries in the same year (and letting the year evolve along an available  period). This allows comparing the relative evolution of inequality in the two countries.
\end{enumerate}
EU-SILC offers microdata from several European countries on a span of more than a decade. Due to obvious proximity reasons, as an example of (i), we have chosen to focus on the evolution of Spain along the 2008--2019 period. Even though there are income microdata from Spain available for the years previous to 2008, the Spanish Instituto Nacional de Estad\'{\i}stica (INE), in charge of the official statistics in Spain and actual source of the Spanish EU-SILC microdata, clearly states that income data before 2008 are not comparable to those after 2008, owing to a methodological change in the data collection process. Regarding the case (ii), we illustrate the comparison between two countries with various examples: we first consider two countries which are almost ordered with respect to the Lorenz dominance; this is the case of Finland and Greece. The last two examples correspond to countries which are not ordered, i.e., their Lorenz curves cross each other. Examples of this situation are Spain vs.\ Portugal and France vs.\ Germany.

\subsection{Yearly evolution of inequality in Spain with respect to 2008} \label{Subsection.RealData.Spain}

Year 2008 was the onset of a severe Spanish financial and economic crisis, officially ending in 2014. It was triggered by the world financial crisis of 2007--08, but one of its main causes was the heavy dependency of the Spanish economy and labour market on low-productivity activities such as construction and services (see, e.g., \citet[Chapter 4]{Royo-2020}). The existence of a housing bubble and a record level of family indebtedness had a snowball effect. There was a first recession period between 2008 and 2010 and a second one between 2011 and 2013. In 2012, Spain had to apply for a 100 billion rescue package provided by the European Stability Mechanism. Due to the resulting steep rise of unemployment rate, in 2013 more than half a million immigrants returned from Spain to their countries of origin. Figure~\ref{fi:Spain0819IIncomeGini} plots the evolution of the mean equivalised disposable income and the Gini index from 2008 to 2019 and clearly reflects this abrupt crisis. Figure~\ref{fi:Spain0819IIncomeGini} also shows the hard climb towards a recovery of the pre-crisis level, which has taken more than 9 years (\citet{Economist-2018}, \citet{IMF-2017}) and has been abruptly ended by the COVID-19 pandemic (\citet{IMF-2020}). Although in 2017 Spanish GDP went beyond its pre-crisis peak of 2007 and many indicators reflected the impressive recovery, it was generally agreed that the country was more unequal than in 2008 (\citet{Economist-2018}, \citet{IMF-2017}). Thus, it is interesting to analyze the evolution of inequality in Spanish society from 2008 onwards (see \citet{Blavier-2017}), in particular to compare the distribution of income between 2008 (held fixed) and the following years. Microdata from INE and EU-SILC cover up to year 2019 (included). Income data from 2020 are not yet available, but they will undoubtedly reflect the severe economic contraction induced by the pandemic and an increase of socio-economic disparities in the population.

\begin{figure}[h]
\begin{center}
\includegraphics[width=11cm]{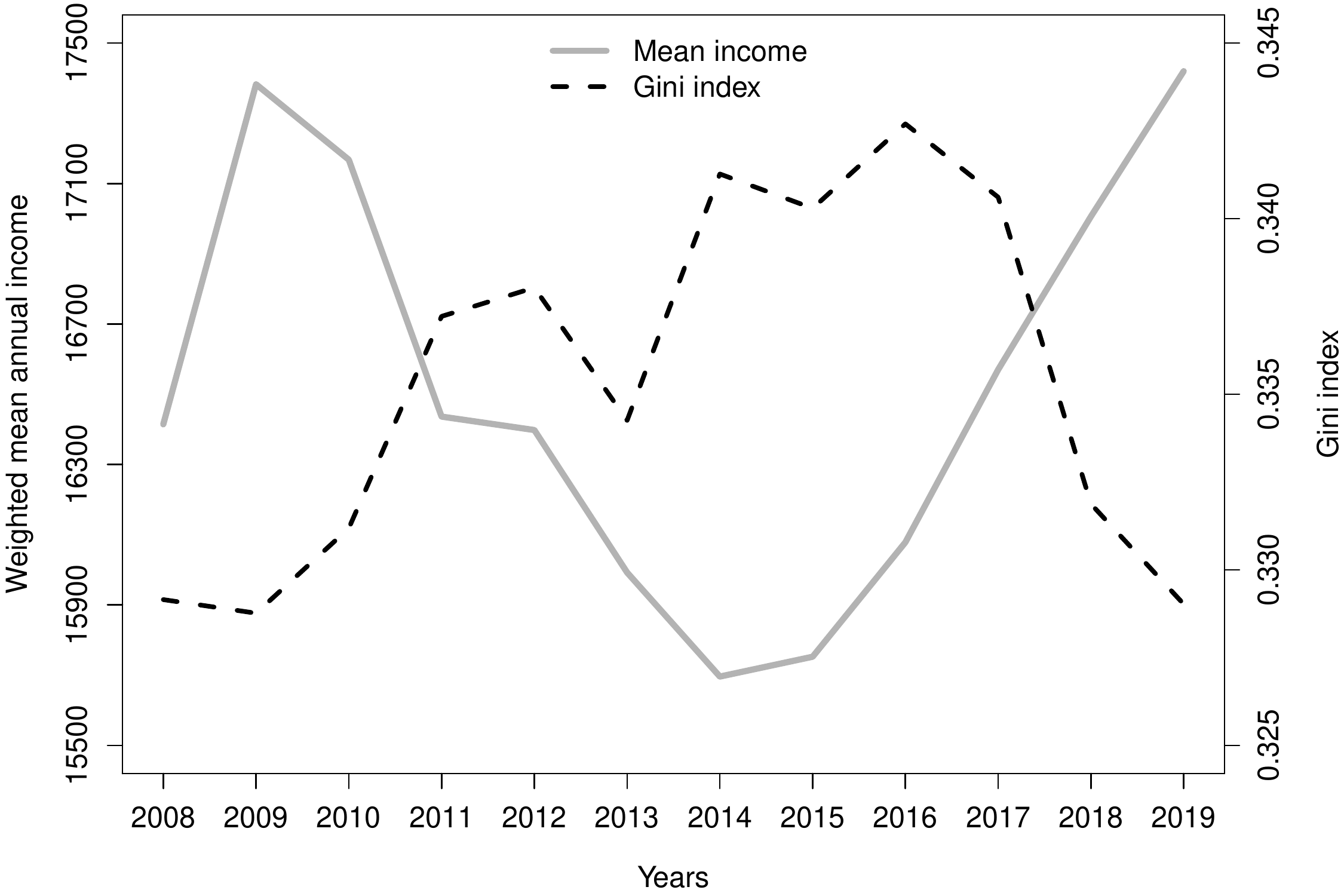}
\end{center}
\caption{The evolution of the mean income and the Gini coefficient in Spain between the years 2008 and 2019.}
\label{fi:Spain0819IIncomeGini}
\end{figure}

To this end, we have computed the estimation \(\hat{\mathcal I}\) of the bidimensional inequality index (\ref{I0}) and one of its normalized versions \(\hat{\mathcal I}_*\) for the equivalised disposable income in Spain in 2008 (\(X_1\)) and in any of the years in the span 2009--2019 (\(X_2\)). Observe that the index \(\hat{\mathcal I}_*\) separates the points more than \(\hat{\mathcal I}\), especially those with similar Gini coefficient (near the vertical axis). The resulting indices (see Figure~\ref{fi:Spain0819I}) show the devastating effects of the crisis on the distribution of income. From 2011 to 2017 the Lorenz curves of the corresponding years were either on the frontier $L_1$ (years 2011 and 2016) or very near it, meaning the curves were strictly ordered $\ell_1\geq\ell_2$ (or almost so) and income was distributed more equitably in 2008 than in 2011 or 2016. The curve $\ell_2$ for the rest of the years from 2011 to 2017 is below $\ell_1$ except for a rightmost interval contained in [0.8,1] where $\ell_1(t)<\ell_2(t)$ (see the plots of all the Lorenz curves and their scaled differences with respect to that of 2008 in the Supplementary Material).
Income distribution in 2008 is therefore {\em almost} more equitable than that of the years 2012, 2013\ldots (in the sense defined by \cite{Zheng2018}).
This would support the generalized social perception that the 2008 crisis in Spain stroke not only the lowest income class but also the middle class (see \cite{Alonso-etal-2017}), broadening the gap between both groups and the richest (last income decile).
The bidimensional indices corresponding to years 2018 and 2019 are remarkably near the index of 2009, meaning that income distribution was slowly approaching the pre-crisis level. Note that, although the Gini indices are nearly the same in 2008 and 2019, the Lorenz curves are still not coincident as the value of the index (between 2008 and 2019) does not lie in $(0,0)$.
In 2019, the poorest half of society has a lower proportion of the total income than in 2008 (see the Supplementary Material).

%

\begin{figure}[h]
\begin{center}
\begin{tabular}{c}
\includegraphics[width=11cm]{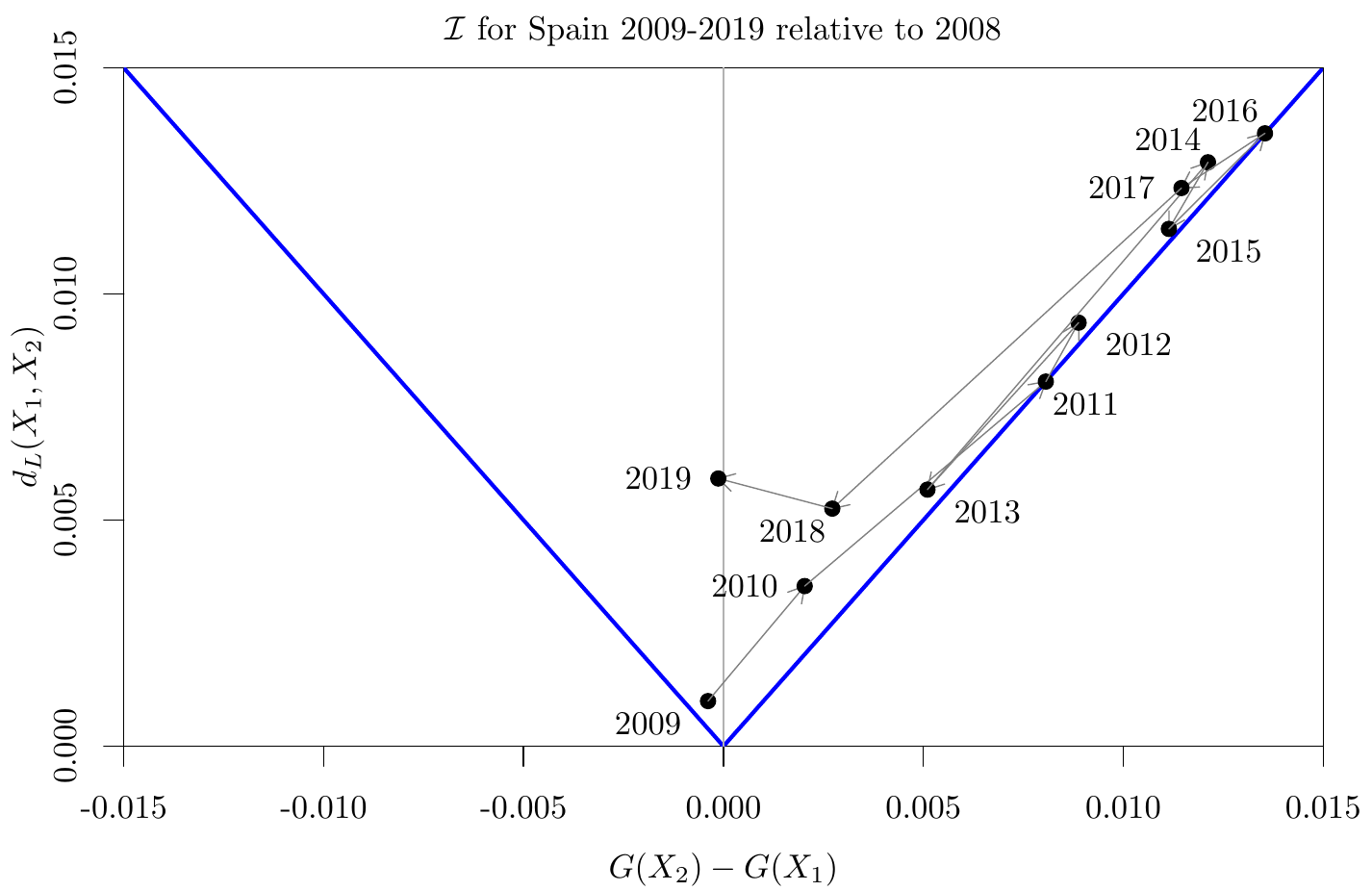} \\
\includegraphics[width=11cm]{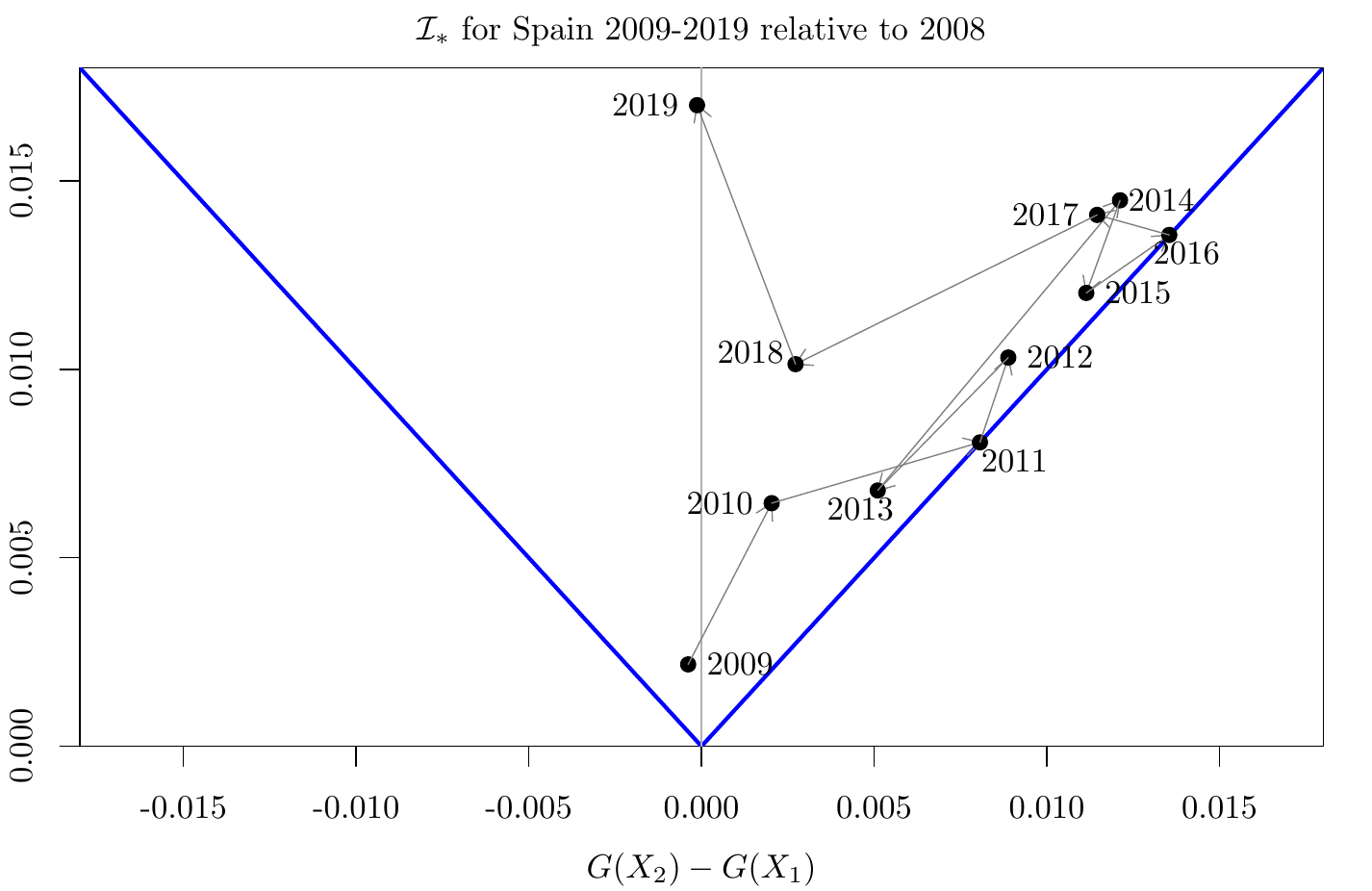}
\end{tabular}
\end{center}
\caption{The bidimensional inequality index \(\mathcal I\) (top figure) and \(\mathcal{I}_*\) (bottom) based on Spanish income data for the period 2009--2019 compared to 2008.}
\label{fi:Spain0819I}
\end{figure}

\subsection{Comparing inequality between Finland and Greece} \label{Subsection.RealData.FinGr}

Another collateral effect of the 2007--08 world financial crisis was the abrupt deterioration of the Greek sovereign-debt crisis. In 2009 the newly elected Greek government announced that its predecessor had underreported national debt levels and deficits. The consequent loss of confidence in the Hellenic economy, its structural weaknesses and other problems such as tax evasion triggered a chain reaction: the increase of national bond yields and the recession resulted in the downgrading of Greek bonds to junk status and a threat of sovereign default in 2010. Successive international bailout loan programs (in 2010, 2012 and 2015) came at the cost of severe austerity measures in Greece. As a result, a huge number of businesses were bankrupt and the unemployment rate rose without control, thus entering a spiral of economic implosion and population impoverishment. Surprisingly, the effects of this deep and prolonged crisis on inequality in Greece were not as dramatic as one could expect (see the evolution of the Greek Gini index in Figure~\ref{IncomeGiniELFI0419} and the Supplementary Material). As noted by \cite{Mitrakos-2014}, Greece already entered the crisis with a high level of income inequality and, also, the thousands of people who ended up homeless as a result of the crisis were not part of the Household Budget Surveys. The effects of the austerity measures are noticeable in the sharp decline of the household disposable income (see Figure~\ref{IncomeGiniELFI0419} and the Supplementary Material). In 2018 Greece exited the last of the bailouts, still owing a debt-to-GDP ratio of more than 150\%.

In contrast, Finnish economy has been growing steadily since the country joined the euro zone and is stable, diversified and competitive. The negative effects of the 2007--08 world crisis on Finnish economy were not severe. Income inequality is among the lowest in the EU (see Figure~\ref{IncomeGiniELFI0419}).
We compare the evolution of the bidimensional inequality index \(\mathcal I\) between these two extreme countries of the EU. Our aim is to check the ability of the index to reflect that the Greek and Finnish income distributions are ordered (or almost so) in all the years of the period. Indeed, in Figure~\ref{fi:GreeceFinland0419I} we can see that the index $\mathcal I$ is mainly on the left frontier $L_2$ of the region $\Delta$ or extremely close to it, that is, Finland (almost) uniformly distributes income more evenly than Greece. However, observe that in 2018 and 2019 the distance between the two countries has greatly diminished indicating an improvement in the Greek distribution of income.

In the Supplementary Material we have compared Greece with Portugal, a country facing economic problems well before the world crisis of 2007--08 and whose income inequality was even greater than that of Greece at the start of the crisis. In this case, the bidimensional index is also lying on the $L_1$ frontier of $\Delta$ (or very near it).

\begin{figure}[h]
\begin{center}
\includegraphics[width=12cm]{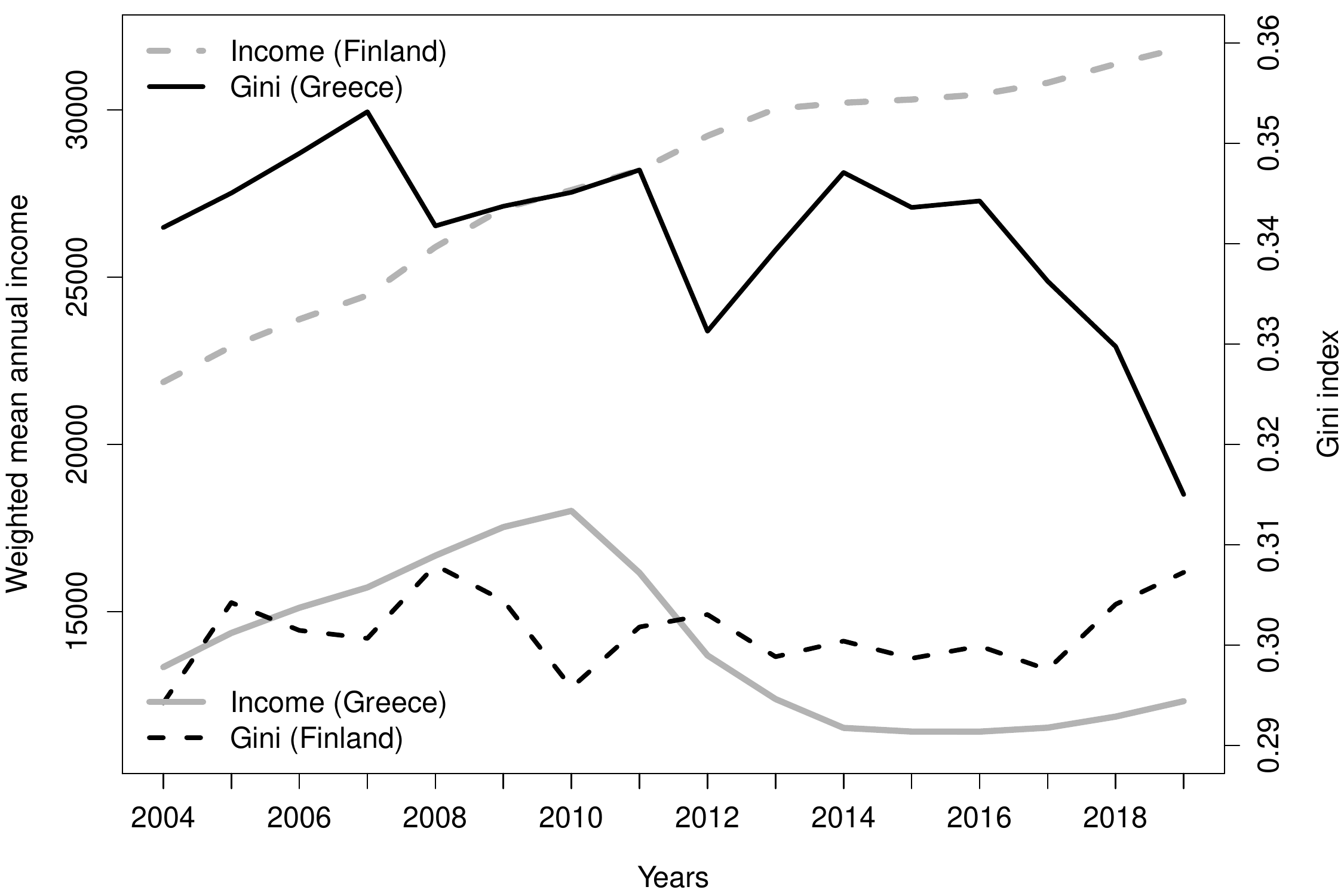}
\end{center}
\caption{Weighted mean annual income and Gini index in Greece and Finland from 2004 to 2019.}
\label{IncomeGiniELFI0419}
\end{figure}

\begin{figure}[h]
\begin{center}
\includegraphics[width=13cm]{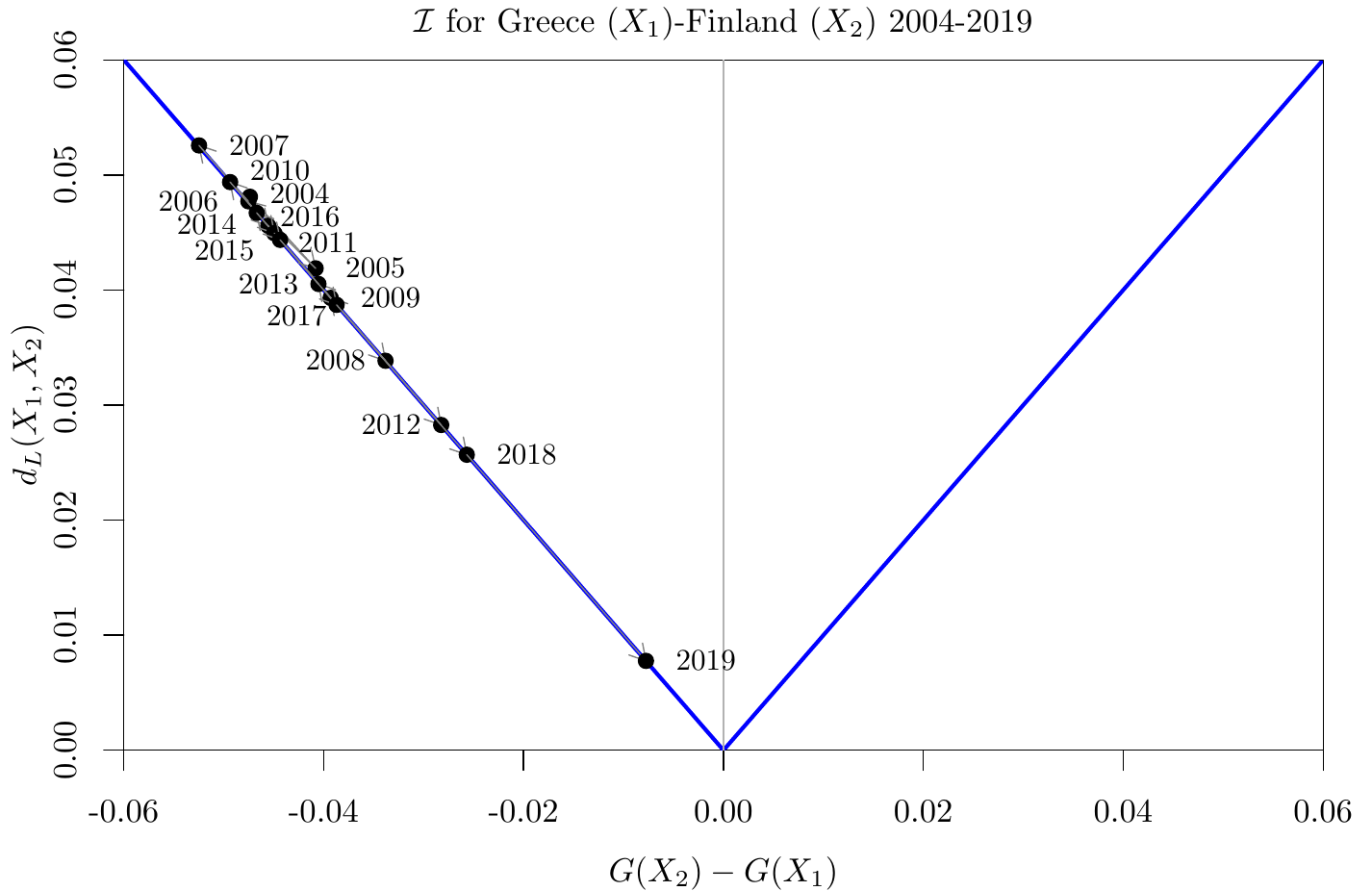}
\end{center}
\caption{Bidimensional inequality index $\mathcal I$ for income data from Greece ($X_1$) and Finland ($X_2$).}
\label{fi:GreeceFinland0419I}
\end{figure}

\subsection{Relative inequality between Spain and Portugal} \label{Subsection.RealData.ESPT}

In Figures \ref{IncomeGiniESPT0819} and \ref{fi:SpainPortugal0819I} we can see the relative evolution of inequality in Spain and Portugal from 2008 to 2018. Observe that in 2008 the inequality index \(\hat{\mathcal I}\) was almost on the right frontier of the region $\Delta$, indicating that income distribution in Portugal was nearly ordered with respect to that of Spain (Spain uniformly distributed income better than Portugal).
But, as the crisis struck in Spain, the Portuguese economy cut the distance with the Spanish one and the index \(\hat{\mathcal I}\) moved towards the vertical line of equal Gini coefficients, reaching it in 2018. Further, in 2019 the index moves to the left of this line, implying a fairer distribution of income in Portugal than in Spain according to the Gini index (see also the Lorenz curves of both countries in 2008 and 2019 in the Supplementary Material).

\begin{figure}[h]
\begin{center}
\includegraphics[width=12cm]{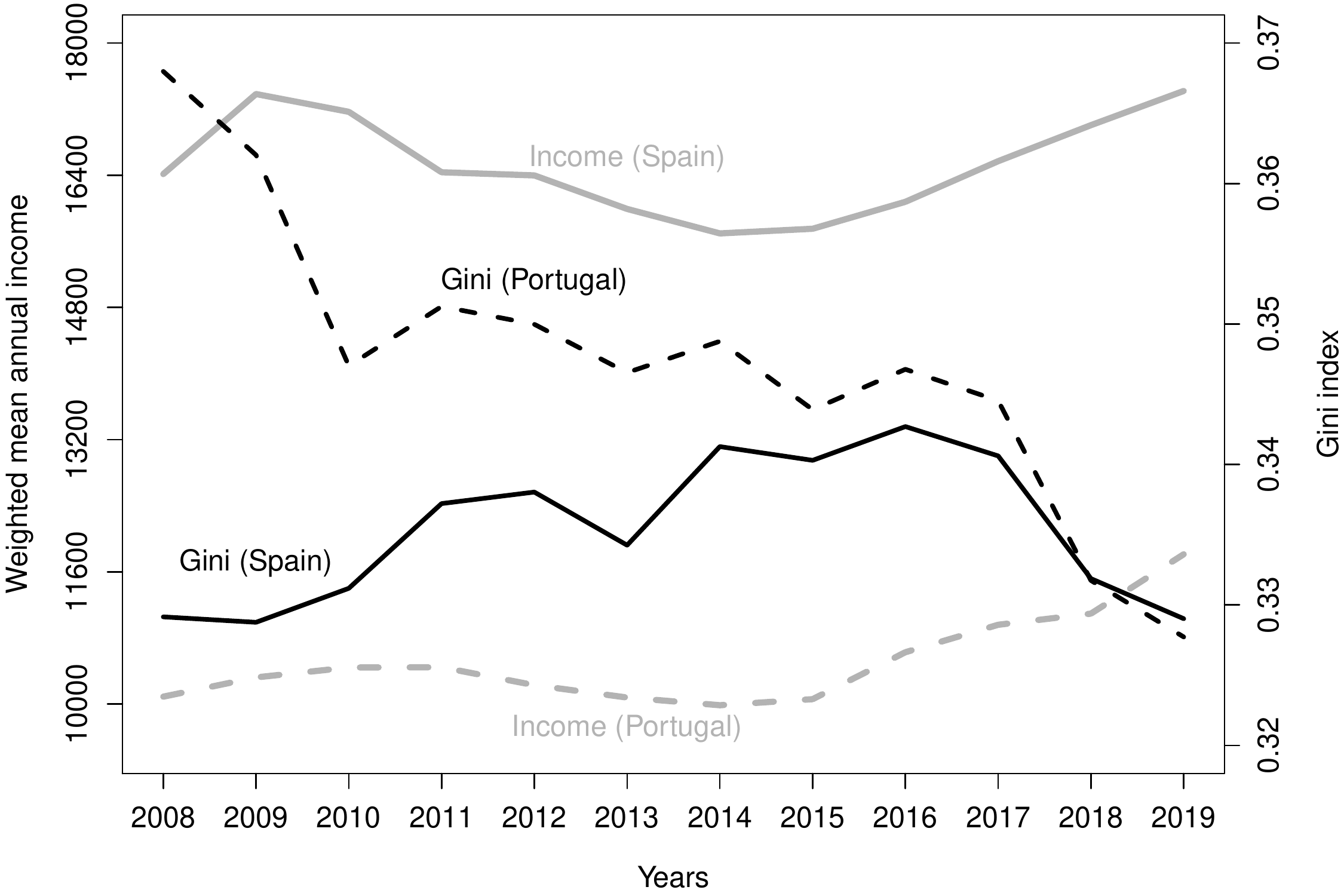}
\end{center}
\caption{Weighted mean annual income and Gini index in Spain and Portugal from 2008 to 2019.}
\label{IncomeGiniESPT0819}
\end{figure}

\begin{figure}[h]
\begin{center}
\includegraphics[width=13cm]{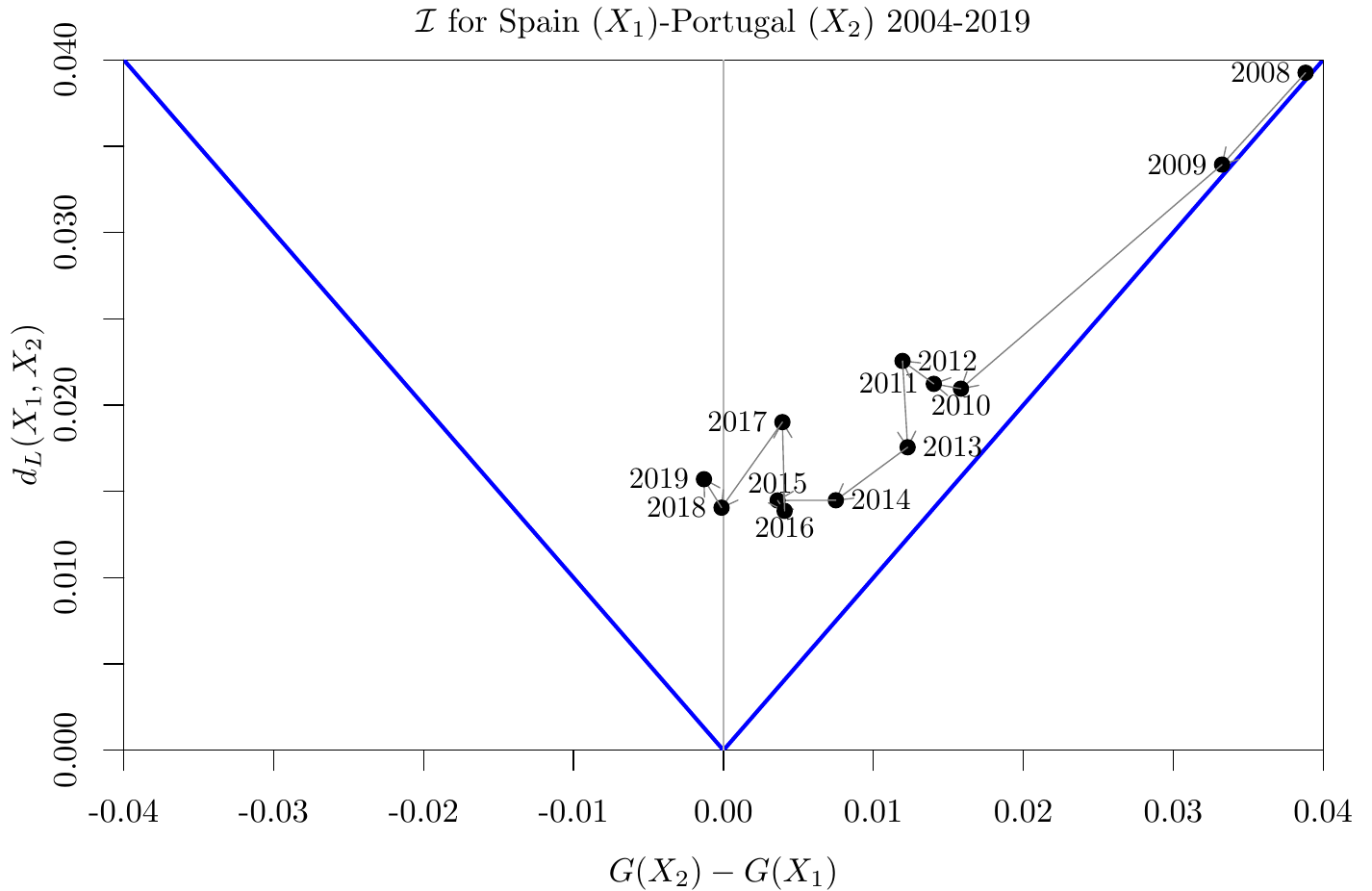}
\end{center}
\caption{Bidimensional inequality index $\mathcal I$ for income data from Spain ($X_1$) and Portugal ($X_2$).}
\label{fi:SpainPortugal0819I}
\end{figure}

\subsection{Relative inequality between France and Germany} \label{Subsection.RealData.FRDE}

We examine an example of two countries, Germany and France, whose relative inequality has great variations as reflected in the position of the bidimensional index \(\mathcal I\) (see Figure~\ref{fi:GermanyFrance0519I}). The value and the evolution of the Gini index is heavily dependent on the variable (e.g., disposable equivalised household income or personal labour income) under study (see \citet{Battisti-Felbermayr-Lehwald-2014}). The subject of growing inequality in Germany has been a matter of interesting discussions (see, e.g., \citet{Dao-2020}): corporate investments and assets revenue benefit the richest and have widened top income inequality; a decrease of unemployment has increased the variability and range of wages. The German reunification caused an inequality increase, but the Gini index has been stable since the mid-2000s (see Figure~\ref{IncomeGiniDEFR0519}). Inequality in France is a matter of great concern. Extensive redistribution of wealth and income through taxes and social transfers is carried out with the aim of correcting poverty and reducing income disparities. This explains the tendency of the bidimensional index $\mathcal I$ to be in the left half of the region $\Delta$, indicating a more unequal distribution of income in Germany than in France.

\begin{figure}[h]
\begin{center}
\includegraphics[width=12cm]{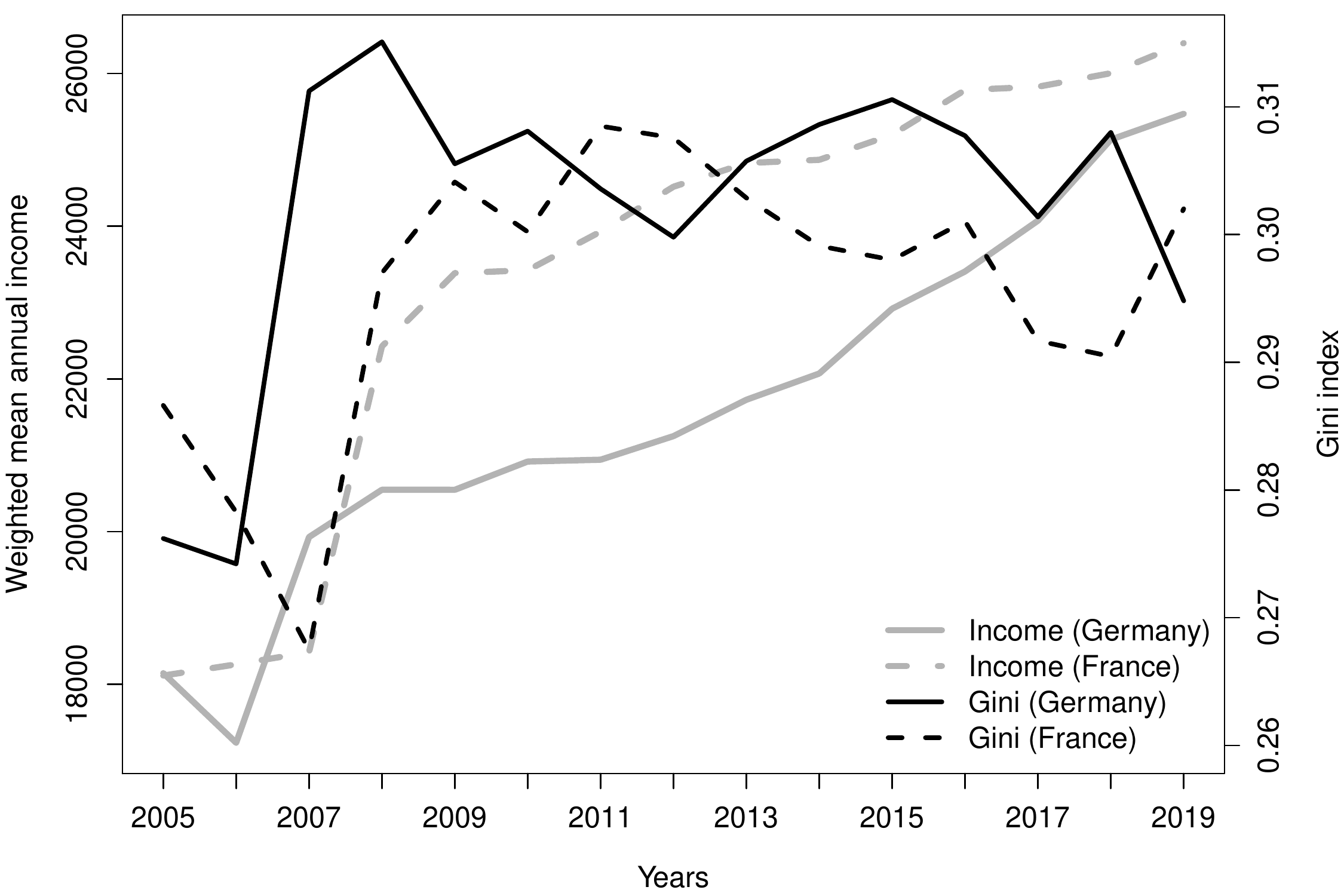}
\end{center}
\caption{Weighted mean annual income and Gini index in Germany and France from 2005 to 2019.}
\label{IncomeGiniDEFR0519}
\end{figure}

\begin{figure}[h]
\begin{center}
\includegraphics[width=13cm]{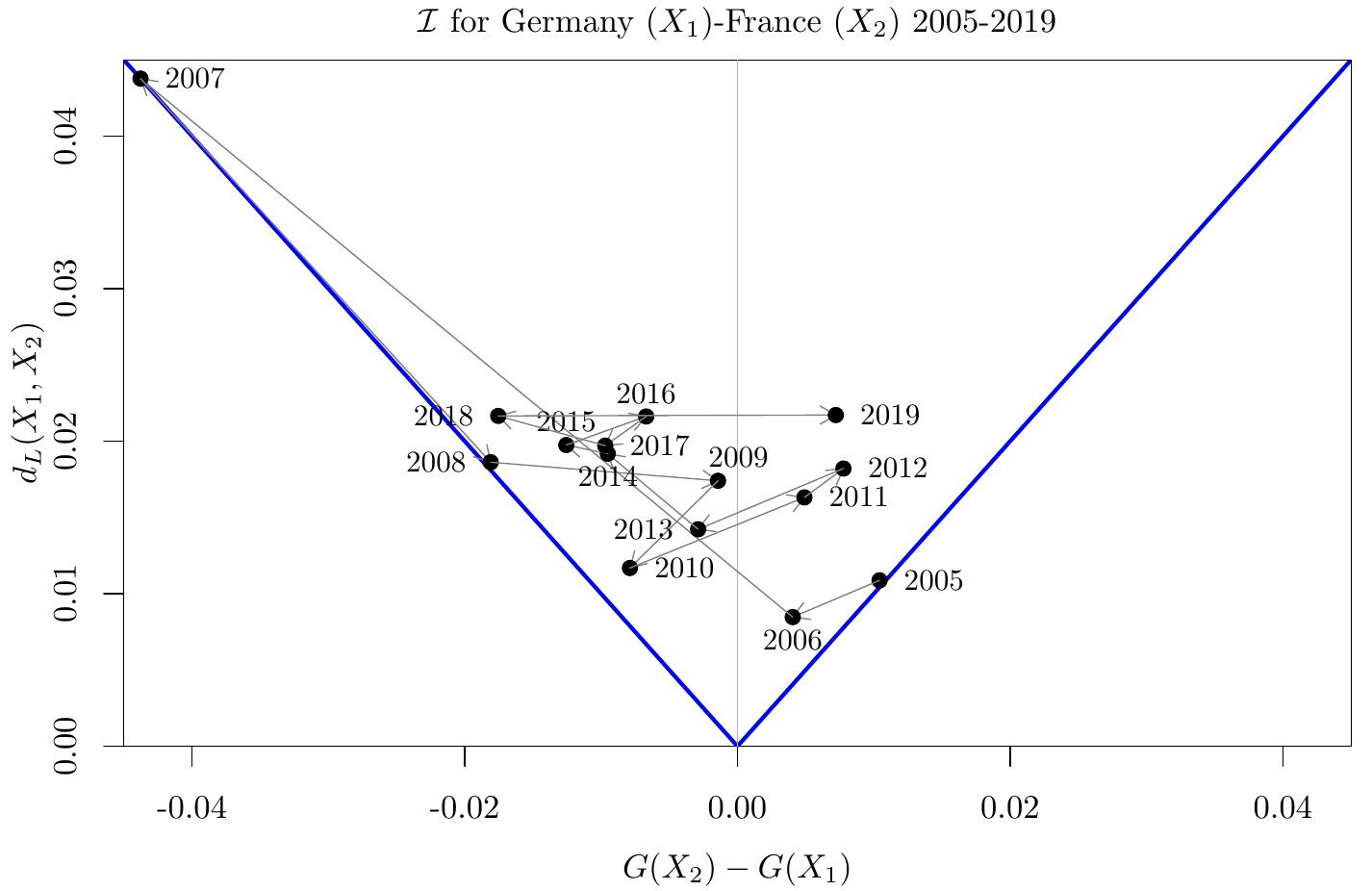}
\end{center}
\caption{Bidimensional inequality index $\mathcal I$ for income data from Germany ($X_1$) and France ($X_2$).}
\label{fi:GermanyFrance0519I}
\end{figure}

\section{Appendix}\label{Section.Appendix}

Here we collect the proofs of the main results stated in Sections \ref{Section.Extremal}--\ref{Section-Asymptotics}. First, we enumerate below some regularity properties of the functions in the set $\mathcal{L}$ defined in \eqref{L} that will be useful throughout the appendix.

\subsection{Regularity properties and compactness of $\mathcal{L}_a$}

We start with a slight change in the definition of the functions in the set $\mathcal{L}$ of \eqref{eq:L}.
Given $\ell \in \mathcal{L}$ we redefine the value of $\ell$ at $1$ as $\ell (1) = \sup_{[0,1)} \ell$.
This redefinition is motivated by the fact that, as shown in the following proposition, functions in $\mathcal{L}$ become continuous in $[0,1]$.
In addition, the convexity of $\mathcal{L}$ and $\mathcal{L}_a$ remains true.
With this definition, $\mathcal{L}$ becomes the set of convex $\ell:[0,1]\to[0,1]$ such that $\ell(0)=0$ and $\ell (1) = \sup_{[0,1)} \ell$.

In the following proposition we denote by $W^{1,1} (0,1)$ the Sobolev space $W^{1,1}$ in the interval $(0,1)$, which is equivalent to the set of absolutely continuous functions in $[0,1]$; it is endowed with the norm
\[
 \left\| \ell \right\|_{W^{1,1} (0,1)} = \left\| \ell \right\| + \left\| \ell' \right\| ,
\]
where $\ell'$ is the distributional derivative of $\ell$, which coincides a.e.\ with the derivative of $\ell$.
For $\alpha \in (0,1)$ we denote by $W^{1,\infty} (0, \alpha)$ the Sobolev space $W^{1,\infty}$ in the interval $(0,\alpha)$, which is equivalent to the set of Lipschitz continuous functions in $[0, \alpha]$; it is endowed with the norm
\[
 \left\| \ell \right\|_{W^{1,\infty} (0,\alpha)} = \sup_{(0,\alpha)} \left| \ell \right| + \esssup_{(0,\alpha)} \left| \ell' \right| .
\]
See, e.g., \citet[Chapter 8]{Brezis-2011} or \citet[Chapter 4]{Evans-Gariepy-1992} for the definition and properties of these spaces.

\begin{proposition}\label{pr:LaContinuous}
Let $a \in [0,1]$ and $\ell \in \mathcal{L}_a$.
Then
\begin{enumerate}[(a)]

\item The function $\ell$ is non-decreasing, absolutely continuous in $[0,1]$ and Lipschitz in $[0, \alpha]$ for each $\alpha \in (0,1)$.
Moreover,
\[
 \| \ell \|_{W^{1,1} (0,1)} \leq 1 + \frac{1-a}{2} ,
\]
and for each $\alpha \in (0,1)$,
\[
 \| \ell \|_{W^{1,\infty} (0, \alpha)} \leq 1 + \frac{1-a}{(1-\alpha)^2} .
\]

\item The function $\ell'$ is locally of bounded variation, the right derivative $\ell'(x^+)$ exists for all $x \in [0, 1)$ and is non-decreasing. Moreover, $\ell' (0^+) \geq 0$.

\item $\ell''$ is a non-negative Radon measure.

\end{enumerate}
\end{proposition}
\begin{proof}
Convex functions are locally Lipschitz (see \citet[Theorem 6.3.1]{Evans-Gariepy-1992} or \citet[Theorem 1.19]{Simon-2011}), have a first derivative locally of bounded variation (see \citet[Theorem 6.3.3]{Evans-Gariepy-1992}) and have a second derivative in the sense of distributions (see \citet[Theorem 6.3.2]{Evans-Gariepy-1992} or \citet[Theorem 1.29]{Simon-2011}), which in fact is a non-negative Radon measure.
Further, the right derivative $\ell' (x^+)$, exists for all $x\in[0,1)$ and is non-decreasing (see \citet[Theorem 1.26]{Simon-2011}).
As $\ell (0) = 0$ and $\ell \geq 0$, we necessarily have that $\ell' (0^+)\geq 0$, so $\ell'(x^+) \geq 0$ for all $x \in [0, 1)$.
By the version of the fundamental theorem of calculus for convex functions (see  \citet[Theorem 1.28]{Simon-2011}), $\ell$ is non-decreasing.
In addition, the derivative of $\ell$ exists a.e.\ and coincides a.e.\ with the right derivative, so $\ell' \geq 0$ a.e.
In particular,
\[
 \left\| \ell' \right\|= \int_{0}^{1}  \ell' (t)\, \d t = \ell(1)-\ell(0)\leq 1.
\]
We conclude that $\| \ell \|_{W^{1,1} (0,1)} \leq 1 + \frac{1-a}{2}$.

On the other hand, we observe that the affine function $s : [0,1] \to \R$ given by $s(x)=\ell'(\alpha^+) (x-\alpha)+\ell(\alpha)$ is a supporting line of $\ell$ at the point $(\alpha,\ell(\alpha))$.
By convexity, we hence have that $s\leq \ell$ and then,
\begin{align*}
   \frac{1-a}{2} = \left\| \ell \right\| & \geq \int_{\alpha}^1 \ell (t) \, \d t\geq \int_{\alpha}^{^1}  s(t)\,\d t
      = \ell'(\alpha^+)\frac{(1-\alpha)^2}{2} +\ell(\alpha) (1-\alpha)
       \geq   \ell'(\alpha^+)\frac{(1-\alpha)^2}{2}.
  \end{align*}
 We conclude that
\[
 \esssup_{(0, \alpha)} \ell' \le \ell'(\alpha^+) \le \frac{1-a}{(1-\alpha)^2} .
\]
Consequently, $\| \ell \|_{W^{1,\infty} (0, \alpha)} \leq 1 + \frac{1-a}{(1-\alpha)^2}$.
\end{proof}


We are now ready to prove that $\mathcal{L}_a$ is a compact set of $L^1$.

\medskip

\textbf{Proof of Proposition \ref{pr:LaCompact} in Section \ref{Section.Extremal} (Compactness of $\mathcal{L}_a$).}
The convexity of $\mathcal{L}_a$ is straightforward. Let us prove its compactness.
 Let $\{ \ell_n \}_{n \in \N}$ be a sequence in $\mathcal{L}_a$.
By Proposition \ref{pr:LaContinuous}, $\{ \ell_n \}_{n \in \N}$ is bounded in $W^{1,1} (0,1)$, so by the Rellich--Kondrachov theorem (see \citet[Theorem 8.8]{Brezis-2011}), there exists a subsequence (not relabelled) and an $\ell \in L^1$ such that $\ell_n \to \ell$ in $L^1$ as $n \to \infty$. This also implies that $G(\ell)=a$.

On the other hand, for each $\alpha \in (0, 1)$ we have by Proposition \ref{pr:LaContinuous} (a) that $\{ \ell_n \}_{n \in \N}$ is bounded in $W^{1, \infty} (0, \alpha)$, so by the Ascoli--Arzel\`a theorem (see \citet[Theorems 4.25 and 8.8]{Brezis-2011}), for a further subsequence, $\ell_n \to \ell$ uniformly in $[0, \alpha]$ as $n \to \infty$.
In particular, $\ell (0) = 0$ and $0 \leq \ell \leq 1$ in $[0, \alpha]$.
As the pointwise limit of convex function is a convex function, we obtain that $\ell$ is convex in $[0, \alpha]$.
Therefore, $0 \leq \ell \leq 1$ in $[0, 1)$ and $\ell$ is convex in $[0,1)$.
We redefine $\ell(1)$ as $\ell(1) = \sup_{[0,1)}$, so that $\ell$ becomes continuous in $[0,1]$.
We also obtain that $0 \leq \ell \leq 1$ in $[0, 1]$ and $\ell$ is convex in $[0,1]$. Therefore, $\ell\in \mathcal{L}_a$ and the proof is finished.
\hfill $\square$

\subsection{Proof of Theorem \ref{th:extLa} in Section \ref{Section.Extremal} (Extreme points of $\mathcal{L}_a$)}

We will use an alternative description of the elements in $\mathcal{L}_a$ in terms of positive measures concentrated on the interval $(0,1)$. The main idea is based on the following fact: any curve $\ell \in \mathcal{L}_a$ is univocally determined by its second derivative, $\ell''$, together with the conditions $\ell(0) = 0$ and $G(\ell)=a$ (or, equivalently, $\| \ell \| = \frac{1-a}{2}$).

Given $a \in [0,1]$, we denote by $\mathcal{M}_a$ be the set of non-negative Radon measures $\mu$ concentrated on the interval $(0,1)$ and such that
\begin{equation}\label{eq:restrictionsMa}
 \int_0^1 (1-s)^2 \, \d \mu (s) \leq 1-a \quad \text{and} \quad \int_0^1 s (1-s) \, \d \mu (s) \leq a .
\end{equation}

\begin{proposition}\label{pr:iso}
For $a \in [0,1]$, the map $T_a : \mathcal{L}_a \to \mathcal{M}_a$ defined by $T_a (\ell) = \ell''$ is an affine isomorphism with inverse $T_a^{-1} : \mathcal{M}_a \to \mathcal{L}_a$ given by
\begin{equation}\label{eq:MaLa}
T_a^{-1} \mu (x)= \left[ 1 - a - \int_0^1 (1-s)^2 \, \d \mu (s) \right] x + \int_0^x (x-s) \, \d \mu (s) , \quad x \in [0,1] .
\end{equation}
\end{proposition}
\begin{proof}
First we see that the map $T_a$ is well defined.
Given $\ell \in \mathcal{L}_a$, we have from Proposition \ref{pr:LaContinuous} that $\ell''$ is a non-negative Radon measure.
As $\ell'$ is locally of bounded variation,
\[
 \ell'(t) = \ell'(0^+) + \int_0^t \, \d \ell'' (s) , \qquad \text{a.e. } t \in (0,1) .
\]
As $\ell$ is locally Lipschitz and $\ell (0) = 0$,
for $x \in [0,1]$, we have that
\begin{equation}\label{eq:ellx1}
\begin{split}
    \ell (x)  & = \int_0^x \ell'(t) \, \d t = \int_0^x \left[ \ell'(0^+) + \int_0^t \, \d \ell'' (s) \right] \d t  \\
       & = \ell'(0^+) x + \int_0^x (x-s) \, \d \ell'' (s),
\end{split}
\end{equation}
where for the last equality we have used Fubini's theorem.
Integrating in $x \in (0,1)$ equality \eqref{eq:ellx1} (and by Fubini's theorem again) we obtain the restriction
\begin{equation}\label{eq:ellx2}
 \frac{1-a}{2} = \| \ell \| =  \frac{\ell'(0^+)}{2} + \frac{1}{2} \int_0^1 (1-s)^2 \, \d \ell''(s).
\end{equation}
As $\ell'(0^+) \geq 0$, from \eqref{eq:ellx2} we directly obtain the first inequality of \eqref{eq:restrictionsMa}.
On the other hand, \eqref{eq:ellx1} and \eqref{eq:ellx2} show that
\begin{equation}\label{eq:ellell''}
 \ell (x) = \left[ 1 - a - \int_0^1 (1-s)^2 \, \d \ell''(s) \right] x + \int_0^x (x-s) \, \d \ell'' (s) , \quad x \in [0,1].
\end{equation}
Hence, imposing $\ell(1) \leq 1$ we have the second inequality of \eqref{eq:restrictionsMa}.

Now, for $\mu \in \mathcal{M}_a$, we define $\ell$ as in the right-hand side of \eqref{eq:MaLa} and we will check that $\ell \in \mathcal{L}_a$. First, $\ell (0) = 0$ and
\[
 \int_0^1 \ell (x) \, \d x = \frac{1}{2} \left[ 1 - a - \int_0^1 (1-s)^2 \, \d \mu (s) \right] + \int_0^1 \int_s^1 (x-s) \, \d x \, \d \mu (s) = \frac{1-a}{2} .
\]
Further, thanks to the second inequality of \eqref{eq:restrictionsMa}, we obtain that
\[
 \ell(1) = 1 - a + \int_0^1 s (1-s) \, \d \mu (s) \leq 1.
\]
By Leibniz integral rule (differentiation under the integral sign), it can also be checked that
\begin{equation}\label{eq:ell'x}
 \ell' (x) =  1 - a - \int_0^1 (1-s)^2 \, \d \mu (s) + \int_0^x \, \d \mu (s) , \quad \text{a.e. } x \in (0,1)
\end{equation}
and hence
\begin{equation}\label{eq:ell''mu}
 \ell'' = \mu \quad \text{as measures.}
\end{equation}
As $\mu$ is positive, from \eqref{eq:ell'x} we have that $\ell'$ is essentially non-decreasing, $\ell$ is convex and
\[
 \ell' (x) \geq 1 - a - \int_0^1 (1-s)^2 \, \d \mu (s) \geq 0 , \quad \text{a.e. } x \in (0,1) ,
\]
by the first inequality of \eqref{eq:restrictionsMa}, so $\ell$ is non-decreasing.
In particular, $\ell \geq 0$.
This shows that $\ell \in \mathcal{L}_a$.

Finally, we prove that the maps $T_a$ and \eqref{eq:MaLa} are mutually inverse.
Given $\ell \in \mathcal{L}_a$, if we apply first $T_a$ and then \eqref{eq:MaLa} we get $\ell$ back thanks to \eqref{eq:ellell''}.
Conversely, given $\mu \in \mathcal{M}_a$, if we apply first \eqref{eq:MaLa} and then $T_a$ we recover $\mu$ by \eqref{eq:ell''mu}. Since $T_a$ is affine, the proof is concluded.
\end{proof}

Next we calculate $\ext (\mathcal{M}_a)$. We denote by $\delta_x$ the Dirac measure at $x \in [0,1]$.

\begin{proposition}\label{pr:extMa}
For $a \in [0,1]$, we have that
\begin{align*}
 \ext (\mathcal{M}_a) = & \{ 0 \} \cup \left\{ \frac{1-a}{(1-x_1)^2} \delta_{x_1} : x_1 \in (0, a] \right\} \cup \left\{ \frac{a}{x_2 (1-x_2)} \delta_{x_2} : x_2 \in (a, 1)  \right\} \\
 & \cup \left\{ \frac{x_2-a}{(1-x_1) (x_2-x_1)} \delta_{x_1} + \frac{a-x_1}{(x_2-x_1)(1-x_2)} \delta_{x_2} : \, x_1 \in (0,a) , \, x_2 \in (a,1) \right\}.
\end{align*}
\end{proposition}

\begin{proof}
The proof is divided into several smaller results.

\smallskip

\emph{Step 1: The null measure $\mu\equiv 0 \in \ext (\mathcal{M}_a)$.}

This is direct as all the measures in $\mathcal{M}_a$ are non-negative.

\smallskip

\emph{Step 2: For all $x_1 \in (0, a]$, the measure $\mu = \frac{1-a}{(1-x_1)^2} \delta_{x_1} \in \ext(\mathcal{M}_a)$.}

Clearly, $\mu \in \mathcal{M}_a$.
Assume that $\mu =  t_1 \mu_1 + t_2 \mu_2$, for some $t_1, t_2 >0$ with $t_1 + t_2 = 1$ and $\mu_1, \mu_2 \in \mathcal{M}_a$.
Then $\mu_i = \beta_i \delta_{x_1}$ and, due to \eqref{eq:restrictionsMa},
\[
 \beta_i \geq 0, \quad \beta_i (1-x_1)^2 \leq 1-a , \quad \beta_i x_1 (1-x_1) \leq a , \qquad i=1,2 .
\]
Thus,
\[
 \frac{1-a}{(1-x_1)^2} = t_1 \beta_1 + t_2 \beta_2 \leq t_1 \frac{1-a}{(1-x_1)^2} + t_2 \frac{1-a}{(1-x_1)^2} = \frac{1-a}{(1-x_1)^2} ,
\]
so, necessarily, $\beta_1 = \beta_2 = \frac{1-a}{(1-x_1)^2}$, and, hence, $\mu_1 = \mu_2$.
Therefore, $\mu \in \ext (\mathcal{M}_a)$.

\smallskip

\emph{Step 3: For all $x_2 \in (a, 1)$, the measure $\frac{a}{x_2 (1-x_2)} \delta_{x_2} \in \ext(\mathcal{M}_a)$.}

The proof is similar to the one of the previous step and it is therefore omitted.

\smallskip

\emph{Step 4: If for some $x_1 \in (0, a]$ and $\alpha \in \R \setminus \{ 0, \frac{1-a}{(1-x_1)^2}\}$, $\mu= \alpha \delta_{x_1} \in \mathcal{M}_a$, then $\mu \notin \ext (\mathcal{M}_a)$.}

The fact $\mu \in \mathcal{M}_a$ implies that $0 < \alpha < \frac{1-a}{(1-x_1)^2}$.
Therefore, for $\varepsilon >0$ small enough, we have that $\mu \pm \varepsilon \delta_{x_1} \in \mathcal{M}_a$ since both are positive measures and the restrictions \eqref{eq:restrictionsMa} are satisfied; indeed,
\[
  \int_0^1 (1-s)^2 \, \d (\mu \pm \varepsilon \delta_{x_1}) (s) = (\alpha \pm \varepsilon) (1 - x_1)^2 < 1-a
\]
and
\[
 \int_0^1 s (1-s) \, \d (\mu \pm \varepsilon \delta_{x_1}) = (\alpha \pm \varepsilon) x_1 (1 - x_1) < \frac{1-a}{1-x_1}x_1 \leq a .
\]
Finally, we can write $\mu = \frac{1}{2} (\mu + \varepsilon \delta_{x_1}) + \frac{1}{2} (\mu - \varepsilon \delta_{x_1})$, and, hence, $\mu \notin \ext (\mathcal{M}_a)$.

\smallskip

\emph{Step 5: If for some $x_2 \in (a,1)$ and $\alpha \in \R \setminus \{ 0, \frac{a}{x_2 (1-x_2)}\}$, $\mu = \alpha \delta_{x_2} \in \mathcal{M}_a$, then $\mu \notin \ext (\mathcal{M}_a)$.}

The proof is similar to that of Step 4 and it is left to the reader.

\smallskip

\emph{Step 6: For all $x_1 \in (0,a)$ and $x_2 \in (a,1)$, $\mu = \frac{x_2-a}{(1-x_1) (x_2-x_1)} \delta_{x_1} + \frac{a-x_1}{(x_2-x_1)(1-x_2)} \delta_{x_2} \in \ext(\mathcal{M}_a)$.}

It is immediate to check that
\[
 \int_0^1 (1-s)^2 \, \d \mu (s) = 1-a \quad \text{and} \quad \int_0^1 s (1-s) \, \d \mu (s) = a.
\]
Therefore, $\mu \in \mathcal{M}_a$. Moreover, if $\mu =  t_1 \mu_1 + t_2 \mu_2$ for some $t_1, t_2 >0$ with $t_1 + t_2 = 1$ and $\mu_1, \mu_2 \in \mathcal{M}_a$, then
\[
 \int_0^1 (1-s)^2 \, \d \mu_i (s) = 1-a \quad \text{and} \quad \int_0^1 s (1-s) \, \d \mu_i (s) = a , \qquad i=1,2 .
\]
Furthermore, as $\mu_i = \sum_{j=1}^2 \beta_{ij} \delta_{x_j}$ for some $\beta_{ij} \geq 0$ (for $i,j = 1,2$), we have that
\[
 \sum_{j=1}^2 \beta_{ij} (1-x_j)^2 = 1-a \quad \text{and} \quad \sum_{j=1}^2 \beta_{ij} x_j (1-x_j) = a , \qquad i=1,2 ,
\]
and, hence,
\[
 \beta_{i1} = \frac{x_2-a}{(1-x_1) (x_2-x_1)} , \quad \beta_{i2} = \frac{a-x_1}{(x_2-x_1)(1-x_2)} , \qquad i=1,2 .
\]
Therefore, $\mu_1 = \mu_2$.
Consequently, $\mu \in \ext (\mathcal{M}_a)$.

\smallskip

\emph{Step 7: If for some $\alpha_1, \alpha_2 >0$ and $0 < x_1 < x_2 < 1$ with
\[
 \alpha_1 \neq \frac{x_2 - a}{(1-x_1) (x_2 - x_1)} \quad \text{or} \quad \alpha_2 \neq \frac{a - x_1}{(x_2 - x_1) (1-x_2)}
\]
the measure $\mu = \sum_{i=1}^2 \alpha_i \delta_{x_i}\in \mathcal{M}_a$,  then $\mu \notin \ext (\mathcal{M}_a)$.}

By \eqref{eq:restrictionsMa}, we have that
\begin{equation}\label{eq:rest2deltas}
 \sum_{i=1}^2  \alpha_i (1-x_i)^2 \leq 1-a \quad \text{and} \quad \sum_{i=1}^2 \alpha_i x_i (1-x_i) \leq a .
\end{equation}
If both inequalities in \eqref{eq:rest2deltas} were equalities, we necessarily have that
\[
 \alpha_1 = \frac{x_2 - a}{(1-x_1) (x_2 - x_1)} \quad \text{and} \quad \alpha_2 = \frac{a - x_1}{(x_2 - x_1) (1-x_2)} ,
\]
against our assumption.
Therefore, at least one of the two inequalities of \eqref{eq:rest2deltas} is strict.
If $\sum_{i=1}^2  \alpha_i (1-x_i)^2 < 1-a$, then we consider the signed measure defined by
\[
 \mu_0 = x_2 (1-x_2) \delta_{x_1} - x_1 (1-x_1) \delta_{x_2} .
\]
Then, it is straightforward to check that, for small enough $\varepsilon >0$, $\mu \pm \varepsilon \mu_0 \in \mathcal{M}_a$  and
\begin{equation}\label{eq:mumu0}
 \mu = \frac{1}{2} (\mu + \varepsilon \mu_0) + \frac{1}{2} (\mu - \varepsilon \mu_0) .
\end{equation}
Therefore, $\mu \notin \ext(\mathcal{M}_a)$.

If, instead, $\sum_{i=1}^2  \alpha_i x_i (1-x_i)^2 < a$, we  then consider the signed measure
\[
 \mu_0 = (1-x_2)^2 \delta_{x_1} - (1-x_1)^2 \delta_{x_2} .
\]
Again, we have that, for small enough $\varepsilon >0$, $\mu \pm \varepsilon \mu_0 \in \mathcal{M}_a$ and equality \eqref{eq:mumu0} holds. We conclude that $\mu \notin \ext (\mathcal{M}_a)$.

\smallskip

\emph{Step 8: If $\mu \in \mathcal{M}_a$ is supported in more than two points, then $\mu \notin \ext (\mathcal{M}_a)$.}

In this case, there exist Borel disjoint sets $A_i \subset (0,1)$ such that $\mu|_{A_i} \neq 0$, for $i=1,2,3$.
Let $(\alpha_1, \alpha_2, \alpha_3) \in \R^3 \setminus \{ (0,0,0) \}$ be such that
\[
 \sum_{i=1}^3 \alpha_i \int_{A_i} (1-s)^2 \, \d \mu (s) = 0 \quad \text{and} \quad \sum_{i=1}^3 \alpha_i \int_{A_i} s (1-s)^2 \, \d \mu (s) = 0 .
\]
We consider the signed measure $\mu_0 = \sum_{i=1}^3 \alpha_i \mu|_{A_i}$.
For $\varepsilon >0$ small enough, define $\mu^+$ and $\mu^-$ as $\mu^{\pm} = \mu \pm \varepsilon \mu_0$.
Then $\mu = \frac{1}{2} \mu^+ + \frac{1}{2} \mu^-$ with $\mu^{\pm} \neq \mu$.
Moreover, $\mu^{\pm}$ are positive measures since

$$\mu^{\pm} = \mu|_{A_1^c \cap A_2^c\cap A_3^c} + \sum_{i=1}^3 (1\pm \varepsilon \alpha_i) \mu|_{A_i},$$
where $A^c$ stands for the complement of the set $A$ in $(0,1)$. In fact, $\mu^{\pm} \in \mathcal{M}_a$ since
\[
 \int_0^1 (1-s)^2 \, \d \mu^{\pm} (s) = \int_0^1 (1-s)^2 \, \d \mu (s) \leq 1-a \text{ and } \int_0^1 s (1-s) \, \d \mu^{\pm} (s) = \int_0^1 s (1-s) \, \d \mu (s) \leq a .
\]
Therefore, we conclude that $\mu \notin \ext (\mathcal{M}_a)$.
\smallskip

The eight steps above complete the proof.
\end{proof}

Step 8 of the previous proof is related to the works by \cite{Winkler-1988} and \cite{Pinelis-2016}, where they analyze the set of extreme points of subset of measures defined through some inequalities.

\textbf{Proof of Theorem \ref{th:extLa} in Section \ref{Section.Extremal} (Extreme points of $\mathcal{L}_a$).}  By Proposition \ref{pr:iso}, we have the equality
$$\ext(\mathcal{L}_a) = T_a^{-1} (\ext (\mathcal{M}_a)).$$
Now, we can use Proposition \ref{pr:extMa} to determine the set $\ext(\mathcal{L}_a)$. By \eqref{eq:MaLa} and Proposition \ref{pr:extMa}, we obtain three families of extreme curves in $\mathcal{L}_a$.  First, for $x_1 \in (0,a]$, let $\ell^a_{x_1}=T_a^{-1} \left(\frac{1-a}{(1-x_1)^2} \delta_{x_1}\right)$ and $\ell^a_0=T_a^{-1} \left(0 \right)$. More explicitly, we have that, for $x_1 \in [0,a]$
\[
\ell^a_{x_1} (x)   = \frac{1-a}{(1-x_1)^2} \max \{ 0 , x - x_1 \} , \quad x \in [0,1] .
\]
Second, for $x_2 \in (a, 1)$, we set $m^a_{x_2}=T_a^{-1} \left(\frac{a}{x_2 (1-x_2)} \delta_{x_2}\right)$. We obtain that
\[
 m^a_{x_2} (x)   = \frac{1}{x_2} \left[ (x_2 - a) x + \frac{a}{1-x_2} \max \{ 0 , x - x_2 \} \right], \quad x \in [0,1] .
\]
Finally, for $x_1 \in (0,a)$ and $x_2 \in (a, 1)$, let $n^a_{x_1, x_2} = T_a^{-1} \left( \frac{x_2-a}{(1-x_1) (x_2-x_1)} \delta_{x_1} + \frac{a-x_1}{(x_2-x_1)(1-x_2)} \delta_{x_2}\right)$. In this case we have that
\begin{equation*}
n^a_{x_1, x_2} (x)  =  \frac{1}{x_2 - x_1} \left[ \frac{x_2 - a}{1 - x_1} \max \{ 0 , x - x_1 \} + \frac{a - x_1}{1 - x_2} \max \{ 0 , x - x_2 \} \right] , \quad x \in [0,1] .
\end{equation*}
These curves admit the characterization as piecewise affine functions given in \eqref{eq:extremePW}.

Therefore, the proof of Theorem \ref{th:extLa} is complete. \hfill $\square$

\subsection{Proof of Theorem \ref{Theorem.Mab} in Section \ref{Section.Maximum} (Maximal distance)}

Here we compute the exact value of $M(a,b)$ in \eqref{Mab}. The proof of this theorem is long and we have divided it into several results. It is based on following proposition.

\begin{proposition}\label{pr:Mabextreme}
For $a, b \in [0,1]$, let $M (a,b)$ be defined in \eqref{Mab}. We have that
\begin{equation}\label{eq:Mab-ext}
M(a,b)=\max \{ d_{\rm L} (\ell_1, \ell_2)  : \ell_1 \in\ext (\mathcal{L}_a) \text { and }\ell_2 \in\ext (\mathcal{L}_b) \}.
\end{equation}
\end{proposition}

\begin{proof}
The distance $d_{\rm L} : \mathcal{L}_a \times \mathcal{L}_b \to \R$ is a convex and continuous functional in $L^1$. Further, by Proposition \ref{pr:LaCompact}, the convex sets $\mathcal{L}_a$ and $\mathcal{L}_b$ are  compact in $L^1$. Therefore, by Bauer's maximum principle  (see, e.g., \citet[Proposition 16.6]{Phelps-2001}), the maximum of $d_{\rm L}$ is attained at the set $\ext (\mathcal{L}_a \times \mathcal{L}_b) =  \ext (\mathcal{L}_a) \times \ext (\mathcal{L}_b)$.
\end{proof}

Proposition \ref{pr:Mabextreme} together with Theorem \ref{th:extLa} reduce the calculation of $M (a,b)$ to a finite-dimensional problem, in fact, to several problems of dimension at most $4$. Although in principle these problems can be solved using elementary analytic techniques, the computations are extremely cumbersome.
For this reason, in the following we present several auxiliary results to simplify the calculations.

\subsubsection{Proof of Proposition \ref{Proposition.Extremal.Third} in Section \ref{Section.Maximum} (Maximality property)}

Let $a\in[0,1]$ and $\ell\in\mathcal{L}_a$. To prove the inequalities in \eqref{Extremality-la}, we will first show that there exist $c^+, c^- \in[0,1]$ (depending on $\ell$) such that
\begin{align}
  \ell(x) \le  \ell_a^+(x) \text{ if } 0\le x \le c^+\quad\text{and}\quad & \ell(x) \ge  \ell_a^+(x) \text{ if } c^+\le x \le 1,\label{inequalities-1}  \\
  \ell_a^-(x) \le  \ell(x) \text{ if } 0\le x \le c^-\quad\text{and}\quad & \ell_a^- (x) \ge  \ell (x) \text{ if } c^-\le x \le 1.\label{inequalities-2}
\end{align}
To check \eqref{inequalities-1}, we note that the right derivative of $\ell$ at $0$, $\ell'(0^+)$, is necessarily less than or equal to $1-a$. Further, $\ell'(0^+)=1-a$ if and only if $\ell=\ell_a^+$, as in this case $\ell_a^+$ is a supporting line of $\ell$ at $0$. If $\ell'(0^+)<1-a$, then $\ell-\ell_a^+$ is a continuous and convex function in $[0,1)$, starting at $0$, with negative derivative at $0$, and zero integral. Hence, there exists $c^+$ satisfying \eqref{inequalities-1}. To prove \eqref{inequalities-2}, observe that $\ell_a^--\ell\le 0$ in $[0,a)$. Also, $\ell(a)=0$ if and only if $\ell=\ell_a^-$. If $\ell(a)>0$, then $\ell_a^--\ell$ is a continuous and concave function in $[a,1)$ such that $\ell_a^-(1)-\ell(1^-)\ge 0$. As $\ell_a^--\ell$ has zero integral, there exists $c^-\in (a,1)$ satisfying \eqref{inequalities-2}.

Now, for $0\le t\le c^+$, by \eqref{inequalities-1}, we directly have that $\int_{0}^{t} \ell_a^+(x)\,\d x \ge \int_{0}^{t} \ell(x)\,\d x$. For  $c^+\le t\le 1$, as $\ell$ and $\ell_a^+$ have the same integral in $(0,1)$, we have that
\begin{align*}
  \int_{0}^{t} \ell_a^+(x)\,\d x & =  \int_{0}^{1} \ell_a^+(x)\,\d x - \int_{t}^{1} \ell_a^+(x)\,\d x \\
   & \ge  \int_{0}^{1} \ell(x)\,\d x - \int_{t}^{1} \ell(x)\,\d x \\
   & = \int_{0}^{t} \ell(x)\,\d x.
\end{align*}
An analogous reasoning shows the second inequality in \eqref{Extremality-la} and the proof is complete.
\hfill $\square$


\subsubsection{The case in which $\ell-m$ has a most one sign change}

Firstly, we carry out a detailed analysis of the maximum value of the distance $d_{\rm L} (\ell,m)$ (for $\ell\in\mathcal{L}_a$ and $m\in\mathcal{L}_b$) when the number of crossing points of the functions $\ell$ and $m$ is less than one. In the following lemma we consider the simpler case in which the Lorenz curves are ordered.

\begin{lemma}\label{lemma:ordered}
For $a,b\in[0,1]$, let $\ell \in \mathcal{L}_a$ and $m \in \mathcal{L}_b$. We have that $d_{\rm L}(\ell,m) = |a-b| $ if and only if $\ell\le m$ or $m\le \ell$. In addition,
$$\min\{ d_{\rm L} (\ell, m)  : \ell \in\mathcal{L}_a \text { and }m \in\mathcal{L}_b \} =| a - b|.$$
\end{lemma}
\begin{proof}
By \eqref{Gini-Lorenz} and the triangular inequality, we always have that
\begin{equation*}
  |a-b| =  2 |\Vert \ell_1 \Vert -\Vert \ell_2 \Vert   |\le 2 \Vert\ell_1-\ell_2 \Vert = d_{\rm L} (\ell_1, \ell_2)\le M(a,b).
\end{equation*}
Moreover, if the inequality above is an equality, we conclude that $\ell$ and $m$ are ordered since $|\ell-m|-||\ell|-|m||$ is a continuous and non-negative function on $(0,1)$. The last part follows from the fact that $ d_{\rm L} (\ell_a^+, \ell_b^+)=|a-b|$, where $\ell_a^+$ and $\ell_b^+$ are defined as in \eqref{la+-la-}.
\end{proof}

Observe that the minimal distance is attained when the curves $\ell \in \mathcal{L}_a$ and $m \in \mathcal{L}_b$ are (pointwise) ordered, i.e., when the underlying variables are ordered in the Lorenz sense.

The second result (Lemma \ref{le:signswitch} below) analyzes the maximum distance between pairs of curves with only one sign switch. We need the following auxiliary lemma.



\begin{lemma}\label{le:convexorder}
For any convex function $\varphi : [0, 1] \to \R$, we have that
\begin{equation}\label{convex-order}
 \int_0^1 \varphi \, \d \ell_a^- \leq \int_0^1 \varphi \, \d \ell \leq \int_0^1 \varphi \, \d \ell_a^+ .
\end{equation}
\end{lemma}
\begin{proof}
For this proof we consider the original definition of $\mathcal{L}$ given by \eqref{eq:L}, i.e., any $\ell \in \mathcal{L}$ satisfies $\ell (1) = 1$ instead of $\ell (1) = \sup_{[0,1)} \ell$. In this way, each $\ell\in\mathcal{L}_a$ is itself a distribution function of a random variable, say $X_\ell$, concentrated on the interval $[0,1]$. The inequalities in \eqref{Extremality-la} together with the fact that all variables $X_\ell$ ($\ell\in\mathcal{L}_a$) have the same expectation imply \eqref{convex-order}; see \cite[Theorem 3.A.1]{Shaked-Shanthikumar-2006}.
\end{proof}

We now establish the maximum value of the Lorenz distance for curves with one crossing point.

\begin{lemma}\label{le:signswitch}
Let $a, b \in [0,1]$ and consider $\ell \in \mathcal{L}_a$, $m \in \mathcal{L}_b$ and $\ell_a^\pm$, $\ell_b^\pm$ the curves defined in \eqref{la+-la-}. We have that
\begin{enumerate}[(a)]
\item\label{item:Asignswitch} If for some $t_0 \in (0,1)$ we have $\ell \leq m$ in $[0, t_0]$ and $\ell \geq m$ in $[t_0, 1]$, then $ d_{\rm L} (\ell, m)   \leq  d_{\rm L} (\ell_a^-, \ell_b^+)$.

\item\label{item:Bsignswitch}  If for some $t_0 \in (0,1)$ we have $\ell \geq m$ in $[0, t_0]$ and $\ell \leq m$ in $[t_0, 1]$, then $d_{\rm L} (\ell, m)  \leq  d_{\rm L} (\ell_a^+,  \ell_b^-)$.
\end{enumerate}
Moreover, it holds that
\begin{equation}\label{eq:maximal-distance}
d_{\rm L} (\ell_a^-, \ell_b^+)  =d_{\rm L} (\ell_a^+,  \ell_b^-) =\frac{(1-a)b^2 +(1-b)a^2}{a+b-ab}.
\end{equation}
\end{lemma}
\begin{proof}
As in Lemma \ref{le:convexorder}, in this proof we consider the definition of $\mathcal{L}$ given by \eqref{eq:L}, i.e., any $\ell \in \mathcal{L}$ satisfies $\ell (1) = 1$. Thus, each function in $\mathcal{L}$ is a distribution function of a random variable concentrated on $[0,1]$.

We only show \emph{(\ref{item:Asignswitch})}, as the proof of \emph{(\ref{item:Bsignswitch})} is analogous.
Let us consider the function $s_{t_0} : [0,1] \to \R$ defined as
\[
 s_{t_0} (x) = t_0 - x + 2 \max \{ 0, x- t_0 \} .
\]
Clearly, the function $s_{t_0}$ is convex, Lipschitz and its distributional derivative is given by $s_{t_0}' = - 1_{[0, t_0]} + 1_{(t_0, 1]}$, where $1_A$ stands for the indicator function of the set $A$.
By integration by parts in the Lebesgue--Stieltjes integral, we obtain that
\[
 \int_0^1 s_{t_0} \, \d (m - \ell) = \int_0^1 (\ell - m) \, \d s_{t_0} = \int_0^{t_0} (m - \ell) + \int_{t_0}^1 (\ell - m) = \left\| \ell - m \right\| .
\]
As $s_{t_0}$ is convex, by Lemma \ref{le:convexorder}, we conclude that
\[
 \int_0^1 s_{t_0} \, \d (m - \ell) \leq \int_0^1 s_{t_0} \, \d \ell_b^+ - \int_0^1 s_{t_0} \, \d \ell_a^- .
\]
Now, from the expression of the Wasserstein distance between probability distributions on the line (see \cite{Vallender-1974}) and by virtue of the Kantorovich--Rubinstein duality (see \citet[eq.\ (6.3)]{Villani-2009}), we have that
\[
 \left\| \ell_b^+ - \ell_a^- \right\| = \sup \left\{ \int_0^1 f \, \d (\ell_b^+ - \ell_a^-) : \, f \in W^{1,\infty} (0,1) , \, \| f' \|_{\infty} \leq 1 \right\}.
\]

Putting together the relations above we obtain that
\[
 \left\| \ell - m \right\| = \int_0^1 s_{t_0} \, \d (m - \ell) \leq \int_0^1 s_{t_0} \, \d (\ell_b^+ - \ell_a^-) \leq \left\| \ell_b^+ - \ell_a^- \right\|.
\]
Therefore, part \emph{(\ref{item:Asignswitch})} of this lemma is proved.

To finish, we note that the curves $\ell_a^+$ and $\ell_b^-$ have one sign switch at the point $t_0=a/(a+b-ab)$. The Lorenz distance in \eqref{eq:maximal-distance} can be directly computed by elementary geometry (as twice the sum of the areas of two triangles; see Figure \ref{fi:maximales}),
\begin{equation*}
d_{\rm L} (\ell_a^-, \ell_b^+)  = a(1-b) t_0 + b (1-t_0)=\frac{(1-a)b^2 +(1-b)a^2}{a+b-ab},
\end{equation*}
which completes the proof of the lemma.
\end{proof}

\subsubsection{Symmetry properties of the set $\mathcal{L}_a$}

Another observation that simplifies to a great extend the calculations is a symmetry reasoning. Given the graph of a function in $\mathcal{L}_a$, its symmetry along the line $y = 1-x$ corresponds to the graph of another curve in $\mathcal{L}_a$.
The precise definition, statement and proof are as follows.
\begin{lemma}\label{le:tilde}
For $\ell \in \mathcal{L}$, define $\ell^{-1} : [0,1] \to [0,1]$ as in \eqref{eq:F-1}. Let us consider the function $\tilde{\ell} : [0,1] \to [0,1]$ given by
\[
 \tilde{\ell} (x) = 1 - \ell^{-1} (1-x), \quad x \in [0,1].
\]
We have that the map $\ell \mapsto \tilde{\ell}$ is a bijective isometry from $(\mathcal{L},d_{\rm L})$ to $(\mathcal{L},d_{\rm L})$ whose inverse is itself. That is, for any $\ell$, $m \in \mathcal{L}$, we have that $d_{\rm L} (\ell, m) = d_{\rm L} (\tilde{\ell} ,\tilde{m})$.

Moreover, for $a\in[0,1]$, the map $\ell \mapsto \tilde{\ell}$ is also a bijective isometry from $(\mathcal{L}_a,d_{\rm L})$ to $(\mathcal{L}_a,d_{\rm L})$.
\end{lemma}
\begin{proof}
As in the proof of Lemma \ref{le:signswitch}, we consider the original definition of $\mathcal{L}$ given by \eqref{eq:L}, i.e., any $\ell \in \mathcal{L}$ satisfies $\ell (1) = 1$.

Let $\ell\in \mathcal{L}$. To check the first assertion we will first show that $\tilde{\ell}\in \mathcal{L}$. Let us consider $x_{\ell} = \sup \{ x \in [0,1] : \ell (x) = 0 \}$.
We observe that $\ell : (x_{\ell}, 1 ) \to (0, \ell (1^-))$ is a bijection and
\begin{equation}\label{eq:ell-1}
 \ell^{-1} (y) = \begin{cases}
 0 & \text{if } y = 0 , \\
 \left( \ell|_{(x_{\ell}, 1)} \right)^{-1} (y) & \text{if } y \in (0, \ell (1^-)) , \\
 1 & \text{if } y \in [\ell (1^-) , 1 ] .
 \end{cases}
\end{equation}
Indeed, by Proposition \ref{pr:LaContinuous}, $\ell$ is non-decreasing.
We shall see that $\ell : [x_{\ell}, 1 ] \to \R$ is strictly increasing.
Assume, by contradiction, that there exist $x_1 , x_2$ with $x_{\ell} \leq x_1 < x_2 \leq 1$ and $\ell (x_1) = \ell (x_2)$.
Then, the constant function $\ell (x_2)$ is a supporting line of $\ell$ at $(x_2, \ell (x_2))$, so, by convexity, $\ell \geq \ell (x_2)$.
Since $x_2 > x_{\ell}$, by definition of $x_{\ell}$ we have that $\ell (x_2) > 0$.
In particular $\ell (0) > 0$, which is a contradiction. We conclude that $\ell : [x_{\ell}, 1 ] \to \R$ is strictly increasing. By Proposition \ref{pr:LaContinuous}, $\ell|_{[0,1)}$ is continuous and, in particular, $\ell : (x_{\ell}, 1 ) \to (0, \ell (1^-))$ is a bijection. Moreover, $\ell|_{[0, x_{\ell}]} = 0$ and $\ell (1) = 1$.
With these properties of $\ell$, it is immediate to check that $\ell^{-1}$ is given as in \eqref{eq:ell-1}.
Now, from \eqref{eq:ell-1}, we obtain that 
\begin{equation}\label{eq:tildel}
 \tilde{\ell} (x) = \begin{cases}
 0 & \text{if } x \in [0, 1 - \ell(1^-)] , \\
 1 - \left( \ell|_{(x_{\ell}, 1)} \right)^{-1} (1-x) & \text{if } x \in (1 - \ell(1^-), 1) , \\
 1 & \text{if } x = 1 .
 \end{cases}
\end{equation}
The inverse of an increasing convex function is concave (see \citet[Example 1.6]{Simon-2011}). Hence, the function $\ell^{-1}$ in \eqref{eq:ell-1} is concave, and this implies that $\tilde{\ell}$ in \eqref{eq:tildel} is convex.
In particular, it follows that $\tilde{\ell}\in \mathcal{L}$.

Next we will check that $\widetilde{\tilde{\ell}} = \ell$. It can be easily seen that $\tilde{\ell} (1^-) = 1 - x_{\ell}$ and $x_{\tilde{\ell}} = 1 - \ell (1^-)$. Moreover, $\left( \tilde{\ell}|_{(x_{\tilde{\ell}}, 1)} \right)^{-1} : (0, 1 - x_{\ell}) \to (x_{\tilde{\ell}}, 1)$ is given by $\left( \tilde{\ell}|_{(x_{\tilde{\ell}}, 1)} \right)^{-1} (y) = 1 - \ell (1-y)$. Therefore, by \eqref{eq:tildel}, we conclude that $\widetilde{\tilde{\ell}} = \ell$.

Next we need to prove that the map $\ell \mapsto \tilde{\ell}$ is an isometry. We consider $m \in \mathcal{L}$ and observe first that  $\| \tilde{\ell} - \tilde{m} \|=\| \ell^{-1} - m^{-1} \|$. Further, as $\ell$ and $m$ are distribution functions of random variables concentrated on the interval $[0,1]$, by using the expression of the Wasserstein distance for real-valued random variables (see \cite{Vallender-1974}), we have that $\| \ell^{-1} - m^{-1} \|= \| \ell - m \|$. Hence, we conclude that $d_{\rm L} (\ell, m) = d_{\rm L} (\tilde{\ell} ,\tilde{m})$.

Finally, to check the last assertion of the lemma it is enough to verify that, for $a\in[0,1]$ and $\ell\in\mathcal{L}_a$, one has $G( \tilde{\ell})=a$. Let us fix $\ell\in\mathcal{L}_a$. We will equivalently show that $\| \tilde{\ell} \|=(1-a)/2$. We apply Laisant's formula (see, for instance, \cite{Parker-1955}) for the integral of the inverse to obtain
\[
 \int_0^{\ell (1^-)} \left( \ell|_{(x_{\ell}, 1)} \right)^{-1} (y) \, \d y = \ell (1^-) - \int_{x_{\ell}}^1 \ell (x) \, \d x = \ell (1^-) - \frac{1-a}{2} .
\]
Therefore,
\[
 \int_0^1 \ell^{-1} (y) \, \d y = \int_0^{\ell (1^-)} \left( \ell|_{(x_{\ell}, 1)} \right)^{-1} (y) \, \d y + 1 - \ell (1^-) = 1 - \frac{1-a}{2}.
\]
From \eqref{eq:tildel}, we finally obtain that
\[
 \int_0^1 \tilde{\ell} (y) \, \d y = 1 - \int_0^1 \ell^{-1} (y) \, \d y = \frac{1-a}{2},
\]
which completes the proof of the lemma.
\end{proof}

Lemma \ref{le:tilde} helps us to disregard some cases in the computation of the value of $M(a,b)$, as the next result shows.

\begin{lemma}\label{le:symmetry}
For $a, b \in [0,1]$, let $M (a,b)$ be defined in \eqref{Mab}. We have that
\begin{equation*}
M(a,b)=\max \{ d_{1} (a, b),\, d_{2} (a, b),\, d_{3} (a, b)\},
\end{equation*}
where
\begin{equation}\label{eq:d123}
\begin{split}
  d_{1} (a, b) &  = \max\{ d_{\rm L} (\ell^a_{x_1} , \ell^b_{y_1}) : x_1, y_1\in[0,a]  \}, \\
  d_{2} (a, b) & = \max\{ d_{\rm L} (m^a_{x_2} , m^b_{y_2}) : x_2, y_2\in (a,1)  \}, \\
  d_{3} (a, b) & = \max\{ d_{\rm L} (\ell^a_{x_1} , m^b_{x_2}) : x_1\in[0,a], x_2\in (a,1)  \}.
\end{split}
\end{equation}
\end{lemma}

\begin{proof}
We describe first how the extreme points of Theorem \ref{th:extLa} are affected under the isometry $\ell \mapsto \tilde{\ell}$ defined in Lemma \ref{le:tilde}. From \eqref{eq:extremePW} and \eqref{eq:tildel}, we find that they are the piecewise affine functions such that
\[
 \tilde{\ell}^a_{x_1} : \begin{cases}
  0 \mapsto 0 \\
 1- \frac{1-a}{1-x_1} \mapsto 0 \\
 1 \mapsto 1-x_1 ,
 \end{cases} \qquad
 \tilde{m}^a_{x_2} : \begin{cases}
 0 \mapsto 0 \\
 1- (x_2 -a) \mapsto 1- x_2  \\
 1 \mapsto 1,
 \end{cases} \qquad
 \tilde{n}^a_{x_1,x_2} : \begin{cases}
 0 \mapsto 0 \\
 1-\frac{x_2 -a}{1-x_1} \mapsto 1- x_2  \\
 1 \mapsto 1 - x_1 .
 \end{cases}
\]
In other words,
\begin{equation*}
\tilde{\ell}^a_{x_1} = \ell^a_{\frac{a-x_1}{1-x_1}} , \qquad \tilde{m}^a_{x_2} = m^a_{a + 1 - x_2}, \qquad \tilde{n}^a_{x_1,x_2} \notin \ext(\mathcal{L}_a).
\end{equation*}
As $\tilde{n}^a_{x_1,x_2} \notin \ext(\mathcal{L}_a)$, by Proposition \ref{pr:Mabextreme} and since $\ell \mapsto \tilde{\ell}$ is an isometry (see Lemma \ref{le:tilde}), we can exclude the functions ${n}^a_{x_1,x_2}$ in the computation of the maximum given in \eqref{eq:Mab-ext}.
\end{proof}

\subsubsection{Computation of $M(a,b)$}

Here, we will combine all the previous results to prove Theorem \ref{Theorem.Mab}. We start with the following lemma.
\begin{lemma}\label{le:d123}
 For $a, b \in [0,1]$, let $d_1(a,b)$ and $d_2(a,b)$ be defined in \eqref{eq:d123}. We have that
 $$d_1(a,b),\, d_2(a,b)\le d_{\rm L} (\ell_a^- , \ell_b^+)= d_{\rm L} (\ell_a^+ , \ell_b^-),$$
 where the curves $\ell_a^\pm$, $\ell_b^\pm$ are defined as in \eqref{la+-la-}.
\end{lemma}

\begin{proof}
It is easy to see that the curves $\ell^a_{x_1}$ and $\ell^b_{y_1}$ cannot have two crossing points, and neither can $m^a_{x_2}$ and $m^b_{y_2}$. Therefore, the conclusion follows from Lemmas \ref{lemma:ordered} and \ref{le:signswitch}.
\end{proof}

We are then led to the computation of $d_3(a,b)$, i.e., the maximum of value of $d_{\rm L} (\ell^a_{x_1} , m^b_{x_2})$, when $x_1\in[0,a]$ and $x_2 \in (b,1)$. This question is addressed in the following result.
\begin{lemma}\label{le:d3}
For $a, b \in [0,1]$, let $d_3(a,b)$ be defined in \eqref{eq:d123}. We have that
 $$d_3(a,b)\le d_{\rm L} (\ell_a^- , \ell_b^+)= d_{\rm L} (\ell_a^+ , \ell_b^-).$$
\end{lemma}
\begin{proof}
Let $x_1\in [0,a]$ and $x_2\in (b,1)$. Thanks to Lemmas \ref{lemma:ordered} and \ref{le:signswitch}, we can assume that the curves $\ell^a_{x_1}$ and $m^b_{x_2}$ cross each other twice; Figure \ref{fi:two-crossing} represents this situation. This happens if and only if the triangle with vertices
\[
 \left( x_1, 0 \right) , \quad \left( x_2, x_2 -b \right) , \quad \left( 1, \frac{1-a}{1-x_1} \right)
\]
has positive orientation; equivalently, if and only if
\[
 (1-a) (x_2-x_1) - (x_2-b) (1-x_1) + x_1 (1-x_1) (x_2-b) >0 .
\]
In this case, they cross each other at the points $(x_1^*, y_1^*)$ and $(x_2^*, y_2^*)$, with $x_1^* \leq x_2^*$ and
\begin{align*}
 & x_1^* = \frac{(1-a) x_1 x_2}{(1-a) x_2 - (1-x_1)^2 (x_2-b)} , \quad y_1^* = \frac{(1-a) x_1 (x_2 - b)}{(1-a) x_2 - (1-x_1)^2 (x_2-b)} , \\
 & x_2^* = \frac{(1-a) x_1 x_2 (1-x_2) - b (1-x_1)^2 x_2}{(1-a) x_2 (1-x_2) - (1-x_1)^2 (x_2-b) (1-x_2) - b (1-x_1)^2} ,  \\
 & y_2^* = \frac{(1-a) \left[ x_1 (x_2-b) (1-x_2) - b (x_2 - x_1) \right]}{(1-a) x_2 (1-x_2) - (1-x_1)^2 (x_2-b) (1-x_2) - b (1-x_1)^2} .
\end{align*}
\begin{figure}[h]\centering{
\begin{tikzpicture}[scale=8]
\draw [ultra thick] (0,0) -- (0.3,0) -- (1,1/1.4) -- (1,1);
\draw [blue, ultra thick] (0,0) -- (0.7,0.2) -- (1,1);
\draw [thick, black] (0,0) -- (1,0) -- (1,1) -- (0,1)-- (0,0);
\draw [black,dashed, thick] (0,0) -- (1,1);
\draw[fill] (5/12, 5/42) circle [radius=0.01];
\draw[fill] (100/121, 65/121) circle [radius=0.01];
\draw [thick] (.3,-0.02) node[below]{$x_1$} -- (.3,0.02);
\draw [thick] (.7,-0.02) node[below]{$x_2$} -- (.7,0.02);
\draw [thick] (1,-0.02) node[below]{1} -- (1,0.02);
\draw [thick] (0,-0.02) node[below]{0} -- (0,0.02);

\draw [black,dashed, thick] (0.7,0) -- (0.7,0.2);
\draw [black,dashed, thick] (0.7,0.2) -- (1,0.2);

\draw [thick] (0.98,0.20) -- (1.02 ,0.20);
\node [right] at (1.02,0.20) {$x_2-b$};

\draw [thick] (0.98,5/7) -- (1.02 ,5/7);
\node [right] at (1.02,5/7) {$\displaystyle \frac{1-a}{1-x_1}$};

\node [right] at (5/12-0.19 , 5/42+0.03) {$(x_1^*,y_1^*)$};
\node [right] at (100/121-0.19 , 65/121+0.03) {$(x_2^*,y_2^*)$};

\draw [thick] (5/12,-0.02) node[below]{$x_1^*$} -- (5/12,0.02);
\draw [black,dashed, thick] (5/12,0) -- (5/12,5/42);

\draw [thick] (100/121,-0.02) node[below]{$x_2^*$} -- (100/121,0.02);
\draw [black,dashed, thick] (100/121,0) -- (100/121,65/121);

\draw [thick] (0.98,5/42) -- (1.02 ,5/42);
\node [right] at (1.02,5/42) {$y_1^*$};
\draw [black,dashed, thick] (5/12,5/42) -- (1,5/42);

\draw [thick] (0.98,65/121) -- (1.02 ,65/121);
\node [right] at (1.02,65/121) {$y_2^*$};
\draw [black,dashed, thick] (100/121,65/121) -- (1,65/121);

\draw [thick] (0.98,0) -- (1.02 ,0);
\node [right] at (1.02,0) {$0$};

\draw [thick] (0.98,1) -- (1.02 ,1);
\node [right] at (1.02,1) {$1$};

\node [above] at (0.55 , 0.28) {$\ell_{x_1}^a$};
\node [above, blue] at (0.62 , 0.18) {$m_{x_2}^a$};

\end{tikzpicture}}
\caption{A graphical representation of the functions $\ell_{x_1}^a$ (in black) and $m_{x_2}^a$ (in blue), as well as their crossing points, $(x_1^*,y_1^*)$ and $(x_2^*,y_2^*)$.}
\label{fi:two-crossing}
\end{figure}
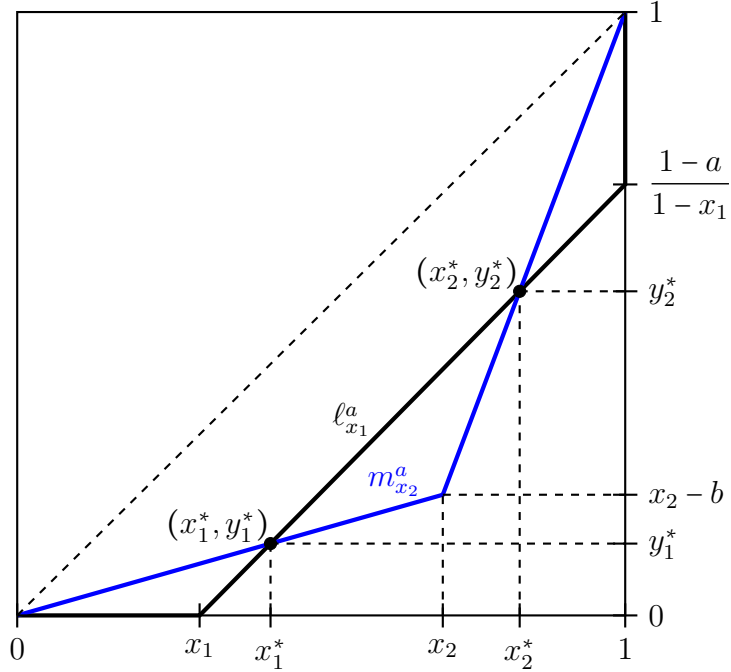

For simplicity, from now on we will call $A=A(x_1,x_2)=d_{\rm L} \left( \ell^a_{x_1} , m^b_{x_2} \right)$. The value of $A$ can be therefore calculated as twice the sum of the areas of the three triangles (see Figure \ref{fi:two-crossing})
\[
 (0, 0), (x_1, 0), (x_1^*, y_1^*) , \quad (x_1^*, y_1^*) , (x_2, x_2 - b) , (x_2^*, y_2^*) , \quad (x_2^*, y_2^*) , \left(1, \frac{1-a}{1-x_1}\right) , (1, 1) .
\]
Elementary geometry shows that
\[
 A = x_1 \, y_1^* + x_1^* (x_2 - b - y_2^*) + x_2 (y_2^* - y_1^*) + x_2^* (y_1^* - (x_2-b)) + \frac{(a - x_1) (1- x_2^*)}{1 - x_1} .
\]
So are led to the maximization of $A= A(x_1, x_2)$ under the constraints
\begin{equation}\label{eq:restrictions}
 0 \leq x_1 \leq a , \quad b < x_2 < 1 , \quad  (1-a) (x_2-x_1) - (x_2-b) (1-x_1) + x_1 (1-x_1) (x_2-b) >0 .
\end{equation}
Obviously, the supremum under \eqref{eq:restrictions} coincides with the maximum under the analogous inequalities of \eqref{eq:restrictions} but replacing the `$<$' with `$\leq$', and `$>$' with `$\geq$'.
Observe that $A$ is a rational function, and, hence, the computation of its maximum in region \eqref{eq:restrictions} can be done by elementary techniques.
Nevertheless, the computations are extremely long, so we have used the programme  \cite{Mathematica} in the rest of the proof to avoid unnecessary details.

We find that the only solution of $\frac{\p A}{\p x_1} = \frac{\p A}{\p x_2} = 0$ under \eqref{eq:restrictions} is
\begin{equation*}
 x_1 = 1 - \sqrt{1 - a}, \quad x_2 = \frac{1 + b}{2} ,
\end{equation*}
which in fact requires $a + b^2 < 2b$.
The corresponding value of $A$ is
\[
 A \left( 1 - \sqrt{1 - a}, \frac{1 + b}{2} \right) = \frac{4 \left(\sqrt{1-a}-1\right) b+a (b-2)-4 \sqrt{1-a}+b^2+4}{b} ,
\]
which is seen to be less than the value of $d_{\rm L} (\ell_a^- , \ell_b^+)$ (computed in Lemma \ref{le:signswitch}), thanks to the restriction $a + b^2 < 2b$.

After having checked the value of $A$ at the critical points in the interior of region \eqref{eq:restrictions}, we analyze the
value of $A$ on the boundary.
Describing the boundary of region \eqref{eq:restrictions} is cumbersome since it involves several cases, according to the values of $a,b$.
In any case, the boundary is clearly contained in the set
\begin{align*}
 \bigl\{  (x_1, x_2) \in [0,a] \times [b,1] : \ & x_1 = 0 \text{ or } x_1 = a \text{ or } x_2 = b \text{ or } x_2 = 1 \\
 & \text{ or } (1 - a) (x_2 - x_1) - (x_2 - b) (1 - x_1) + x_1 (1 - x_1) (x_2 - b) = 0 \bigr\} .
\end{align*}
When $(1 - a) (x_2 - x_1) - (x_2 - b) (1 - x_1) + x_1 (1 - x_1) (x_2 - b) = 0$ the curves $\ell^a_{x_1}$ and $m^b_{x_2}$ do not have a proper crossing, so, by Lemma \ref{lemma:ordered}, $A = |b-a|$, which does not release a maximum.
Therefore, we are led to the maximization of $A(x_1,x_2)$ in the set
\[
 \bigl\{  (x_1, x_2) \in [0,a] \times [b,1] : \ x_1 = 0 \text{ or } x_1 = a \text{ or } x_2 = b \text{ or } x_2 = 1 \bigr\}
\]
The value of $A$ when $x_1 = 0$ is
\[
 A (0, x_2) = a+b -\frac{2 a b}{a+b-a x_2} ,
\]
which is decreasing in $x_2$, so the maximum is attained at $x_2 = b$ and equals
\[
A (0, b) = a+b- \frac{2ab}{a+b-ab} = \frac{(1-a)b^2 +(1-b)a^2}{a+b-ab}= d_{\rm L} (\ell_a^- , \ell_b^+).
\]

The value of $A$ when $x_1 = a$ is
\[
 A (a, x_2) = a+b-\frac{2 a b}{a x_2 + b - a b}
\]
which is increasing in $x_2$, so the maximum is attained at $x_2 = 1$ and equals
\[
 A (a, 1) =\frac{(1-a)b^2 +(1-b)a^2}{a+b-ab} = d_{\rm L} (\ell_a^- , \ell_b^+).
\]

The value of $A$ when $x_2 = b$ is
\[
 A (x_1, b) = \frac{a^2 b - a^2 + a b^2 - b^2 + \left( - 4 a b + 2 a + 2 b \right) x_1 + \left( a + b -2 \right) x_1^2}{a b-a-b +2 x_1 -x_1^2} ,
\]
which is decreasing in $x_1$, so the maximum is attained at $x_1 = 0$ and equals
\[
 A (0, b) = \frac{(1-a)b^2 +(1-b)a^2}{a+b-ab} = d_{\rm L} (\ell_a^- , \ell_b^+).
\]

The value of $A$ when $x_2 = 1$ is
\[
 A (x_1, 1) = \frac{a^2 b - a^2 + a b - b^2 + \left( a  + 2 b^2 - 3 a b \right) x_1 +  \left( a  - b^2 \right) x_1^2+  \left( b - 1 \right) x_1^3}{a b - a - b +2 x_1 - x_1^2} ,
\]
which is increasing in $x_1$.
Therefore, the maximum is attained at $x_1 =a$ and equals
\[
 A (a, 1) = \frac{(1-a)b^2 +(1-b)a^2}{a+b-ab} .
\]
This concludes the proof.
\end{proof}


We finally observe that the proof of Theorem \ref{Theorem.Mab} directly follows from Lemmas \ref{le:signswitch}, \ref{le:symmetry}, \ref{le:d123}, and  \ref{le:d3}.


\subsubsection{Proof of Corollary \ref{Coro.range} in Section \ref{Section.Maximum} (Range of the distance)}

The minimum and maximum of $d_{\rm L} (\ell, m)$ among $\ell \in \mathcal{L}_a$ and $m \in\mathcal{L}_b$ are $\left| b-a \right|$ and $M (a,b)$, as calculated in Lemma \ref{lemma:ordered} and Theorem \ref{Theorem.Mab}, respectively. Now, the set $\{ d_{\rm L} (\ell, m) : \, \ell \in \mathcal{L}_a \text{ and } m  \in \mathcal{L}_b \}$
 is compact and connected as it is the image by the continuous function $d_{\rm L}$ of the compact and connected set $\mathcal{L}_a \times \mathcal{L}_b$ (in the space $L^1$). The result follows.   \hfill $\square$

\medskip

\color{black}
\subsection{Proofs of the results in Section \ref{Section.Index}}

\textbf{Proof of Proposition \ref{Proposition:rangeDelta} (Values of $\mathcal{I}$).} This result follows from Corollaries \ref{Coro.range} and \ref{Corollary.maximal}.

\medskip

\textbf{Proof of Proposition \ref{Proposition-range-Delta*} (The set $\Delta^*$).} This result follows from Corollary \ref{Coro.range}.

\medskip

\textbf{Proof of Proposition \ref{Proposition.Properties} (Properties of the indices).}
We will only consider the index $\mathcal{I}_*$ as the corresponding proofs for $\mathcal{I}^*$ are analogous. Parts (i) and (ii) are fulfilled by construction. Parts (1) and (2) of (iii) are consequences of Lemma \ref{lemma:ordered}, while (iii) (3) is fulfilled by the definition of extremal Lorenz curves (see Definition \ref{Definition.extremal}).
Part (iv) (1) is trivial while (iv) (2) and (3) follow from (iii). Finally, part (v) holds by dominated convergence.
\hfill $\square$

\subsection{Proofs of the results in Section \ref{Section-Asymptotics}}

\textbf{Proof of Proposition \ref{Proposition.Strong.Consistency} (Strong consistency).} By the triangle inequality, we obtain that
\begin{equation}\label{consistency-inequalities}
\left|  ( \| \hat{\ell}_1 \| - \| \hat{\ell}_2 \|)-( \|{\ell}_1 \|- \|{\ell_2} \| )    \right|,\, \left|   \| \hat{\ell}_1 -  \hat{\ell}_2 \| - \|{\ell}_1 - {\ell}_2 \|     \right|   \le \|  \hat{\ell}_1-\ell_1 \| +  \|  \hat{\ell}_2-\ell_2 \|.
\end{equation}
From \cite{Goldie-1977}, we have that $\hat\ell_j \to \ell_j$ (uniform convergence) a.s., as $n_j\to \infty$ (for $j=1,2$). Using dominated convergence, we therefore have that the right-hand side of \eqref{consistency-inequalities} goes to 0 a.s. as $n_j\to \infty$, and, consequently, $\hat \I (\ell_1,\ell_2) \to \I (\ell_1,\ell_2)$ a.s. Moreover, as the maps $t_*$ in \eqref{t*} and $t^*$ in \eqref{t} are continuous, we also conclude that $\hat \I_* (\ell_1,\ell_2) \to \I_* (\ell_1,\ell_2)$ and $\hat \I^* (\ell_1,\ell_2) \to \I^* (\ell_1,\ell_2)$ a.s.
\hfill $\square$
\medskip

\textbf{Proof of Proposition \ref{Proposition-Asymptotics} (Asymptotic behaviour).} Let us consider the map $\phi: C([0,1])\times C([0,1])\to \R^2$ defined by
\begin{equation*}
  \phi(f,g)=2\, \left(   \Vert f \Vert - \Vert g \Vert ,\, \Vert f-g \Vert     \right),\quad  f,g\in C([0,1]).
\end{equation*}
From \eqref{I0}, we obviously have that
\begin{equation}\label{nornalized-index-representation}
 \sqrt{\frac{n_1 n_2}{n_1+n_2}} \left( \I (\hat \ell_1,\hat \ell_2) -   \I ( \ell_1, \ell_2)        \right)  = \sqrt{\frac{n_1 n_2}{n_1+n_2}} \left( \phi(\hat\ell_1,\hat\ell_2) -    \phi(\ell_1,\ell_2)      \right).
\end{equation}
Further, from Lemma \ref{Lemma-Hadamard}, it is easy to check that $\phi$ is Hadamard directionally differentiable at $(\ell_1,\ell_2)$ with derivative given by
  \begin{equation*}
\phi^\prime_{(\ell_1,\ell_2)}(h_1,h_2)=2\left(    \delta^\prime_{\ell_1}(h_1)-  \delta^\prime_{\ell_2}(h_2),\, \delta^\prime_{\ell_1-\ell_2}(h_1-h_2)            \right),\quad \text{for } h_1,h_2\in C([0,1]).
\end{equation*}

On the other hand, by Lemma \ref{Lemma-Lorenz-convergence} (and the independence assumption of the samples) we also have that
\begin{equation}\label{convergence-delta}
  \sqrt{\frac{n_1 n_2}{n_1+n_2}} \left(  \hat\ell_1 - \ell_1 ,\, \hat\ell_2 -\ell_2      \right) \rightsquigarrow \left( \sqrt{1-\lambda}\, \mathbb{L}_1    ,\, \sqrt{\lambda} \, \mathbb{L}_2    \right)\quad \text{in } C([0,1])\times C([0,1]).
\end{equation}
Therefore, from \eqref{nornalized-index-representation}--\eqref{convergence-delta}, the application of the extended version of the functional delta method (see \citet[Theorem 2.1]{Shapiro-1990}) yields the result.
\hfill $\square$
\medskip

\textbf{Proof of Corollary \ref{Corollary-normal} (Asymptotic normality).} The equivalence between (a) and (b) follows from \eqref{delta-derivative}. As $\mathbb{L}$ is a centered Gaussian process, the distribution in \eqref{limit-distribution-2} is normally distributed with mean $(0,0)$, so (b) implies (c). Conversely, if (c) holds, we have that the variable
  \begin{equation*}
   \delta^\prime_{\ell_1-\ell_2} (   \mathbb{L}  ) =  \int_{\{\ell_1=\ell_2\}} |\mathbb{L}|+ \int_{\{\ell_1  \ne \ell_2\}} \mathbb{L}\cdot \sgn(\ell_1-\ell_2)
  \end{equation*}
  has zero mean normal distribution. Therefore, the set $\{\ell_1=\ell_2\}$ has zero Lebesgue measure, since otherwise, the first summand in the previous equation would have normal distribution, which is clearly not possible.
\hfill $\square$


\section*{Acknowledgements}

A. Ba\'illo and J. C\'{a}rcamo are supported by the Spanish MCyT grant PID2019-109387GB-I00. C. Mora-Corral is supported by the Spanish MCyT grant MTM2017-85934-C3-2-P.

This paper is based on data from Eurostat, European Union Statistics on Income and Living Conditions (EU-SILC), 2004--2019. The responsibility for all conclusions drawn from the data lies entirely with the authors.



\newpage

\begin{center}
\textbf{\Large Extremal points of Lorenz curves and} \\ [2 mm]
\textbf{\Large applications to inequality analysis} \\ [2 mm]
\textbf{\Large Supplementary material} \\ [2 mm]
Amparo Ba\'illo, Javier C\'{a}rcamo, and Carlos Mora-Corral \\  [2 mm]
Departamento de Matem\'{a}ticas, Universidad Aut\'{o}noma de Madrid, 28049 Madrid (SPAIN) \\  [2 mm]
\today
\end{center}

\setcounter{section}{6}


\section{Estimation and asymptotic properties}\label{SuppMatSection.Asymptotics}

We have carried out a simulation study to illustrate how the asymptotic results (in Section 7 of the main manuscript) behave in practice with finite samples. Using the R package GB2, for each model we have generated 1000 Monte Carlo samples with sample size \(n\), that is, \(n_1=n_2=n\) from two variables \(X_1\sim\text{GB2}(a_1,b_1,p_1,q_1)\) and \(X_2\sim\text{GB2}(a_2,b_2,p_2,q_2)\) (such that \(\E X_1 =\E X_2 = 1\)). In each Monte Carlo run, we have computed the value of \(\hat{\mathcal I}(\ell_1,\ell_2) = \mathcal I (\hat{\ell}_1,\hat{\ell}_2) \), thus obtaining 1000 realizations of \(\hat{\mathcal I}\). We have chosen two values for \(n\), 10000 and 50000 in Models 1, 2, 3 and 5, and sizes $n=10^4$ and $n=10^5$ in the Monte Carlo samples of Model 4.

\newpage

\textbf{Model 1:}

\begin{figure}[H]
\begin{center}
\begin{tabular}{cc}
\includegraphics[width=0.47\textwidth]{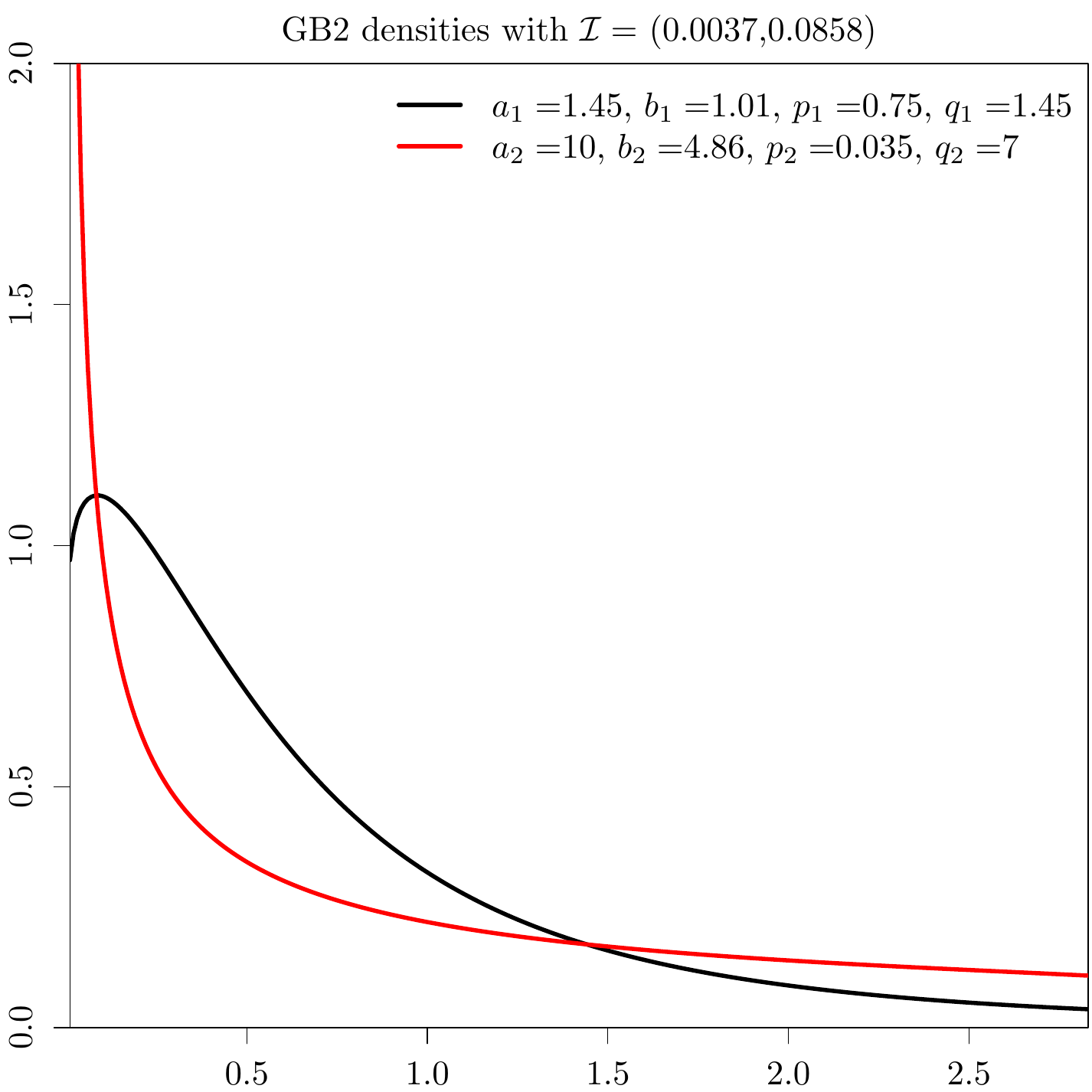} & \includegraphics[width=0.47\textwidth]{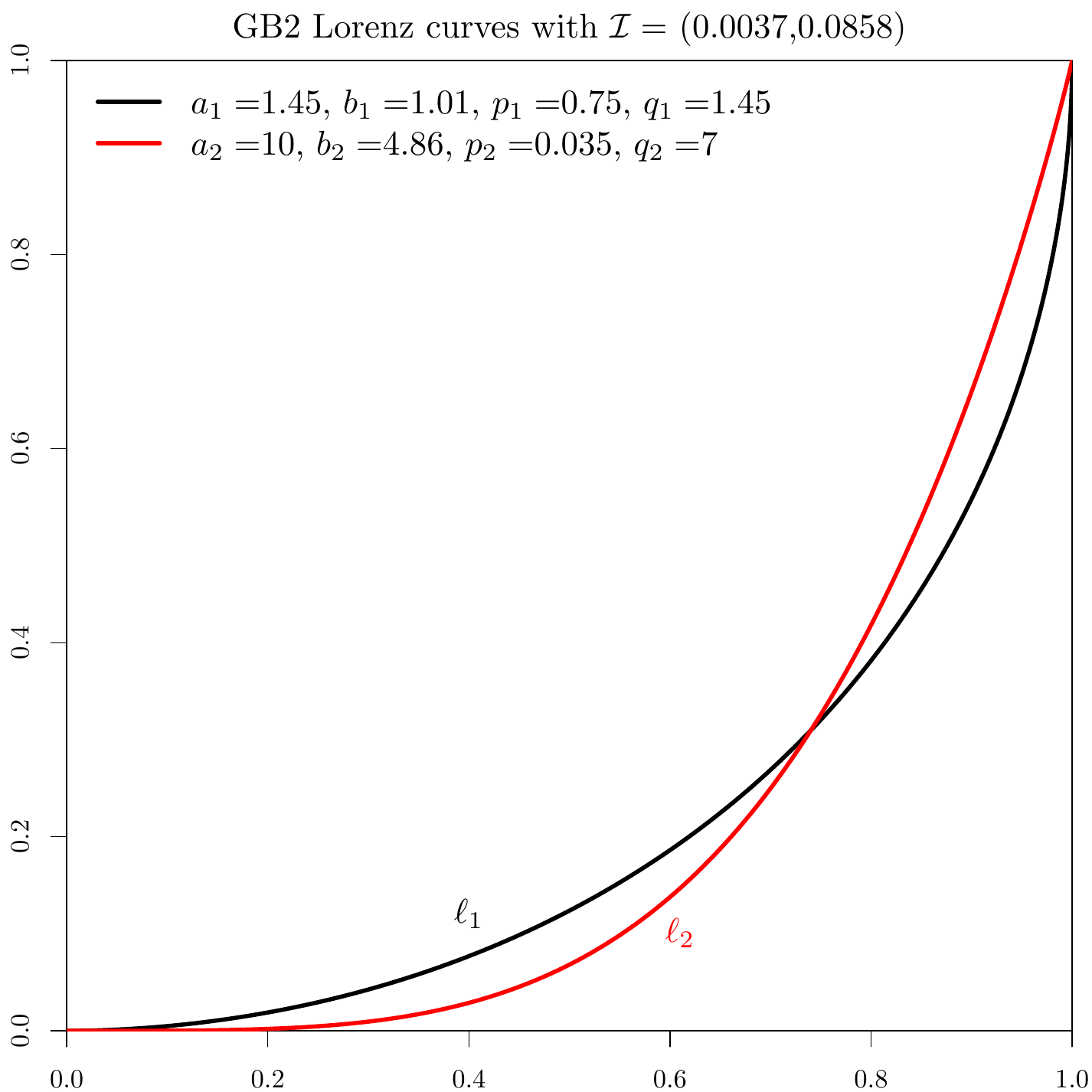}
\end{tabular}
\end{center}
\caption{Densities (left) and Lorenz curves (right) of the GB2 distributions in Model 1.}
\label{SuppMatfi:Mod1Simul}
\end{figure}

The simulated values of the normalized empirical index, \(\sqrt{\frac{n_1 n_2}{n_1+n_2}} \left( \I (\hat \ell_1,\hat \ell_2) - \I ( \ell_1, \ell_2) \right) \), appear in Figure~\ref{SuppMatfi:SimulMod1}. The green point is the origin (0,0). The level sets in red are those of a mixture of normal densities fit (with the R library mclust) to the 1000 simulations of the inequality index. The number of components in the mixture and the parameterization of the covariance matrices of the Gaussian components was chosen via the Bayesian Information Criterion (BIC). The points in the mixture components with weight lower than 10\% are colored in blue: we are thus able to detect the points in the data cloud which contribute most to the lack of normality. The convergence to normality is clear as \(n\) increases.

\begin{figure}[H]
\begin{center}
\begin{tabular}{cc}
\includegraphics[width=0.47\textwidth]{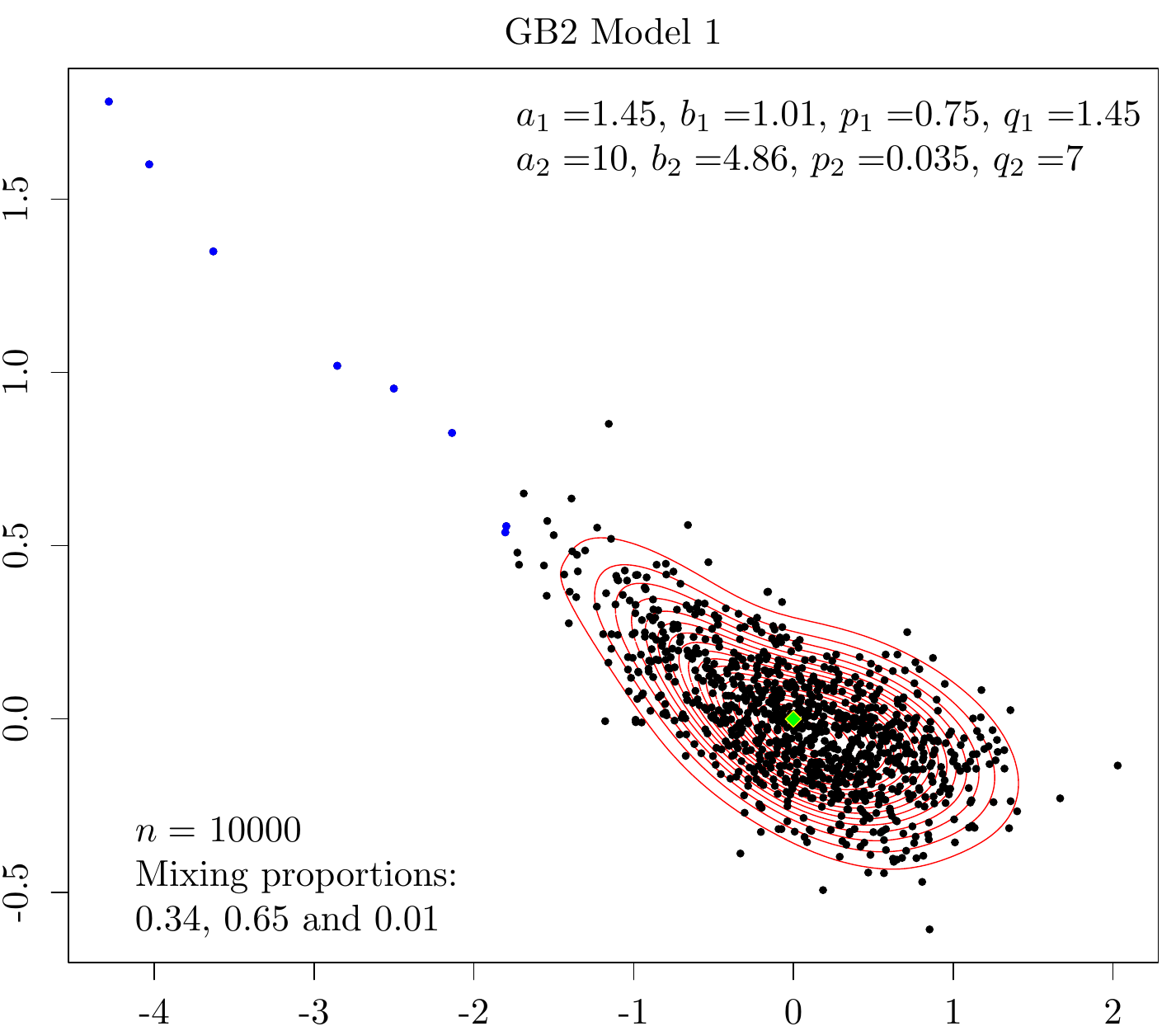} & \includegraphics[width=0.47\textwidth]{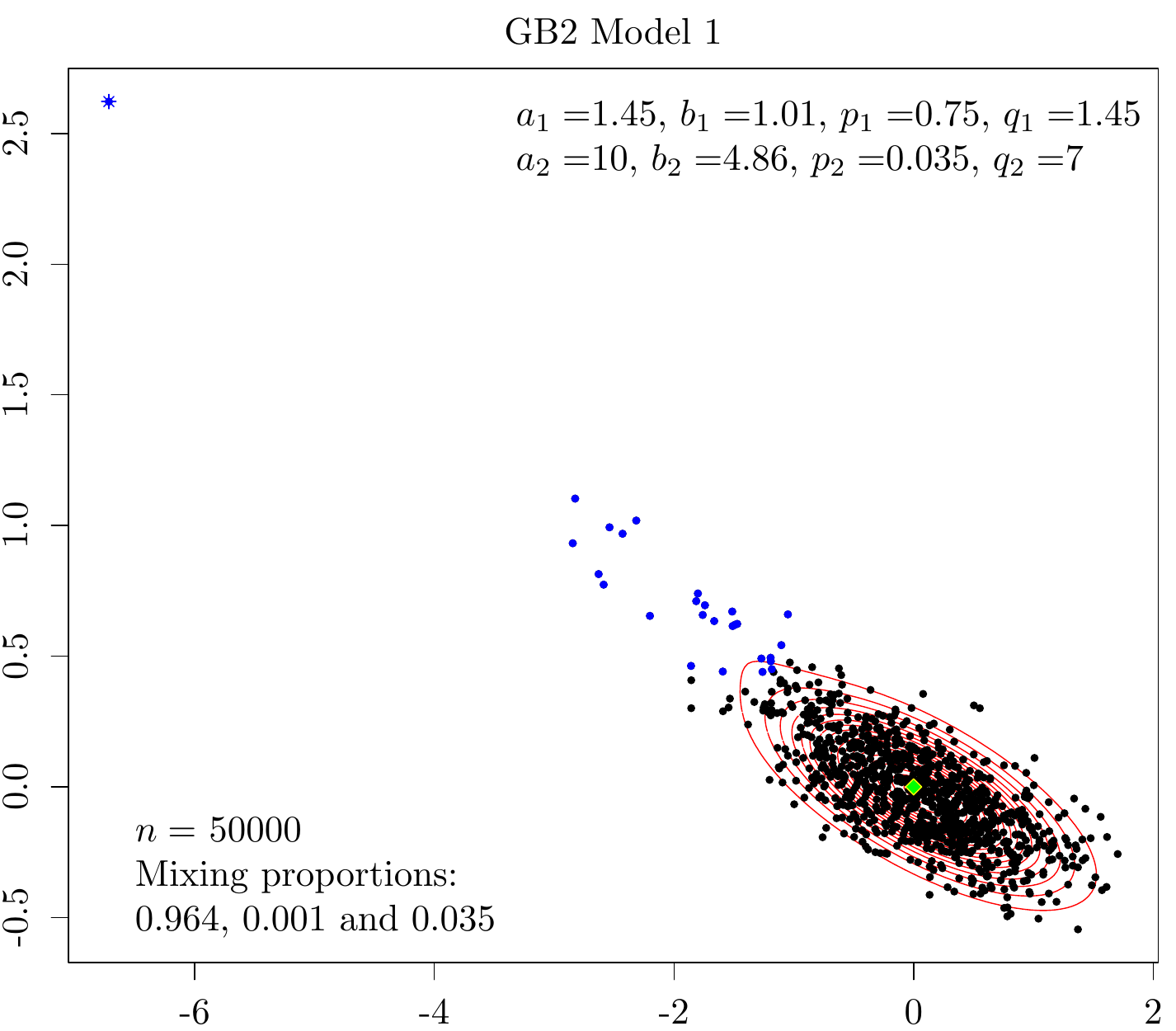}
\end{tabular}
\end{center}
\caption{Normalized values of \(\hat{\mathcal I}\) for simulated samples of Model 1 with \(n=10000\) (left) and \(n=50000\) (right).}
\label{SuppMatfi:SimulMod1}
\end{figure}

\newpage

\textbf{Model 2:} Since in this model the variables \(X_1\) and \(X_2\) are ordered (\(X_2 \Lo X_1\), see Figure~\ref{SuppMatfi:Mod2Simul}), the bidimensional index \(\mathcal I\) lies on \(L_2=\{ (-x,x) : x\in[0,1] \}\) and only one of the components of \(\mathcal I\) is of interest. We take the second component, \(d_{\rm L}(\ell_1,\ell_2)\), as it is positive.

\begin{figure}[H]
\begin{center}
\begin{tabular}{cc}
\includegraphics[width=0.47\textwidth]{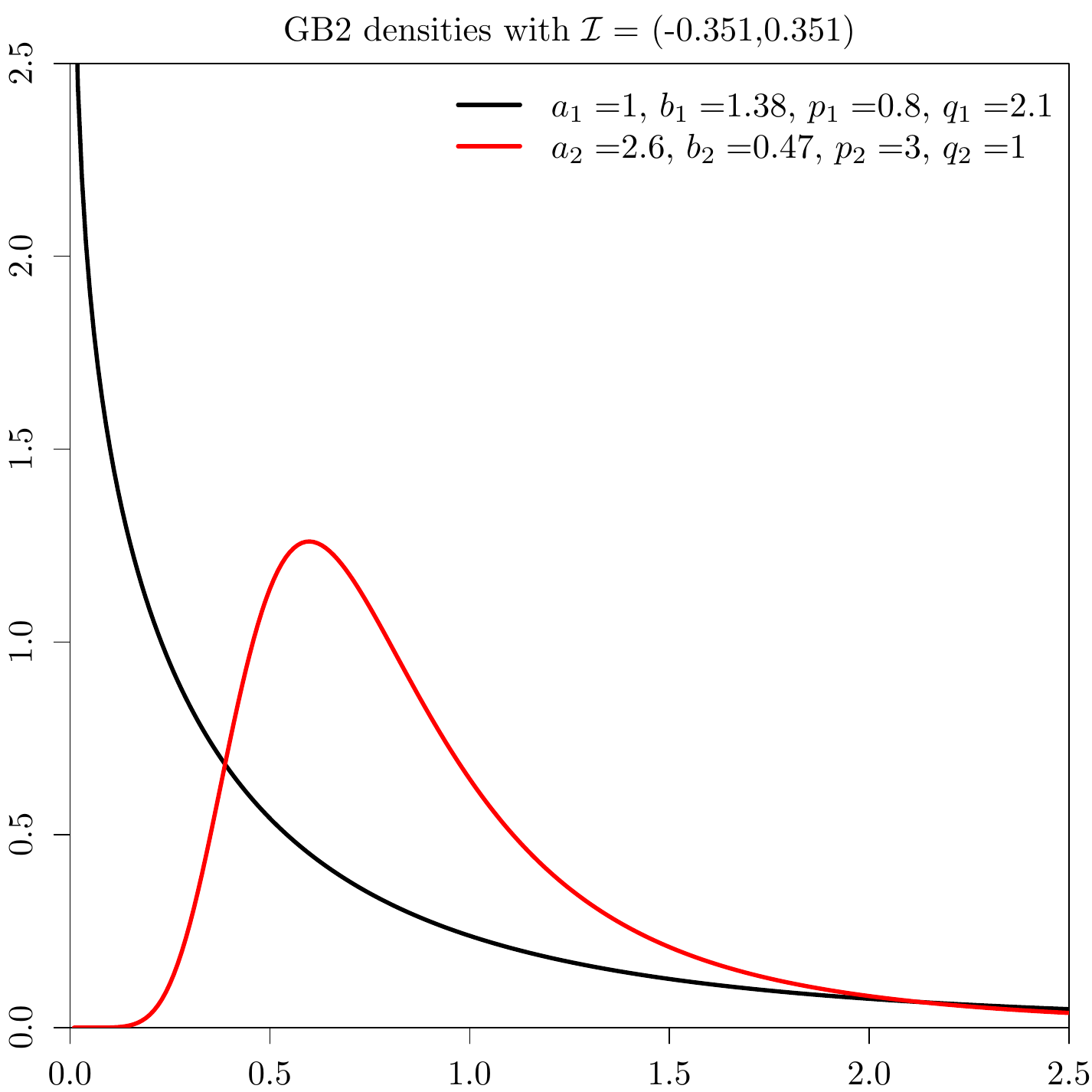} & \includegraphics[width=0.47\textwidth]{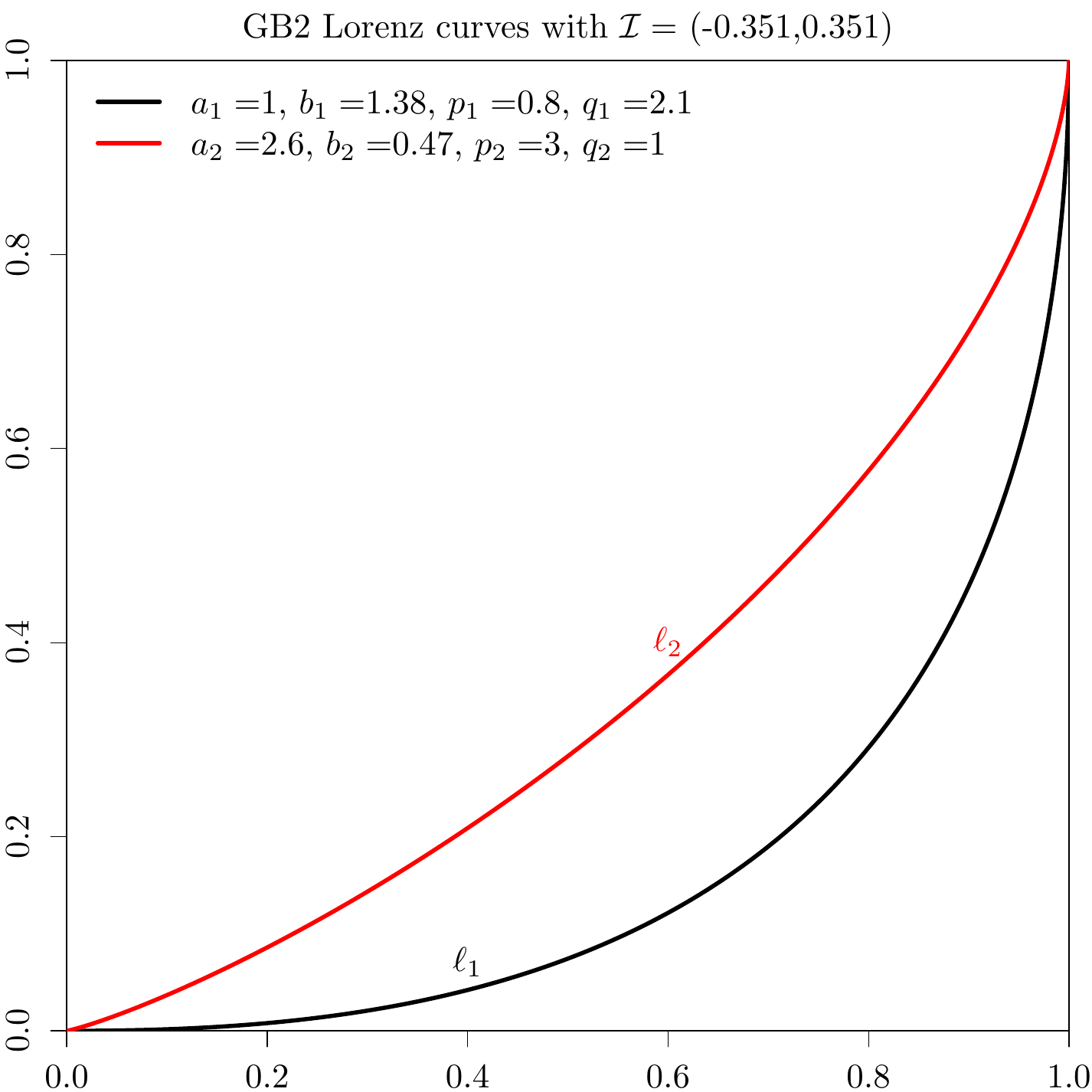}
\end{tabular}
\end{center}
\caption{Densities (left) and Lorenz curves (right) of the GB2 distributions in Model 2.}
\label{SuppMatfi:Mod2Simul}
\end{figure}

Figure~\ref{SuppMatfi:SimulMod2} displays the histograms of 1000 simulated values of \( \sqrt{n/2}(d_{\rm L}(\hat\ell_1,\hat\ell_2)-d_{\rm L}(\ell_1,\ell_2))\) for \(n=10000\) and \(n=50000\). We have superimposed a mixture of normal densities (red line with higher weight and blue line with lower weight), whose number of components and homo/heteroscedasticity were determined with the BIC. Even though the number of components was 2 for both sample sizes, the histograms have Gaussian appearance and the component of the mixture with the highest weight (black line) is almost coincident with the mixture density. In Model 2, as the index \(\mathcal I\) is one-dimensional, the convergence to normality is faster in \(n\).

\begin{figure}[H]
\begin{center}
\begin{tabular}{cc}
\includegraphics[width=0.47\textwidth]{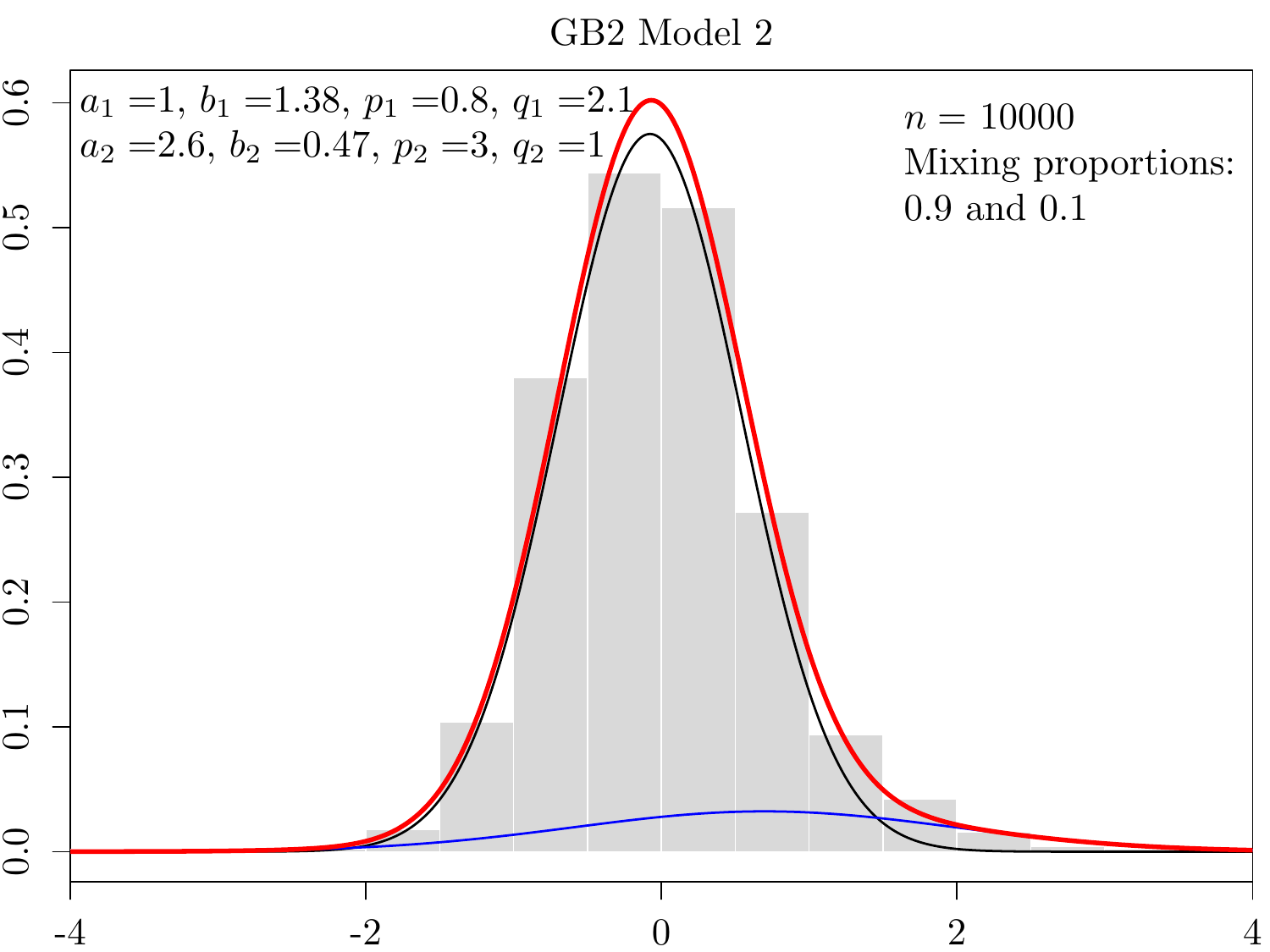} & \includegraphics[width=0.47\textwidth]{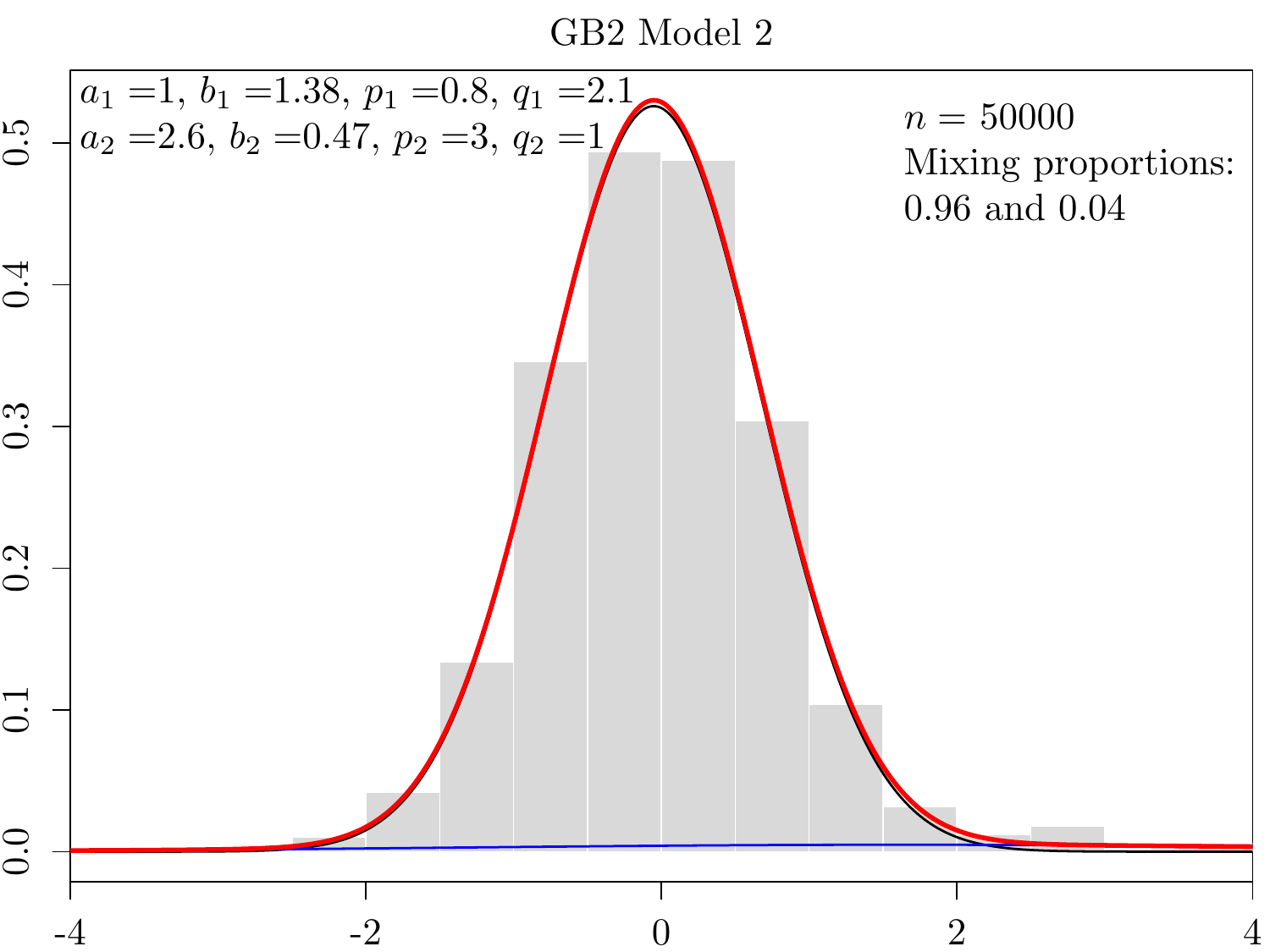}
\end{tabular}
\end{center}
\caption{Normalized values of \(d_{\rm L}(\ell_1,\ell_2)\) and level sets of normal mixture for simulated samples of Model 2 with \(n=10000\) (left) and \(n=50000\) (right).}
\label{SuppMatfi:SimulMod2}
\end{figure}

\newpage

\textbf{Model 3:}

\begin{figure}[H]
\begin{center}
\begin{tabular}{cc}
\includegraphics[width=0.47\textwidth]{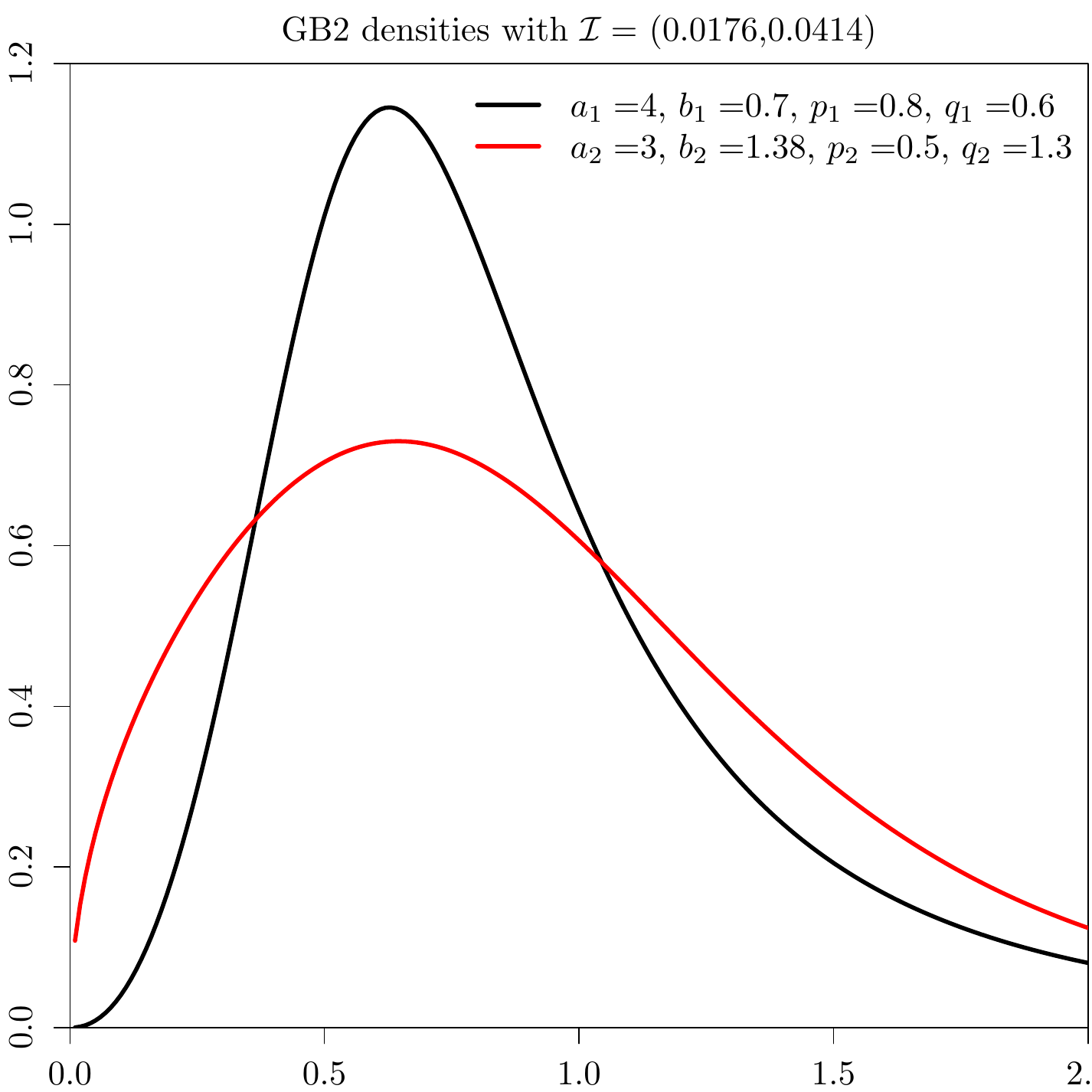} & \includegraphics[width=0.47\textwidth]{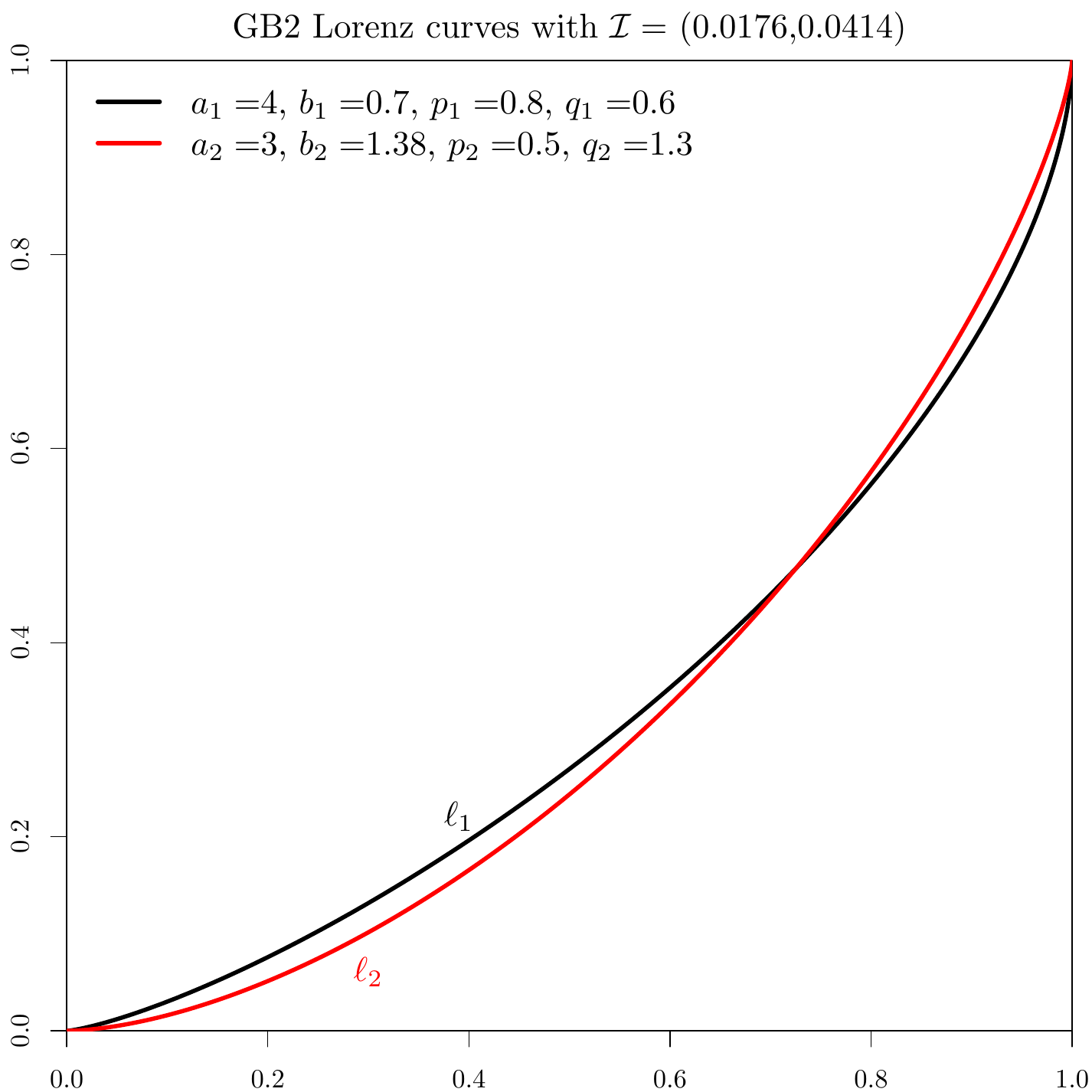}
\end{tabular}
\end{center}
\caption{Densities (left) and Lorenz curves (right) of the GB2 distributions in Model 3.}
\label{SuppMatfi:Mod3Simul}
\end{figure}

The structure of the graphics is analogous to those of Model 1.
The 1000 simulated values of the normalized \(\hat{\mathcal I}\) appear in Figure~\ref{SuppMatfi:SimulMod3}. We have superimposed the level sets (in red) of a mixture of Gaussian densities fit to these points. As before, BIC was used to choose the optimal parameters in the mixture. The points in the mixture component with weight lower than 10\% are colored in blue. The convergence to normality is clear as \(n\) increases, but it is slower than in Model 1 as the Lorenz curves in Model 3 are closer to each other than in Model 1.

\begin{figure}[H]
\begin{center}
\begin{tabular}{cc}
\includegraphics[width=0.47\textwidth]{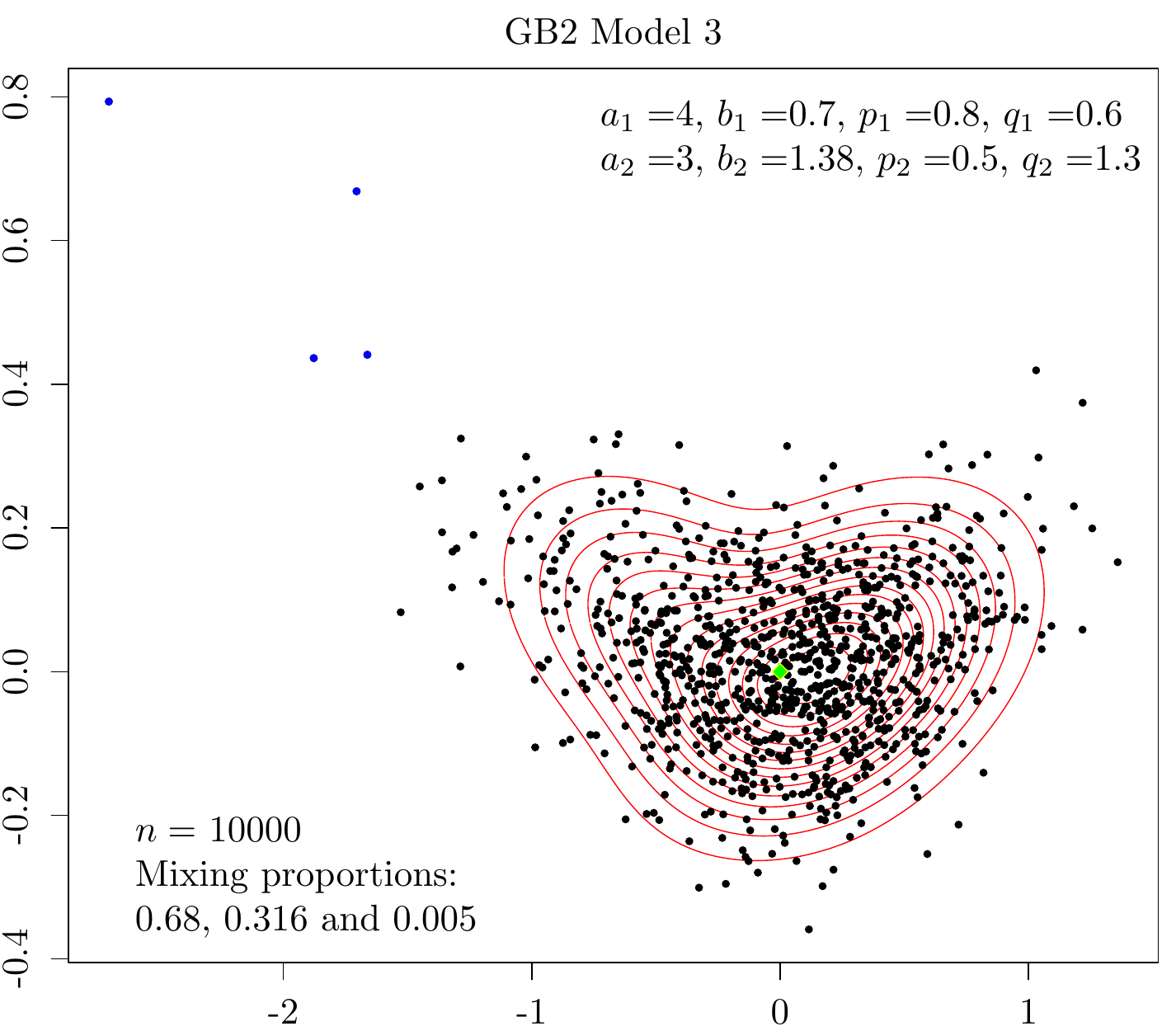} & \includegraphics[width=0.47\textwidth]{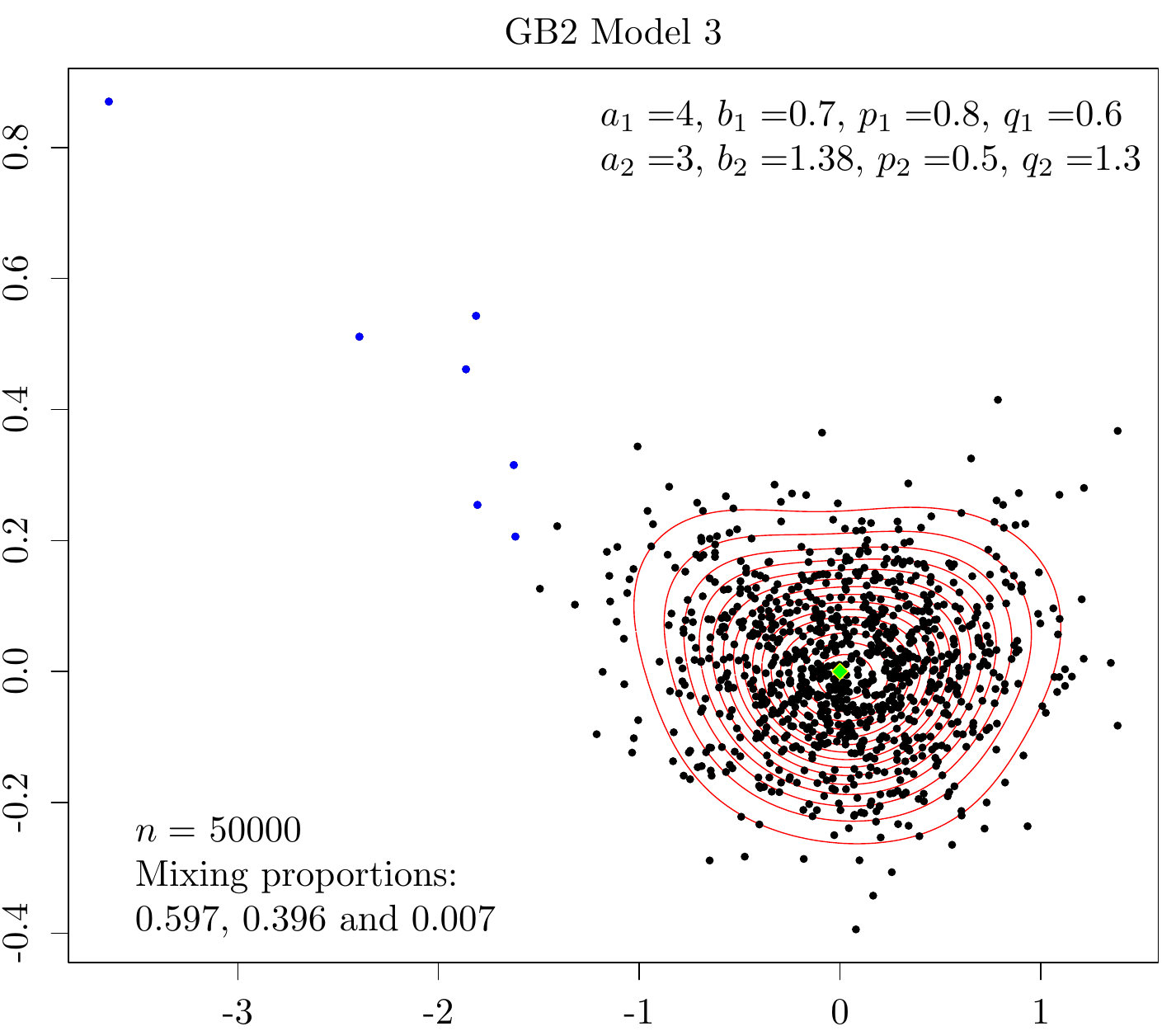}
\end{tabular}
\end{center}
\caption{Normalized values of \(\hat{\mathcal I}\) and level sets of normal mixture for simulated samples of Model 3 with \(n=10000\) (left) and \(n=50000\) (right).}
\label{SuppMatfi:SimulMod3}
\end{figure}

\newpage

\textbf{Model 4:}

\begin{figure}[H]
\begin{center}
\begin{tabular}{cc}
\includegraphics[width=0.47\textwidth]{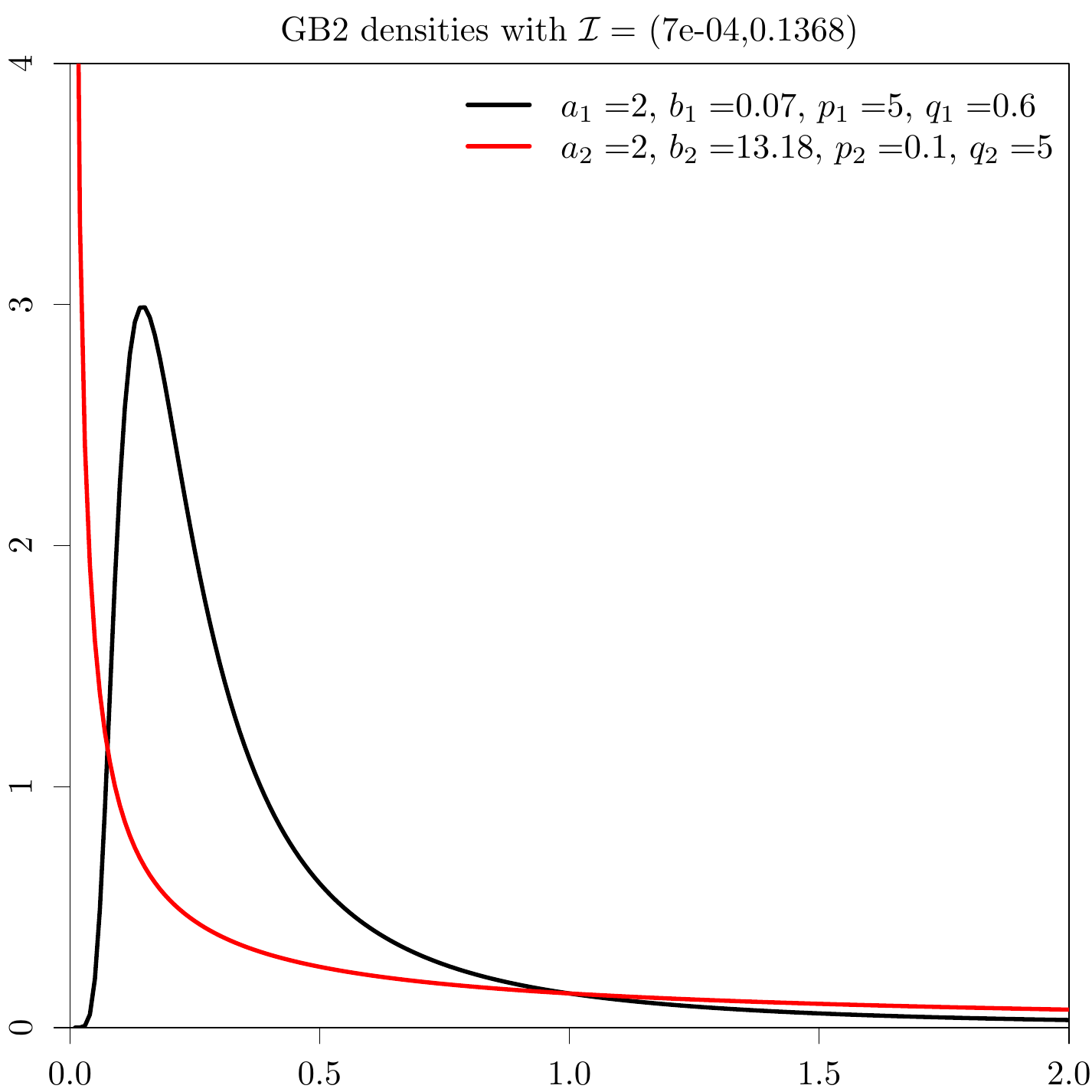} & \includegraphics[width=0.47\textwidth]{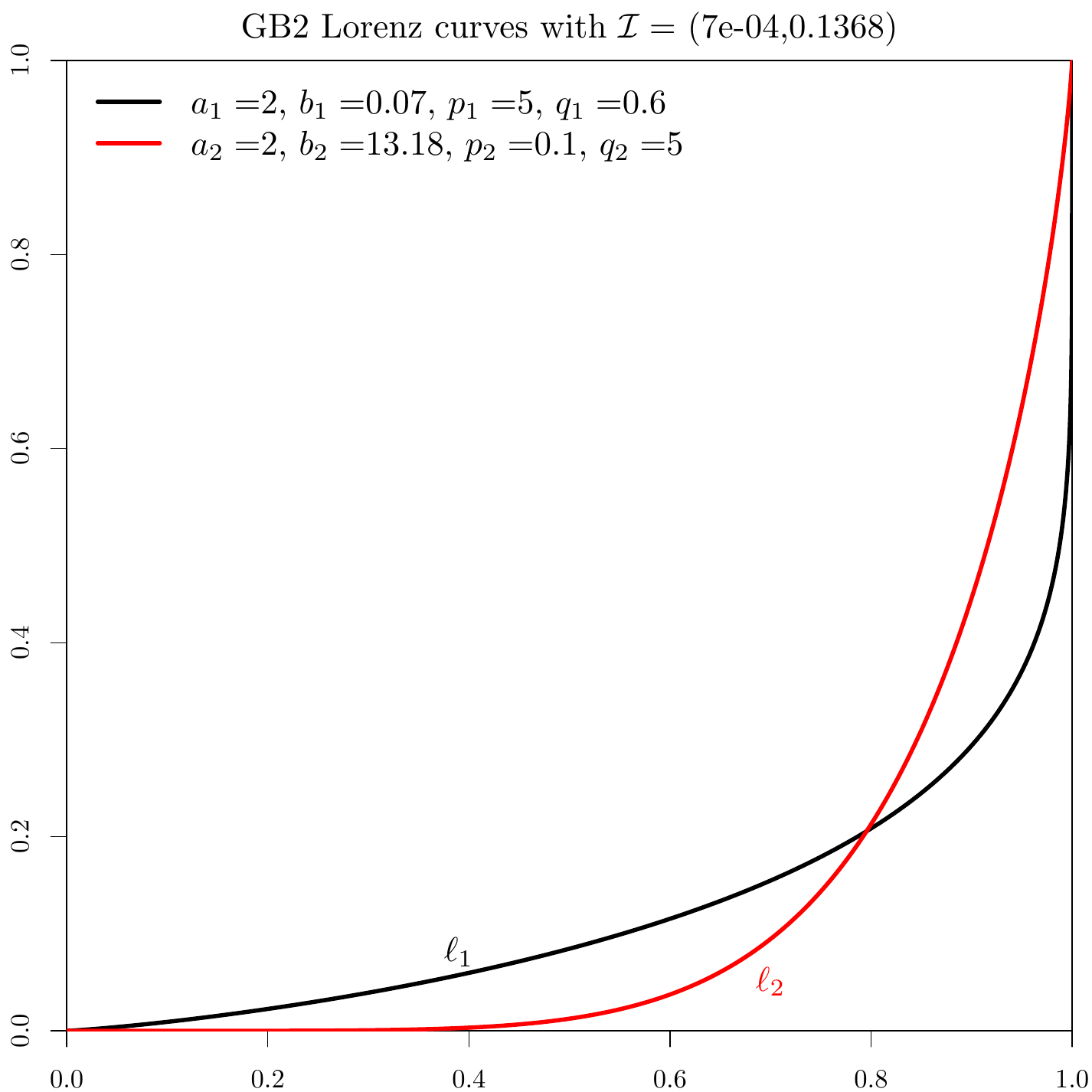}
\end{tabular}
\end{center}
\caption{Densities (left) and Lorenz curves (right) of the GB2 distributions in Model 4.}
\label{SuppMatfi:Mod4Simul}
\end{figure}

In this case we have taken the sample sizes equal to \(n=10^4\) and \(n=10^5\) to illustrate that there is no convergence in distribution to normality.
The 1000 simulated values of the empirical inequality index \(\hat{\mathcal I}\) appear in Figures~\ref{SuppMatfi:SimulMod4} and~\ref{SuppMatfi:SimulMod4ZOOM}. The green dot marks the population index \(\mathcal I\). We see that there is consistency of the empirical inequality index \(\hat{\mathcal I}\), but the rate of convergence seems slow.

\begin{figure}[H]
\begin{center}
\begin{tabular}{cc}
\includegraphics[width=0.47\textwidth]{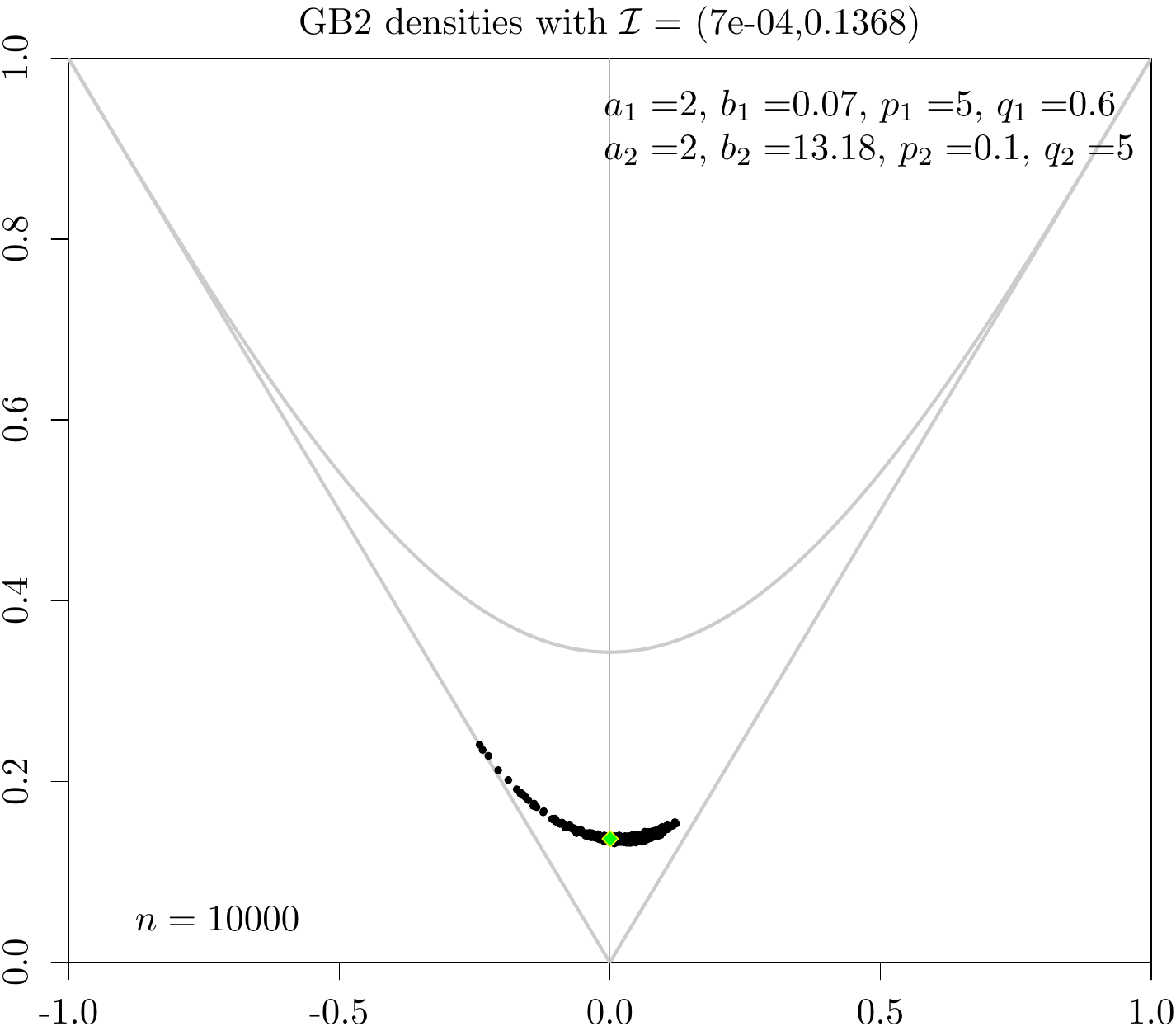} & \includegraphics[width=0.47\textwidth]{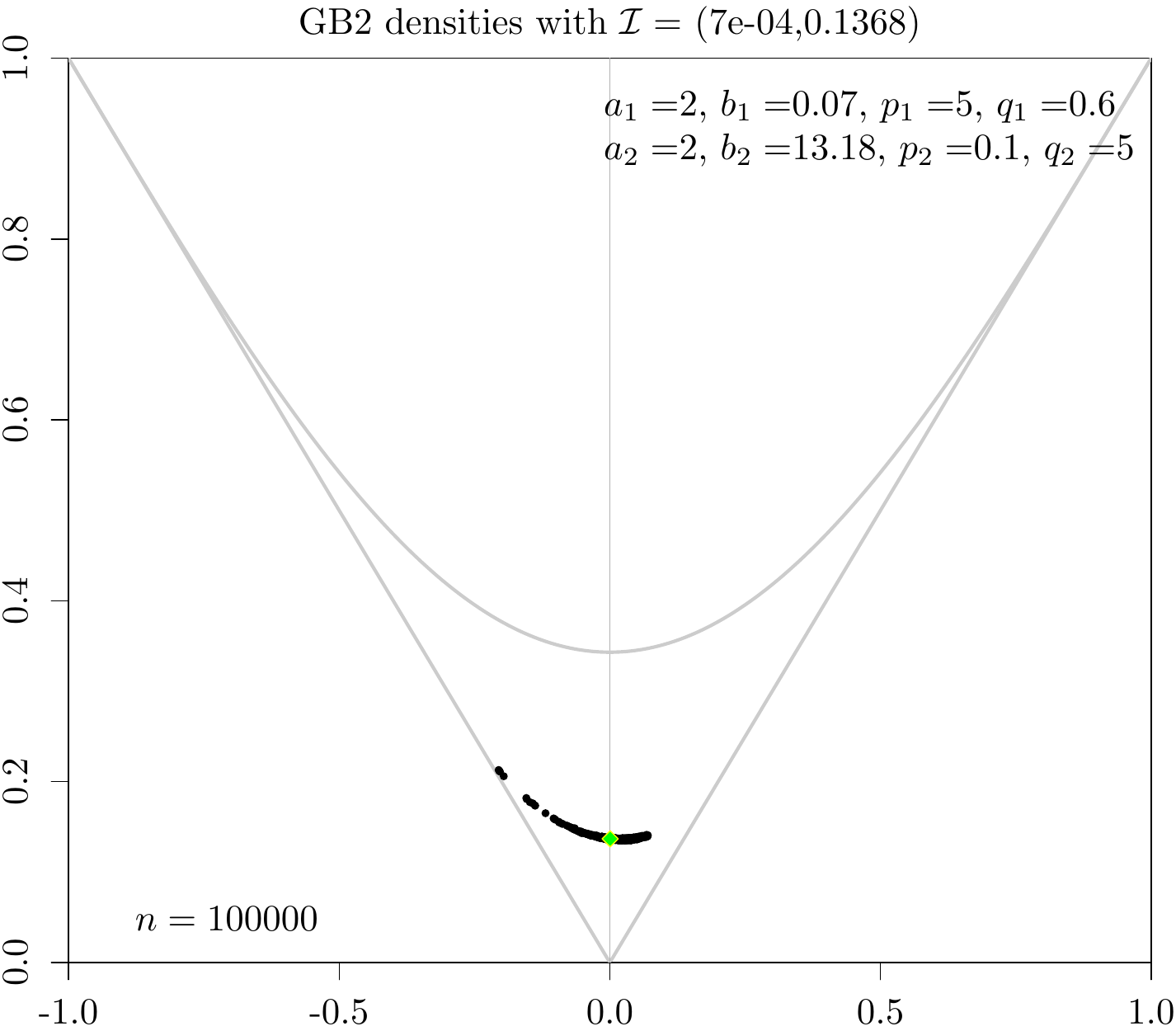}
\end{tabular}
\end{center}
\caption{Values of \(\hat{\mathcal I}\) for samples of Model 4 with \(n=10000\) (left) and \(n=100000\) (right).}
\label{SuppMatfi:SimulMod4}
\end{figure}

\begin{figure}[H]
\begin{center}
\begin{tabular}{cc}
\includegraphics[width=0.47\textwidth]{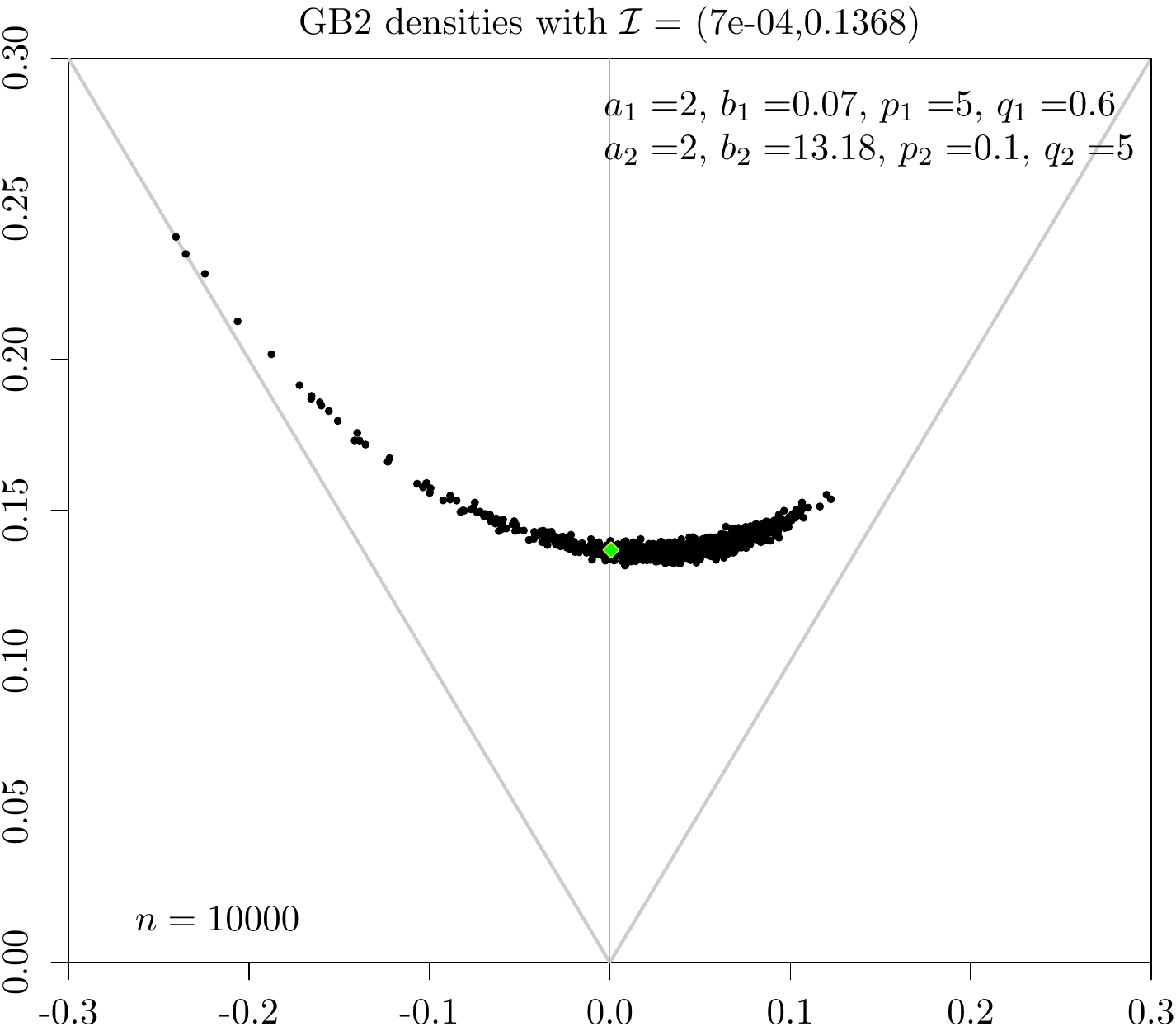} & \includegraphics[width=0.47\textwidth]{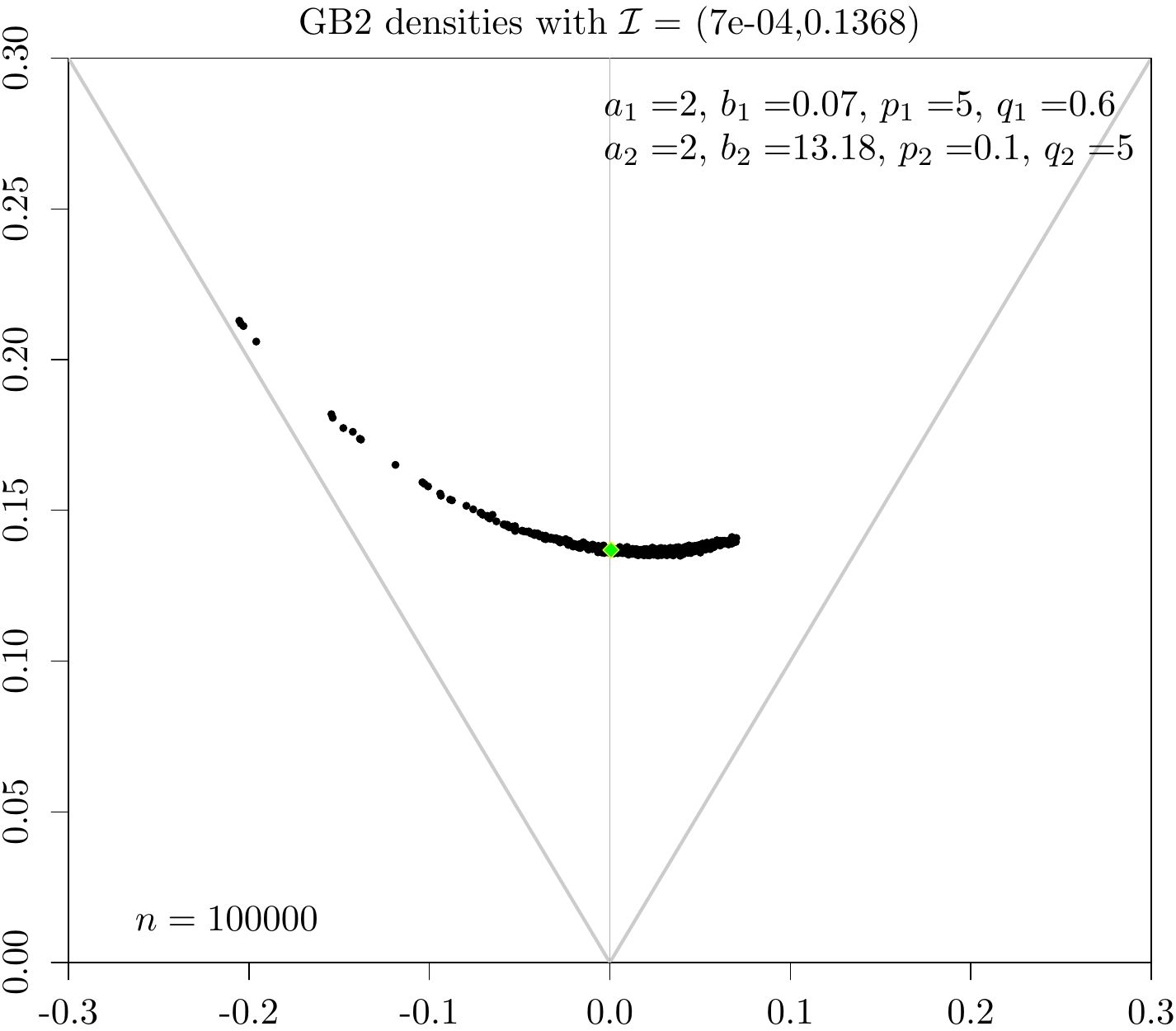}
\end{tabular}
\end{center}
\caption{Values of \(\hat{\mathcal I}\) for samples of Model 4 with \(n=10000\) (left) and \(n=100000\) (right). Enlarged version of Figure~\ref{SuppMatfi:SimulMod4}.}
\label{SuppMatfi:SimulMod4ZOOM}
\end{figure}

\newpage

\textbf{Model 5:}

Recall that in Model 5 the variables $X_1$ and $X_2$ follow the same distribution, so the real value of \(\mathcal I\) is (0,0).
The 1000 simulated values of the normalized empirical index \(\hat{\mathcal I}\) appear in Figure~\ref{SuppMatfi:SimulMod5}.

\begin{figure}[H]
\begin{center}
\begin{tabular}{cc}
\includegraphics[width=0.47\textwidth]{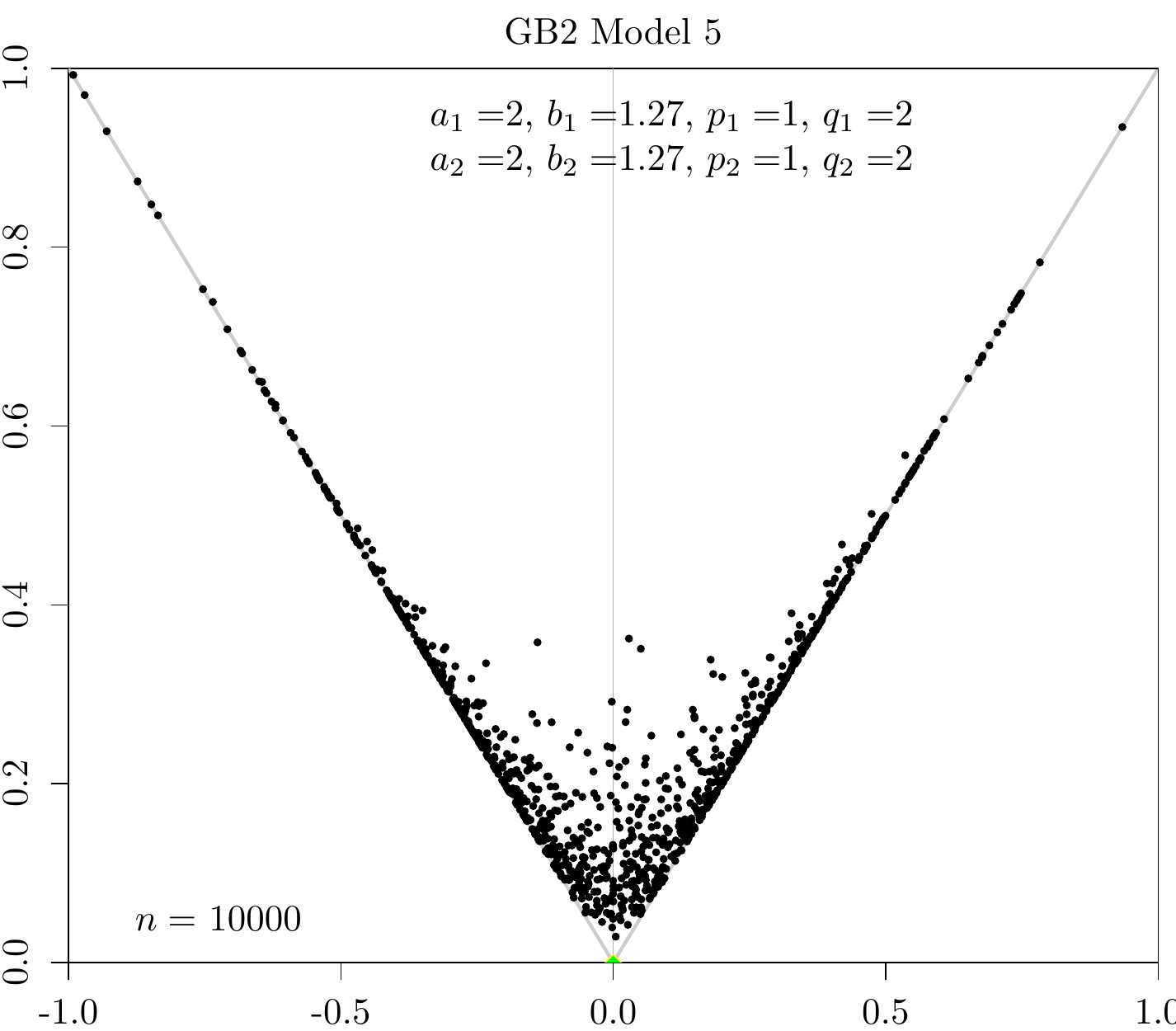} & \includegraphics[width=0.47\textwidth]{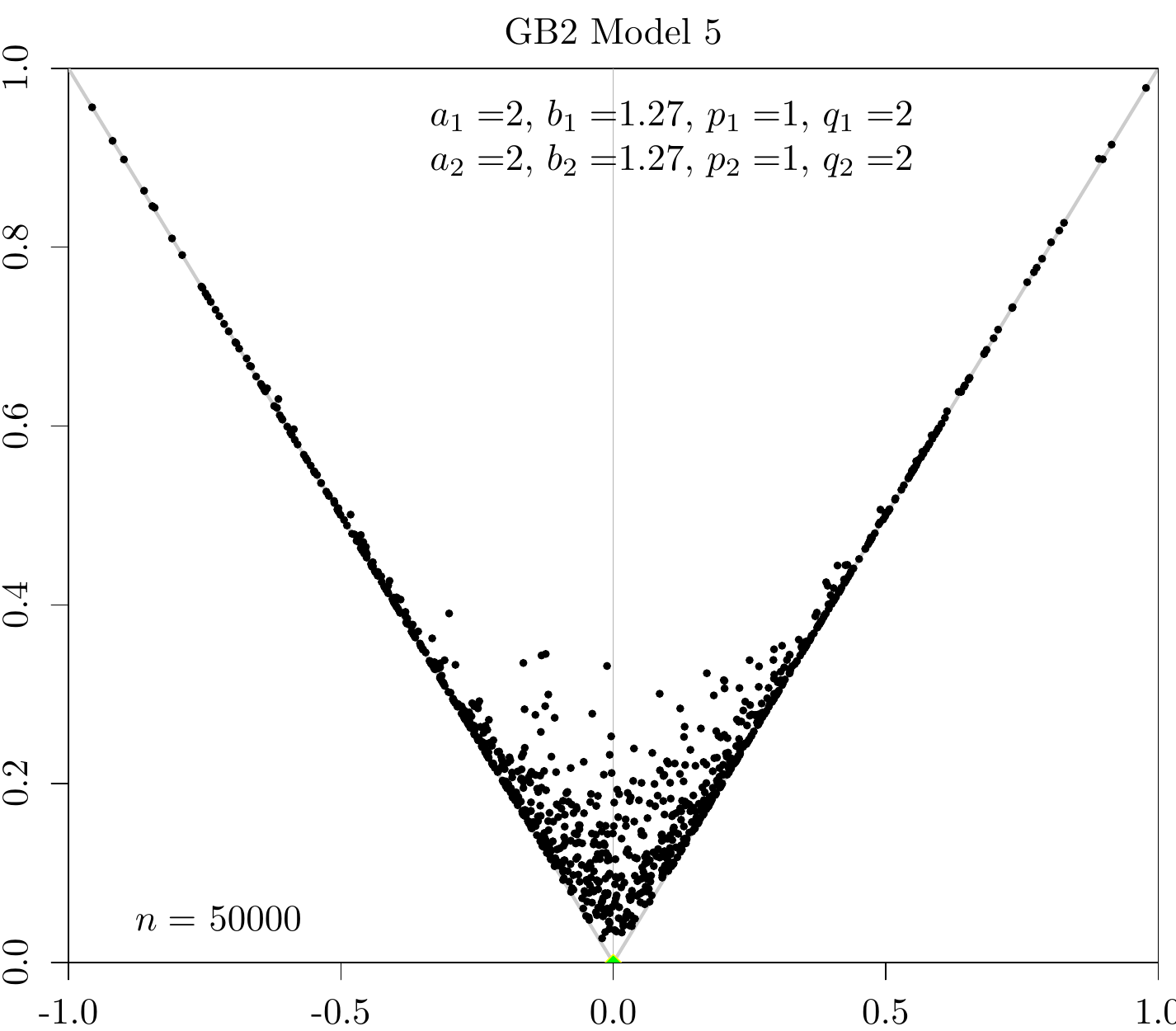}
\end{tabular}
\end{center}
\caption{Normalized values of \(\hat{\mathcal I}\)  for simulated samples of Model 5 with \(n=10000\) (left) and \(n=50000\) (right).}
\label{SuppMatfi:SimulMod5}
\end{figure}

\newpage


\section{Application to EU-SILC income data} \label{SuppMatSection.RealData}

\subsection{Yearly evolution of inequality in Spain with respect to 2008}

In Figures~\ref{SuppMatfi:SpainLorenz1}--\ref{SuppMatfi:SpainLorenz4}, on the left we display the Lorenz curves, \(\hat\ell_1\) and \(\hat\ell_2\), of Spanish income for 2008 and a year between 2009 and 2019, respectively. Since the two Lorenz curves are always very similar and it is difficult to appreciate the change from 2008 to the other year, on the right we plot the difference of the two Lorenz curves, \(\hat\ell_1-\hat\ell_2\), scaled by the supremum norm of this difference, \(\|\hat\ell_1-\hat\ell_2\|_\infty\).

\begin{figure}[H]
\begin{center}
\begin{tabular}{cc}
\includegraphics[width=0.47\textwidth]{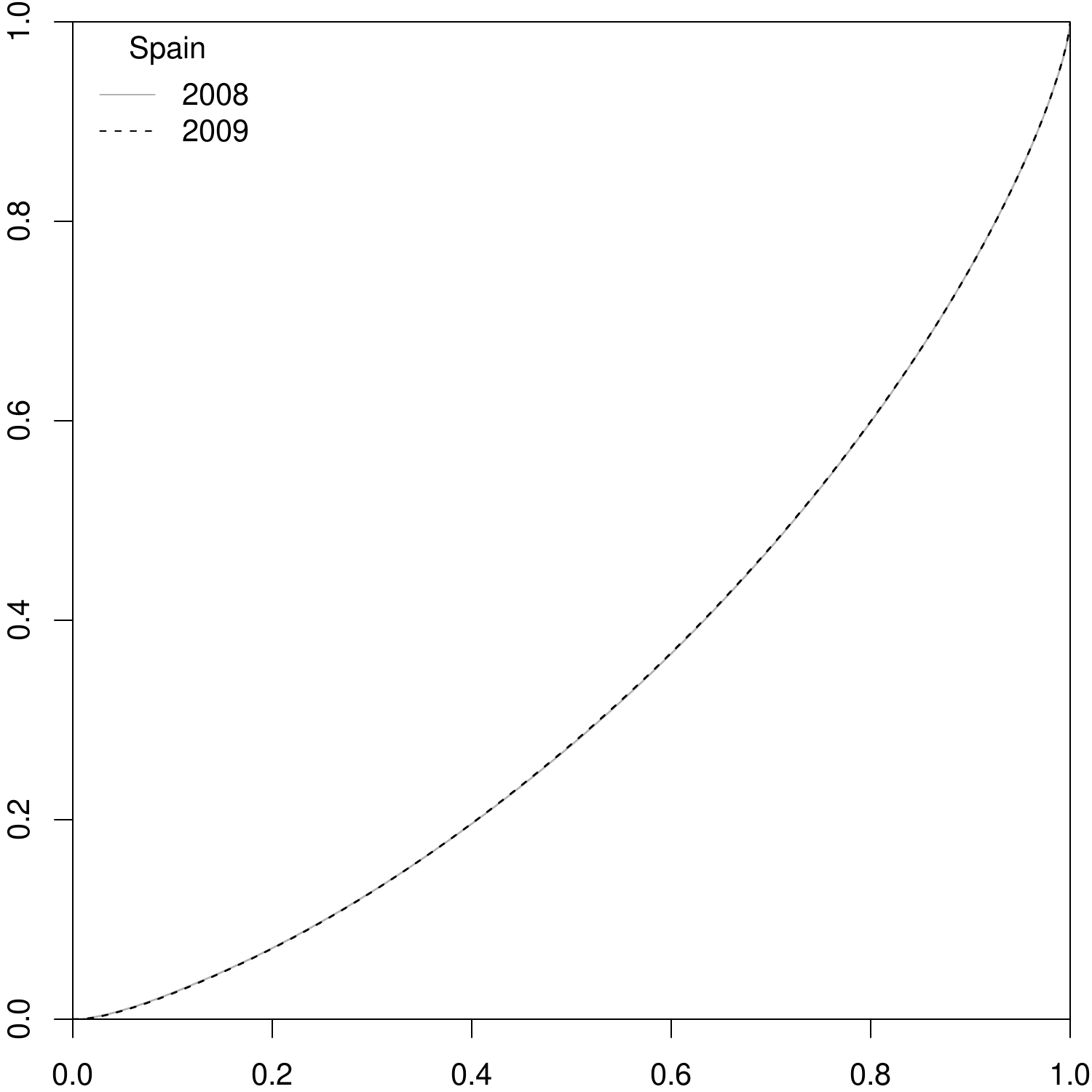} & \includegraphics[width=0.47\textwidth]{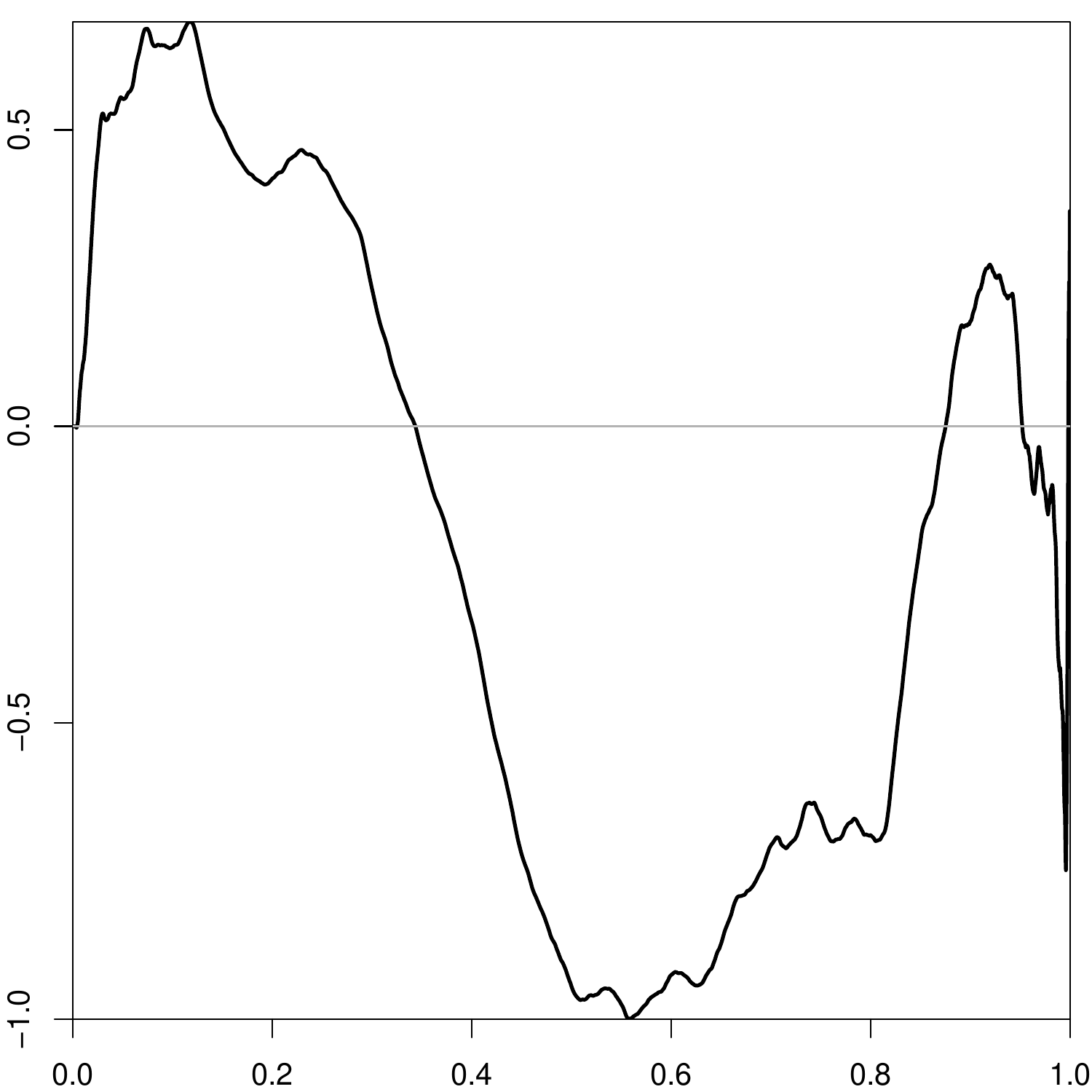}  \\
\includegraphics[width=0.47\textwidth]{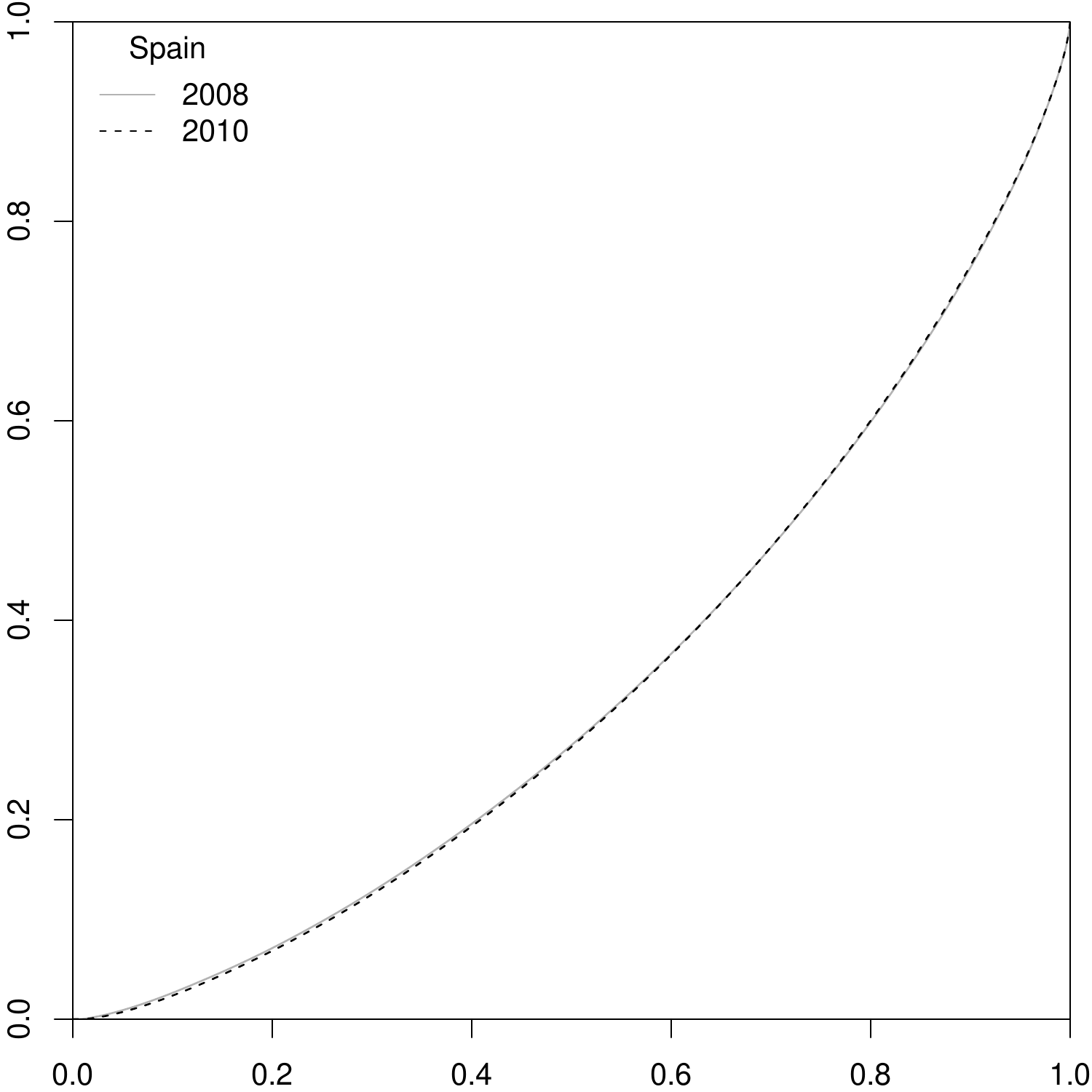} & \includegraphics[width=0.47\textwidth]{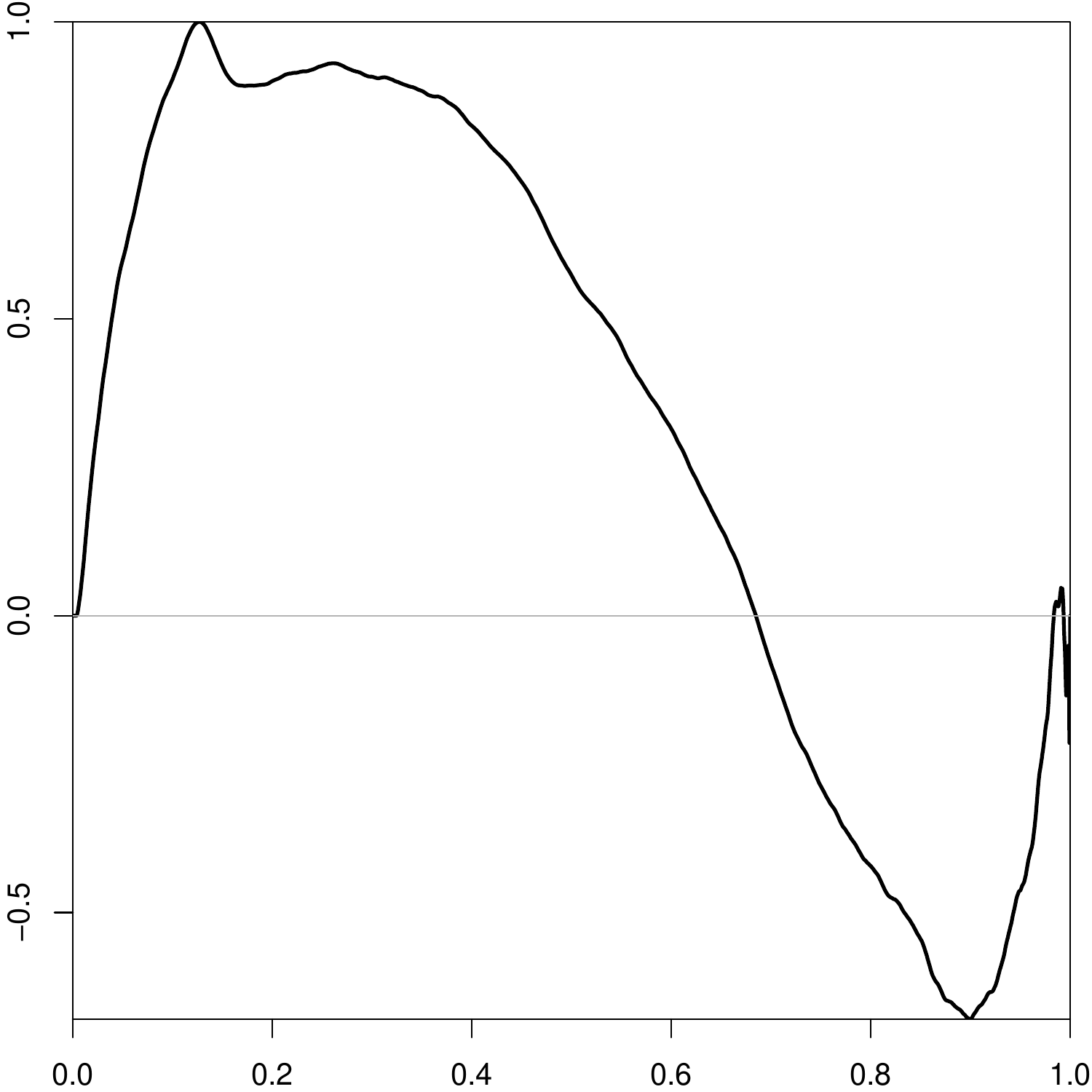}
\end{tabular}
\end{center}
\caption{The Lorenz curves $\hat\ell_1$ and $\hat\ell_2$ (left column) of income distribution and their difference scaled by the maximum absolute difference, $(\hat\ell_1-\hat\ell_2)/\|\hat\ell_1-\hat\ell_2\|_\infty$, (right column) corresponding to the pairs of years formed by 2008 ($X_1$) and one in the span 2009--2019.}
\label{SuppMatfi:SpainLorenz1}
\end{figure}

\begin{figure}[H]
\begin{center}
\begin{tabular}{cc}
\includegraphics[width=0.47\textwidth]{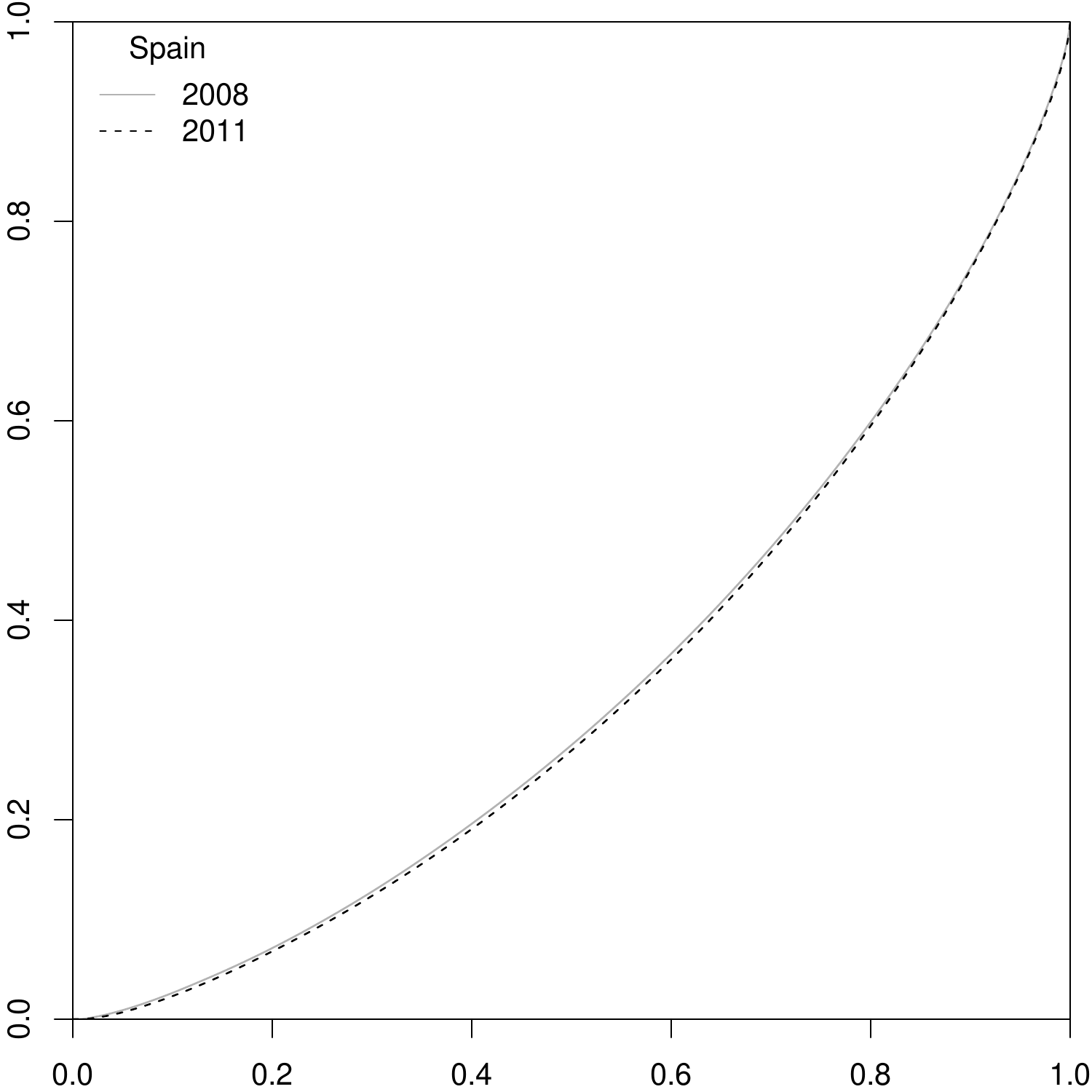} & \includegraphics[width=0.47\textwidth]{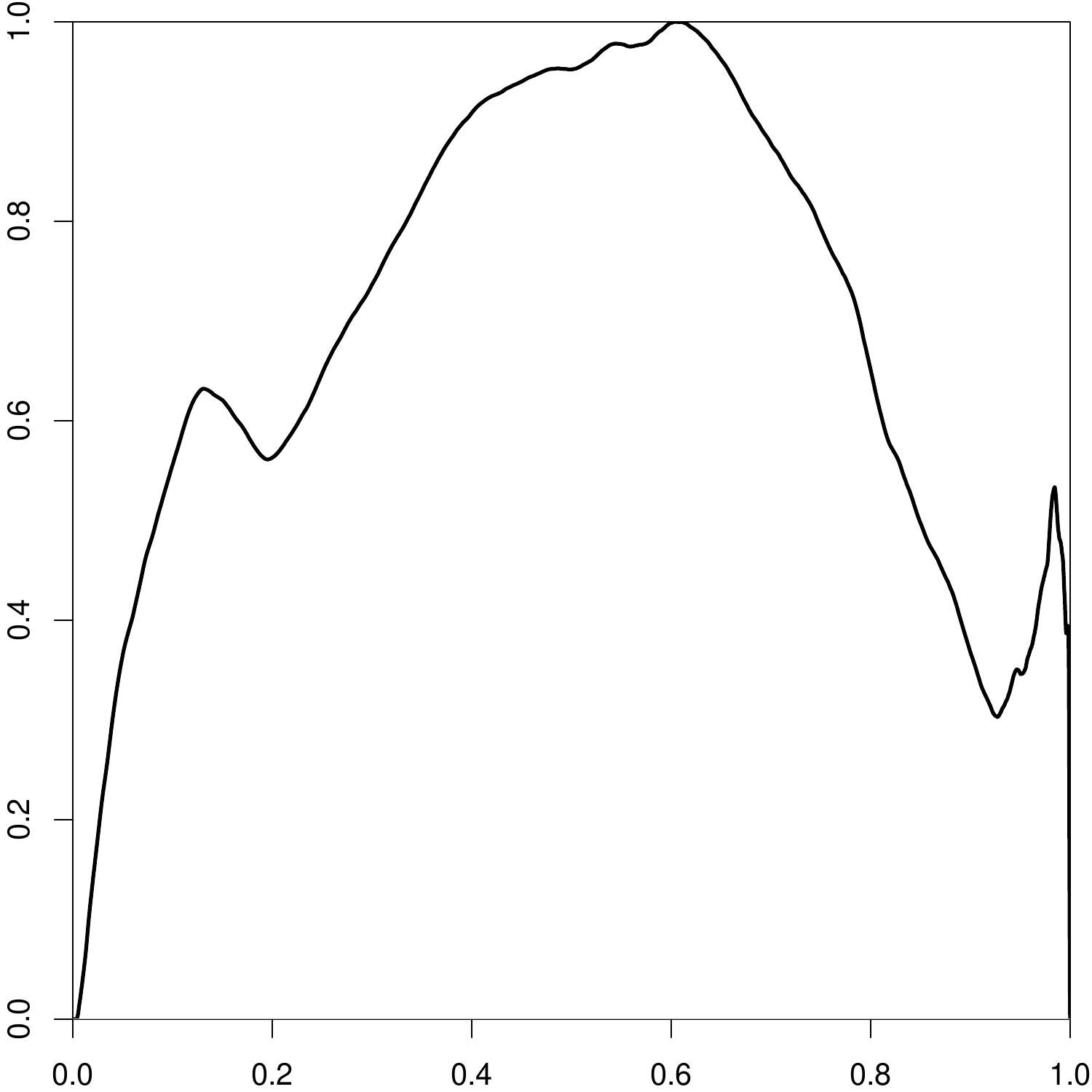} \\
\includegraphics[width=0.47\textwidth]{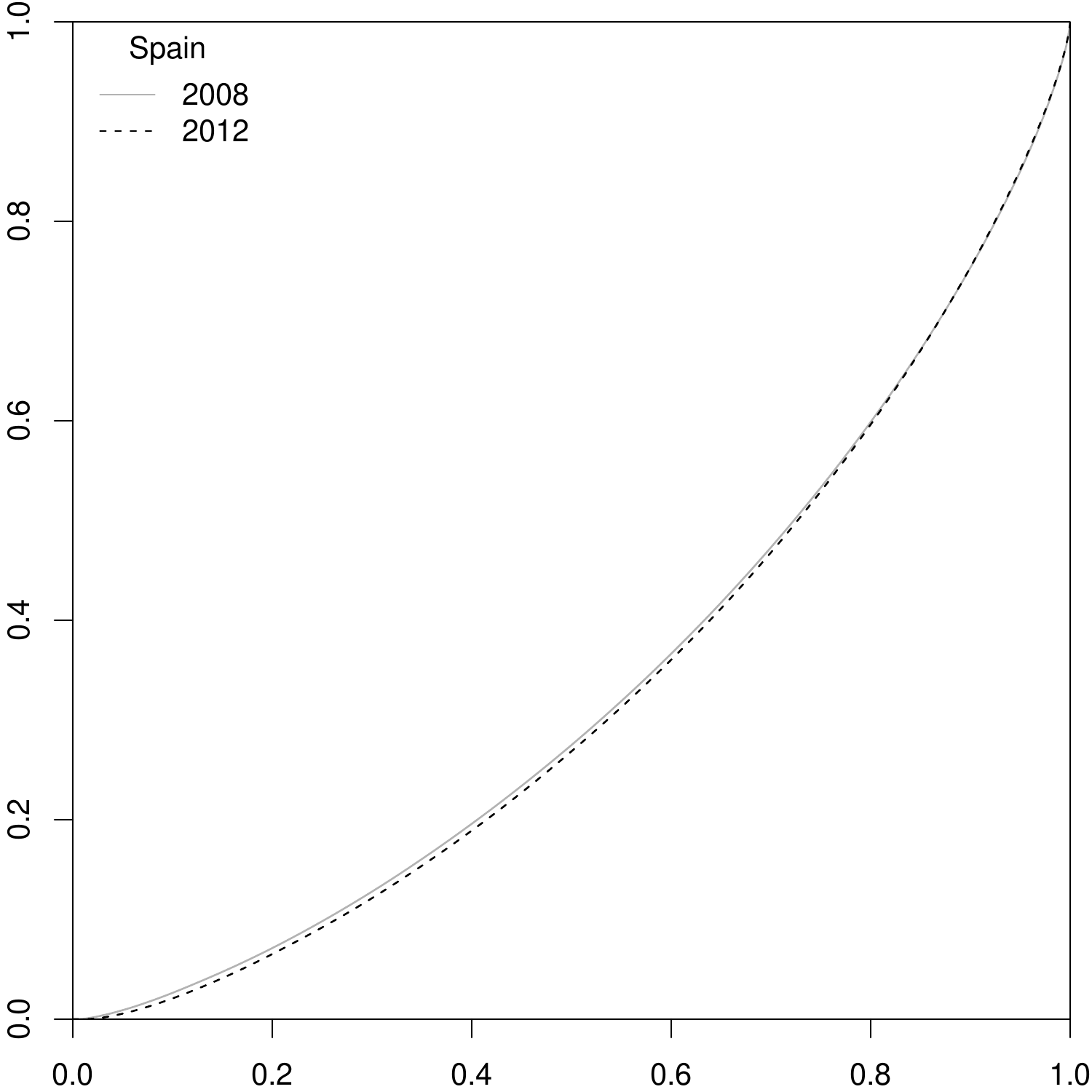} & \includegraphics[width=0.47\textwidth]{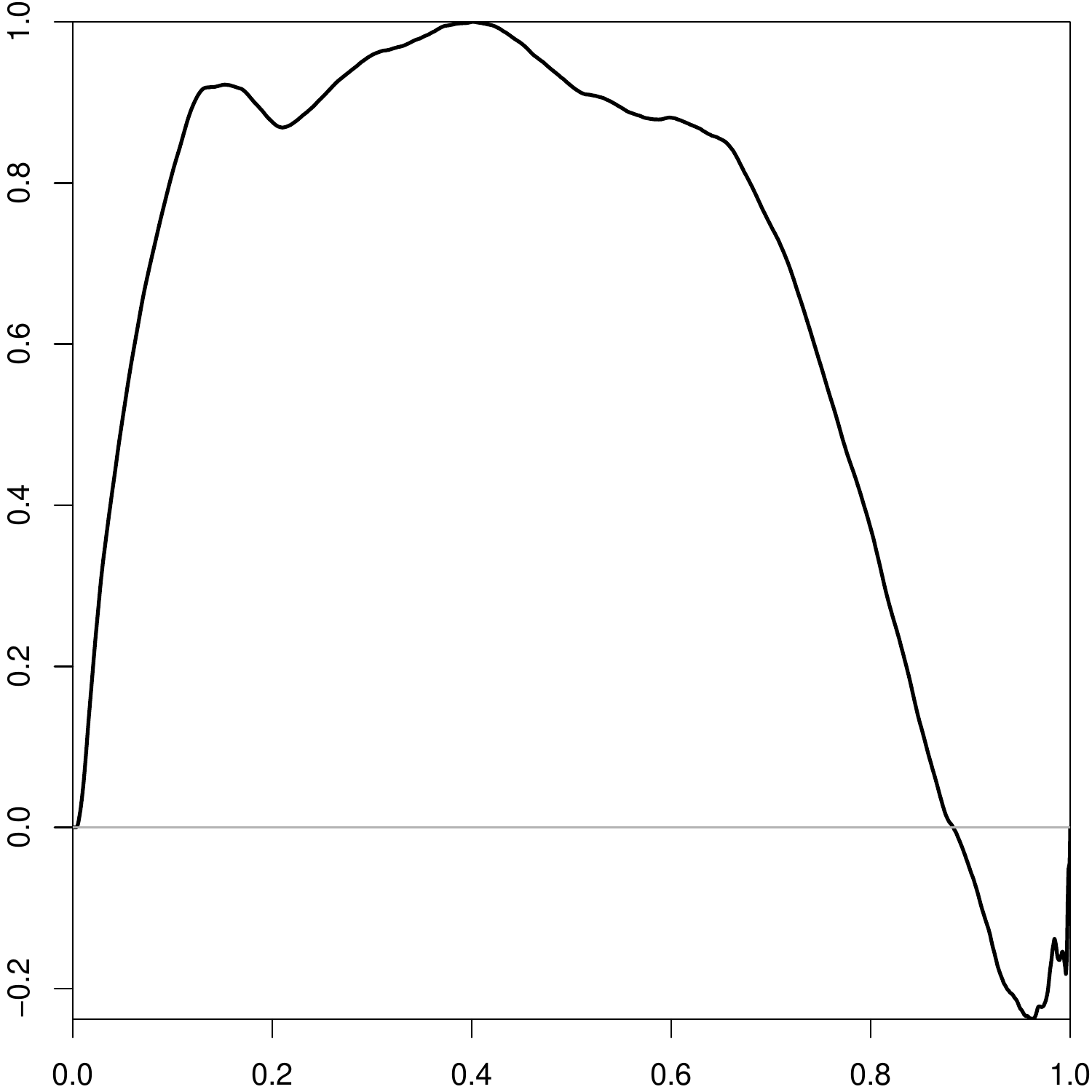} \\
\includegraphics[width=0.47\textwidth]{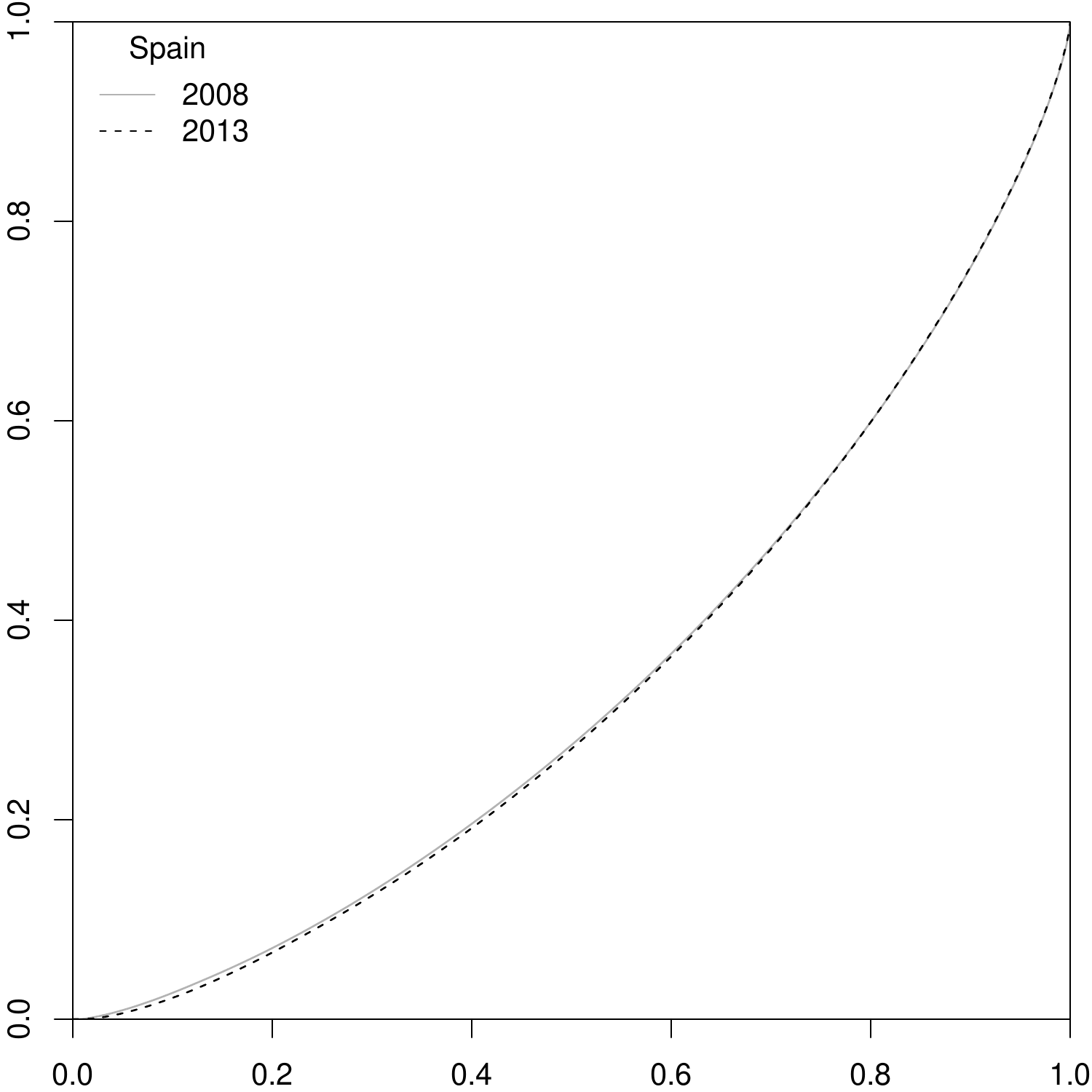} & \includegraphics[width=0.47\textwidth]{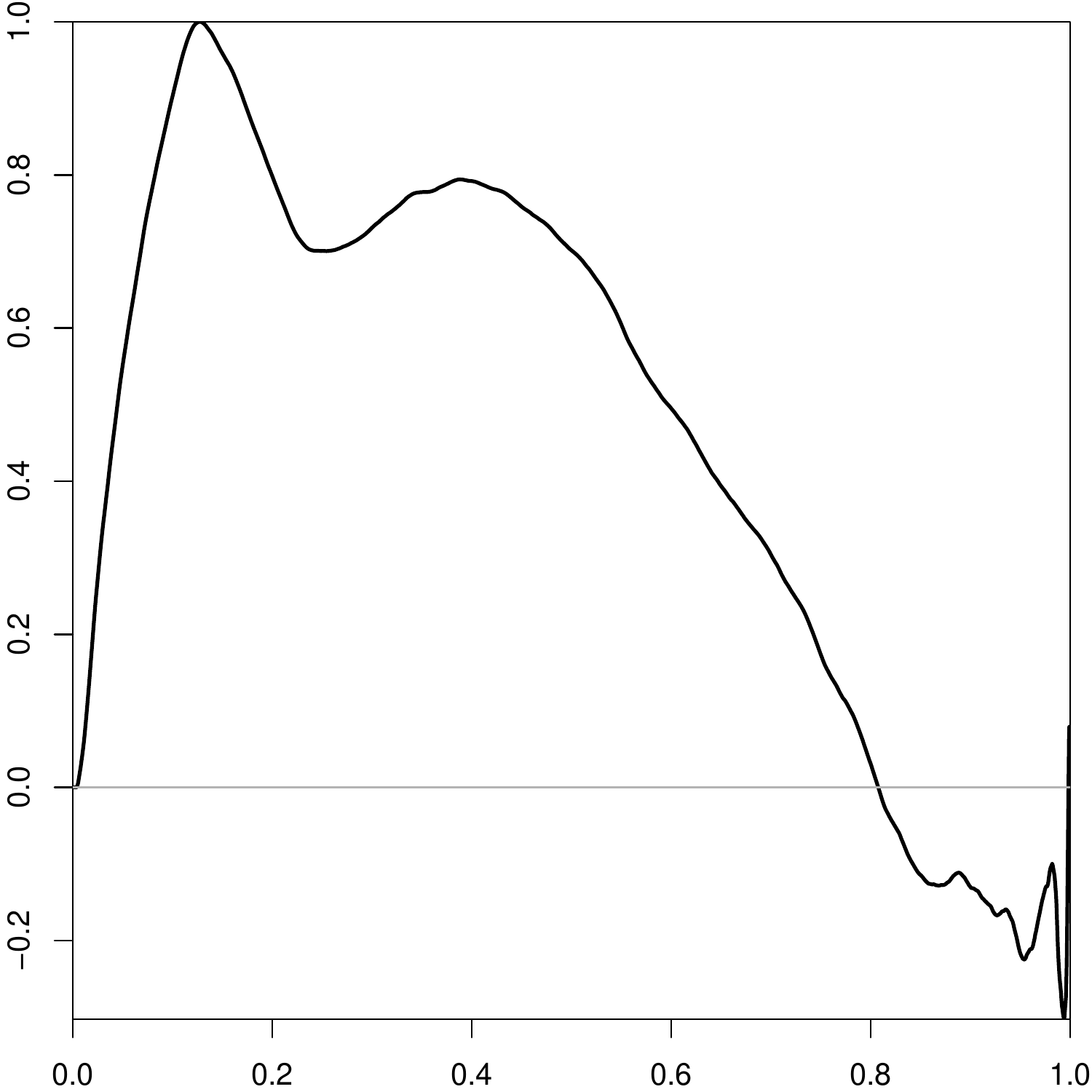}
\end{tabular}
\end{center}
\caption{The Lorenz curves $\hat\ell_1$ and $\hat\ell_2$ (left column) of income distribution and their difference scaled by the maximum absolute difference, $(\hat\ell_1-\hat\ell_2)/\|\hat\ell_1-\hat\ell_2\|_\infty$, (right column) corresponding to the pairs of years formed by 2008 ($X_1$) and one in the span 2009--2019.}
\label{SuppMatfi:SpainLorenz2}
\end{figure}

\begin{figure}[H]
\begin{center}
\begin{tabular}{cc}
\includegraphics[width=0.47\textwidth]{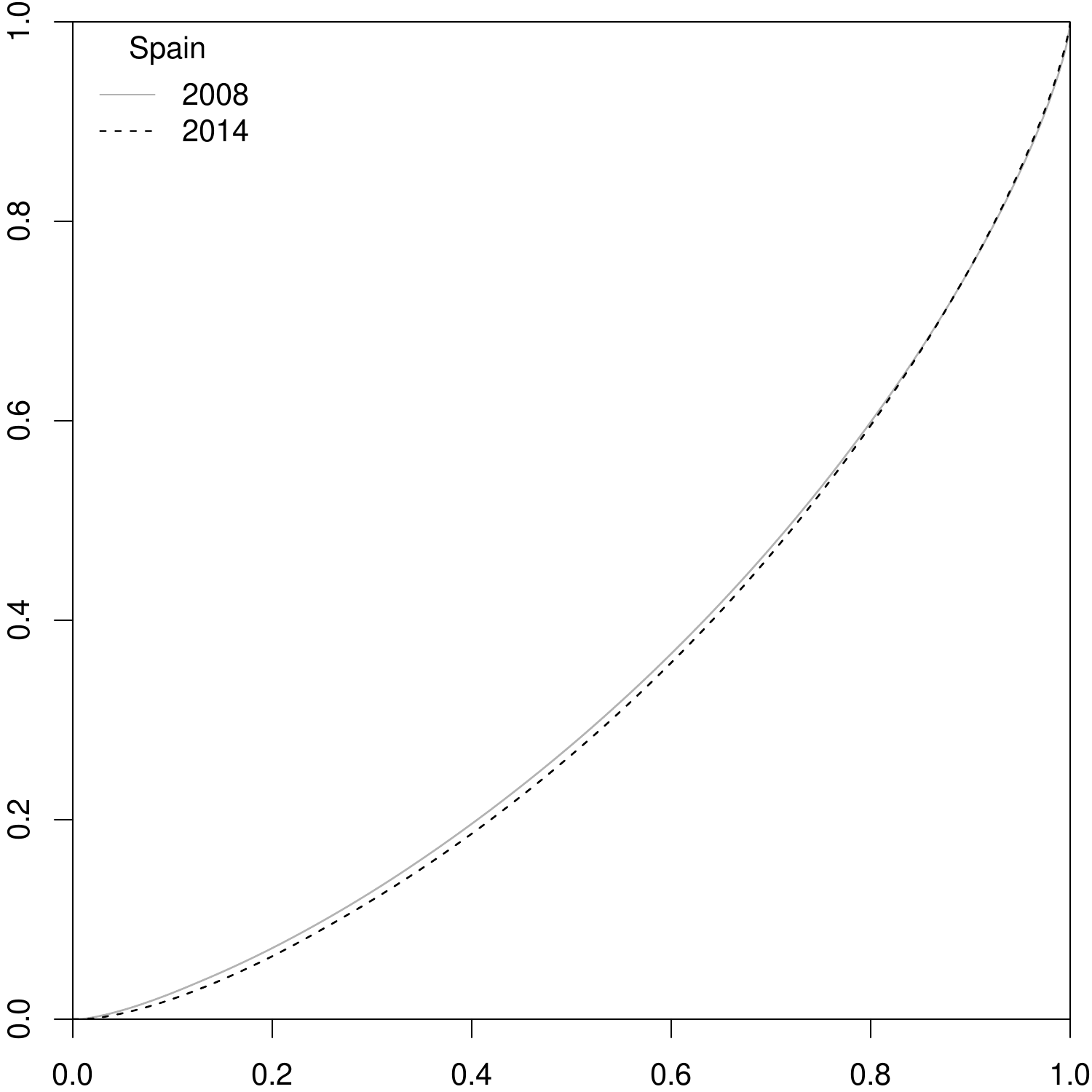} & \includegraphics[width=0.47\textwidth]{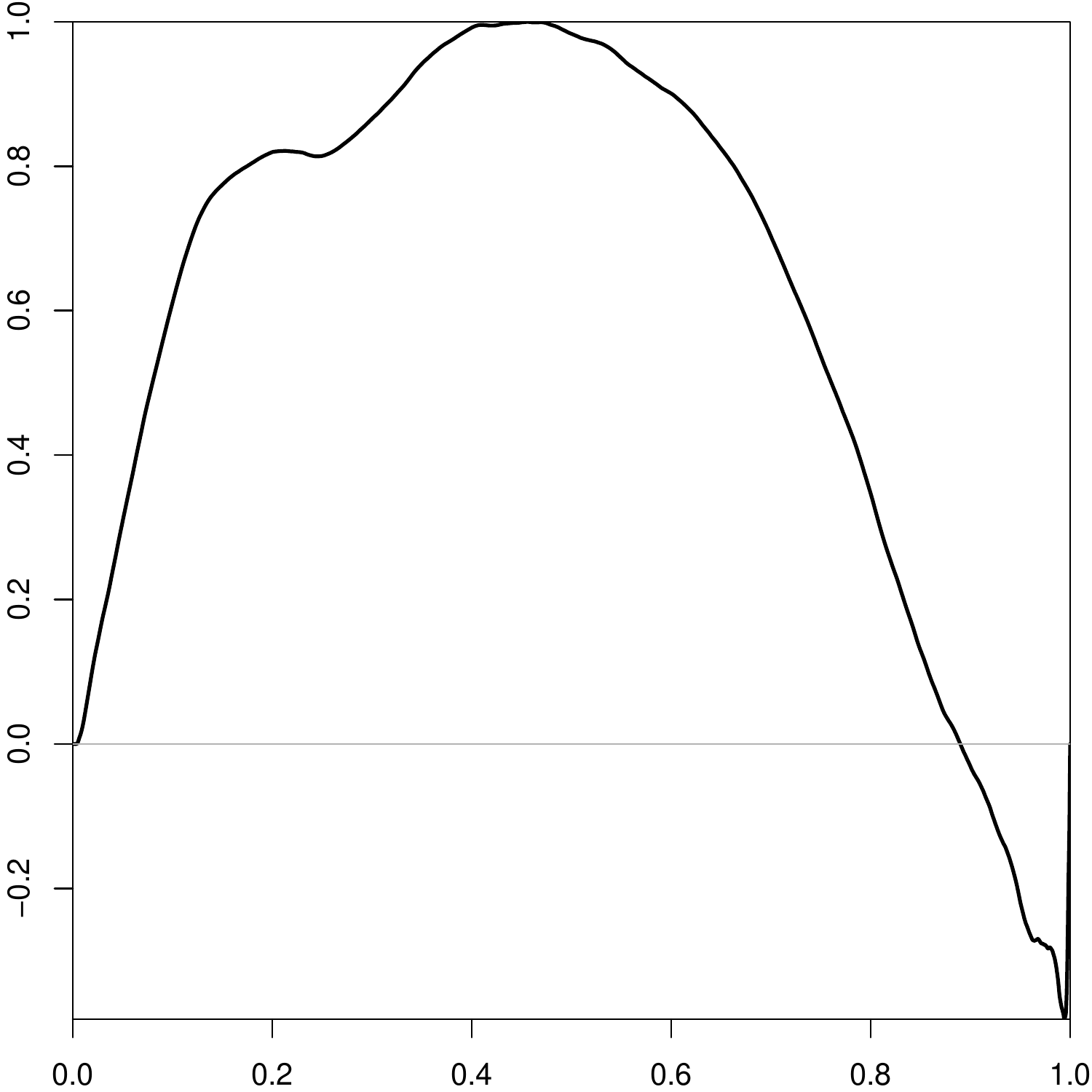} \\
\includegraphics[width=0.47\textwidth]{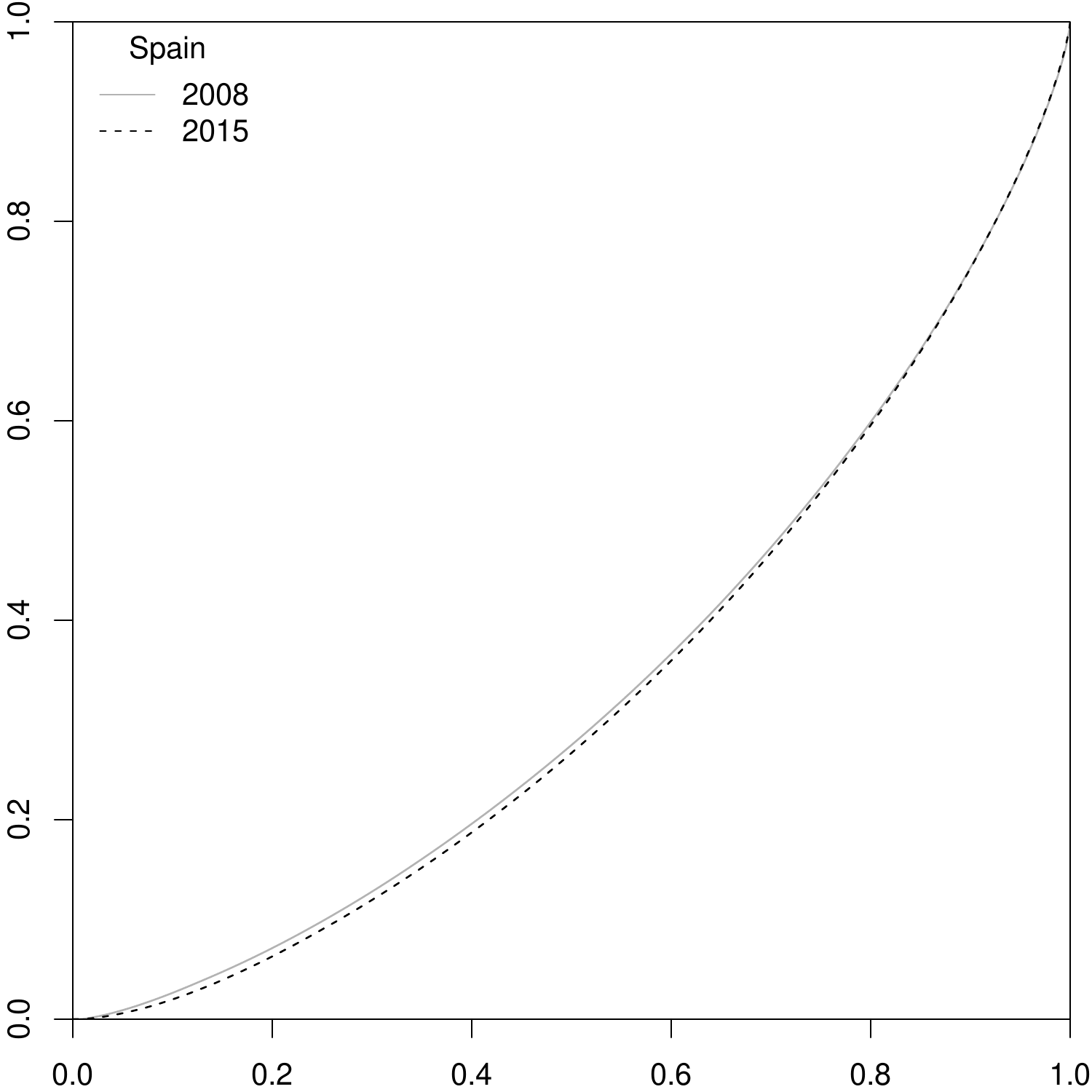} & \includegraphics[width=0.47\textwidth]{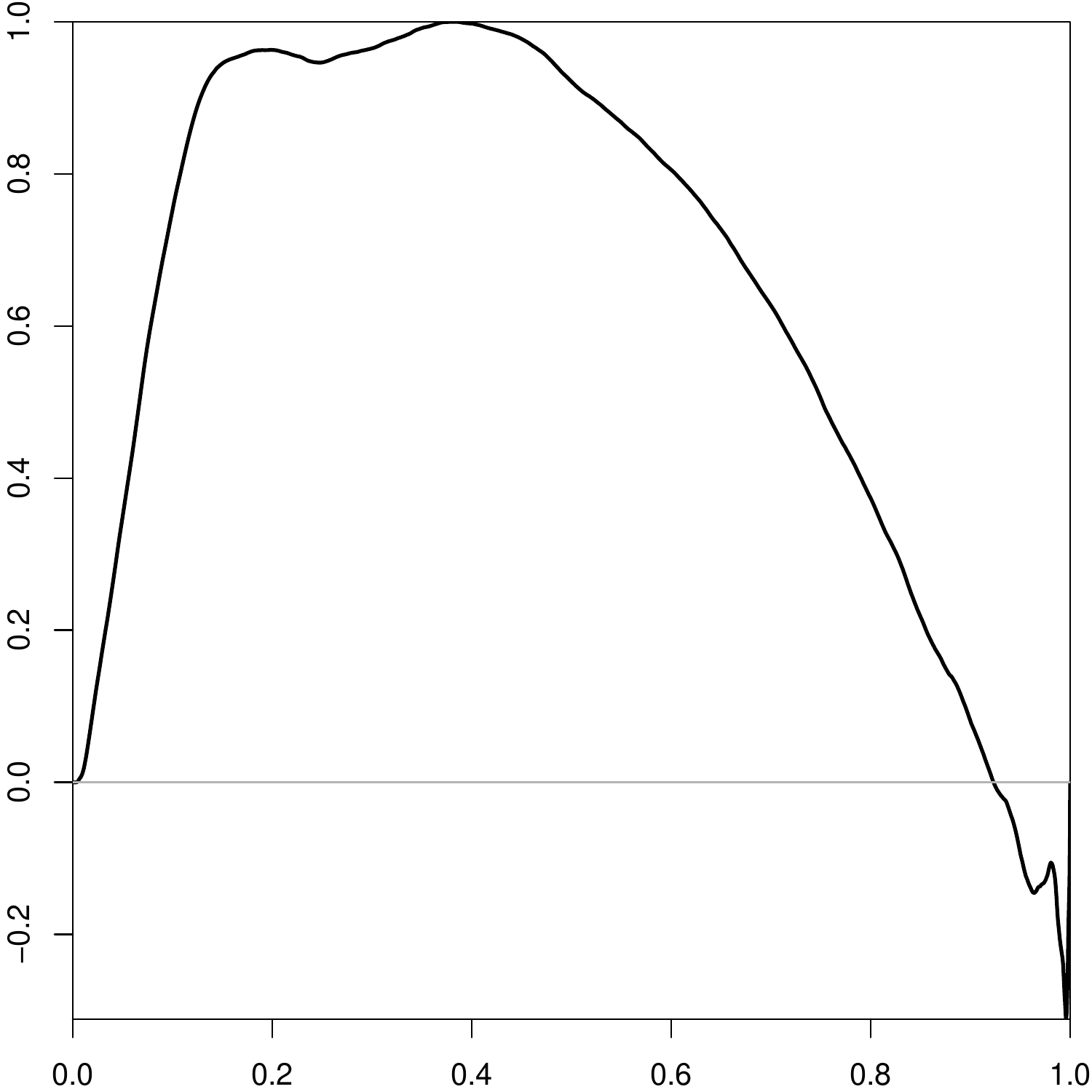} \\
\includegraphics[width=0.47\textwidth]{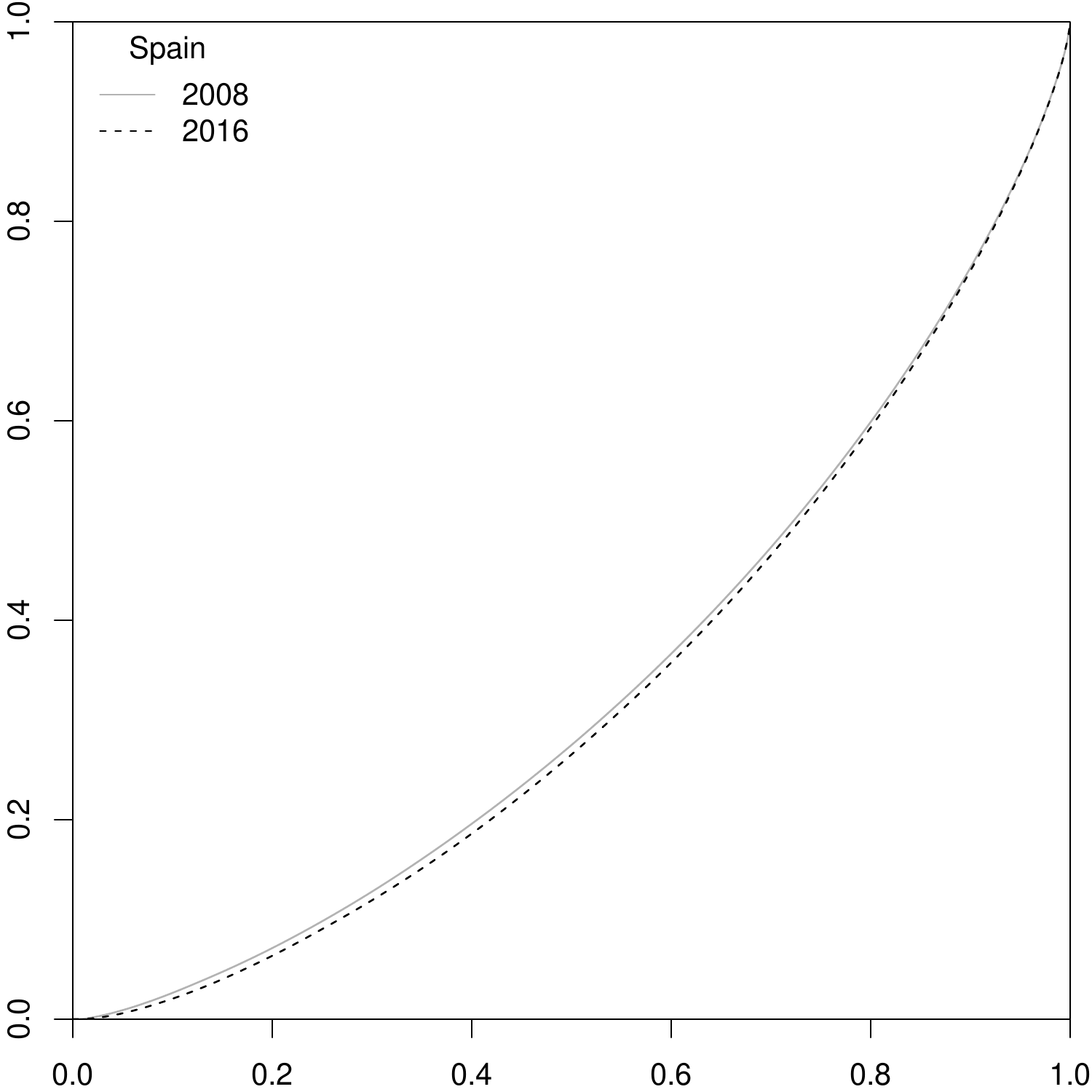} & \includegraphics[width=0.47\textwidth]{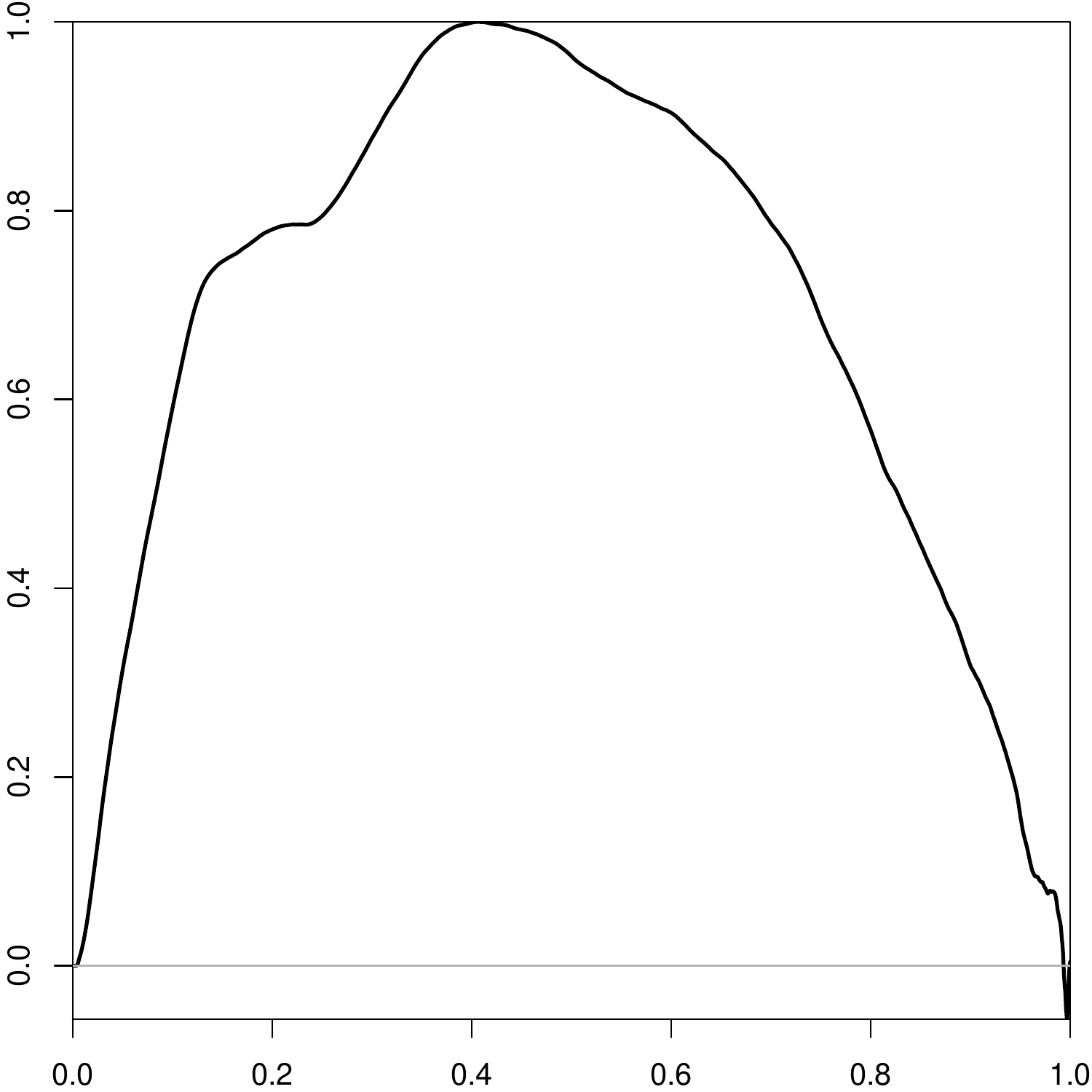}
\end{tabular}
\end{center}
\caption{The Lorenz curves $\hat\ell_1$ and $\hat\ell_2$ (left column) of income distribution and their difference scaled by the maximum absolute difference, $(\hat\ell_1-\hat\ell_2)/\|\hat\ell_1-\hat\ell_2\|_\infty$, (right column) corresponding to the pairs of years formed by 2008 ($X_1$) and one in the span 2009--2019.}
\label{SuppMatfi:SpainLorenz3}
\end{figure}

\begin{figure}[H]
\begin{center}
\begin{tabular}{cc}
\includegraphics[width=0.47\textwidth]{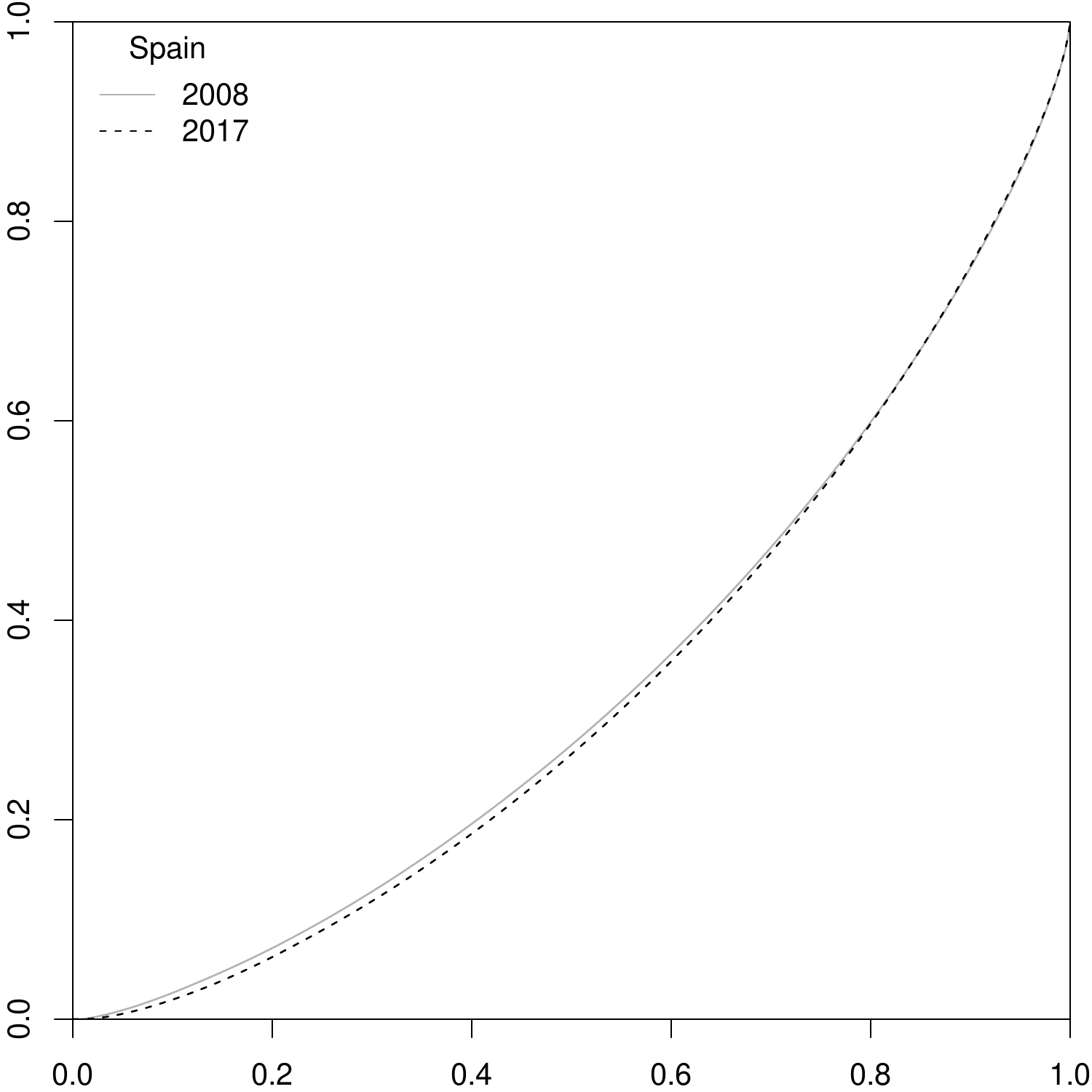} & \includegraphics[width=0.47\textwidth]{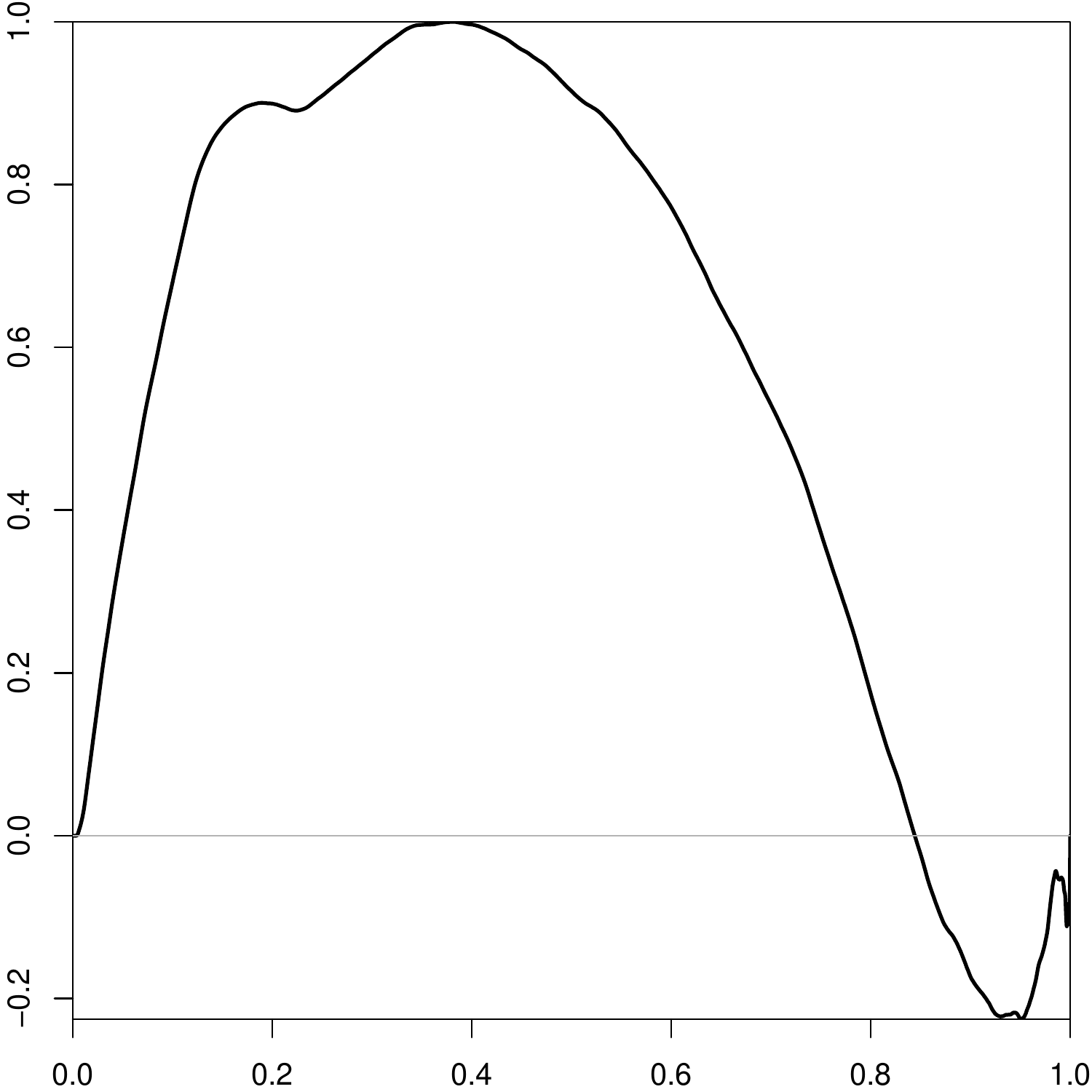} \\
\includegraphics[width=0.47\textwidth]{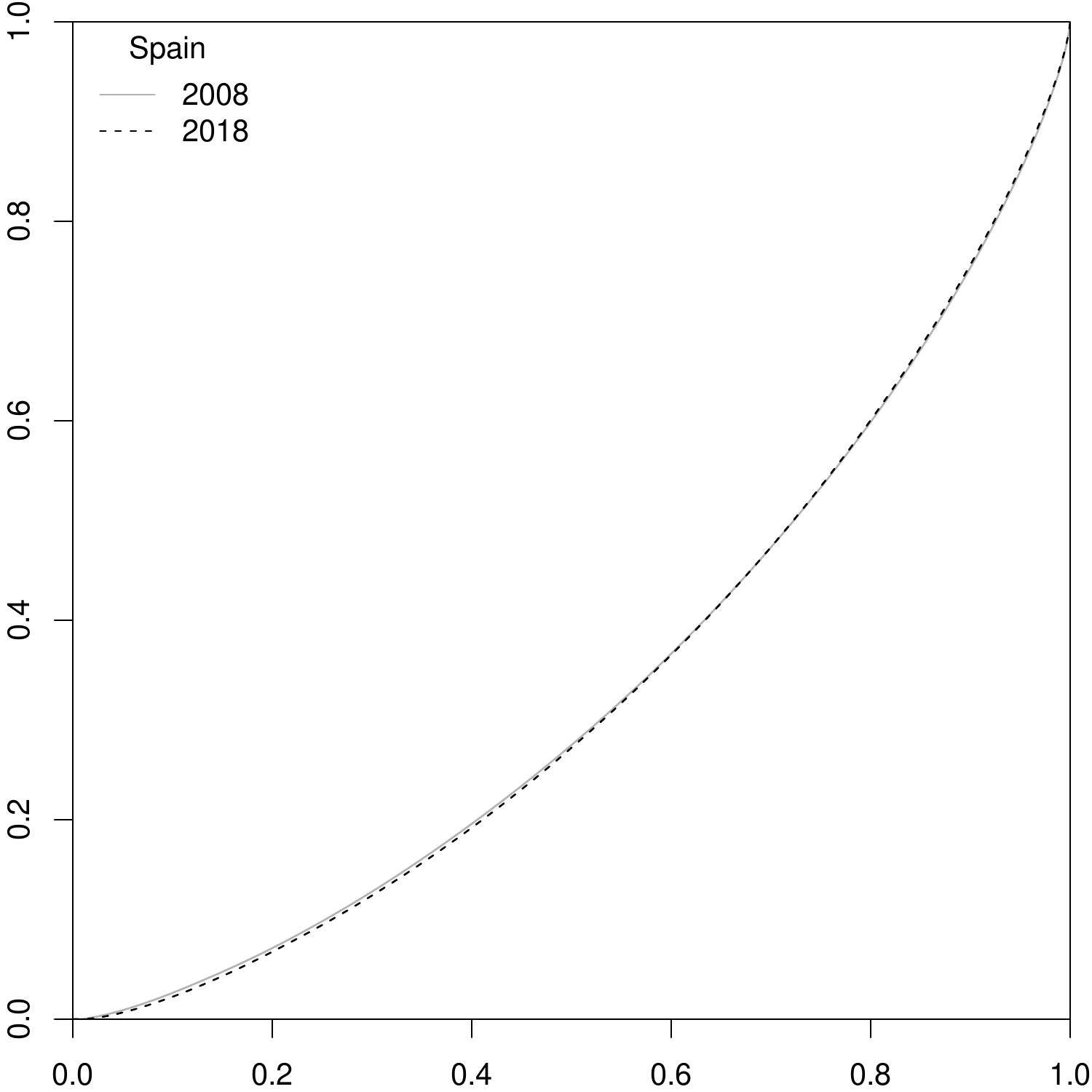} & \includegraphics[width=0.47\textwidth]{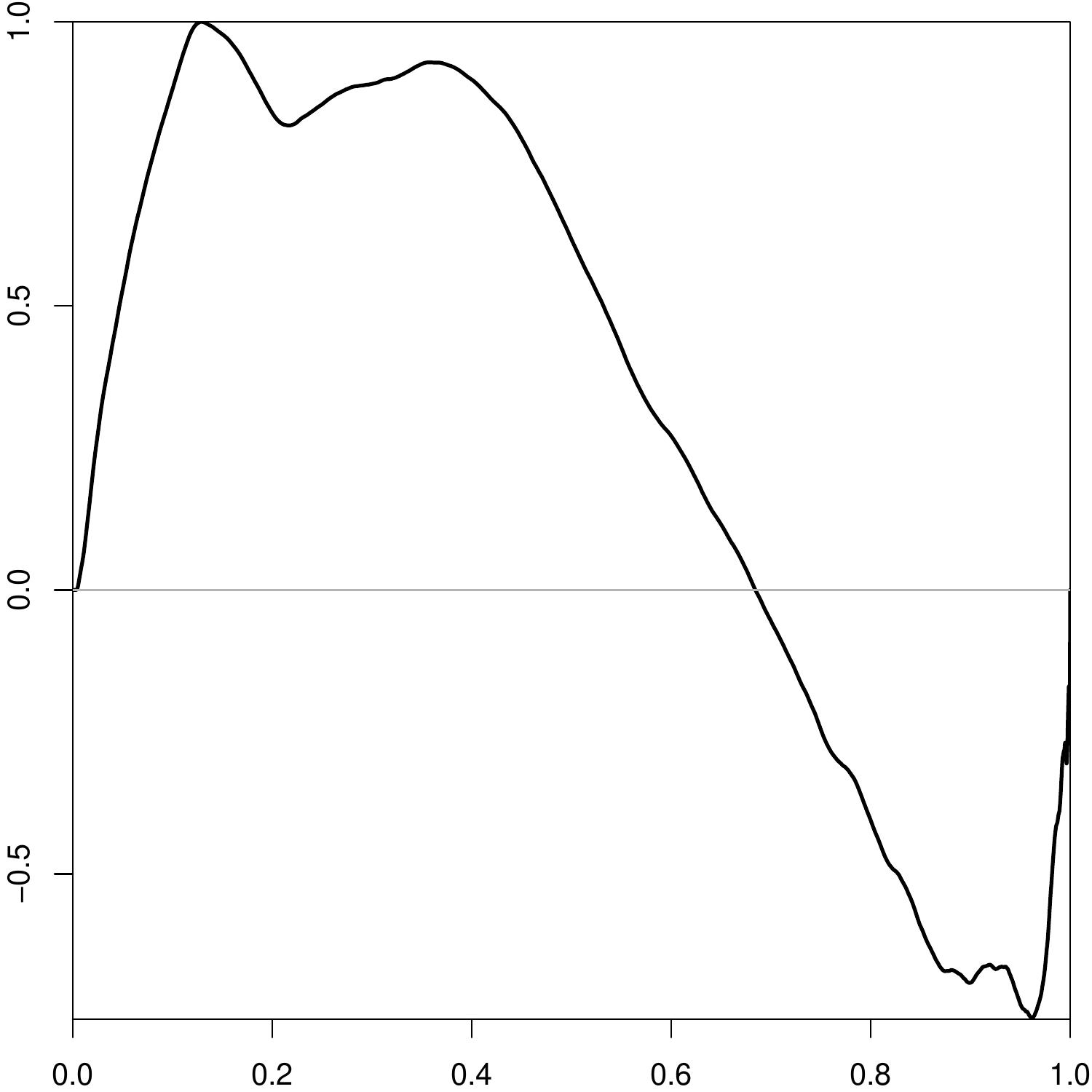} \\
\includegraphics[width=0.47\textwidth]{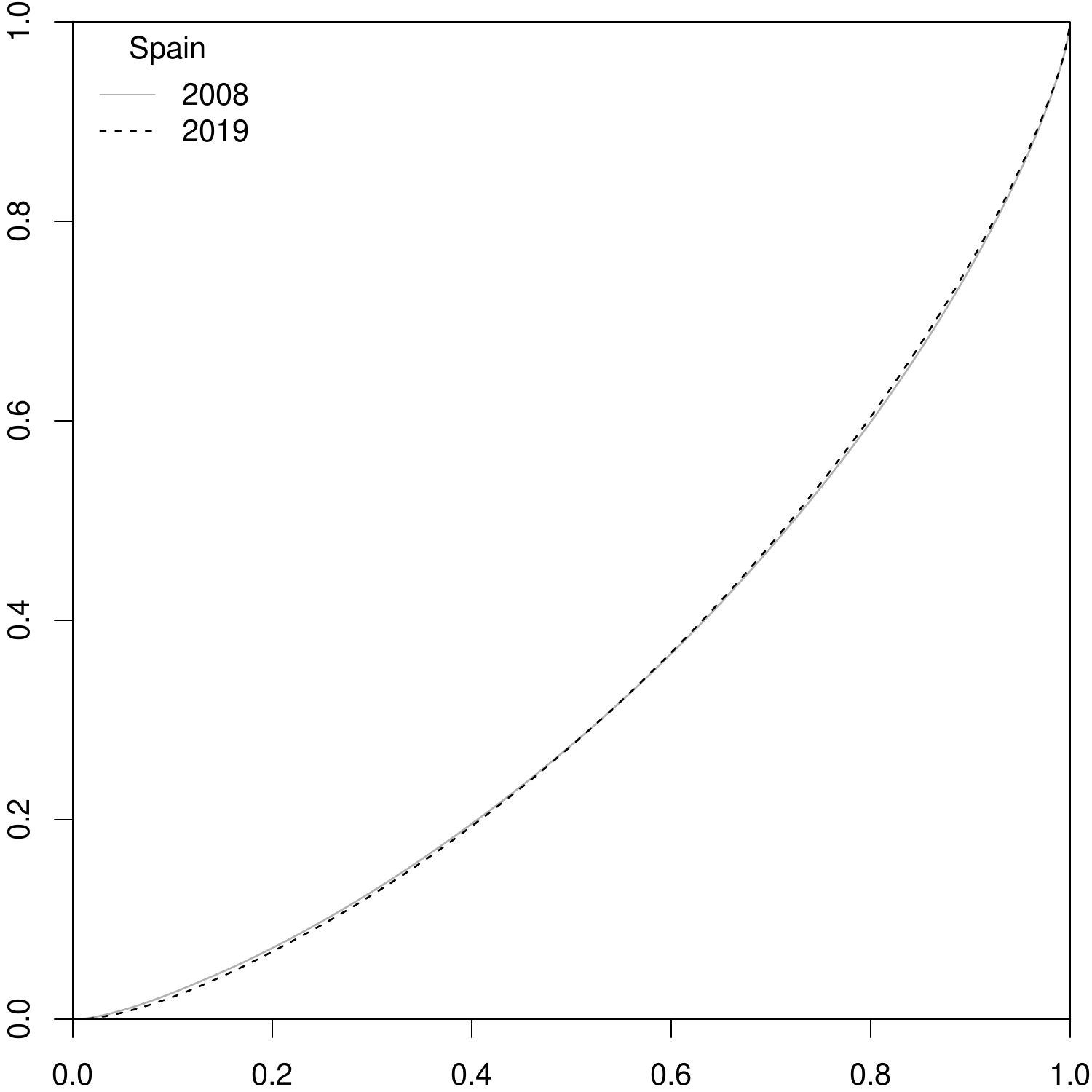} & \includegraphics[width=0.47\textwidth]{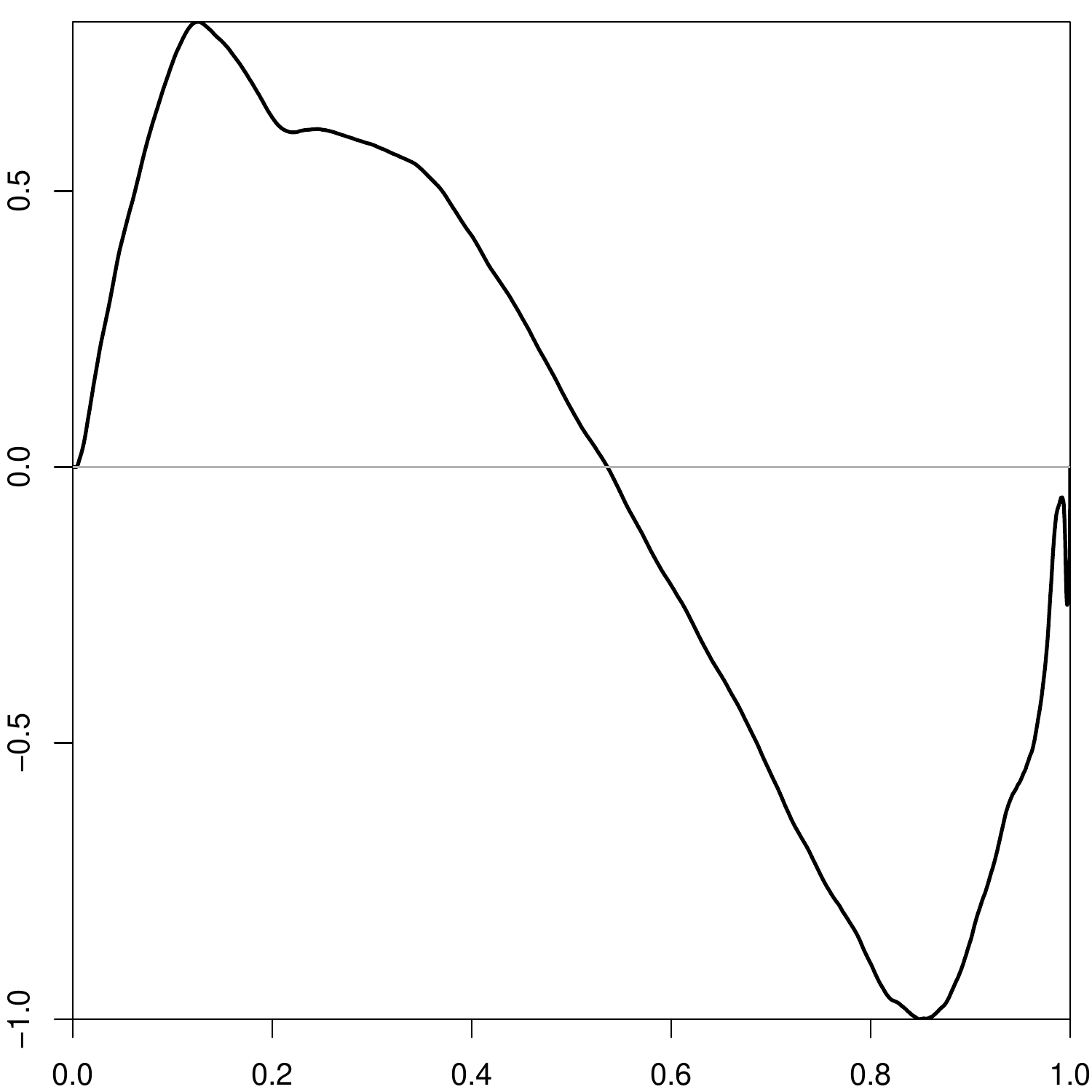}
\end{tabular}
\end{center}
\caption{The Lorenz curves $\hat\ell_1$ and $\hat\ell_2$ (left column) of income distribution and their difference scaled by the maximum absolute difference, $(\hat\ell_1-\hat\ell_2)/\|\hat\ell_1-\hat\ell_2\|_\infty$, (right column) corresponding to the pairs of years formed by 2008 ($X_1$) and one in the span 2009--2019.}
\label{SuppMatfi:SpainLorenz4}
\end{figure}

\newpage

\begin{table}[h] \label{SuppMatTa:DescriptivesIncomesES0819}
\begin{center}
\begin{tabular}{cccc}
Year & Weighted mean annual income & Sample size & Gini index (in \%) \\ \hline
2008 & 16414.18 & 12987 &  32.9 \\
2009 & 17382.66 & 13310 &  32.9 \\
2010 & 17167.77 & 13544 &  33.1 \\
2011 & 16436.31 & 13052 &  33.7 \\
2012 & 16398.07 & 12661 &  33.8 \\
2013 & 15992.40 & 12091 &  33.4 \\
2014 & 15696.13 & 11912 &  34.1 \\
2015 & 15752.62 & 12332 &  34.0 \\
2016 & 16078.59 & 14212 &  34.3 \\
2017 & 16570.21 & 13716 &  34.1 \\
2018 & 17006.14 & 13344 &  33.2 \\
2019 & 17419.96 & 15861 &  32.9 \\ \hline
\end{tabular}
\end{center}
\caption{Summaries for the annual household equivalised disposable incomes in Spain.}
\end{table}

\subsection{Comparing inequality between Finland and Greece} \label{Subsection.RealData.FinGr}

In Figures~\ref{SuppMatfi:GRFILorenz1}, \ref{SuppMatfi:GRFILorenz3} and~\ref{SuppMatfi:GRFILorenz5}, on the left we display the Lorenz curves, \(\hat\ell_1\) and \(\hat\ell_2\), of the equivalised household income in Greece and Finland, respectively, for a year between 2004 and 2019. The two Lorenz curves are close to each other and it is difficult to appreciate the full detail of their differences (for instance, one could think that Finnish income is less than Greek income in the Lorenz order for any of the years). This is why on the right we plot the difference of the two Lorenz curves, \(\hat\ell_1-\hat\ell_2\), scaled by \(\|\hat\ell_1-\hat\ell_2\|_\infty\). The resulting plots show that the empirical distribution of income in Finland and Greece is indeed ordered in 2016 and 2018, but only almost ordered in 2004.

\begin{table}[h] \label{SuppMatTa:DescriptivesOfIncomesGRFI}
\begin{center}
\begin{tabular}{ccccccc}
     & \multicolumn{2}{c}{Weighted mean annual income} & \multicolumn{2}{c}{Sample size} & \multicolumn{2}{c}{Gini index} \\
Year & Greece & Finland & Greece & Finland & Greece & Finland \\ \hline
2004 & 13340.20 & 21861.31 &   6210 & 	11191 & 34.2 & 29.4	\\
2005 & 14362.08 & 22926.73 &   5547 & 	11220 & 34.5 & 30.4	\\
2006 & 15116.40 & 23740.09 &   5685 & 	10860 & 34.9 & 30.1	\\
2007 & 15723.72 & 24453.90 &   5622 & 	10612 & 35.3 & 30.1	\\
2008 & 16674.89 & 25900.65 &   6481 & 	10463 & 34.2 & 30.8	\\
2009 & 17528.87 & 27062.00 &   6986 & 	10126 & 34.4 & 30.4	\\
2010 & 18013.19 & 27603.45 &   6979 & 	10981 & 34.5 & 29.6	\\
2011 & 16160.25 & 28196.67 &   5998 & 	 9346 & 34.7 & 30.2	\\
2012 & 13691.87 & 29223.84 &   5560 & 	10303 & 33.1 & 30.3	\\
2013 & 12386.56 & 30033.81 &   7386 & 	11364 & 33.9 & 29.9	\\
2014 & 11524.52 & 30214.28 &   8607 & 	11029 & 34.7 & 30.0	\\
2015 & 11414.76 & 30311.59 &  14000 & 	10725 & 34.4 & 29.9	\\
2016 & 11413.31 & 30459.17 &  18101 & 	10620 & 34.4 & 30.0	\\
2017 & 11531.00 & 30811.17 &  22524 & 	10208 & 33.6 & 29.8	\\
2018 & 11862.39 & 31373.71 &  24190 &    9830 & 33.0 & 30.4 \\
2019 & 12323.37 & 31874.73 &  17829 &    9644 & 31.5 & 30.7 \\ \hline
\end{tabular}
\end{center}
\caption{Summaries for the annual household equivalised disposable incomes in Greece and Finland.}
\end{table}

\begin{figure}[H]
\begin{center}
\begin{tabular}{cc}
\includegraphics[width=0.47\textwidth]{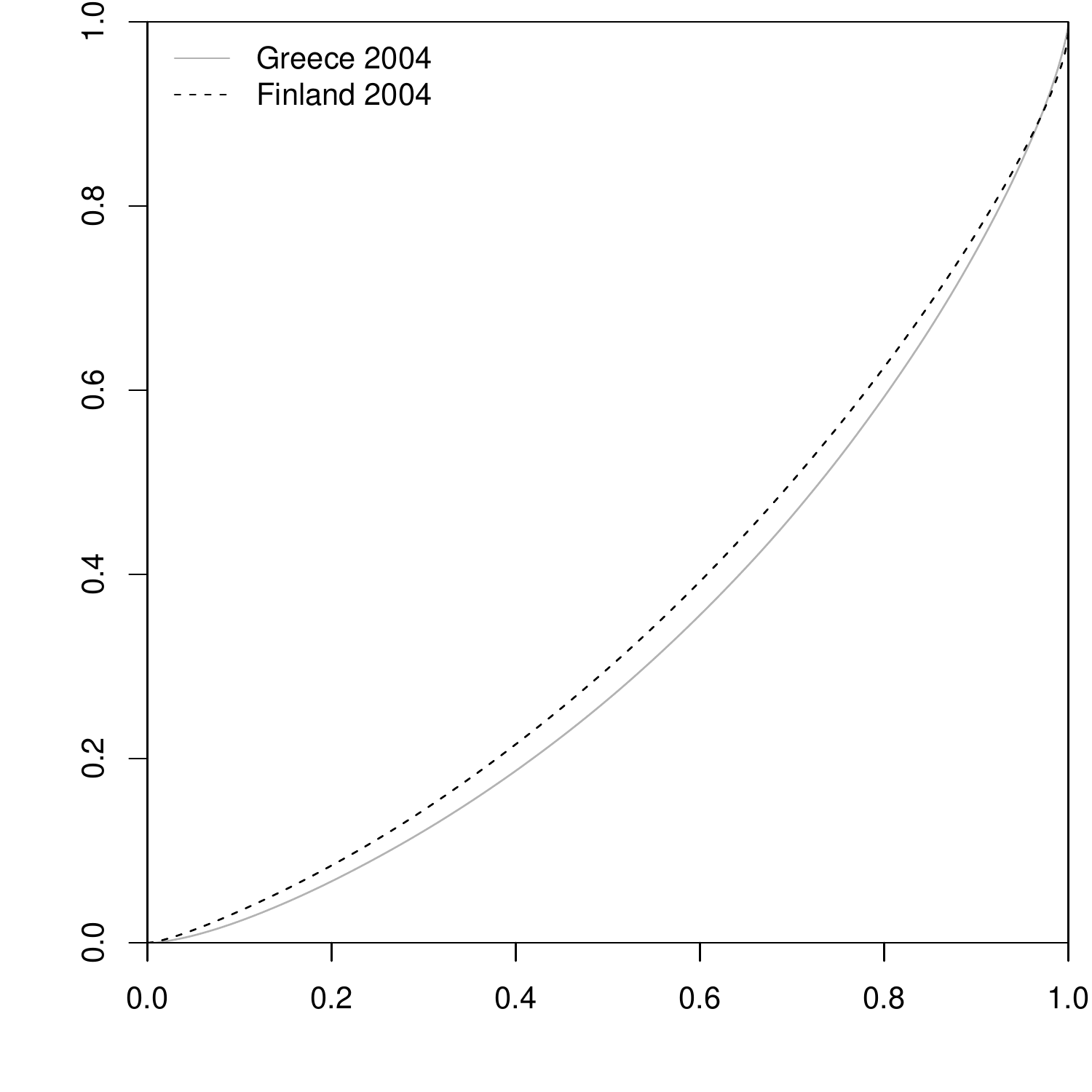} & \includegraphics[width=0.47\textwidth]{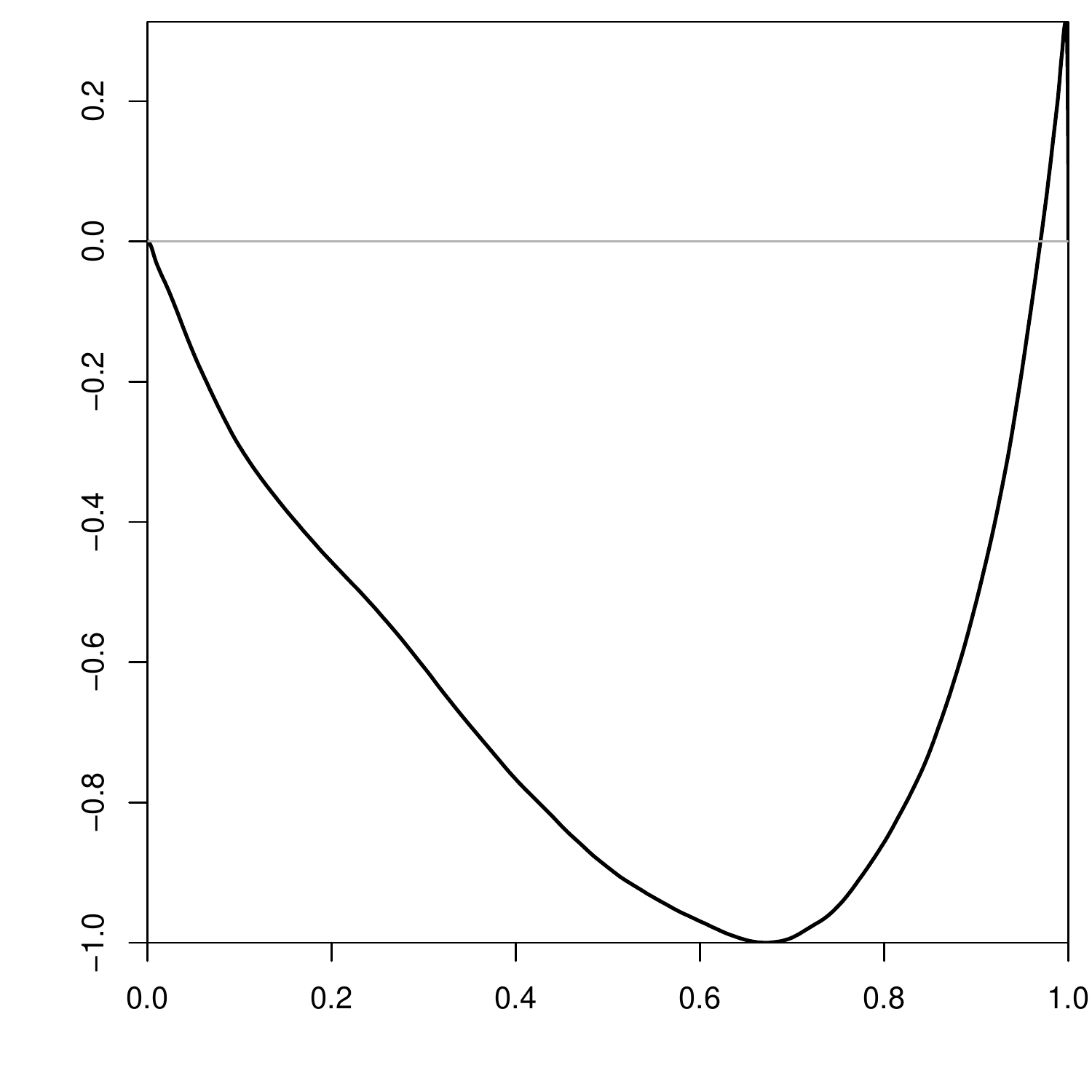}  \\
\includegraphics[width=0.47\textwidth]{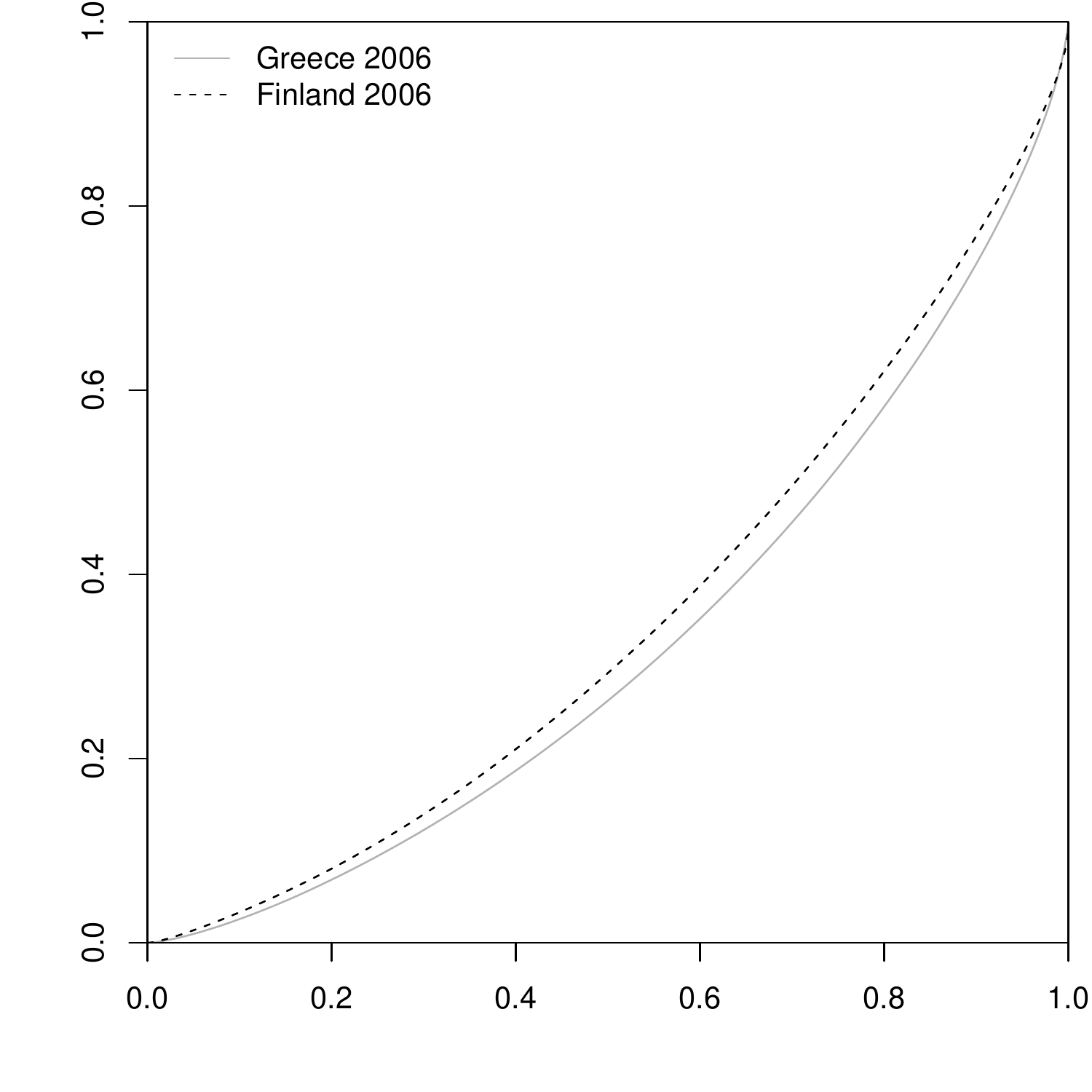} & \includegraphics[width=0.47\textwidth]{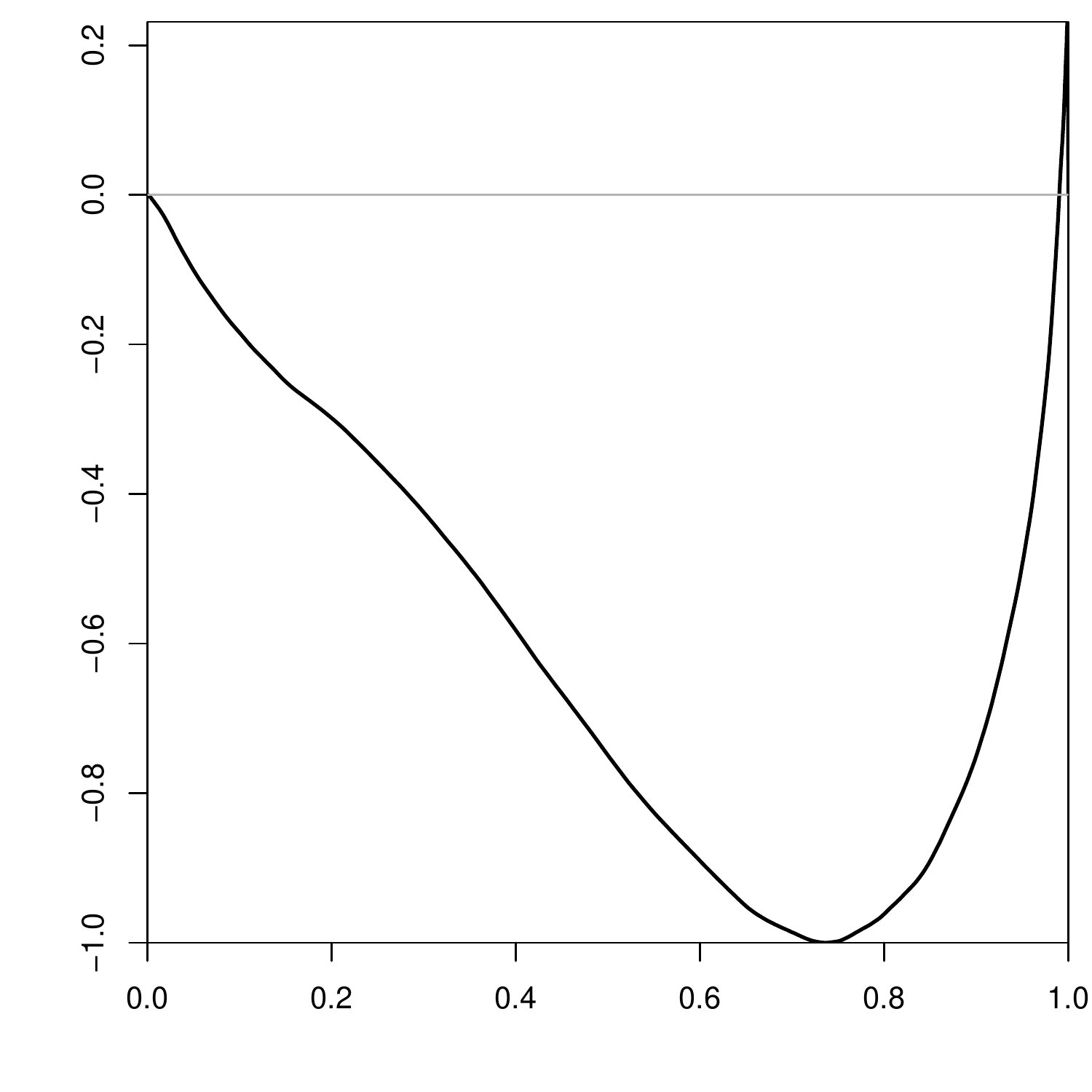} \\
\includegraphics[width=0.47\textwidth]{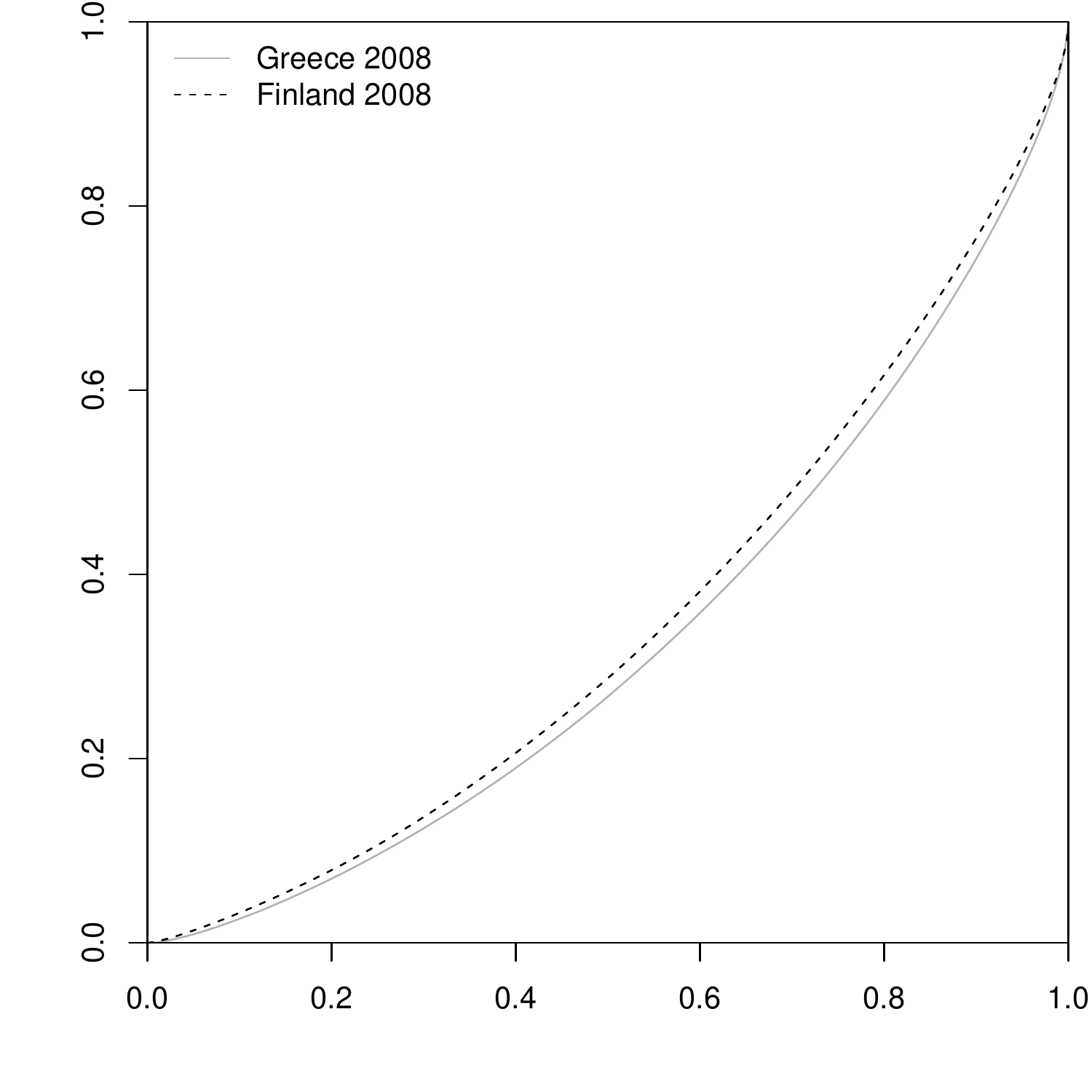} & \includegraphics[width=0.47\textwidth]{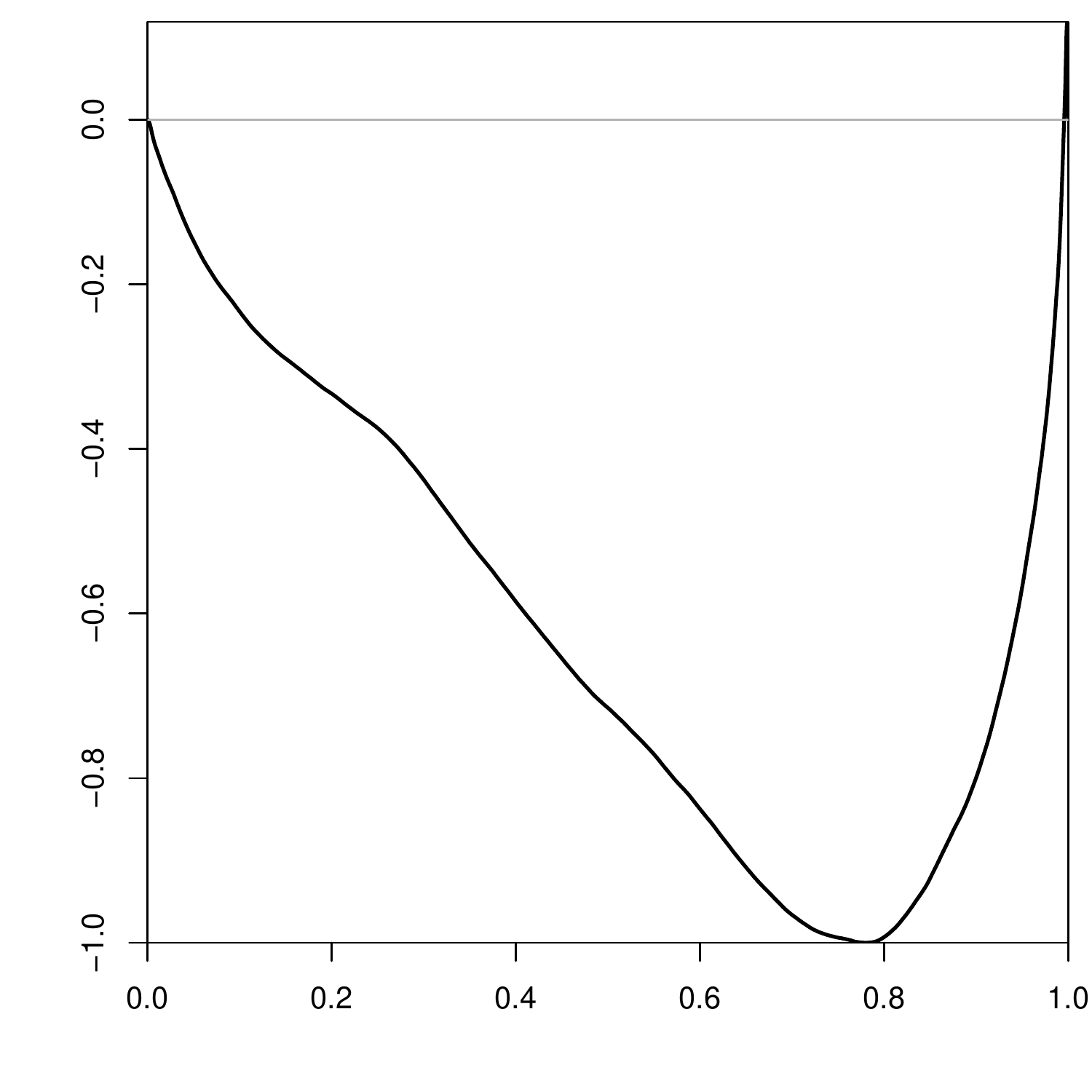}
\end{tabular}
\end{center}
\caption{The Lorenz curves $\hat\ell_1$ and $\hat\ell_2$ (left column) of income distribution and their difference scaled by the maximum absolute difference, $(\hat\ell_1-\hat\ell_2)/\|\hat\ell_1-\hat\ell_2\|_\infty$, (right column) corresponding to Greece ($X_1$)  and Finland ($X_2$) in a year of the span 2004--2019.}
\label{SuppMatfi:GRFILorenz1}
\end{figure}


\begin{figure}[H]
\begin{center}
\begin{tabular}{cc}
\includegraphics[width=0.47\textwidth]{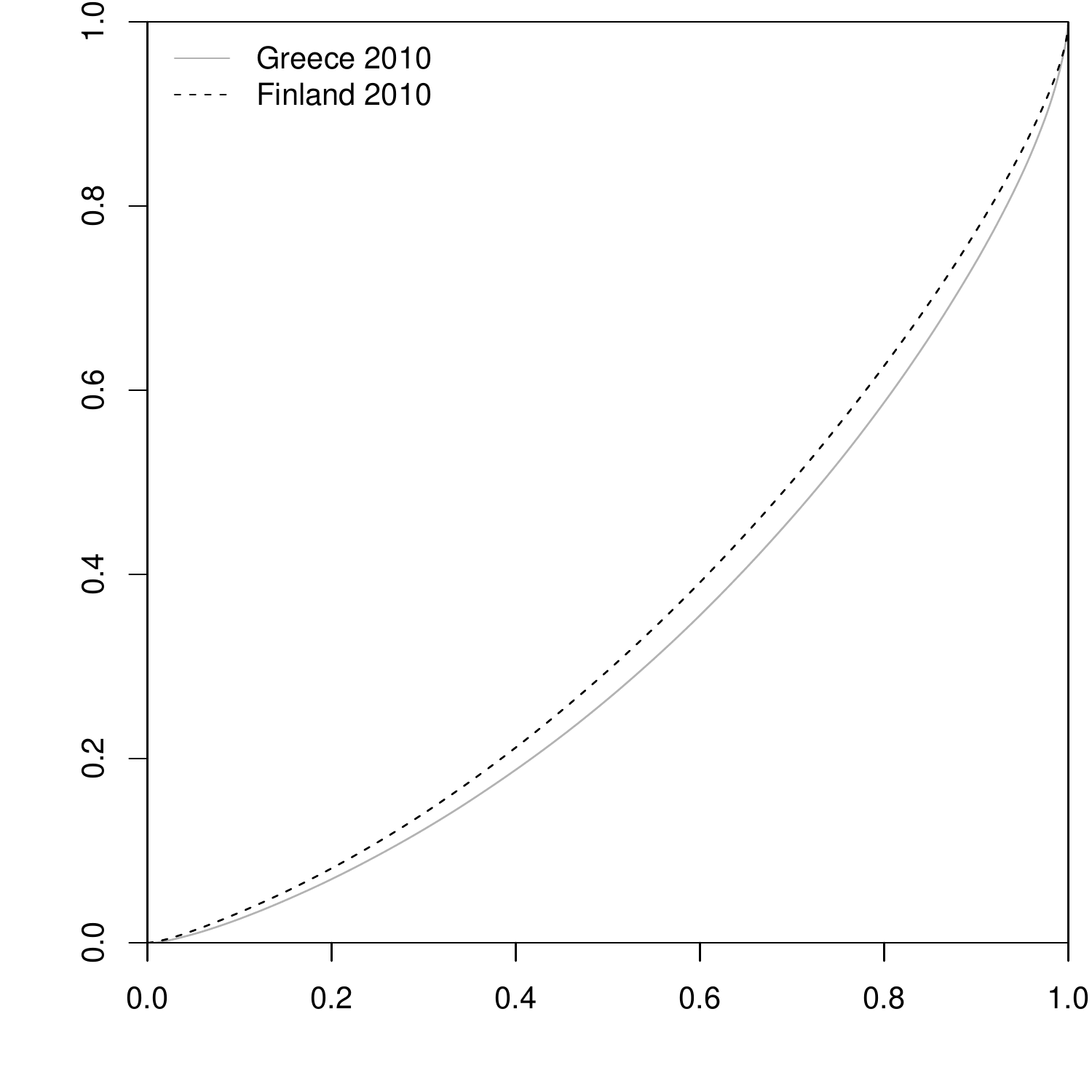} & \includegraphics[width=0.47\textwidth]{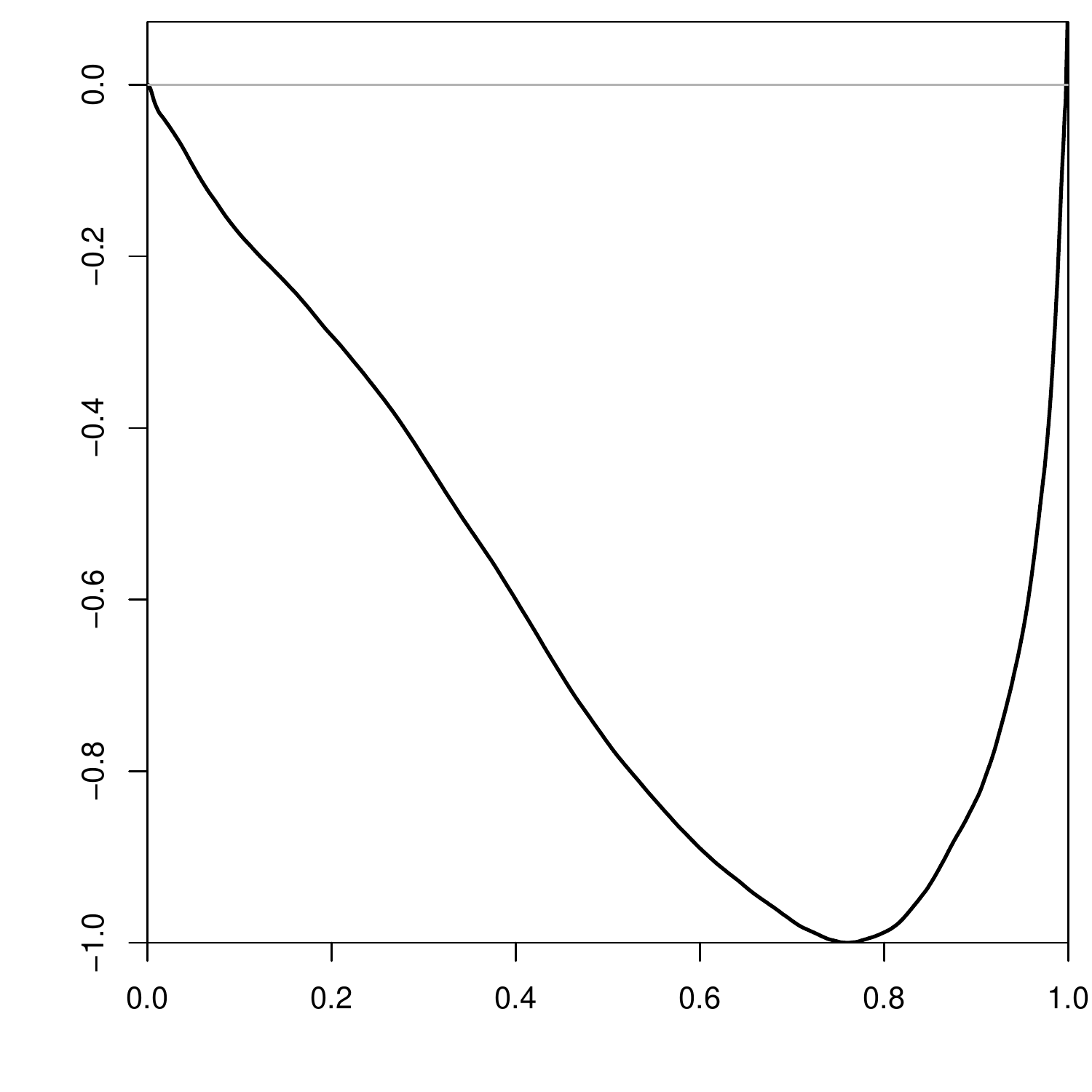}  \\
\includegraphics[width=0.47\textwidth]{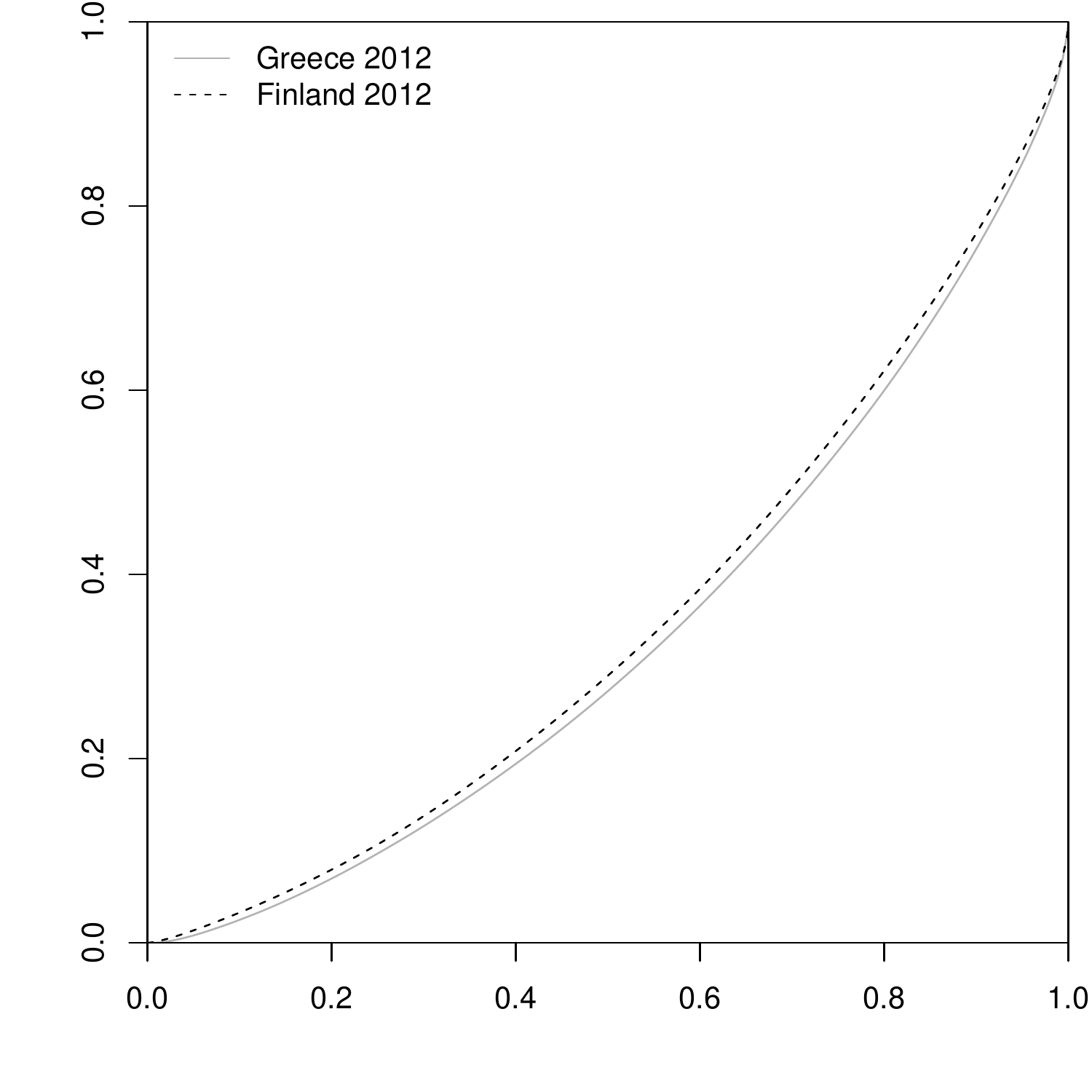} & \includegraphics[width=0.47\textwidth]{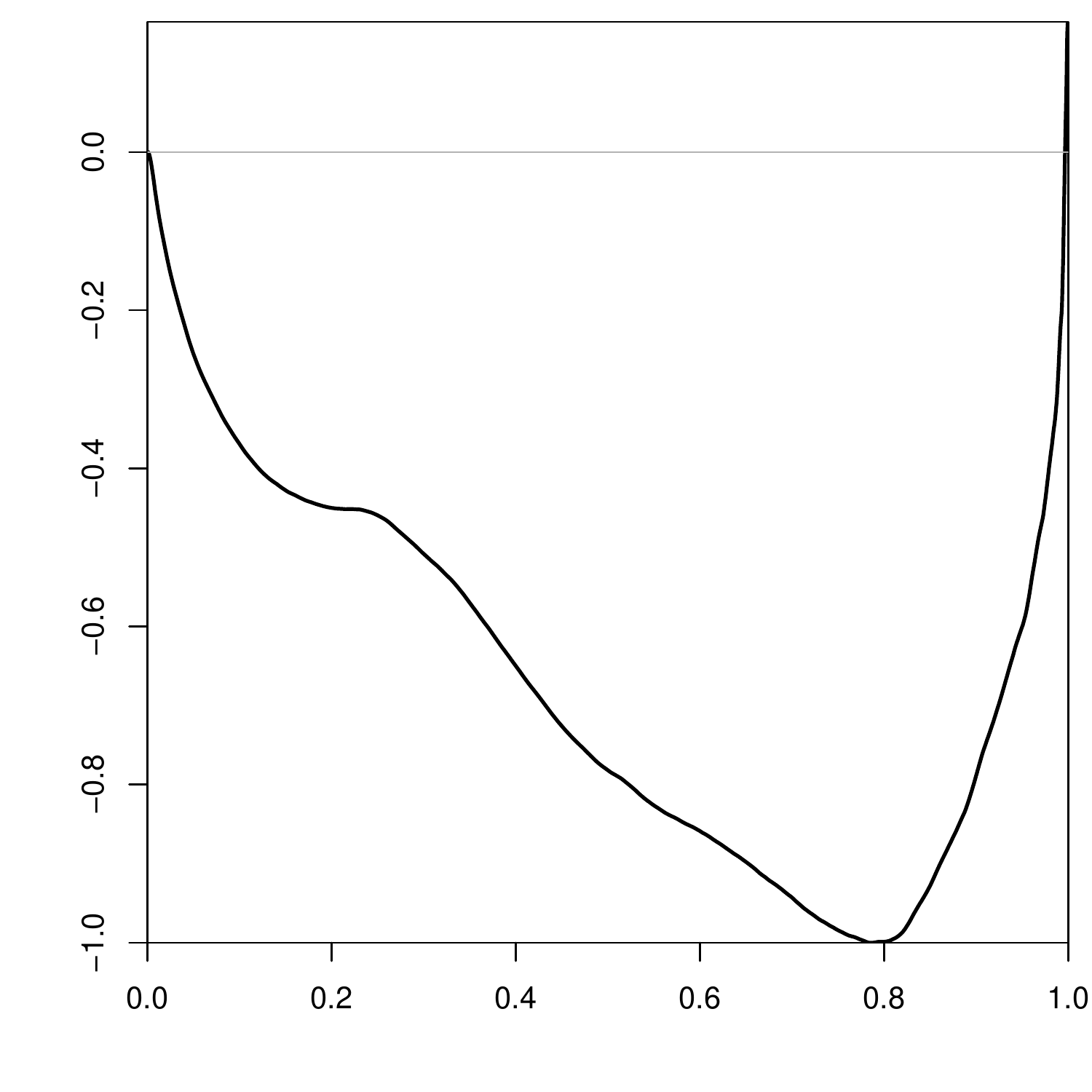} \\
\includegraphics[width=0.47\textwidth]{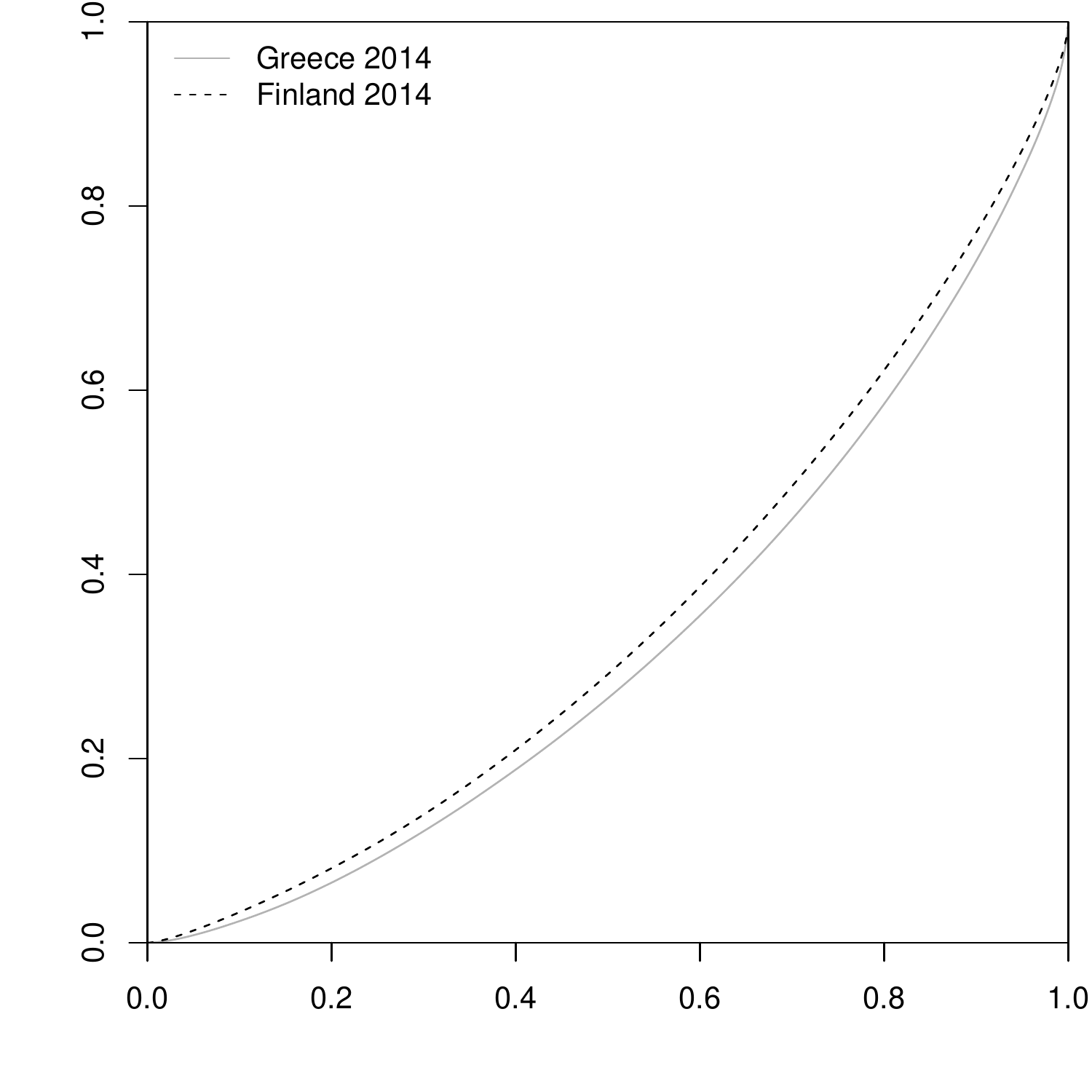} & \includegraphics[width=0.47\textwidth]{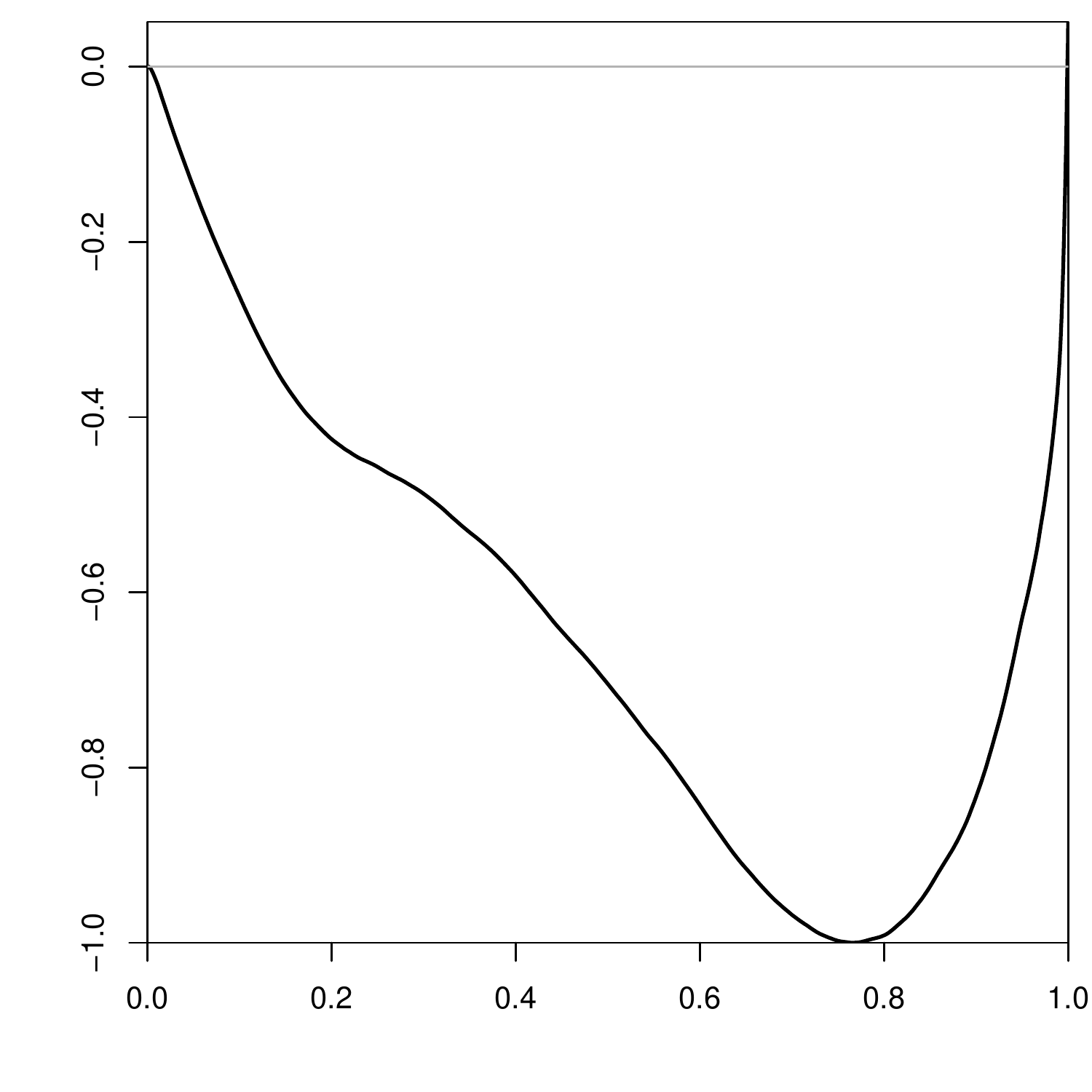}
\end{tabular}
\end{center}
\caption{The Lorenz curves $\hat\ell_1$ and $\hat\ell_2$ (left column) of income distribution and their difference scaled by the maximum absolute difference, $(\hat\ell_1-\hat\ell_2)/\|\hat\ell_1-\hat\ell_2\|_\infty$, (right column) corresponding to Greece ($X_1$)  and Finland ($X_2$) in a year of the span 2004--2019.}
\label{SuppMatfi:GRFILorenz3}
\end{figure}


\begin{figure}[H]
\begin{center}
\begin{tabular}{cc}
\includegraphics[width=0.47\textwidth]{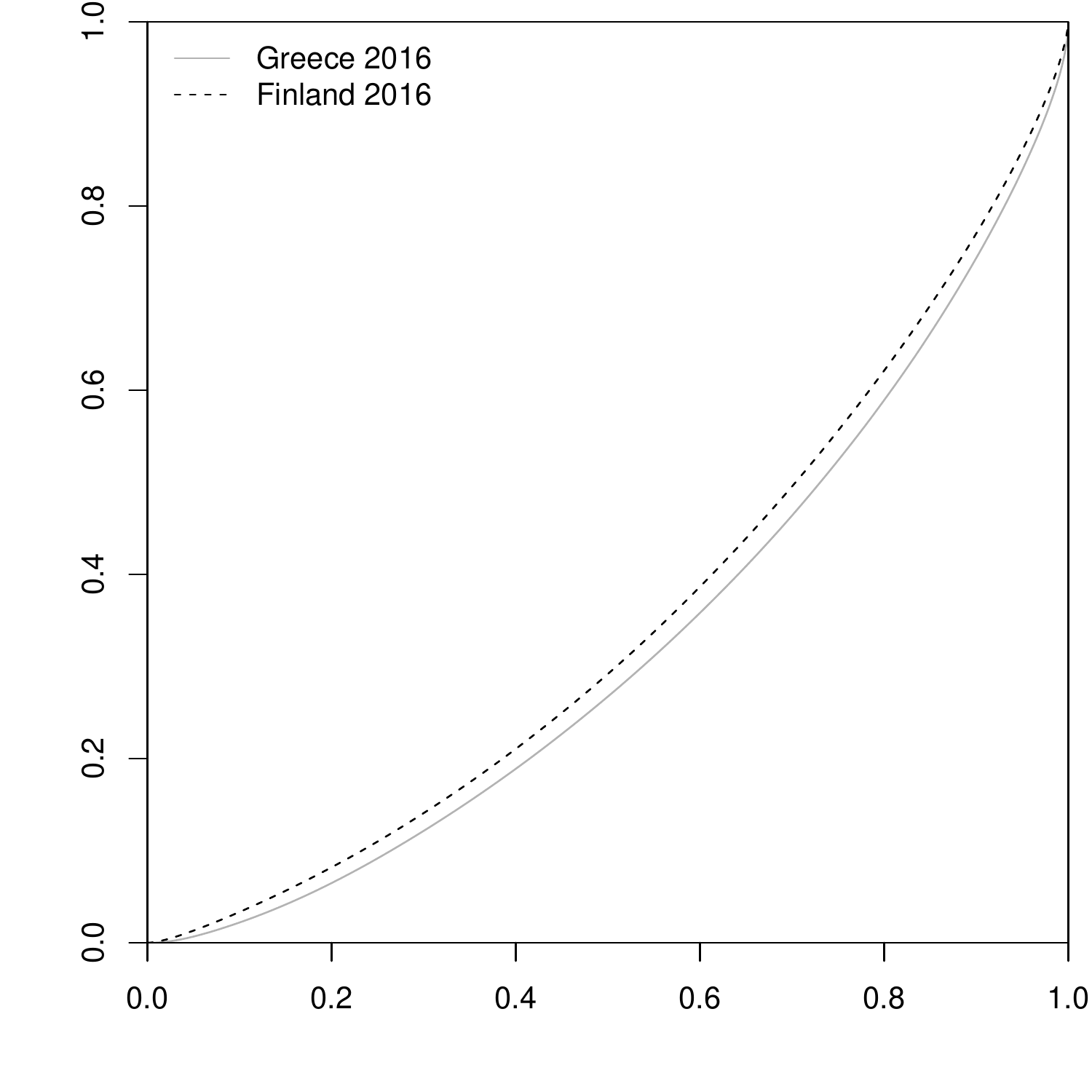} & \includegraphics[width=0.47\textwidth]{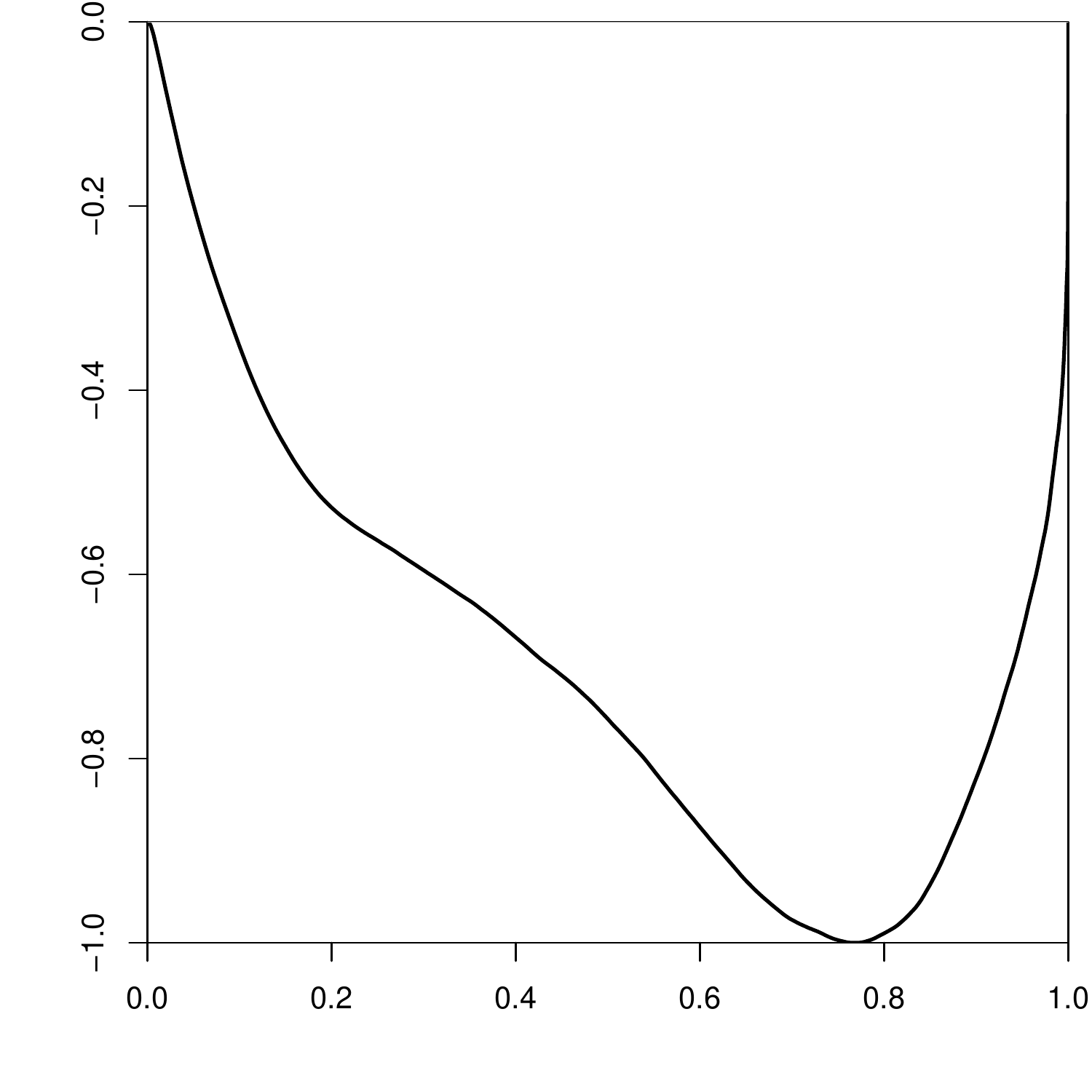}  \\
\includegraphics[width=0.47\textwidth]{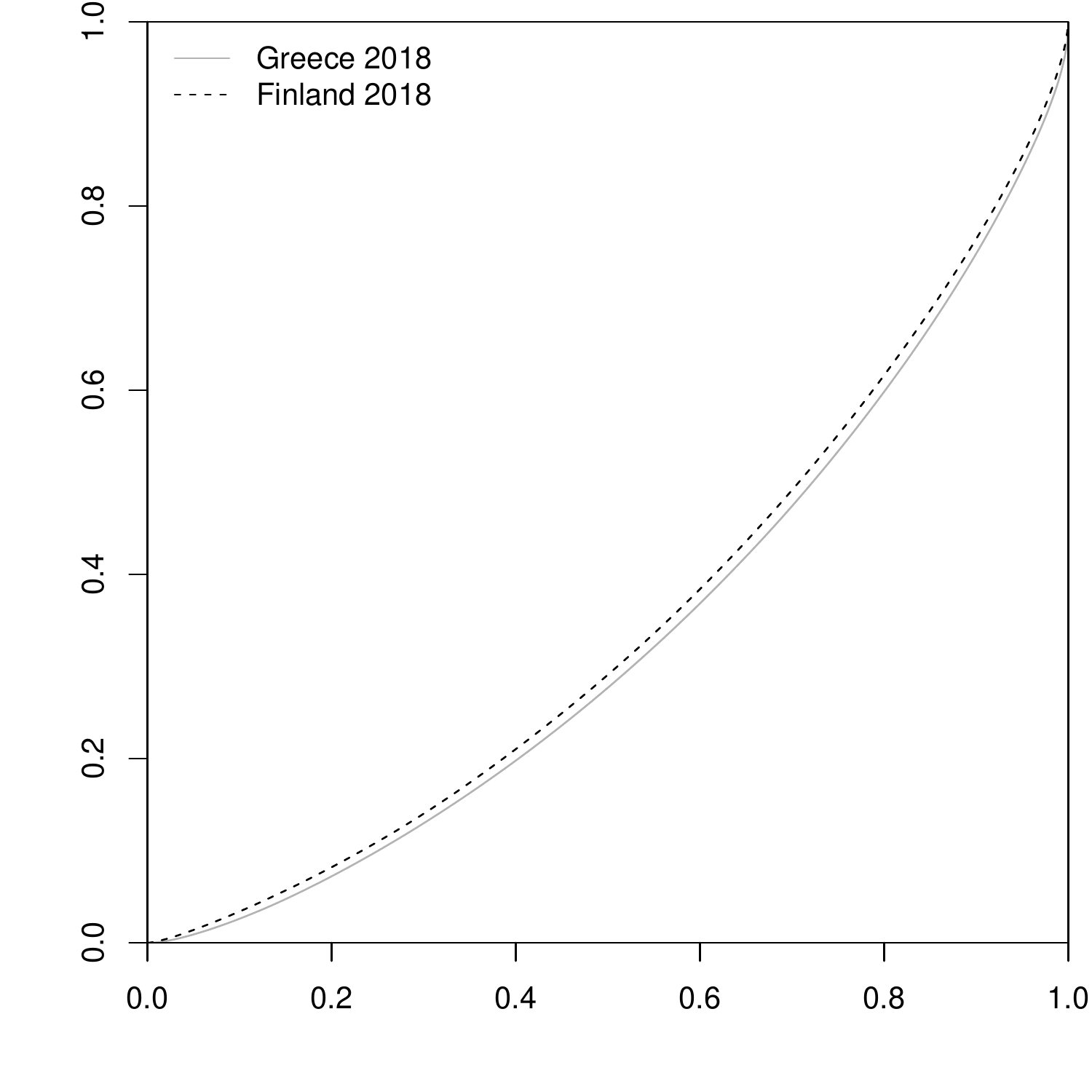} & \includegraphics[width=0.47\textwidth]{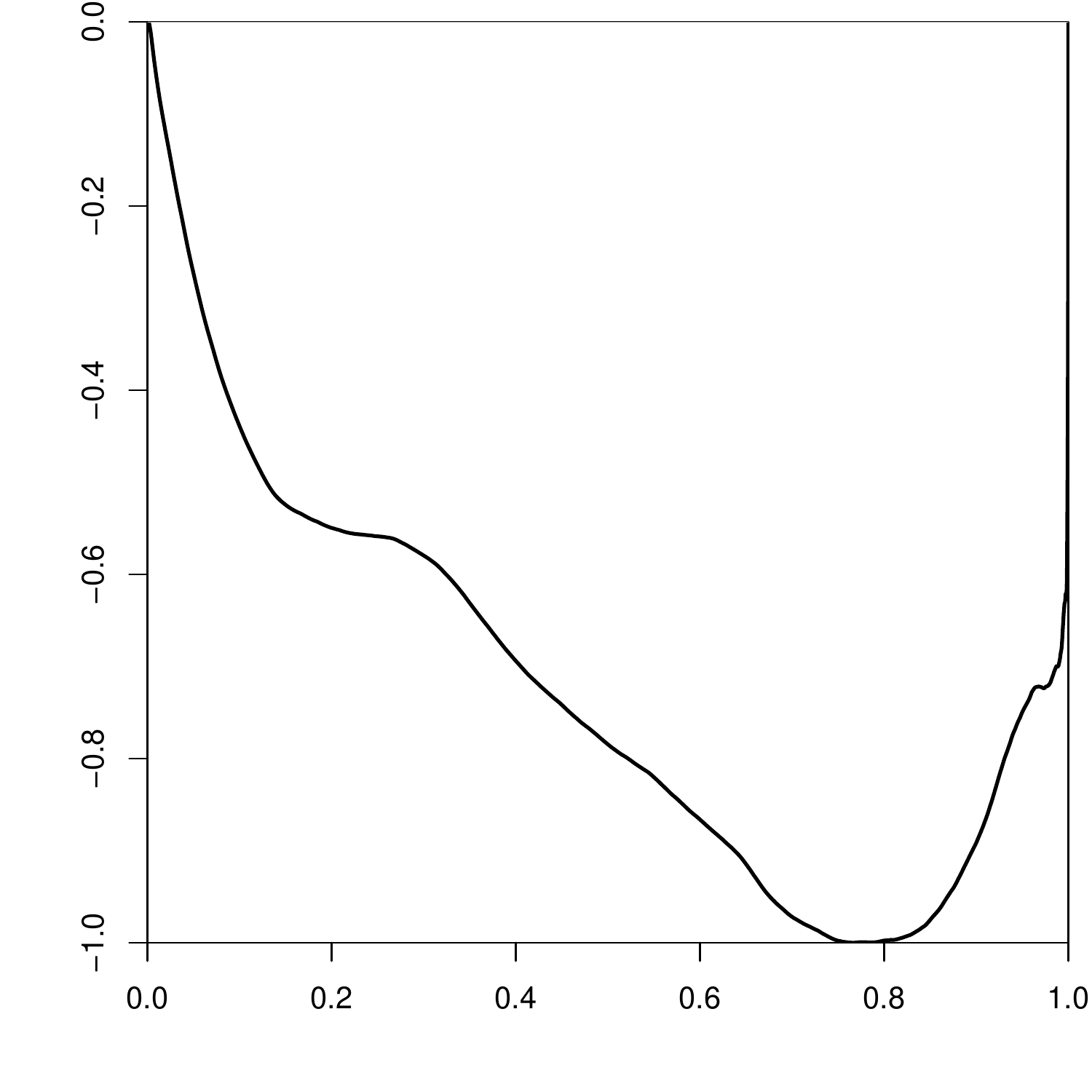} \\
\includegraphics[width=0.47\textwidth]{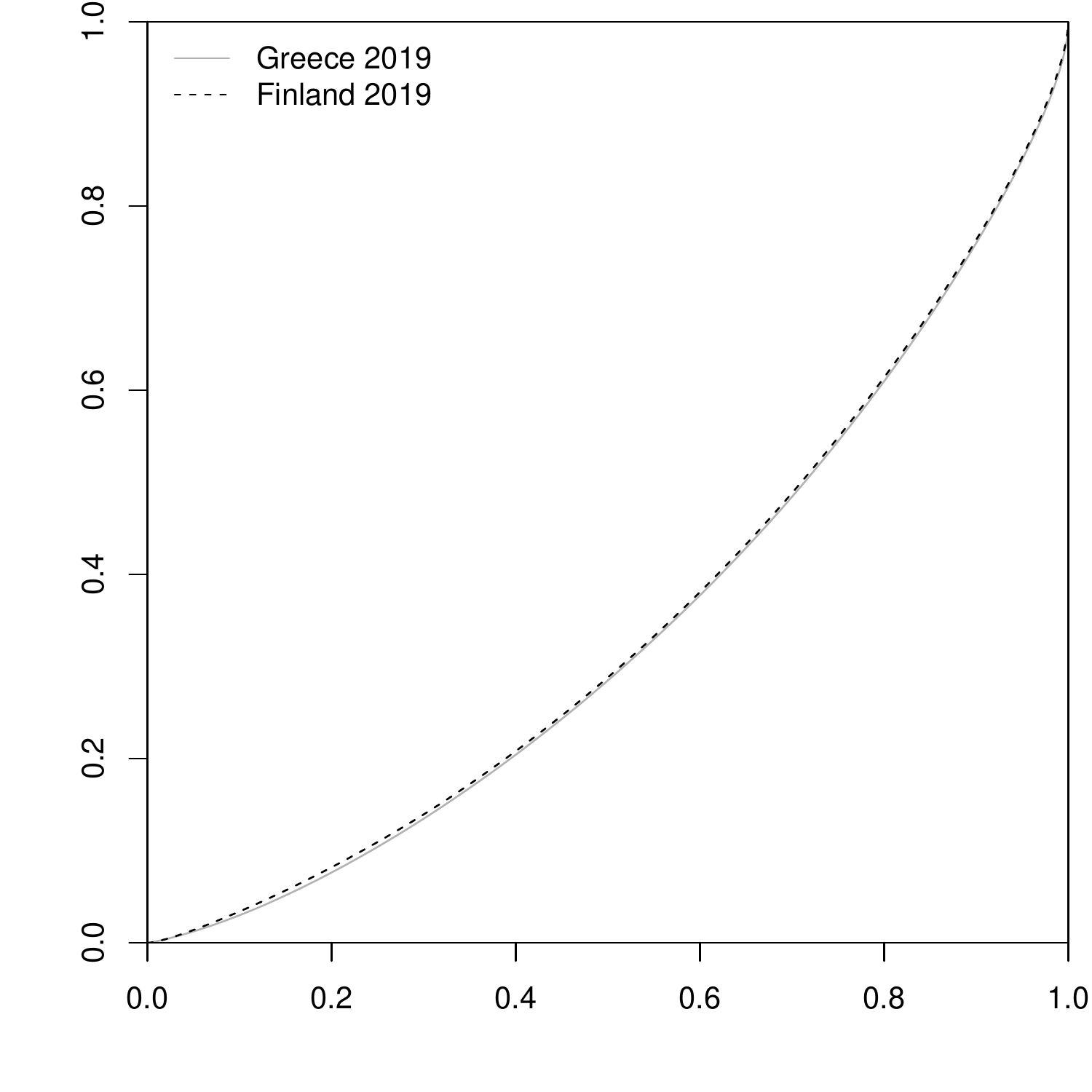} & \includegraphics[width=0.47\textwidth]{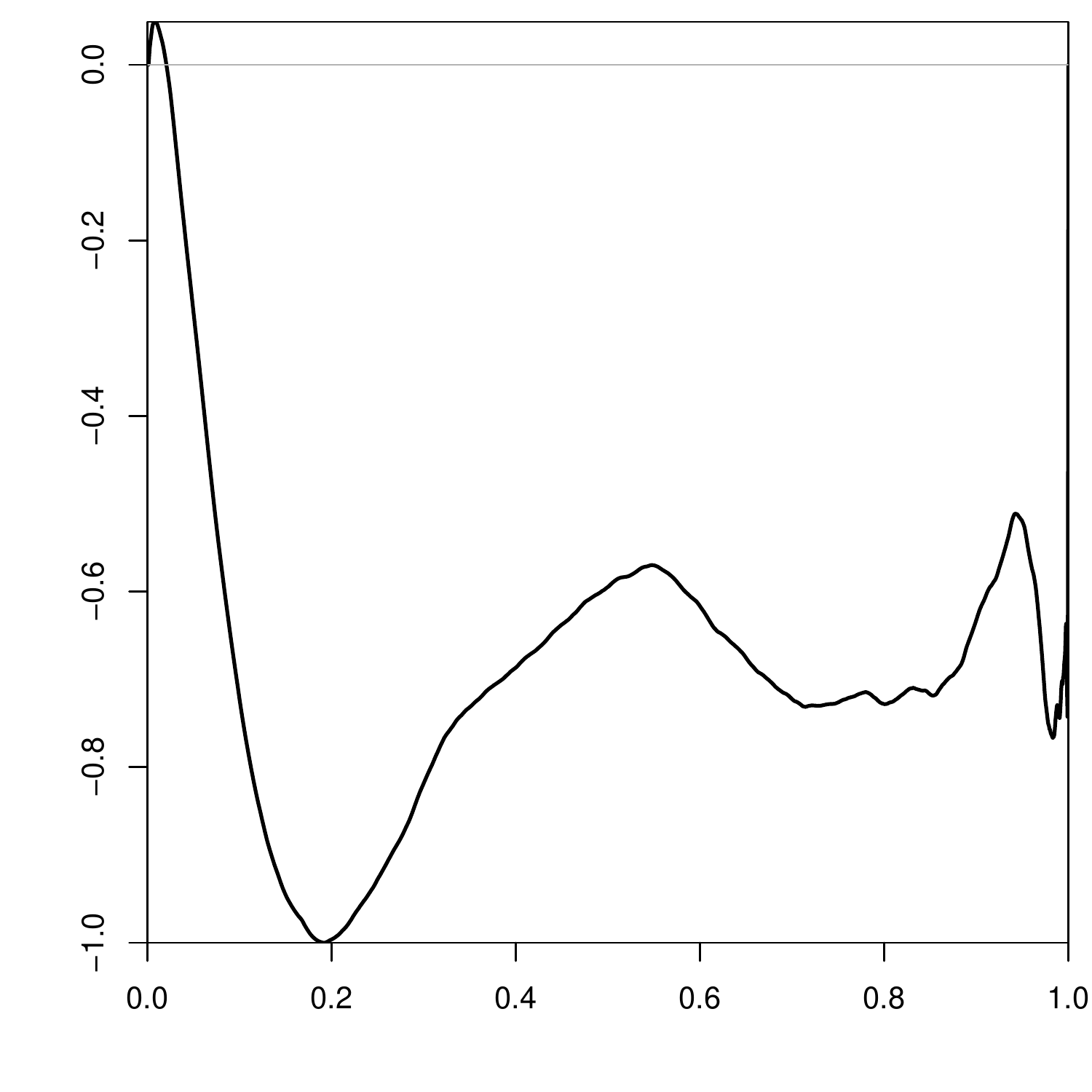}
\end{tabular}
\end{center}
\caption{The Lorenz curves $\hat\ell_1$ and $\hat\ell_2$ (left column) of income distribution and their difference scaled by the maximum absolute difference, $(\hat\ell_1-\hat\ell_2)/\|\hat\ell_1-\hat\ell_2\|_\infty$, (right column) corresponding to Greece ($X_1$)  and Finland ($X_2$) in a year of the span 2004--2019.}
\label{SuppMatfi:GRFILorenz5}
\end{figure}

\subsection{Comparing inequality between Greece and Portugal}

\begin{figure}[h]
\begin{center}
\includegraphics[width=12cm]{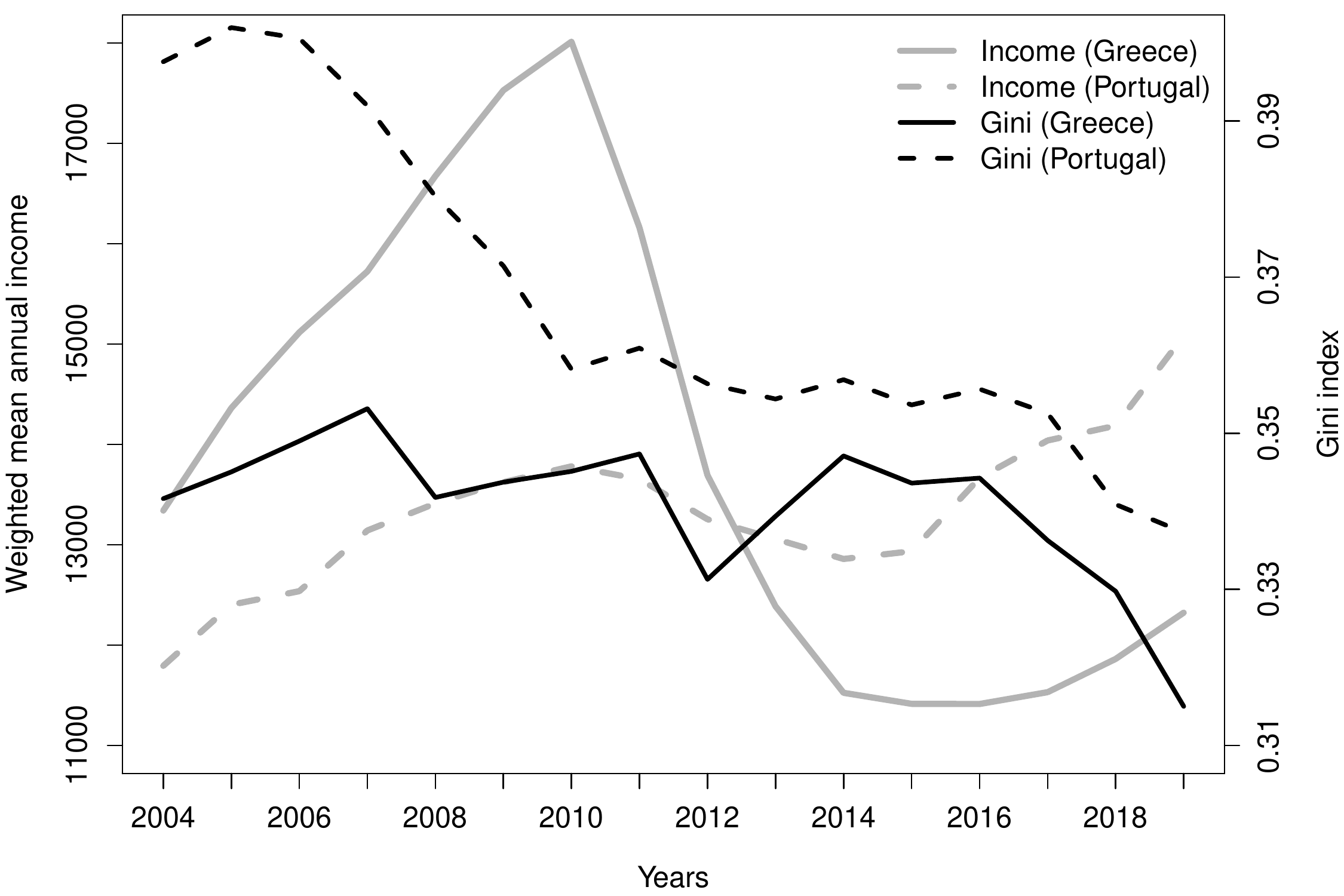}
\end{center}
\label{IncomeGiniELPT0419}
\caption{Weighted mean annual income and Gini index in Greece and Portugal from 2004 to 2019.}
\end{figure}

\begin{figure}[H]
\begin{center}
\includegraphics[width=13cm]{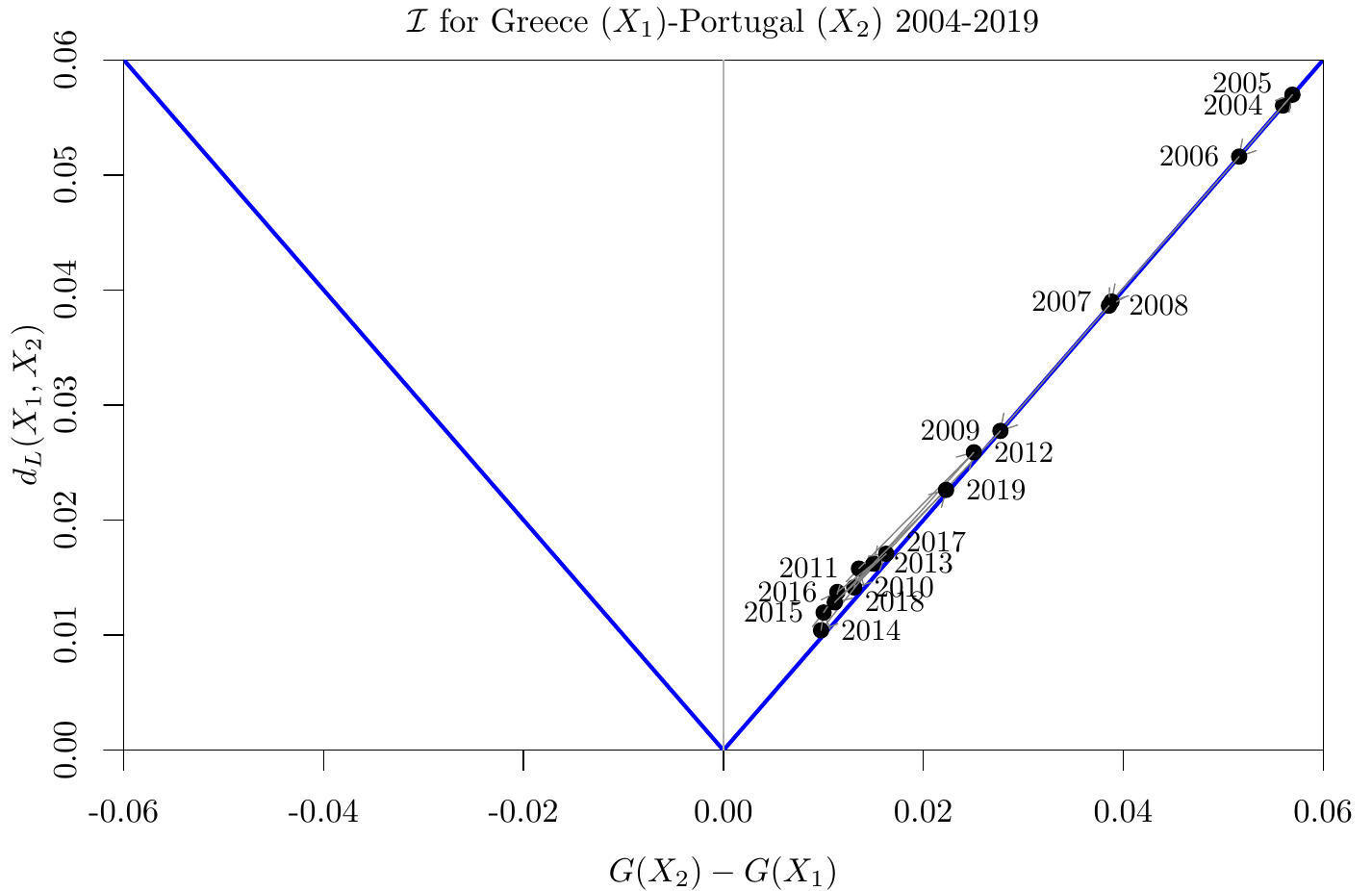}
\end{center}
\caption{Bidimensional inequality index $\mathcal I$ for income data from Greece ($X_1$) and Portugal ($X_2$).}
\label{fi:GreecePortugal0419I}
\end{figure}

\subsection{Relative inequality between Spain and Portugal} \label{Subsection.RealData.ESPT}

In Figure~\ref{SuppMatfi:ESPTLorenz}, on the left we display the Lorenz curves, \(\hat\ell_1\) and \(\hat\ell_2\), of income in Spain and Portugal, respectively, for a year between 2008 and 2019. Since the two Lorenz curves are close to each other, on the right we plot the difference of the two Lorenz curves, \(\hat\ell_1-\hat\ell_2\), scaled by \(\|\hat\ell_1-\hat\ell_2\|_\infty\).

\begin{figure}[H]
\begin{center}
\begin{tabular}{cc}
\includegraphics[width=0.47\textwidth]{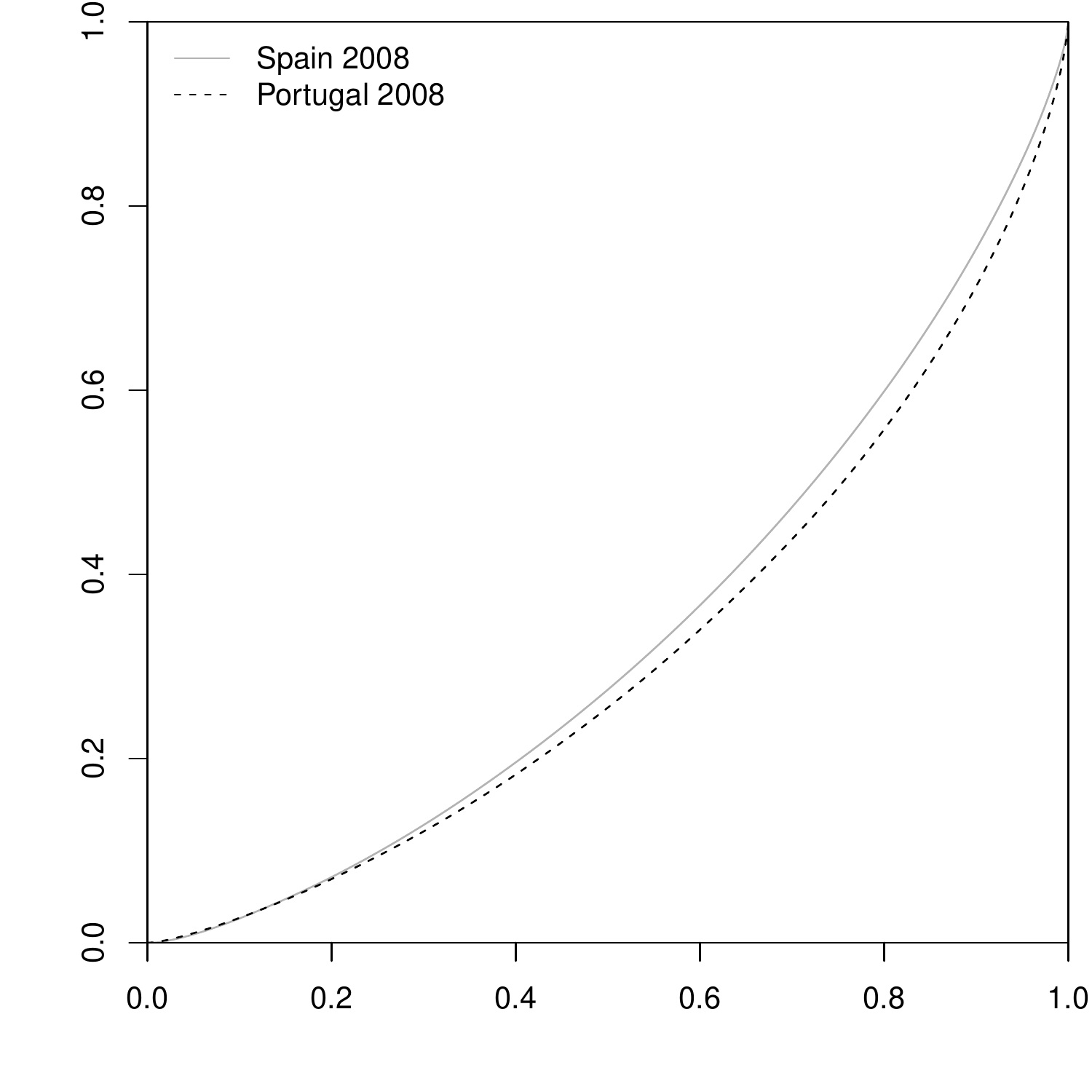} & \includegraphics[width=0.47\textwidth]{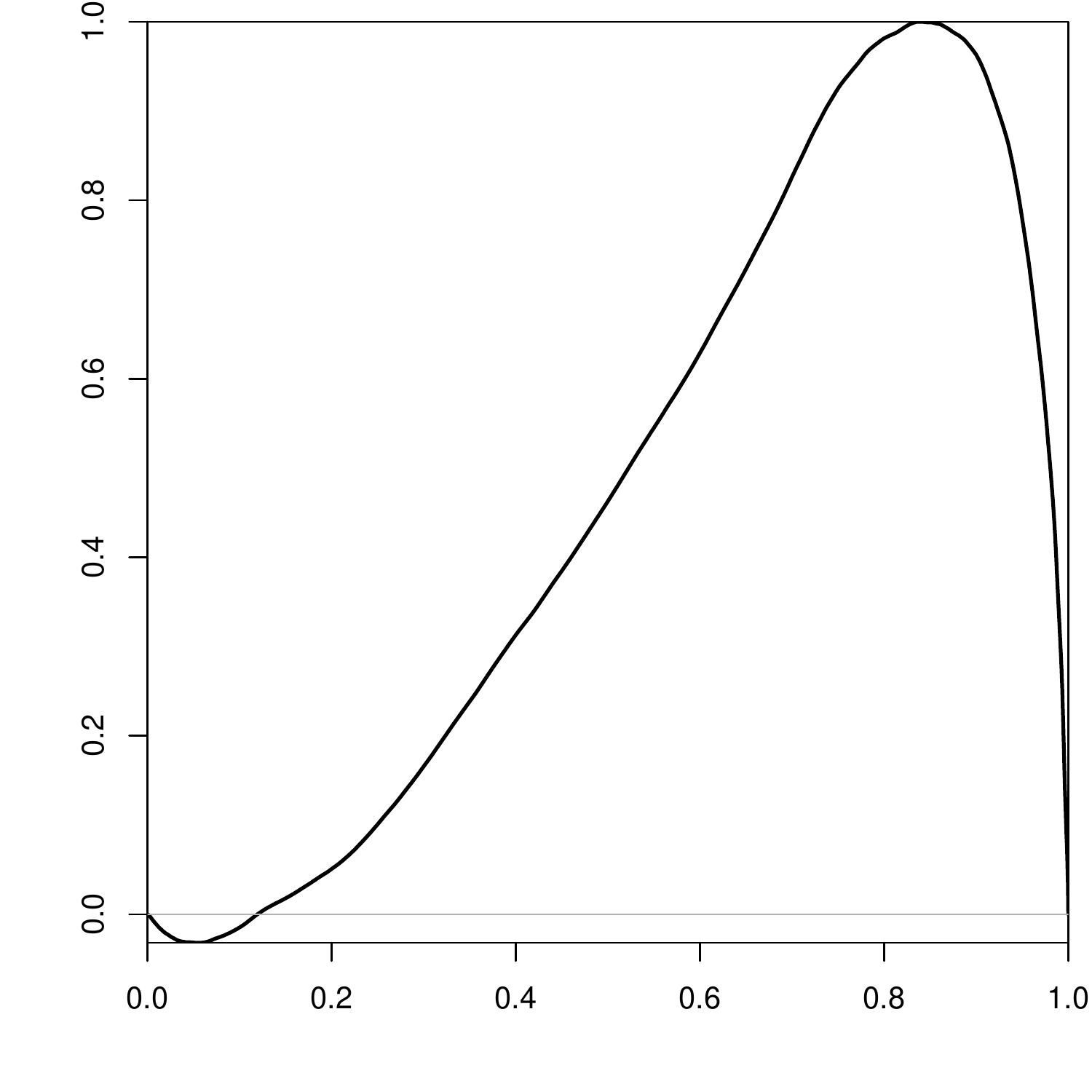}  \\
\includegraphics[width=0.47\textwidth]{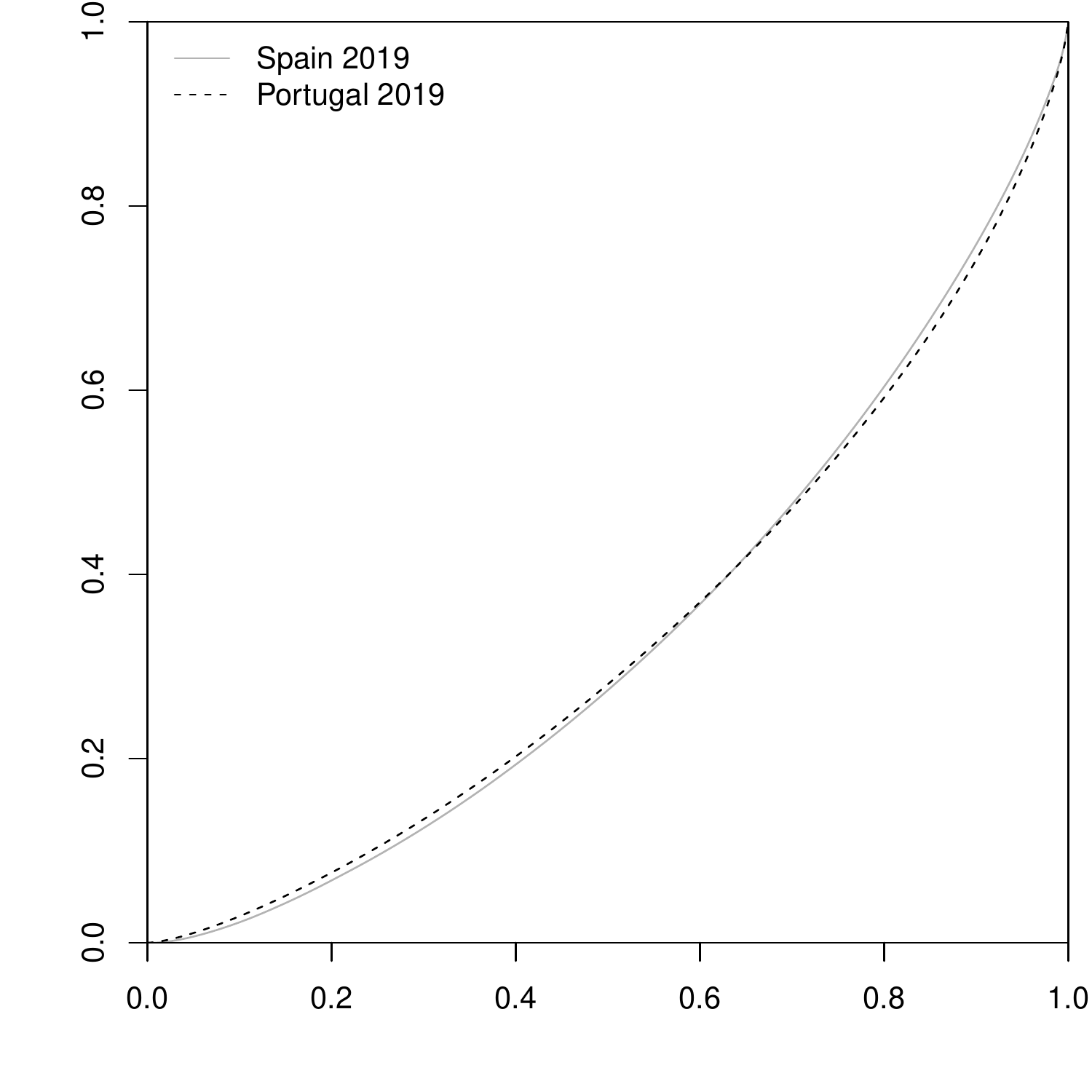} & \includegraphics[width=0.47\textwidth]{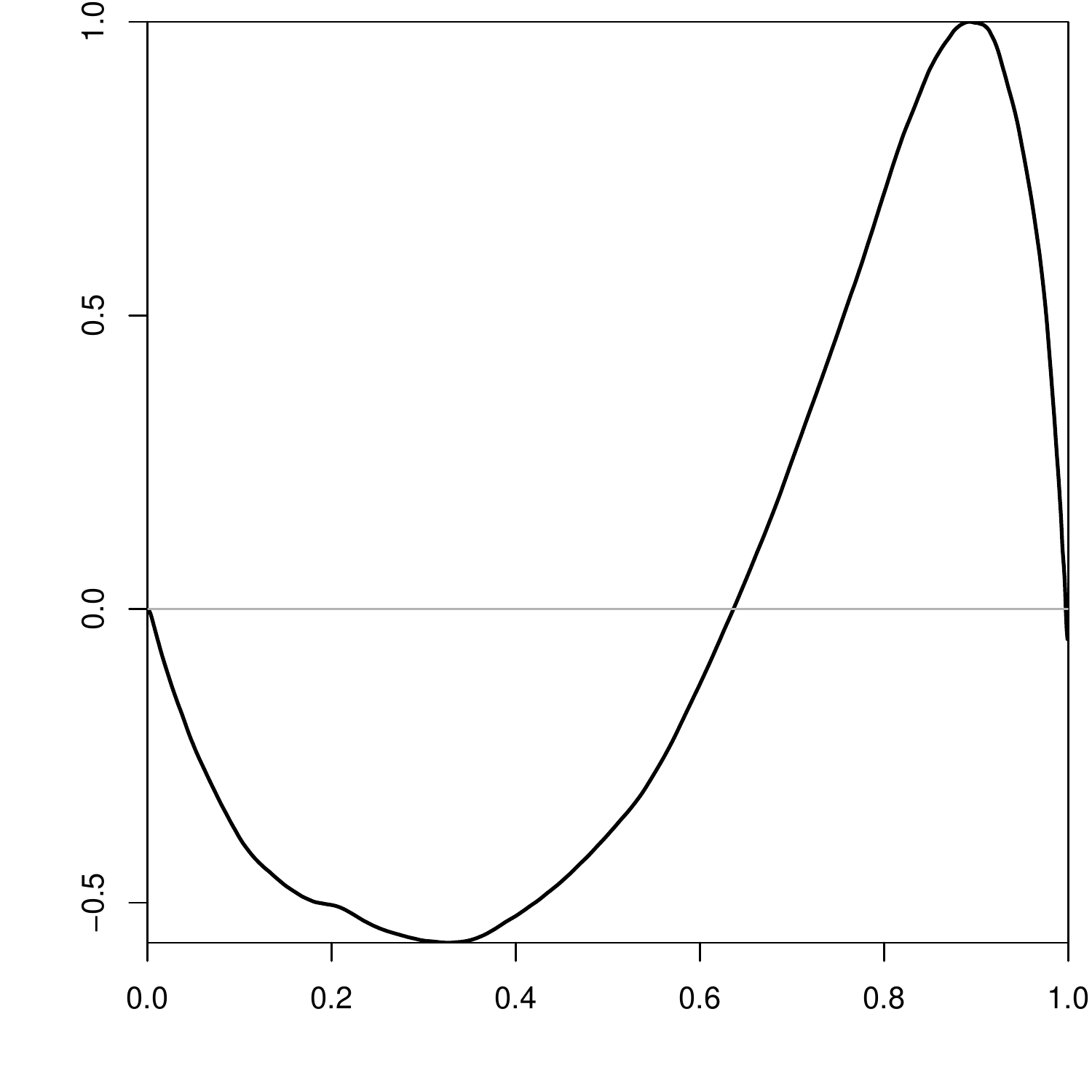}
\end{tabular}
\end{center}
\caption{The Lorenz curves $\hat\ell_1$ and $\hat\ell_2$ (left column) of income distribution and their difference scaled by the maximum absolute difference, $(\hat\ell_1-\hat\ell_2)/\|\hat\ell_1-\hat\ell_2\|_\infty$, (right column) corresponding to Spain ($X_1$)  and Portugal ($X_2$) in 2008 and 2019.}
\label{SuppMatfi:ESPTLorenz}
\end{figure}

\end{document}